\documentclass[11pt]{article}
\pdfoutput=1
\usepackage{amsmath,amssymb,epsfig,amsfonts,mathrsfs}
\usepackage{graphicx}
\usepackage{subfig}
\usepackage[usenames, dvipsnames]{color}
\usepackage{cite}
\usepackage{multirow}
\usepackage{hyperref}
\usepackage{verbatim} 
\usepackage{enumitem}
\usepackage{rotating}
\usepackage{xcolor}
\usepackage{multirow}
\usepackage{bm}
\usepackage[normalem]{ulem}
\usepackage{cleveref}
\usepackage{slashed}
\usepackage{float}
\usepackage{parskip}

\usepackage[fleqn,tbtags]{mathtools}

\usepackage{ytableau}

\usepackage{ulem}

\usepackage{stackengine, array}

 \usepackage{tikz} 

\usepackage{tikz,bm}
\usetikzlibrary{arrows}
\usetikzlibrary{patterns}
\usetikzlibrary{decorations.pathmorphing}
\usepackage{xcolor}
\usepackage{amsmath,amssymb,amsthm,mathdots}

\usetikzlibrary{positioning}
\usetikzlibrary{chains}
\usetikzlibrary{arrows, arrows.meta ,fit,decorations.pathreplacing}
\tikzstyle{every picture}+=[remember picture]
\tikzstyle{na} = [baseline]

\usetikzlibrary{arrows, decorations.markings, calc, fadings, decorations.pathreplacing, patterns, decorations.pathmorphing, positioning}
\tikzset{>={Latex[width=1.5mm,length=1.5mm]}}

\tikzset{flavor/.style={regular polygon,regular polygon sides=4,inner sep=2.5pt, draw}}
\tikzset{gauge/.style={circle, draw,inner sep=2.5pt}}
\tikzset{bd/.style={circle, draw=black, inner sep=0pt, fill=black, minimum size=2mm}}
\tikzset{sbd/.style={circle, draw=black, inner sep=0pt, fill=black, minimum size=1mm}}
\tikzset{wd/.style={circle, draw=black, inner sep=0pt, fill=white, minimum size=2mm}}
\tikzstyle{ligne}=[draw, very thick] 
\tikzstyle{ligne2}=[draw, line width=0.3mm] 
\tikzstyle{brane}=[draw] 
\tikzset{sevenb/.style={circle, draw,fill=white,inner sep=2.5pt}}
\tikzstyle{gridline}=[draw, black!20!] 
\tikzset{hasse/.style={circle, fill,inner sep=2pt}}
\tikzset{
  ->-/.style={
    decoration={
      markings,
      mark=at position 0.5 with {
        \pgftransformxshift{2.2pt}
        \arrow{Stealth}
      }
    },
    postaction={decorate}
  },
  -<-/.style={
    decoration={
      markings,
      mark=at position 0.5 with {
        \pgftransformxshift{-2.2pt}
        \arrowreversed{Stealth}
      }
    },
    postaction={decorate}
  }
}

\makeatletter

\usepackage{tikz}
\usetikzlibrary{arrows}
\usetikzlibrary{arrows.meta}
\usetikzlibrary{shapes.geometric,calc,arrows, positioning,shapes.misc,decorations.markings}
\tikzset{
  big arrow/.style={
    decoration={markings,mark=at position 1 with {\arrow[scale=2,#1]{>}}},
    postaction={decorate},
    shorten >=0.4pt},
  big arrow/.default=black}
  
\pgfdeclarelayer{edgelayer}
\pgfdeclarelayer{nodelayer}
\pgfsetlayers{edgelayer,nodelayer,main} 
\tikzstyle{none}=[inner sep=0pt] 

\tikzstyle{NodeCross}=[draw, shape=circle, cross out, inner sep=0pt, minimum size=6pt,line width=0.25mm]
\tikzstyle{Circle}=[draw, shape=circle, black,  fill=black, inner sep=0pt, minimum size=6pt]
\tikzstyle{Star}=[draw, shape=star, fill=black, star points=8, inner sep=0pt, minimum size=8pt]

\tikzstyle{DashedLine}=[-, densely dashed, line width=0.25mm]
\tikzstyle{DottedLine}=[-, dotted, line width=0.25mm]
\tikzstyle{ThickLine}=[-, line width=0.25mm]
\tikzstyle{ArrowLineRight}=[-, -{Stealth[scale=1.75]}, line width=0.1mm, scale=5]
\tikzstyle{RedLine}=[-, draw={rgb,255: red,191; green,0; blue,0}, fill=none, line width=0.25mm]
\tikzstyle{DottedRed}=[-, dotted, draw={rgb,255: red,191; green,0; blue,0}, fill=none, line width=0.25mm]
\tikzstyle{DashedLineThin}=[-, densely dashed, line width=0.125mm, fill=none, draw=black]
\tikzstyle{ArrowLineRed}=[-, -{Stealth[scale=1.75]}, draw={rgb,255: red,191; green,0; blue,0}, line width=0.1mm, scale=5]

\tikzset{gauge/.style={rounded rectangle, draw=black!100, thick, minimum size=5mm},  gaugeD/.style={rounded rectangle, draw=black!100,double,thick,minimum size=5mm},  empty/.style={rounded rectangle, draw=white!100, thick, minimum size=5mm}, flavor/.style={rectangle, draw=black!100, thick, minimum size=5mm},flavorD/.style={rectangle, draw=black!100, double,thick, minimum size=5mm}}

\newtheorem{Theorem}{Theorem}[section]

\newcommand{\be}{\begin{equation}}
\newcommand{\ee}{\end{equation}}
\newcommand{\ba}{\begin{aligned}}
\newcommand{\ea}{\end{aligned}}

\newcommand{\C}{\mathbb{C}}
\renewcommand{\P}{\mathbb{P}}

\newcommand{\cN}{\mathcal{N}}

\newcommand{\nn}{\nonumber}
\newcommand{\bea}{\begin{eqnarray}}
\newcommand{\eea}{\end{eqnarray}}

\newcommand{\Z}{{\mathbb Z}}

\def\unit{{1\kern-.65ex {\rm l}}}
\def\1{{1\kern-.65ex {\rm l}}}

\def\CF{{\cal F}}

\definecolor{myblue}{RGB}{30,70,170}
\definecolor{mygreen}{RGB}{20,110,70}
\definecolor{mypink}{RGB}{210,70,130}

\newcount\hour \newcount\minute
\hour=\time \divide \hour by 60
\minute=\time
\def\now{%
\ifnum \hour<13
  \ifnum \hour=0 \advance \hour by 12 \number\hour:\else \number\hour:\fi%
     \ifnum \minute<10 0\fi%
     \number\minute%
\ A.M.%
\else \advance \hour by -12 \number\hour:%
  \ifnum \minute<10 0\fi%
  \number\minute%
  \ P.M.%
\fi%
}

\makeatother

\def\bp{\begin{pmatrix}}
\def\ep{\end{pmatrix}}

\begin{document}

\baselineskip=18pt  
\numberwithin{equation}{section}  
\allowdisplaybreaks  

\thispagestyle{empty}

\vspace*{0.8cm} 
\begin{center}

{\huge $G_2$-Manifolds from 4d $\mathcal{N}=1$ Quivers}

\vspace*{1.5cm}
{Andreas P. Braun$\,^1$, 
Oscar Lewis$\,^2$, 
Matteo Sacchi$\,^3$, 
Sakura Sch\"afer-Nameki$\,^2$}\\
\bigskip
{\it $^1$ Department of Mathematical Sciences, Durham University,\\ Upper
Mountjoy Campus, Stockton Rd, Durham DH1 3LE, UK} \\
{\it $^2$ Mathematical Institute, University of Oxford, \\
Andrew-Wiles Building,  Woodstock Road, Oxford, OX2 6GG, UK}\\
{\it $^3$ Simons Center for Geometry and Physics, Stony Brook University, \\
Stony Brook, NY 11794, USA}

\vspace*{0.8cm}
\end{center}
\vspace*{.5cm}

\noindent
Inspired by quantum field-theoretic constructions of 4d $\mathcal{N}=1$  quiver gauge theories that flow to superconformal field theories (SCFTs), we construct 7d manifolds of $G_2$-holonomy, which geometrically engineer these quivers in M-theory. 
Field theoretically, the 4d quivers are obtained by 
flux torus compactifications of 6d $\mathcal{N}=(1,0)$ SCFTs. 
The 6d theory compactified on a circle gives rise to a 5d KK-theory, which has a geometric realization as M-theory on a non-compact elliptically fibered Calabi-Yau threefold.
Following the field-theoretical prescription, these local geometries are fibered over a circle to realize (topological) $G_2$-holonomy manifolds.
We carry this out concretely in the case of the rank 1 E-string theory and its 4d quivers and construct new families of $G_2$-holonomy spaces.

\newpage


\tableofcontents


\section{Introduction}

Geometric engineering of quantum field theories using non-compact geometries with special holonomy is a fruitful enterprise, for physics and geometry alike. On the one hand, it provides a systematic approach to the study and classification of strongly coupled higher-dimensional renormalization group fixed points, in particular 5d and 6d superconformal field theories (SCFTs) \cite{Seiberg:1996bd,Seiberg:1996vs, DelZotto:2014hpa, Bhardwaj:2015xxa, Heckman:2013pva,Heckman:2015bfa, Xie:2017pfl, Jefferson:2018irk, Bhardwaj:2018yhy, Bhardwaj:2018vuu, Apruzzi:2019vpe, Apruzzi:2019opn,Apruzzi:2019enx, Apruzzi:2019kgb}. 
For this application the powerful results on canonical singularities of Calabi-Yau threefolds play a crucial role. 
On the other hand, it can give insights into geometries that are otherwise much more difficult to construct, namely $G_2$ and Spin$(7)$ holonomy manifolds.  
Due to the absence of a Yau-theorem for such geometries, examples 
remain sporadic and reliant on specific constructions. 

In this paper we will show how insights from physics can motivate the construction of vast classes of $G_2$-holonomy manifolds. Such physics motivated approaches have been insightful in the past, see for example \cite{Acharya:2001gy, Halverson:2014tya, Braun:2017ryx, Braun:2016igl, Braun:2017uku,  Braun:2017csz, Braun:2018joh, Braun:2018vhk, Braun:2018fdp, Acharya:2018nbo, Barbosa:2019bgh, Acharya:2019svi, Acharya:2020vmg, Acharya:2021rvh,  Acharya:2023bth, Braun:2023fqa}. 
{Our approach is motivated by} recent constructions in field theory of 4d $\cN=1$ SCFTs from 6d $\cN= (1,0)$ theories on a circle to 5d KK-theories, and then subsequently to 4d including flux tubes and duality domain walls, {studied e.g.~in~\cite{Gaiotto:2015usa,Ohmori:2015pua,Ohmori:2015pia,Razamat:2016dpl,Bah:2017gph,Kim:2017toz,Kim:2018bpg,Kim:2018lfo,Razamat:2018gro,Zafrir:2018hkr,Ohmori:2018ona,Chen:2019njf, Pasquetti:2019hxf,Baume:2021qho, Sabag:2022hyw,Giacomelli:2023qyc}, see in particular \cite{Razamat:2022gpm} for a review. }
The geometry underlying this is the elliptically fibered Calabi-Yau threefold that constructs the 6d theory in F-theory. 
We will fiber this over its extended K\"ahler cone {(field theoretically, we fiber it over the extended Coulomb branch of the 5d KK-theory)}, and apply automorphisms of the fiber, which then result in a geometry that is overall a 7-manifold. 
As the field theory predicts 4d $\cN=1$ supersymmetry, we conjecture that the resulting geometry is a (topological\footnote{By topological $G_2$-manifold, we mean a 7-dimensional topological space described at the level of the singularity structure and curve volumes. We do not construct torsion-free $G_2$-holonomy metrics on any of the spaces constructed in this work.}) $G_2$ holonomy manifold.

We build on the previous construction in  \cite{Braun:2023fqa}, where we geometrized S-duality domain walls, to conjecture new topological $G_2$-manifolds that geometrize the 4d $\mathcal{N}=1$ quiver theories arising from  flux tori compactifications of 6d $\mathcal{N}=(1,0)$ theories. The flux tori compactifications of such theories naturally goes via 5d.

\paragraph{E-string flux torus quivers.}
Our approach can in principle be applied to any 6d $\cN=(1,0)$ SCFT, but we focus here on the rank-1 E-string and its flux torus compactifications \cite{Kim:2017toz}. The 4d quivers corresponding to torus compactifications of the rank-1 E-string with flux $\mathcal{F}=(-2^{2n},0^{8-2n})$\footnote{The flux is expressed in terms of the Cartan of the $SO(16)$ subgroup of $E_8$, as we explain more in details in the main text.} for the $E_8$ global symmetry are shown in Figure~\ref{fig:Torusq1intro} for the cases $n=4,3,2,1$. 

\begin{figure}
    \centering
       \begin{tikzpicture}[baseline,scale=0.8]
\tikzstyle{every node}=[font=\scriptsize]
       \node[draw, circle,thick] (p1) at (0,-2.2) {\fontsize{12pt}{12pt}\selectfont $2$};
       \node[draw, circle,thick] (p2) at (3,0) {\fontsize{12pt}{12pt}\selectfont $2$};
       \node[draw, circle,thick] (p3) at (6,-2.2) {\fontsize{12pt}{12pt}\selectfont $2$};
       \node[draw, rectangle,thick] (p4) at (3,-2.2) {\fontsize{12pt}{12pt}\selectfont $2n$};
       \node[draw, circle,thick] (p5) at (3,-4.4) {\fontsize{12pt}{12pt}\selectfont $2$};
       \node[draw, rectangle,thick] (p6) at (3,1.6) {\fontsize{12pt}{12pt}\selectfont $8-2n$};
       \draw[-,thick] (p1) to  node[rotate=45]{\large $\times$}    (p2);
       \draw[-,thick] (p2) to  node[rotate=45]{\large $\times$}    (p3);
       \draw[-,thick] (p3) to  node[rotate=45]{\large $\times$}    (p5);
       \draw[-,thick] (p5) to  node[rotate=45]{\large $\times$}    (p1);
       \draw[->-,thick] (p1) to      (p4);
       \draw[->-,thick] (p3) to      (p4);
       \draw[-<-,thick] (p2) to      (p4);
       \draw[-<-,thick] (p5) to      (p4);
       \draw[-<-,thick] (p1) to      (0,1.6);
       \draw[-,thick] (0,1.6) to      (p6);
       \draw[->-,thick] (p3) to      (6,1.6);
       \draw[-,thick] (6,1.6) to      (p6);
       \end{tikzpicture}
    \caption{The 4d $\mathcal{N}=1$ quiver theory corresponding to the compactification of the 6d E-string theory on a torus with flux $\mathcal{F}=(-2^{2n},0^{8-2n})$. The quiver is supplemented by cubic and quartic superpotential interactions corresponding to closed loops in the quiver.}
    \label{fig:Torusq1intro}
\end{figure}
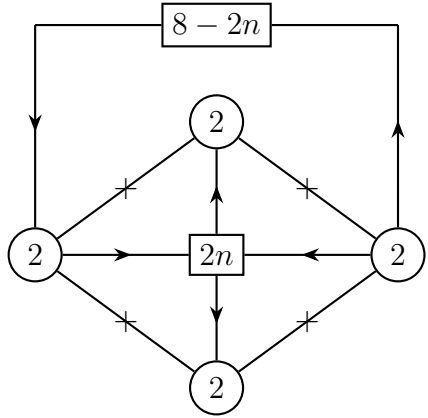

We apply the following logic: first define maps that geometrically capture the action of S-duality domain walls and gluing of quivers, which we call $\mu_s$ and $\mu_{t,n}$, respectively. 
These act as automorphisms on the lattice of curves in the dP$_9$. To $\mu_s$ and $\mu_{t,n}$ we associate so called \textbf{local building blocks} $\mu_s^{\varepsilon}$ and $\mu_{t,n}^{\varepsilon}$, defined as the non-trivial fibration of the threefold over an interval of length $\varepsilon$, with the respective maps applied at the midpoint of the interval. These local building blocks capture the local geometry associated to the domain walls and locations of quiver gluing.  

Next, we concatenate these local building blocks in the same order as one would glue together flux tubes in field theory to get the desired quivers. This produces a non-trivial fibration of the threefold over an interval. Finally, we glue the threefolds at the boundaries of the interval together by applying a suitable automorphism, which then becomes the monodromy map of the threefold fibered over a circle. This action captures the field theory process of closing the flux tubes to build a flux torus. 

\paragraph{Automorphisms.}
Automorphisms of partially resolved singular geometries are not particularly well-studied. Here we will substantiate our claims in two ways: one is to explicitly exhibit the automorphisms we require for a particular choice of complex structure, where the singular fiber is not $II^*$, which realizes the $E_8$ flavor symmetry, but {deformed} to an $I_8$, which corresponds to a manifest $SU(8)$ flavor symmetry. In this {deformed} complex structure we can prove the existence of the required automorphisms due to an enhanced Mordell-Weil group of the elliptic fibration, together with Picard-Lefschetz transformations. The other support comes from the field theory, where we can identify the duality transformation in the subspace where there is a manifest $SU(8)$ flavor symmetry.

\paragraph{Checks.}
There are multiple non-trivial checks of the proposed geometries we can perform beyond reading off the quiver from the volumes of the three cycles in the 7-manifold. First, by tracking the shifts of the $E_8$ roots of the 6d flavor symmetry by multiples of the elliptic fiber $E$ we can determine the sublattice of $E_8$ invariant under the monodromy map. For the above examples the invariant sublattices computed from the geometry correspond to the Lie algebras
\be
\ba
n=4:\quad &\mathfrak{e}_7\oplus \mathfrak{u}(1)\,,\\
n=3:\quad &\mathfrak{e}_6\oplus \mathfrak{su}(2)\oplus \mathfrak{u}(1)\,,\\
n=2:\quad &\mathfrak{so}(14)\oplus \mathfrak{u}(1)\,,\\
n=1:\quad &\mathfrak{e}_7\oplus \mathfrak{u}(1)\,.
\ea
\ee
which matches exactly with the expected IR enhanced non-abelian flavor symmetry of the quiver theories \cite{Kim:2017toz}. Second, it was pointed out in \cite{Sabag:2022hyw} that the mass parameters of the 5d gauge theory get shifted by multiples of the KK mass $m_{\text{KK}}$ when we build a flux tube with multiple domain walls. This causes a breaking of the $\mathfrak{u}(1)$ KK symmetry to a finite subgroup when we identify the two ends of the tube to build the torus model. We will show how the shifts of the $E_8$ roots by the elliptic fiber $E$ in our geometric analysis are in agreement with the shifts of the masses by $m_{\text{KK}}$ in the field theory. Finally, different flux tori quiver theories are actually IR dual to each other if their fluxes can be related by a Weyl transformation of the 6d flavor symmetry. We will recover these field theory dualities from the geometric perspective, by showing the equivalence of the geometries.

Once we have constructed the topological $G_2$ geometries for the quivers in Figure~\ref{fig:Torusq1intro}, we give a general recipe for building topological $G_2$-manifolds starting from a more general flux in the $E_8$ root lattice.\footnote{In this work we focus only on $E_8$ fluxes that are not center fluxes. By this we mean a flux breaking $E_8\to H\times U(1)^k$ where $H$ is non-abelian, which involves only a flux in the abelian part and not for the center of $H$. We leave an investigation of the geometric realization of 4d $\mathcal{N}=1$ quivers corresponding to more general fluxes for future work.} We show that the resulting map between fluxes and monodromies is a group homomorphism, so that addition of fluxes can be expressed as concatenation of domain walls. Not only do we learn the geometry of the 7-manifold associated to a 4d $\mathcal{N}=1$ quiver theory from this construction, but can easily calculate features of the 4d theory such as enhanced non-abelian flavor symmetry. These are not manifest in the quiver and require superconformal index calculations to see from the field theory side, even for the simplest of examples.

\paragraph{Plan of the paper.} The paper is organized as follows: in Section~\ref{sec: Field Theory sec2} we review the relevant field theory construction of flux tubes and tori compactifications of the rank-1 E-string theory with non-trivial $E_8$ flux. In particular, we review the S-duality domain walls and gluing procedures performed on the quiver to build the 4d $\mathcal{N}=1$ theories for flux tori compactifications. 
In Section~\ref{sec: E-string geometry} we review the geometry of the Calabi-Yau threefold given by the total space of the canonical bundle over dP$_9$, which realizes the rank 1 E-string from F-theory, including the K\"ahler form and flavor symmetries. 
In Section~\ref{sec:constructionofXS} we discuss a specific geometric realization of the dP$_9$, its associated threefold, and its corresponding flopped phases. This corresponds to a specific choice of complex structure such that the dP$_9$ surface can be described by a Weierstrass model with $I_8$ fiber, and will be used throughout our construction. We then use the generalized box graphs of~\cite{Sabag:2022hyw} to construct non-trivial automorphisms of the dP$_9$ which will be important for the $G_2$ construction.

With the dictionary between field theory and geometry established, in Section~\ref{sec: Geometric Realization} we begin by concatenating the local building blocks associated to these automorphisms in the analogous way as the quivers are glued in field theory to reproduce quivers for the cases $n=4,3,2,1$ in Figure~\ref{fig:Torusq1intro}. This results in a fibration of the threefold over an interval, which we close by applying a monodromy map. We are able to directly identify the singularities and 3-cycles inside the topological $G_2$ geometry that give rise to the corresponding quivers. 
We then provide a more algebraic description of the setup that generalizes to more general $E_8$ fluxes, and prove that the map we use to glue the threefolds identifies isomorphic geometries. On the level of homology, the monodromy maps are Eichler-Siegel transformations~\cite{heckman2001moduli}, which allows us to show that the identification between fluxes and monodromies is a homomorphism.
To conclude, we lay down the general framework for systematically building topological $G_2$-manifolds from an arbitrary flux of $E_8$.

\section{Quivers from Flux Tubes}\label{sec: Field Theory sec2}

\subsection{The Basic Flux Tube Theory}

In this section we review the field theory constructions of the 4d $\mathcal{N}=1$ theories arising from compactifications of the 6d rank-1 E-string theory on tubes and tori \cite{Kim:2017toz}. When considering the compactification of a generic six-dimensional $\mathcal{N}=(1,0)$ SCFT on a Riemann surface preserving 4d $\mathcal{N}=1$, we have the possibility to turn on a flux through the surface for the global symmetry $G$ of the 6d theory. This has to satisfy Dirac's quantization condition, which forces it to be an element of the coweight lattice of $G$
\begin{equation}
    \mathcal{F}\in\Lambda_{G}^\vee\,.
\end{equation}
For the E-string theory the global symmetry is $E_8$. It is convenient to express the flux vector in terms of the Cartan of the $\mathfrak{so}(16)\subset \mathfrak{e}_8$ subgroup, in a basis where the vector representation ${\bf 16}_{\mathfrak{so}(16)}$ has weights
\be
    \left(\pm 1,0,0,0,0,0,0,0\right)\,+\,\text{permutations}\,.
\ee
The lattice of fluxes can then be written as
\begin{equation} \label{eq: field theory fluxes}
    \Lambda_{E_8}^\vee\cong\Lambda_{E_8}\cong\left\{\mathcal{F}\in \mathbb{Z}^8 :\sum_{i=1}^{8} \mathcal{F}_i \in 2\mathbb{Z}\right\}\cup\left\{\mathcal{F}\in \left(\mathbb{Z}+\frac{1}{2}\right)^8 :\sum_{i=1}^{8} \mathcal{F}_i \in 2\mathbb{Z}\right\}\,.
\end{equation}

One of the basic building blocks to study compactifications of the 6d E-string theory on Riemann surfaces is the theory obtained by compactification on a tube with flux
\be\label{eq:basicflux}
    \mathcal{F}=\left(-\frac{1}{2},-\frac{1}{2},-\frac{1}{2},-\frac{1}{2},-\frac{1}{2},-\frac{1}{2},-\frac{1}{2},-\frac{1}{2}\right)\,.
\ee
This flux tube can be understood as one of the roots of $\mathfrak{e}_8$ that is a weight of the spinor representation ${\bf 128}_{\mathfrak{so}(16)}$ arising from decomposing the adjoint representation of $\mathfrak{e}_8$ with respect to the $\mathfrak{so}(16)$ subgroup
\be
    {\bf 248}_{\mathfrak{e}_8}\to{\bf 120}_{\mathfrak{so}(16)}\oplus{\bf 128}_{\mathfrak{so}(16)}\,,
\ee
where ${\bf 120}_{\mathfrak{so}(16)}$ is instead the adjoint representation of $\mathfrak{so}(16)$. Hence, we see that the flux of the basic tube theory corresponds to one of the Weyl transformations of $\mathfrak{e}_8$ that is not a Weyl transformation of $\mathfrak{so}(16)$.

The 4d $\mathcal{N}=1$ theory corrresponding to this flux tube can be summarized by the quiver in Figure \ref{fig:Flux tube theory}. The theory possesses the global symmetry 
\be\label{eq:globalsymm}
    \mathfrak{su}(2)\oplus\mathfrak{su}(2) \oplus\mathfrak{su}(8)\oplus\mathfrak{u}(1)\oplus\mathfrak{u}(1)\,,
\ee
whose non-abelian part is represented by square boxes in the quiver. This is due to the presence of the following superpotential interaction among the chiral fields:
\be
\mathcal{W}=bQ^2+LQR\,,
\ee
where all flavor indices are understood.

\begin{figure}[t]
    \centering
    \begin{tikzpicture}[baseline,scale=1]
\tikzstyle{every node}=[font=\scriptsize]
       \node[draw, rectangle] (p1) at (0,0) {\fontsize{12pt}{12pt}\selectfont $\,\,2\,\,$};
       \node[draw, rectangle] (p2) at (3,0) {\fontsize{12pt}{12pt}\selectfont $\,\,2\,\,$};
       \node[draw, rectangle] (p3) at (1.5,-2.2) {\fontsize{12pt}{12pt}\selectfont $\,\,8\,\,$};
       \draw[-,thick] (p1) to  node[rotate=0]{\Large $\times$}    (p2);
       \draw[->-,thick] (p1) to      (p3);
       \draw[-<-,thick] (p2) to      (p3);
       \node[left] at (0.6,-1.2)  {\fontsize{12pt}{12pt}\selectfont $L$};
       \node[right] at (2.3,-1.2)  {\fontsize{12pt}{12pt}\selectfont $R$};
       \node[above] at (0.8,0)  {\fontsize{12pt}{12pt}\selectfont $Q$};
       \node[above] at (1.5,0.2)  {\fontsize{12pt}{12pt}\selectfont $b$};
       \end{tikzpicture}
    \caption{The basic flux tube theory. We use a quiver notation where boxes denote flavor symmetries and lines connecting them denote chiral multiplets transforming under such symmetries. In particular, $Q$ is in the bifundamental of $\mathfrak{su}(2)\oplus \mathfrak{su}(2)$, $L$ is in the fundamental of the left $\mathfrak{su}(2)$ and of $\mathfrak{su}(8)$, $R$ is in the fundamental of the right $\mathfrak{su}(2)$ and in the antifundamental of $\mathfrak{su}(8)$, and $b$ is a singlet which is denoted by a cross in the quiver.}
    \label{fig:Flux tube theory}
\end{figure}
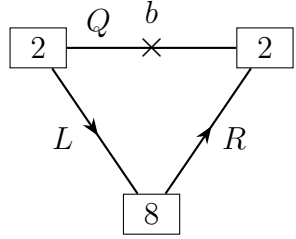

As first explained in \cite{Gaiotto:2015una,Kim:2017toz} (see also \cite{Sabag:2022hyw,Braun:2023fqa}), by viewing the tube as a circle times an interval, this theory can be understood as a domain wall interpolating between two copies of the 5d $\mathcal{N}=1$ gauge theory arising from circle reduction of the 6d E-string theory with an holonomy breaking $\mathfrak{e}_8\to \mathfrak{so}(16)$ \cite{Ganor:1996pc}. The latter is given by an $\mathfrak{su}(2)$ gauge theory with 8 fundamental hypermultiplets. Thanks to the fact that the fundamental representation of $\mathfrak{su}(2)$ is pseudo-real, we can think of the hypers as 16 half-hypermultiplets forming the vector representation of $\mathfrak{so}(16)$. 

The Intriligator-Morrison-Seiberg (IMS) prepotential \cite{Intriligator:1997pq} of this 5d $\mathcal{N}=1$ gauge theory is given by
\be\label{eq:IMSprepotEstring}
\CF_{\text{IMS}}=\frac{m_{\text{KK}}}{2} \phi^2 + \frac{4}{3}\phi^3 -\frac{1}{12}\sum_{i=1}^{8} \sum_{\pm} |\pm\phi +m_i|^3\,,
\ee
where $m_{\text{KK}}=\frac{1}{R_5}>0$ is the scale associated with the circle of the sixth dimension, $\phi$ is the Coulomb branch parameter and $m_i$ are the masses of the hypers in the Cartan of the $\mathfrak{so}(16)$ flavor symmetry. The KK mass measures the mass $m_\lambda$ of the instanton particle, which is related to the inverse gauge coupling squared. For example, in a phase where $\phi=0$ and $m_i>0$ we have
\be
m_\lambda=\frac{1}{2g_{\text{eff}}^2}=\frac{m_{\text{KK}}}{2}-\frac{1}{2}\sum_{i=1}^8m_i\,.
\ee
The domain wall implements the following transformations of the gauge theory parameters, which leaves the prepotential \eqref{eq:IMSprepotEstring} invariant (up to shifts that are constant in $\phi$):
\be\label{eq:DWtransf}
\ba
m_\lambda&\to-m_\lambda\,,\\
\phi&\to\phi+\frac{m_\lambda}{2}\,,\\
m_i&\to m_i+\frac{m_\lambda}{2}\,,\qquad i=1,\cdots,8\,.
\ea
\ee
This is the aforementioned Weyl transformation of $\mathfrak{e}_8$ that is not a Weyl transformation of $\mathfrak{so}(16)$. 

The origin of the various fields in the quiver of Figure \ref{fig:Flux tube theory} is of two types. First, one has to study boundary conditions of the 5d theory at the two extrema of the tube. These essentially define the punctures of the tube. The standard choice is to give Dirichlet boundary conditions to the gauge fields, thus freezing the $\mathfrak{su}(2)$ gauge symmetry to a global symmetry. This is the origin of the two $\mathfrak{su}(2)$ symmetries of the quiver theory, one for each copy of the 5d $\mathfrak{su}(2)$ gauge theory on the two sides of the wall. Moreover, each of the 8 hypermultiplets is split into a pair of 4d $\mathcal{N}=1$ chirals in conjugate representations and for each pair we give Neumann boundary conditions to one and Dirichlet to the other. The fields that are given Neumann boundary conditions survive as dynamical fields and give the chirals $L$, $R$ in the quiver theory. This explains the rest of the global symmetry \eqref{eq:globalsymm}. Specifically, the flux \eqref{eq:basicflux} would break the global symmetry from $\mathfrak{e}_8$ to $\mathfrak{e}_7\oplus \mathfrak{u}(1)$, since this is the subgroup of $\mathfrak{e}_8$ whose roots are orthogonal to the flux vector. However, the presence of the punctures and in particular the choice of boundary conditions further break $\mathfrak{e}_7\to \mathfrak{su}(8)$. The remaining $\mathfrak{u}(1)$ symmetry in \eqref{eq:globalsymm} has been interpreted in \cite{Hwang:2021xyw,Sabag:2022hyw} as the KK symmetry coming from the 6d to 5d circle reduction.

The remaining 4d fields, specifically the $\mathfrak{su}(2)\oplus \mathfrak{su}(2)$ bifundamental $Q$ and the singlet $b$, instead correspond to degrees of freedom that live intrinsically on the domain wall itself. Their presence was argued first in \cite{Gaiotto:2015una} by computations of the instanton partition function, and later in \cite{Kim:2017toz} by checking that the theory in Figure \ref{fig:Flux tube theory} passes several non-trivial and independent tests for it to be the tube compactification with flux of the 6d E-string theory.

\subsection{Gluing Flux-Tubes}

Given the basic flux tube theory, we can concatenate various copies of it with suitable gluing prescriptions to realize tubes of more general fluxes, and eventually identify the two ends to build a flux torus. There exist two main types of gluing operations. More general gluings are achieved by considering mixtures of these two, see e.g.~\cite{Kim:2017toz,Pasquetti:2019hxf} for more details.

\paragraph{$\Phi$-gluing.} Given two copies of the flux tube theory, we can glue them together by introducing an extra $\mathfrak{su}(2)\oplus \mathfrak{su}(8)$ bifundamental chiral $\Phi$ with the interaction
\begin{equation}\label{eq:Mixgluing}
    \delta\mathcal{W}=\Phi\left(R^{(1)}-L^{(2)}\right)\,,
\end{equation}
where $L^{(1)}$, $R^{(1)}$ are the fields of the first tube theory and $L^{(2)}$, $R^{(2)}$ those of the second. This deformation identifies the right $\mathfrak{su}(2)$ symmetry of the first tube with the left $\mathfrak{su}(2)$ of the second tube, as well as their $\mathfrak{su}(8)$ symmetries. Moreover, as part of the gluing prescription, the diagonal $\mathfrak{su}(2)$ symmetry is gauged.

The interaction \eqref{eq:Phigluing} is a mass deformation for $\Phi$ and one combination of $R^{(1)}$, $L^{(2)}$. The surviving combination can be determined from the equation of motion of $\Phi$, which sets
\begin{equation}
    R^{(1)}=L^{(2)}\,.
\end{equation}
This corresponds to a single chiral field in the bifundamental of $\mathfrak{su}(2)\oplus \mathfrak{su}(8)$. The gluing we have just described is usually called $\Phi$-gluing and is summarized in terms of the quiver in Figure \ref{fig:Phi gluing fig}.

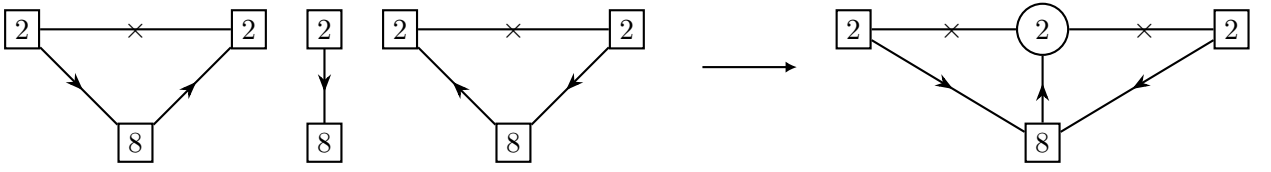
\begin{figure}[t]
    \centering
    \begin{tikzpicture}
    \begin{scope}[shift={(0,0)}]
    \node (p1) at (0,0) [rectangle,draw,thick] {$2$};
    \node (p2) at (3,0) [rectangle,draw,thick] {$2$};
    \node (p3) at (1.5,-1.5) [rectangle,draw,thick] {$8$};
    \draw[black, thick] (p1) -- (p2);
    \draw[black, thick,->-] (p1) -- (p3);
    \draw[black, thick,-<-] (p2) -- (p3);
    \node[below, yshift=0.25cm] at (1.5,0) {$\times$};
    \end{scope}
    \begin{scope}[shift={(4,0)}]
    \node at (0,0) [rectangle,draw,thick] {$2$};
    \node at (0,-1.5) [rectangle,draw,thick] {$8$};
    \draw[black, thick,->-] (0,-0.25) -- (0,-1.25);
    \end{scope}
    \begin{scope}[shift={(5,0)}]
    \node (p1) at (0,0) [rectangle,draw,thick] {$2$};
    \node (p2) at (3,0) [rectangle,draw,thick] {$2$};
    \node (p3) at (1.5,-1.5) [rectangle,draw,thick] {$8$};
    \draw[black, thick] (p1) -- (p2);
    \draw[black, thick,-<-] (p1) -- (p3);
    \draw[black, thick,->-] (p2) -- (p3);
    \node[below, yshift=0.25cm] at (1.5,0) {$\times$};
    \end{scope}
    \begin{scope}[shift={(9,0)}]
    \draw[black, thick,->] (0,-0.5) -- (1.25,-0.5);
    \end{scope}
    \begin{scope}[shift={(11,0)}]
    \node (p1) at (0,0) [rectangle,draw,thick] {$2$};
    \node (p3) at (5,0) [rectangle,draw,thick] {$2$};
    \node (p2) at (2.5,0) [circle,draw,thick] {$2$};
    \node (p4) at (2.5,-1.5) [rectangle,draw,thick] {$8$};
    \draw[black, thick] (p1) -- (p2);
    \draw[black, thick] (p2) -- (p3);
    \draw[black, thick,->-] (p1) -- (p4);
    \draw[black, thick,-<-] (p2) -- (p4);
    \draw[black, thick,->-] (p3) -- (p4);
    \node[below, yshift=0.25cm] at (1.3,0) {$\times$};
    \node[below, yshift=0.25cm] at (3.85,0) {$\times$};
    \end{scope}
    \end{tikzpicture}
    \caption{The $\Phi$-gluing of two copies of the basic flux tube.}
    \label{fig:Phi gluing fig}
\end{figure}

Note that this type of gluing identifies the symmetries of the two tube theories up to the Weyl transformation of $\mathfrak{so}(16)$ that corresponds to conjugating all of its Cartan elements. In fact, the gluing can also be understood as implementing a symmetry transformation of the parameters of the 5d $\mathcal{N}=1$ gauge theory similarly to the domain wall transformation \eqref{eq:DWtransf}. In this case, we have the $\mathfrak{so}(16)$ Weyl transformation
\be
\ba
\phi&\to\phi\,,\\
m_\lambda&\to m_\lambda+\sum_{i=1}^8m_i\,,\\
m_i&\to -m_i\,,\qquad i=1,\cdots,8\,.
\ea
\ee
Moreover, the result of the $\Phi$-gluing can itself be understood as a compactification of the 6d E-string theory on a cylinder, but this time with a flux that is obtained by summing those of the two basic flux tubes that were glued
\be
\ba
    \mathcal{F}&=\left(-\frac{1}{2},-\frac{1}{2},-\frac{1}{2},-\frac{1}{2},-\frac{1}{2},-\frac{1}{2},-\frac{1}{2},-\frac{1}{2}\right)+\left(-\frac{1}{2},-\frac{1}{2},-\frac{1}{2},-\frac{1}{2},-\frac{1}{2},-\frac{1}{2},-\frac{1}{2},-\frac{1}{2}\right)\nn\\
    &=(-1,-1,-1,-1,-1,-1,-1,-1)\,.
    \ea
\ee

\paragraph{$S$-gluing.} Alternatively, we can glue the two tubes without introducing any extra field, but identifying their symmetries via the superpotential
\begin{equation}\label{eq:Sgluing}
    \delta\mathcal{W}=R^{(1)}L^{(2)}\,.
\end{equation}
In this case, all the fields are massive and so there is no chiral left connecting the gauged $\mathfrak{su}(2)$ node with the $\mathfrak{su}(8)$ flavor node. This gluing is usually called $S$-gluing and is summarized in terms of the quiver in Figure \ref{fig:S gluing fig}.

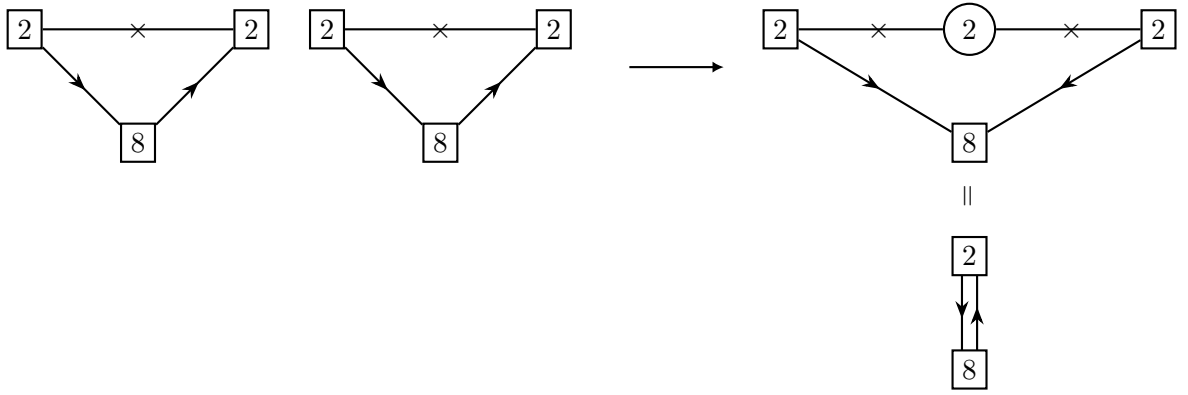
\begin{figure}[t]
    \centering
    \begin{tikzpicture}
    \begin{scope}[shift={(0,0)}]
    \node (p1) at (0,0) [rectangle,draw,thick] {$2$};
    \node (p2) at (3,0) [rectangle,draw,thick] {$2$};
    \node (p3) at (1.5,-1.5) [rectangle,draw,thick] {$8$};
    \draw[black, thick] (p1) -- (p2);
    \draw[black, thick,->-] (p1) -- (p3);
    \draw[black, thick,-<-] (p2) -- (p3);
    \node[below, yshift=0.25cm] at (1.5,0) {$\times$};
    \end{scope}
    \begin{scope}[shift={(4,0)}]
    \node (p1) at (0,0) [rectangle,draw,thick] {$2$};
    \node (p2) at (3,0) [rectangle,draw,thick] {$2$};
    \node (p3) at (1.5,-1.5) [rectangle,draw,thick] {$8$};
    \draw[black, thick] (p1) -- (p2);
    \draw[black, thick,->-] (p1) -- (p3);
    \draw[black, thick,-<-] (p2) -- (p3);
    \node[below, yshift=0.25cm] at (1.5,0) {$\times$};
    \end{scope}
    \begin{scope}[shift={(8,0)}]
    \draw[black, thick,->] (0,-0.5) -- (1.25,-0.5);
    \end{scope}
    \begin{scope}[shift={(10,0)}]
    \node (p1) at (0,0) [rectangle,draw,thick] {$2$};
    \node (p3) at (5,0) [rectangle,draw,thick] {$2$};
    \node (p2) at (2.5,0) [circle,draw,thick] {$2$};
    \node (p4) at (2.5,-1.5) [rectangle,draw,thick] {$8$};
    \draw[black, thick] (p1) -- (p2);
    \draw[black, thick] (p2) -- (p3);
    \draw[black, thick,->-] (p1) -- (p4);
    \draw[black, thick,->-] (p3) -- (p4);
    \node[below, yshift=0.25cm] at (1.3,0) {$\times$};
    \node[below, yshift=0.25cm] at (3.85,0) {$\times$};
    \end{scope}
    \begin{scope}[shift={(12.5,-2)}]
    \node[rotate=90] at (0,-0.2) {$=$};
    \end{scope}
    \begin{scope}[shift={(12.5,-3)}]
    \node at (0,0) [rectangle,draw,thick] {$2$};
    \node at (0,-1.5) [rectangle,draw,thick] {$8$};
    \draw[black, thick,->-] (-0.1,-0.25) -- (-0.1,-1.25);
    \draw[black, thick,-<-] (0.1,-0.25) -- (0.1,-1.25);
    \end{scope}
    \end{tikzpicture}
    \caption{The $S$-gluing of two copies of the basic flux tube.}
    \label{fig:S gluing fig}
\end{figure}

This gluing actually turns out to give a trivial theory. Indeed, the $\mathfrak{su}(2)$ gauge node has only four fundamental chiral multiplets. This theory is known to have a quantum deformed moduli space of vacua and the $\mathfrak{su}(2)\oplus \mathfrak{su}(2)$ global symmetry is consequently broken down to its diagonal $\mathfrak{su}(2)$ subgroup \cite{Seiberg:1994bz}. We are thus left with two pairs of chirals multiplets connecting this $\mathfrak{su}(2)$ with the $\mathfrak{su}(8)$ which are in conjugate representations. This is compatible with a mass term relating these two sets of fields, so that at low energies there is no dynamical degree of freedom left.

Compatibly with this, the flux associated to this composite tube is the difference of that of the components and is thus zero
\be
\ba
    \mathcal{F}&=\left(-\frac{1}{2},-\frac{1}{2},-\frac{1}{2},-\frac{1}{2},-\frac{1}{2},-\frac{1}{2},-\frac{1}{2},-\frac{1}{2}\right)-\left(-\frac{1}{2},-\frac{1}{2},-\frac{1}{2},-\frac{1}{2},-\frac{1}{2},-\frac{1}{2},-\frac{1}{2},-\frac{1}{2}\right)\\
    &=(0,0,0,0,0,0,0,0)\,.
    \ea
\ee
At the level of the parameters of the 5d $\mathcal{N}=1$ gauge theory, the $S$-gluing corresponds to the trivial identification
\be
\ba
\phi&\to\phi\,,\\
m_\lambda&\to m_\lambda\,,\\
m_i&\to m_i\,,\qquad i=1,\cdots,8\,.
\ea
\ee

\paragraph{Mixed gluing. } We can also consider a mixture of $\Phi$ and $S$-gluing. More precisely, we decompose every octet of fields into two groups of $2n$ and $8-2n$ chirals for $n=1,\cdots,4$. We then perform a $\Phi$-gluing for the first set and an $S$-gluing for the second
\begin{equation}\label{eq:Phigluing}
    \delta\mathcal{W}=\sum_{i=1}^{2n}\Phi^i\left(R^{(1)}_i-L^{(2)}_i\right)+\sum_{i=2n+1}^8R^{(1)i}L^{(2)}_i\,,
\end{equation}
where we are now writing the $\mathfrak{su}(8)$ index $i$ explicitly. This deformation thus breaks this $\mathfrak{su}(8)$ to the subgroup
\be\label{eq:SU8decomp}
\mathfrak{su}(8)\to \mathfrak{su}(2n)\oplus \mathfrak{su}(8-2n)\oplus \mathfrak{u}(1)\,.
\ee
Integrating out the massive fields as explained before leads to the quiver of Figure \ref{fig:Mixed gluing fig}.

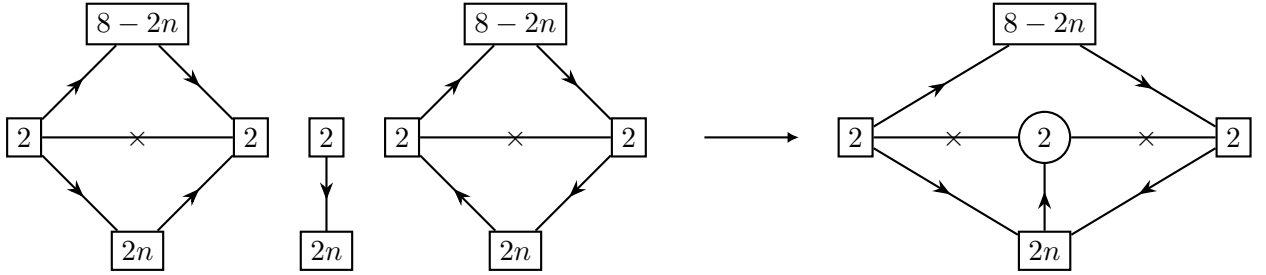
\begin{figure}[t]
    \centering
    \begin{tikzpicture}
    \begin{scope}[shift={(0,0)}]
    \node (p1) at (0,0) [rectangle,draw,thick] {$2$};
    \node (p2) at (3,0) [rectangle,draw,thick] {$2$};
    \node (p3) at (1.5,-1.5) [rectangle,draw,thick] {$2n$};
    \node (p4) at (1.5,1.5) [rectangle,draw,thick] {$8-2n$};
    \draw[black, thick] (p1) -- (p2);
    \draw[black, thick,->-] (p1) -- (p3);
    \draw[black, thick,-<-] (p2) -- (p3);
    \draw[black, thick,->-] (p1) -- (p4);
    \draw[black, thick,-<-] (p2) -- (p4);
    \node[below, yshift=0.25cm] at (1.5,0) {$\times$};
    \end{scope}
    \begin{scope}[shift={(4,0)}]
    \node at (0,0) [rectangle,draw,thick] {$2$};
    \node at (0,-1.5) [rectangle,draw,thick] {$2n$};
    \draw[black, thick,->-] (0,-0.25) -- (0,-1.25);
    \end{scope}
    \begin{scope}[shift={(5,0)}]
    \node (p1) at (0,0) [rectangle,draw,thick] {$2$};
    \node (p2) at (3,0) [rectangle,draw,thick] {$2$};
    \node (p3) at (1.5,-1.5) [rectangle,draw,thick] {$2n$};
    \node (p4) at (1.5,1.5) [rectangle,draw,thick] {$8-2n$};
    \draw[black, thick] (p1) -- (p2);
    \draw[black, thick,-<-] (p1) -- (p3);
    \draw[black, thick,->-] (p2) -- (p3);
    \draw[black, thick,->-] (p1) -- (p4);
    \draw[black, thick,-<-] (p2) -- (p4);
    \node[below, yshift=0.25cm] at (1.5,0) {$\times$};
    \end{scope}
    \begin{scope}[shift={(9,0)}]
    \draw[black, thick,->] (0,0) -- (1.25,0);
    \end{scope}
    \begin{scope}[shift={(11,0)}]
    \node (p1) at (0,0) [rectangle,draw,thick] {$2$};
    \node (p3) at (5,0) [rectangle,draw,thick] {$2$};
    \node (p2) at (2.5,0) [circle,draw,thick] {$2$};
    \node (p4) at (2.5,-1.5) [rectangle,draw,thick] {$2n$};
    \node (p5) at (2.5,1.5) [rectangle,draw,thick] {$8-2n$};
    \draw[black, thick] (p1) -- (p2);
    \draw[black, thick] (p2) -- (p3);
    \draw[black, thick,->-] (p1) -- (p4);
    \draw[black, thick,-<-] (p2) -- (p4);
    \draw[black, thick,->-] (p3) -- (p4);
    \draw[black, thick,->-] (p1) -- (p5);
    \draw[black, thick,-<-] (p3) -- (p5);
    \node[below, yshift=0.25cm] at (1.3,0) {$\times$};
    \node[below, yshift=0.25cm] at (3.85,0) {$\times$};
    \end{scope}
    \end{tikzpicture}
    \caption{The mixed gluing of two copies of the basic flux tube.}
    \label{fig:Mixed gluing fig}
\end{figure}

In this case, the gluing identifies the two tube theories up to the Weyl transformation of $\mathfrak{so}(16)$ that only conjugates $2n$ out of its $8$ Cartan elements. In terms of the parameters of the 5d $\mathcal{N}=1$ gauge theory we have
\be \label{eq:DWtransn}
\ba
\phi&\to\phi\,,\\
m_\lambda&\to m_\lambda+\sum_{i=1}^{2n}m_i\,,\\
m_i&\to -m_i\,,\qquad i=1,\cdots,2n\,,\\
m_i&\to m_i\,,\qquad i=2n+1,\cdots,8\,.
\ea
\ee
At the level of the fluxes, since we are performing a mixture of $\Phi$-gluing for the first $2n$ fields and $S$-gluing for the remaining $8-2n$ fields, we should change the sign of the last $8-2n$ entries of the flux of the second tube before summing it to the one for the first tube
\be\label{eq:fluxmixed}
\ba
    \mathcal{F}&=\bigg(-\frac{1}{2},-\frac{1}{2},-\frac{1}{2},-\frac{1}{2},-\frac{1}{2},-\frac{1}{2},-\frac{1}{2},-\frac{1}{2}\bigg)+\bigg(\underbrace{-\frac{1}{2},\cdots,-\frac{1}{2}}_{2n},\underbrace{\frac{1}{2},\cdots,\frac{1}{2}}_{8-2n}\bigg)\\
    &=\bigg(\underbrace{-1,\cdots,-1}_{2n},\underbrace{0,\cdots,0}_{8-2n}\bigg)\,.
\ea
\ee

\subsection{Flux Tori}
\label{subsec:fluxtori}

Given the above gluing rules, we can next proceed to constructing the most general flux tube theory by concatenating many copies of the basic flux tube theory with a sequence of $\Phi$ and mixed gluings (remember the pure $S$-gluing is trivial). Moreover, if the flux of the resulting tube lives in the $E_8$ lattice \eqref{eq: field theory fluxes}, we can perform one last gluing to identify its two ends and end up with the theory corresponding to a compactification on a torus with flux. 

For simplicity we will for now consider the simplest models that can be built in this way, and discuss the most general configuration in Section \ref{sect:elemfluxtubes}. These consist of the torus models obtained by gluing $4q$ copies of the basic flux tube with an alternation of $\Phi$-gluing and mixed gluing, where $q\in\mathbb{N}$. Equivalently, we can think of this as $\Phi$-gluing together $2q$ copies of the tube of Figure \ref{fig:Mixed gluing fig} obtained from the mixed gluing. The flux of this model is thus given by $2q$ times the one \eqref{eq:fluxmixed} of the tube obtained after mixed gluing 
\be
\mathcal{F}=\bigg(\underbrace{-2q,\cdots,-2q}_{2n},\underbrace{0,\cdots,0}_{8-2n}\bigg) \,,\qquad q\in\mathbb{N}\,.
\ee
In Figure \ref{fig:Torusq1intro} we show the quiver for the model corresponding to the case $q=1$. The generalization to arbitrary $q$ is straightforward: we have a wheel composed of $2q$ gauge nodes $\mathfrak{su}(2)$, which are all connected with rays to the middle flavor $\mathfrak{su}(2n)$, while only alternating $\mathfrak{su}(2)$ gauge nodes are connected to the flavor $\mathfrak{su}(8-2n)$.


\paragraph{Global symmetry. }The manifest continuous global symmetry of this model is
\be
    \mathfrak{su}(2n)\oplus \mathfrak{su}(8-2n)\oplus \mathfrak{u}(1)\oplus \mathfrak{u}(1)\,.
\ee
One abelian symmetry factor results from the decomposition \eqref{eq:SU8decomp} of the original $\mathfrak{su}(8)$ symmetry of the basic flux tube that occurred in the mixed gluing. The second abelian symmetry instead comes from one of the $\mathfrak{u}(1)$s of the global symmetry \eqref{eq:globalsymm} of the basic flux tube. As explained in \cite{Sabag:2022hyw}, the second $\mathfrak{u}(1)$ of the basic flux tube theory corresponding to the KK symmetry is instead broken to a finite subgroup by anomaly and superpotential constraints. We refer to Appendix \ref{app:KKshift} for more details on this point as well as an explanation of how this agrees with the geometric perspective that we will discuss in the next sections.

The actual symmetry of the theory at low energies is however different in general from the one that is manifest in the quiver theory. For example, we have already mentioned that for $n=4$ the $\mathfrak{e}_8$ roots that are orthogonal to the flux vector coincide with the roots of the $\mathfrak{e}_7\oplus \mathfrak{u}(1)$ subgroup. Since there are no punctures anymore, we expect the $\mathfrak{su}(8)$ symmetry of the quiver to enhance to $\mathfrak{e}_7$ in the IR. For different values of $n$, the expected continuous symmetry is
\be \label{eq:preserved flavor symms}
\ba
n=4:\quad &\mathfrak{e}_7\oplus \mathfrak{u}(1)\,,\\
n=3:\quad &\mathfrak{e}_6\oplus \mathfrak{su}(2)\oplus \mathfrak{u}(1)\,,\\
n=2:\quad &\mathfrak{so}(14)\oplus \mathfrak{u}(1)\,,\\
n=1:\quad &\mathfrak{e}_7\oplus \mathfrak{u}(1)\,.
\ea
\ee
One of the main checks of our geometric construction of these models that we will discuss in Section \ref{sec: Geometric Realization} will be to reproduce these global symmetries, as well as the breaking of the $\mathfrak{u}(1)$ KK symmetry.

\paragraph{IR dualities. }
The fact that the cases $n=4$ and $n=1$ preserve the same symmetry is due to their fluxes being related by a Weyl symmetry transformation of $\mathfrak{e}_8$
\begin{equation}
    (-q,-q,-q,-q,-q,-q,-q,-q)\quad\overset{\text{Weyl}}{\longleftrightarrow}\quad(-2q,-2q,0,0,0,0,0,0)\,.
\end{equation}
As a consequence, the two compactifications are actually equivalent and the corresponding quiver theories provide different microscopic descriptions of the same low energy theory, namely they are IR dual. As explained in \cite{Kim:2017toz,Pasquetti:2019hxf}, this can be understood as a local application of Seiberg duality \cite{Seiberg:1994pq} on some nodes of the quiver. Specifically, the 4d $\mathcal{N}=1$ $\mathfrak{su}(2)$ SQCD with 6 fundamental chiral multiplets is dual to a simple Wess--Zumino model, i.e.~a model without gauge group and composed only of chiral fields interacting with a cubic superpotential. This duality can be summarized with quivers as in Figure \ref{fig:Seiberg}, where we have a cubic superpotential for each triangle loop.

\begin{figure}
    \centering
    \begin{tikzpicture}
    \begin{scope}[shift={(0,0)}]
    \node (p1) at (0,0) [rectangle,draw,thick] {$2$};
    \node (p3) at (5,0) [rectangle,draw,thick] {$2$};
    \node (p2) at (2.5,0) [circle,draw,thick] {$2$};
    \node (p4) at (2.5,-1.5) [rectangle,draw,thick] {$2$};
    \draw[black, thick] (p1) -- (p2);
    \draw[black, thick] (p2) -- (p3);
    \draw[black, thick,-] (p2) -- (p4);
    \node[below, yshift=0.25cm] at (1.3,0) {$\times$};
    \node[below, yshift=0.25cm] at (3.85,0) {$\times$};
    \end{scope}
    \begin{scope}[shift={(6.5,0)}]
    \draw[black, thick,<->] (0,-0.5) -- (1.25,-0.5);
    \end{scope}
    \begin{scope}[shift={(9.25,0)}]
    \node (p1) at (0,0) [rectangle,draw,thick] {$2$};
    \node (p2) at (3,0) [rectangle,draw,thick] {$2$};
    \node (p3) at (1.5,-1.5) [rectangle,draw,thick] {$2$};
    \draw[black, thick] (p1) -- (p2);
    \draw[black, thick,-] (p1) -- (p3);
    \draw[black, thick,-] (p2) -- (p3);
    \node[below, yshift=0.25cm] at (1.5,0) {$\times$};
    \end{scope}
    \end{tikzpicture}
    \caption{The basic Seiberg duality.}
    \label{fig:Seiberg}
\end{figure}
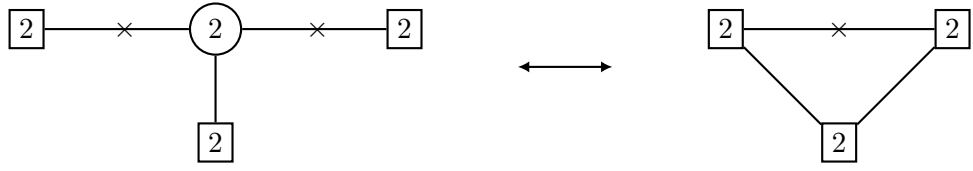

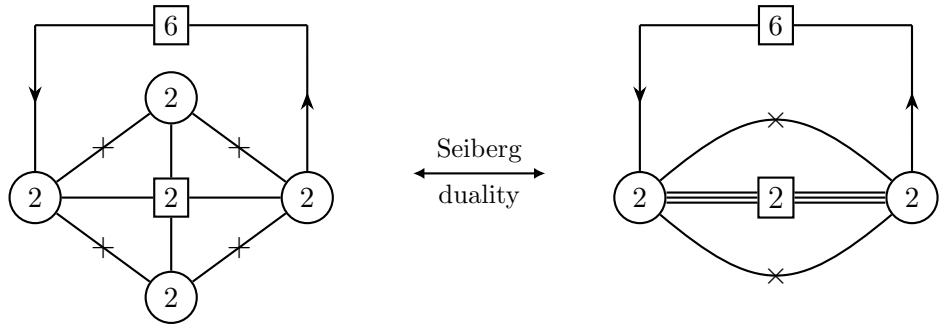
\begin{figure}
    \centering
    \begin{tikzpicture}[
    triple/.style={
        draw,
        double,
        double distance=3pt,
        line width=0.8pt,
        postaction={draw, line width=0.8pt}
    }
]
    \begin{scope}[shift={(0,0)},scale=0.6]
    \node[draw, circle,thick] (p1) at (0,-2.2) {\fontsize{11pt}{11pt}\selectfont $2$};
       \node[draw, circle,thick] (p2) at (3,0) {\fontsize{11pt}{11pt}\selectfont $2$};
       \node[draw, circle,thick] (p3) at (6,-2.2) {\fontsize{11pt}{11pt}\selectfont $2$};
       \node[draw, rectangle,thick] (p4) at (3,-2.2) {\fontsize{11pt}{11pt}\selectfont $2$};
       \node[draw, circle,thick] (p5) at (3,-4.4) {\fontsize{11pt}{11pt}\selectfont $2$};
       \node[draw, rectangle,thick] (p6) at (3,1.6) {\fontsize{11pt}{11pt}\selectfont $6$};
       \draw[-,thick] (p1) to  node[rotate=45]{\large $\times$}    (p2);
       \draw[-,thick] (p2) to  node[rotate=45]{\large $\times$}    (p3);
       \draw[-,thick] (p3) to  node[rotate=45]{\large $\times$}    (p5);
       \draw[-,thick] (p5) to  node[rotate=45]{\large $\times$}    (p1);
       \draw[-,thick] (p1) to      (p4);
       \draw[-,thick] (p3) to      (p4);
       \draw[-,thick] (p2) to      (p4);
       \draw[-,thick] (p5) to      (p4);
       \draw[-<-,thick] (p1) to      (0,1.6);
       \draw[-,thick] (0,1.6) to      (p6);
       \draw[->-,thick] (p3) to      (6,1.6);
       \draw[-,thick] (6,1.6) to      (p6);
    \end{scope}

    \begin{scope}[shift={(5.25,0)}]
       \draw[black, thick,<->] (-0.25,-1) -- node[midway, above] {\small Seiberg} node[midway, below] {\small duality} (1.5,-1);
    \end{scope}

    \begin{scope}[shift={(8,0)},scale=0.6]
       \node[draw, circle,thick] (p1) at (0,-2.2) {\fontsize{11pt}{11pt}\selectfont $2$};
       \node[draw, circle,thick] (p3) at (6,-2.2) {\fontsize{11pt}{11pt}\selectfont $2$};
       \node[draw, rectangle,thick] (p4) at (3,-2.2) {\fontsize{12pt}{12pt}\selectfont $2$};
       \node[draw, rectangle,thick] (p6) at (3,1.6) {\fontsize{11pt}{11pt}\selectfont $6$};

       \draw[triple] (p1) to (p4);
       \draw[triple] (p3) to (p4);

       \draw[-<-,thick] (p1) to (0,1.6);
       \draw[-,thick] (0,1.6) to (p6);
       \draw[->-,thick] (p3) to (6,1.6);
       \draw[-,thick] (6,1.6) to (p6);

       \draw[thick] (p1) to [out=40,in=140,looseness=1.4] node {\large$\times$} (p3);
       \draw[thick] (p1) to [out=-40,in=-140,looseness=1.4] node {\large$\times$} (p3);
    \end{scope}
\end{tikzpicture}
    \caption{Application of the Seiberg duality to the $n=1$ model with $q=1$.}
    \label{fig:Seiberg n=14}
\end{figure}

\begin{figure}
    \centering
    \begin{tikzpicture}[
    triple/.style={
        draw,
        double,
        double distance=3pt,
        line width=0.8pt,
        postaction={draw, line width=0.8pt}
    }
]
    \begin{scope}[shift={(0,0)},scale=0.6]
       \node[draw, circle,thick] (p1) at (0,-2.2) {\fontsize{11pt}{11pt}\selectfont $2$};
       \node[draw, circle,thick] (p3) at (6,-2.2) {\fontsize{11pt}{11pt}\selectfont $2$};
       \node[draw, rectangle,thick] (p4) at (3,-2.2) {\fontsize{11pt}{11pt}\selectfont $2$};
       \node[draw, rectangle,thick] (p6) at (3,1.6) {\fontsize{11pt}{11pt}\selectfont $6$};

       \draw[triple] (p1) to (p4);
       \draw[triple] (p3) to (p4);

       \draw[-<-,thick] (p1) to (0,1.6);
       \draw[-,thick] (0,1.6) to (p6);
       \draw[->-,thick] (p3) to (6,1.6);
       \draw[-,thick] (6,1.6) to (p6);

       \draw[thick] (p1) to [out=40,in=140,looseness=1.4] node {\large$\times$} (p3);
       \draw[thick] (p1) to [out=-40,in=-140,looseness=1.4] node {\large$\times$} (p3);
    \end{scope}
    \begin{scope}[shift={(5.25,0)}]
    \draw[black, thick,->] (-0.25,-1) -- node[midway, above] {\small mass} node[midway, below] {\small deformation} (1.5,-1);
    \end{scope}
    \begin{scope}[shift={(8,0)},scale=0.6]
    \node[draw, circle,thick] (p1) at (0,-2.2) {\fontsize{11pt}{11pt}\selectfont $2$};
       \node[draw, circle,thick] (p3) at (6,-2.2) {\fontsize{11pt}{11pt}\selectfont $2$};
       \node[draw, rectangle,thick] (p4) at (3,-2.2) {\fontsize{11pt}{11pt}\selectfont $2$};
       \node[draw, rectangle,thick] (p6) at (3,1.6) {\fontsize{11pt}{11pt}\selectfont $6$};
       \draw[-,thick] (p1) to      (p4);
       \draw[-,thick] (p3) to      (p4);
       \draw[-<-,thick] (p1) to      (0,1.6);
       \draw[-,thick] (0,1.6) to      (p6);
       \draw[->-,thick] (p3) to      (6,1.6);
       \draw[-,thick] (6,1.6) to      (p6);
       \draw[thick] (p1) to [out=40,in=140,looseness=1.4] node {\large$\times$} (p3);
       \draw[thick] (p1) to [out=-40,in=-140,looseness=1.4] node {\large$\times$} (p3);
    \end{scope}
    \begin{scope}[shift={(8,-2.8)}]
    \node[rotate=90] at (1.8,-0.2) {$=$};
    \end{scope}
    \begin{scope}[shift={(8,-3.4)},scale=0.6]
    \node[draw, circle,thick] (p1) at (0,-2.2) {\fontsize{11pt}{11pt}\selectfont $2$};
       \node[draw, circle,thick] (p3) at (6,-2.2) {\fontsize{11pt}{11pt}\selectfont $2$};
       \node[draw, rectangle,thick] (p4) at (3,-2.2) {\fontsize{11pt}{11pt}\selectfont $8$};
       \draw[-<-,thick] (p1) to      (p4);
       \draw[->-,thick] (p3) to      (p4);
       \draw[thick] (p1) to [out=40,in=140,looseness=1.4] node {\large$\times$} (p3);
       \draw[thick] (p1) to [out=-40,in=-140,looseness=1.4] node {\large$\times$} (p3);
    \end{scope}
    \end{tikzpicture}
    \caption{Implementation of the mass deformation leading to the dual $n=1$ model with $q=\tfrac{1}{2}$.}
    \label{fig:Seiberg n=14bis}
\end{figure}
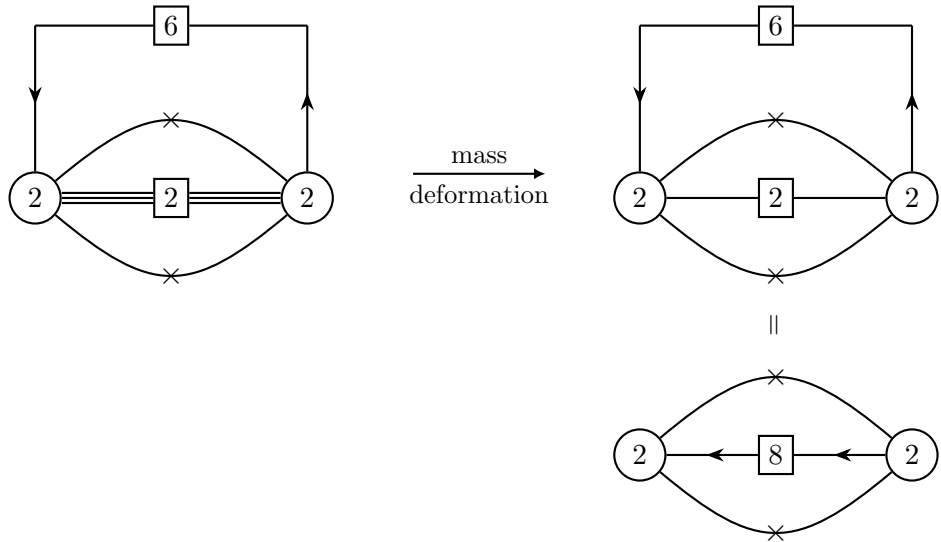
For simplicity, we present the derivation of the duality between flux torus models for the case of
$q=\tfrac{1}{2}$. This consists of two steps. First, we start from the $n=4$ model in Figure \ref{fig:Torusq1intro} for $q=1$ and apply the basic Seiberg duality to the $\mathfrak{su}(2)$ gauge nodes that are connected to only one of the two flavor nodes. The result of the dualization is shown in Figure \ref{fig:Seiberg n=14}.

Some of the original cubic interactions turned into quadratic mass terms for the chirals encoded in the triple links connecting the left and right $\mathfrak{su}(2)$ gauge nodes to the middle $\mathfrak{su}(2)$ flavor node. The second step is then to integrate out the massive fields, which has the effect of removing pairs of links, see Figure \ref{fig:Seiberg n=14bis}. This gives a quiver that can be immediately identified with the dual $n=4$ quiver of Figure \ref{fig:Torusq1intro} with $q=\tfrac{1}{2}$.

\section{E-string Geometric Engineering} \label{sec: E-string geometry}

\subsection{E-string from F-Theory and M/F-Duality}

The canonical way to realize the geometry underlying the E-string theory is to consider an elliptic Calabi-Yau threefold $Y$ over a non-compact base $\C^2$ (with coordinates $u, v\in \C$) given by a Tate form 
\be  \label{eq: II* fiber}
y^2+b_1 u x y+b_3 u^3 y=x^3+b_2 u^2 x^2+b_4 u^4 x+b_6 u^5 v  \, ,
\ee 
for constants $b_i$. This has the following singular structure: along the locus $u=0$ there is a type $II^*$ fiber, which when resolved gives rise to an $E_8$ intersection matrix between the rational curves and the exceptional divisors of the blowup. 
Along $v=0$ there is an $I_1$ fiber which meets the $II^*$ fiber at $u=v=0$. Here, the fiber becomes non-minimal, i.e. we have a non-Kodaira fiber with the vanishing orders in the Weierstrass form exceeding $\deg (f, g, \Delta) \geq (4,6,12)$. 


To resolve this singularity one can proceed in two ways: either by first blowing up the base, i.e. removing the locus $u=v=0$ and thus inserting a $(-1)$-curve. This separates the loci of the $II^*$ and the $I_1$ fibers, so that all singular fibers are of Kodaira type, with the $II^*$ fiber realizing the affine $E_8$ Dynkin diagram. Note that this gives rise to a compact divisor $S$ in this three-fold, obtained from the fiber and base rational curve. 
The resulting smooth CY3 geometry $\widetilde{X}$ has non-compact (flavor) divisors $D_{\alpha_i}$, and associated rational curves $C_{\alpha_i}$ such that 
\be
D_{\alpha_i} \cdot_{\tilde{X}}  C_{\alpha_j} = - \text{CM}^{E_8}_{ij} \, ,
\ee
where $\text{CM}$ is the Cartan Matrix of $E_8$. 

Alternatively we can perform a non-flat resolution of the fiber, which means inserting a complex 2-dimensional surface into the fiber and then further resolving the remaining singularities. 

Both approaches result in a smooth Calabi-Yau threefold $X_{S_{II^*}}$ containing a smooth dP$_9$ or gdP$_9$ surface $S_{II^*}$ -- as discussed in detail in \cite{Apruzzi:2019opn} (blowup sequence (3.14)) with self-triple intersection $[S_{II^*}]^3=0$ and intersections with the non-compact divisors $D_i$ associated to the simple roots of affine $E_8$.

The key point here is to note that the non-compact Calabi-Yau threefold $X_{S_{II^*}}$ can be realized as the total space of the canonical bundle over dP$_9$, so that the geometric engineering follows entirely from the properties of curves within this surface. 

Starting from \eqref{eq: II* fiber}, we can first resolve the base by eliminating the intersection point $u=v=0$, and inserting a $\mathbb{P}^1$ there. In this partially resolved geometry, the dP$_9$ is realized by a Tate model over $\P^1$ with a single $II^*$ and two $I_1$ fibers, which is a well-known family of rational elliptic surfaces. Together with the fact that such a rational elliptic surface has a trivial Mordell-Weil group, we find that $H^2(S_{II^*},\Z)$ is generated by the elliptic fiber $E$, zero section $\sigma_0$ and the fiber components $\alpha_i$ $i=1,\cdots, 8$ of the $II^*$ fiber, so that 
the inner form between curves is 
\begin{equation}\label{eq:I19viaIIstarfib}
H^2(S_{II^*},\Z) = \begin{pmatrix}
    0 & 1 \\
    1 & -1
\end{pmatrix} \oplus (-E_8) \, .
\end{equation}
The first Chern class of $S_{II^*}$ is given by
\begin{equation}
    c_1(S_{II^*}) = E  \, . 
\end{equation}
By Poincar\'e duality, the relation between self-intersection and genus of curves, and the index theorem, $H^2(S_{II^*},\Z)$ is an odd\footnote{In this context, odd means that $x^2$ is not always an even integer for all elements $x$ of the lattice.}, unimodular lattice of signature $(1,9)$,  which uniquely determines it up to isometry \cite{Serre:1993}. In the following, we will denote this lattice by $I_{1,9}$. 

Compactifying M-Theory on $X_{S_{II^*}}$ yields a 5D KK theory which is known to flow to a 5D $\mathcal{N}=1$ $\mathfrak{su}(2)$ gauge theory with eight flavors in the IR. Upon a suitable parametrization of the K\"ahler form $J$ on $X_{S_{II^*}}$, the IMS prepotential of this gauge theory \eqref{eq:IMSprepotEstring} agrees with the geometric prepotential
\begin{equation}
    \CF_{\text{geo}} = -\frac16 \int_{X_{S_{II^*}}} J^3 \, ,
\end{equation}
up to terms independent of the Coulomb branch parameter $\phi$.

\subsection{The Prepotential for $X_{S_{II^*}}$}

In order to make the geometric origin of the IMS prepotential explicit, let us choose a slightly different parametrization of $H^2(S,\Z)$ and write
\begin{equation}\label{eq:innerformS12I}
H^2(S_{II^*},\Z) \simeq I_{1,9} \simeq \begin{pmatrix}
        0 & 1 \\
        1 & 0
    \end{pmatrix} \oplus (-1)^8 \, .
\end{equation}
We will denote the generators by $F,\Sigma$ and $M_i$ for $i=1,\cdots, 8$, so that
\be
    F \cdot \Sigma=1 \, , \qquad M_i \cdot M_j=-\delta_{ij} \, ,
\ee
with all other products vanishing\footnote{We will see this presentation emerge from an explicit example later in Section \ref{sec:constructionofXS}.}. This can be related to our previous basis \eqref{eq:I19viaIIstarfib} by
\begin{equation}\label{eq: C vs alpha}
    \begin{aligned}
       & E = 2F + 2 \Sigma - \sum_{i=1}^8 M_i && \sigma_0 = M_1 \\
        &\alpha_1 = \Sigma - F \qquad &   \alpha_2 = M_8 - M_7 \qquad& \alpha_3 = F - M_7 - M_8 \qquad& \alpha_4 = M_7-M_6 \\ 
        &\alpha_5 = M_6 -M_5   \qquad &   \alpha_6 = M_5 -M_4  \qquad& \alpha_7 = M_4 -M_3 \qquad & \alpha_8 = M_3 -M_2
    \end{aligned}
\end{equation}

We now choose to parameterize the K\"ahler form of the threefold {$X_{S_{II^*}}$} as follows
\be \label{eq: correct Kahler form}
   J = -\phi [S_{II^*}] + m_{\lambda}^0 F + \sum_{i=1}^8 m_i M_i \, ,
\ee
where $[S_{II^*}]$ is the class of the surface $S_{II^*}$ inside the threefold and we use the abbreviation
\begin{equation}
    m_\lambda^0 := \frac{m_{\text{KK}}}{2}-\frac{1}{2}\sum_{i=1}^8m_i\, . 
\end{equation}

By abuse of notation, we also represent by $F$, $\Sigma$, and $M_i$ the non-compact divisor classes of $X_{S_{II^*}}$ associated with the respective curve classes in the dP$_9$. The parameters $\{ \phi,m_{\lambda}^0,m_i \}$ provide a parameterization of the extended Coulomb branch of our 5d $\cN=1$ gauge theory. 

As $X_{S_{II^*}}$ is a Calabi-Yau threefold, adjunction tells us that
\begin{equation}
 [S_{II^*}] = -E=-2\Sigma-2F+\sum_{i=1}^8 M_i,   
\end{equation} 
so that we can rewrite the K\"ahler form as
\be \label{eq: Kahler form}
    J = 2\phi  \Sigma +(2\phi+m_{\lambda}^0)F +\sum_{i=1}^8 (m_i-\phi) M_i \, .
\ee
We can find the volumes of various curves in this parametrization by integrating $J$, where we want to think of the classes in $J$ as divisor classes on $X_{S_{II^*}}$.  
The volume of $F$ captures the Coulomb branch parameter $\phi$ 
\be \label{eq: Vol(F)=Coulomb branch}
    \text{Vol}(F) = 2 \phi \, ,
\ee
and the curves $M_i$ capture the masses of the hypers
\be \label{eq: Vol(M)=Coulomb branch}
\ba
    \text{Vol}(M_i) &= \phi-m_i \cr
    \text{Vol}(F-M_i) &= \phi +m_i \, .
\ea
\ee
Moreover, the volume of the elliptic fiber remains constant, and its volume is precisely the $m_{\text{KK}}$ mass scale 
\be \label{eq: Vol(E) remains constant}
    \text{Vol}(E) = m_{\text{KK}} \, .
\ee

To justify that~\eqref{eq: correct Kahler form} is the correct K\"ahler form, we first show that its geometric prepotential reproduces the IMS prepotential~\eqref{eq:IMSprepotEstring}. The curve classes $M_i$ are all contained in $S_{II^*}$, so that $X_{S_{II^*}}$ describes the gauge theory phase in which $\phi\pm m_i>0$ for all $i=1,\cdots ,8$. As $[S_{II^*}]|_{S_{II^*}} \cdot [S_{II^*}]|_{S_{II^*}}=0$, the geometric prepotential contains no terms cubic in $\phi$. We can now calculate
\be
\ba
    \CF_{\text{geo}} &= -\frac{1}{6} \int_{X_{S_{II^*}}} \left[ -\phi^3[S_{II^*}]^3 + 3 \phi^2 [S_{II^*}]^2\left( m_{\lambda}^0 F + \sum_{i=1}^8 m_i M_i\right) - 3 \phi [S_{II^*}] \left(m_{\lambda}^0 F + \sum_{i=1}^8 m_i M_i\right)^2 \right] \\
& = -\frac{1}{6} \int_{S_{II^*}}\left[ -3 \phi^2 [S_{II^*}]|_{S_{II^*}} \left(m_{\lambda}^0 F + \sum_{i=1}^8 m_i M_i\right)-3 \phi \left(m_{\lambda}^0 F + \sum_{i=1}^8 m_i M_i\right)^2 \right]\\
& = \left( m_{\lambda}^0+ \frac{1}{2}\sum_{i=1}^8 m_i \right) \phi^2 - \frac{1}{2} \phi \sum_{i=1}^8 m_i^2\, ,
\ea
\ee
where we have dropped all terms independent of the Coulomb branch parameter $\phi$, effectively allowing us to restrict the integral to $S$. The above 
coincides exactly with the IMS prepotential $\mathcal{F}_{\text{IMS}}|_{\phi \pm m_i>0}$. 

For the quiver theories we are interested in constructing, we will however need to work in the gauge theory phase $-m_i < \phi < m_i$. This appears to send the curves $M_i$ to negative volume, which geometrically means undergoing a flop transition 
\begin{equation}
    X_{S_{II^*}} \rightarrow X_{S_{II^*}}^f \, ,
\end{equation}
after which all the curves $M_i$ have been flopped out of the dP$_9$ surface $S_{II^*}$. From the perspective of $S_{II^*}$, this blows down all of the curves $M_i$, so that the canonical class of the compact surface $S^f$ in $X_{S_{II^*}}^f$ is now 
\be 
c_1(S^f) = 2\Sigma + 2 F\,,\qquad [S^f]|_S=-2\Sigma-2F\,,
\ee 
inside $X_{S_{II^*}}^f$. In particular, $M_i|_{S^f} = 0$ and $[S^f]_{S^f} \cdot [S^f]|_{S^f} = 8 $, so that the geometric prepotential becomes 
\be
\ba
    \CF_{\text{geo}} := -\frac{1}{6} \int_{X_{S_{II^*}}^f} J^3 &= -\frac{1}{6}\int_{S^f} \left(-\phi^3 [S^f]^2 + 3\phi^2 m_{\lambda}^0[S^f] \cdot F \right) \cr
    &= \frac{4}{3} \phi^3 -\frac{1}{2} m_{\lambda}^0 \phi^2 \int_{S^f} (-2\Sigma-2F) \cdot F \cr
    &= \frac{4}{3}\phi^3+m_{\lambda}^0 \phi^2 \, ,
\ea
\ee
which one can verify matches the gauge theory IMS prepotential $\mathcal{F}_{\text{IMS}}|_{-m_i<\phi<m_i}$.

\subsection{Flavor Symmetry} \label{sec: flavour symmetry}

The hypermultiplets transform under the vector $\mathbf{16}$ of $\mathfrak{so}(16)$, and we can use a basis where their weights, in the weight lattice $\Lambda_w(\mathfrak{so}(16))$, are given by
\begin{equation}
(\pm 1,0,0,0,0,0,0,0) \, ,
\end{equation}
and permutations thereof. We can identify these with curves in the threefold as
\begin{equation} \label{eq: N map applied to effective curves}
\begin{aligned}
N(M_i) =  (-\delta_{ij})_j \, , \qquad
N(F - M_i) = (\delta_{ij})_j \, ,
\end{aligned}
\end{equation}
where we have introduced the surjective homomorphism 
\be \label{eq: N homomorphism definition}
    N:\quad H_2(X_{S_{II^*}},\mathbb{Z})\rightarrow \Lambda_w(\mathfrak{so}(16))\, ,
\ee
to formalize this correspondence and keep track of the relative sign between the representation theoretic conventions and the intersection form in geometry
\begin{equation}
N(\alpha) \cdot N(\beta) = - \alpha \cdot \beta  \, .
\end{equation} 
This map is not injective as it has non-trivial kernel $N(F)=N(E)=0$. 

The charges of the 16 half-hypers under the Cartan $U(1)$'s are given by the intersection product of the matter curves and the non-compact flavor divisors. In particular, they branch as
\begin{equation}
    \mathbf{16} \rightarrow \mathbf{8}^{+1} \oplus \mathbf{\bar{8}}^{-1} \, ,
\end{equation}
with respect to the $\mathfrak{su}(8)\oplus \mathfrak{u}(1)$ subgroup of $\mathfrak{so}(16)$. Again, this determines the roots of $\mathfrak{so}(16)$ to be given by the following curves of $S$
\begin{equation}
\begin{aligned}
N(M_i - M_j) = (-\delta_{ik}+\delta_{jk})_k \, ,  \qquad N(F-M_i - M_j ) = (\delta_{ik}+\delta_{jk})_k  \, ,
\end{aligned}
\end{equation}
in terms of the $\mathfrak{so}(16)$ weight basis. We choose to work in a basis of simple roots $\{ \beta_0 ,\beta_1 ,\cdots ,\beta_8 \}$ of $\mathfrak{su}(2)\oplus \mathfrak{so}(16)$, appearing in the box graph description of the field theory symmetry algebra, given explicitly by
\begin{equation}\label{eq:Eso16roots_so16weights}
\ba
&& N(\beta)&& \mbox{Vol}\\
\beta_0 & = F& (0,0,0,0,0,0,0,0)&& 2\phi \\
\beta_1 & = M_2-M_1 & (1,-1,0,0,0,0,0,0)&& m_1-m_2 \\
\beta_2 & = M_3 - M_2 & (0,1,-1,0,0,0,0,0)&& m_2-m_3\\
\beta_3 & = M_4 - M_3 & (0,0,1,-1,0,0,0,0)&& m_3 - m_4\\
\beta_4 & = M_5 - M_4& (0, 0, 0, 1, -1, 0, 0, 0)&& m_4 - m_5\\
\beta_5 & = M_6 - M_5& (0, 0, 0, 0, 1, -1, 0, 0)&& m_5 - m_6\\
\beta_6 & = M_7 - M_6& (0, 0, 0, 0, 0, 1, -1, 0)&& m_6 - m_7\\
\beta_7 & = M_8 - M_7& (0,0, 0, 0, 0, 0,  1, -1)&& m_7 - m_8\\
\beta_8 & = F-M_7-M_8& (0, 0, 0, 0, 0, 0, 1, 1)&& m_7 +m_8 \, .
\ea
\end{equation}
The $\mathfrak{so}(16)$ root lattice 
\begin{equation}\label{eq D8 lattice}
    \Lambda_{D_{8}} := \langle \beta_1, \cdots, \beta_8 \rangle_{\Z} \, ,
\end{equation}
has an $\mathfrak{su}(8)$ root lattice 
\begin{equation}
    \Lambda_{A_7} := \langle \beta_1, \cdots, \beta_7 \rangle_{\Z}\, ,
\end{equation}
embedded within it. The identification between the geometric affine $E_8$ roots and the $\mathfrak{so}(16)$ weights, as well as the volumes of the associated curves in $S$, are given by
\begin{equation}\label{eq:E8roots_so16weights}
\ba
&& N(\alpha)&& \mbox{Vol}\\
\alpha_1 & = \Sigma - F & (-\tfrac12,-\tfrac12,-\tfrac12,-\tfrac12,-\tfrac12,-\tfrac12,-\tfrac12,-\tfrac12)&& m_\lambda^0\\
\alpha_2 & = M_8 - M_7 & (0,0,0,0,0,0,1,-1)&& m_7 - m_8\\
\alpha_3 & = F - M_7 - M_8 & (0,0,0,0,0,0,1,1)&& m_7 + m_8\\
\alpha_4 & = M_7 - M_6& (0, 0, 0, 0, 0, 1,-1, 0)&& m_6 - m_7\\
\alpha_5 & = M_6 - M_5& (0, 0, 0, 0, 1, -1, 0, 0)&& m_5 - m_6\\
\alpha_6 & = M_5 - M_4& (0, 0, 0, 1, -1, 0, 0, 0)&& m_4 - m_5\\
\alpha_7 & = M_4 - M_3& (0, 0, 1, -1, 0,  0, 0, 0)&& m_3 - m_4\\
\alpha_8 & = M_3 - M_2& (0, 1, -1, 0, 0, 0, 0, 0)&& m_2 - m_3\\
\alpha_0 & = M_2 - M_1& (1, -1, 0, 0, 0, 0, 0, 0)&& m_1 - m_2  \, .
\ea
\end{equation}
where the labeling of roots uses Bourbaki conventions as summarized by\\

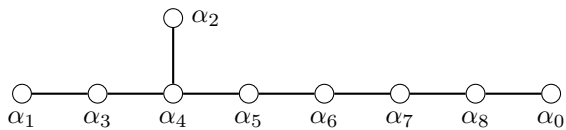
\begin{center}
\begin{tikzpicture}[
    node/.style={circle, draw, fill=white, inner sep=2.5pt},
    lbl/.style={font=\small, text=black}
]

    \node[node] (1) at (0,0) {};
    \node[node] (2) at (1,0) {};
    \node[node] (3) at (2,0) {};
    \node[node] (4) at (3,0) {};
    \node[node] (5) at (4,0) {};
    \node[node] (6) at (5,0) {};
    \node[node] (7) at (6,0) {};
    \node[node] (9) at (7,0) {};
    \node[node] (8) at (2,1) {}; 

    \draw[thick] (1) -- (2) -- (3) -- (4) -- (5) -- (6) -- (7) -- (9);
    \draw[thick] (3) -- (8);

    \node[lbl, below=3pt] at (1) {$\alpha_1$};
    \node[lbl, below=3pt] at (2) {$\alpha_3$};
    \node[lbl, below=3pt] at (3) {$\alpha_4$};
    \node[lbl, below=3pt] at (4) {$\alpha_5$};
    \node[lbl, below=3pt] at (5) {$\alpha_6$};
    \node[lbl, below=3pt] at (6) {$\alpha_7$};
    \node[lbl, below=3pt] at (7) {$\alpha_8$};
    \node[lbl, below=3pt] at (9) {$\alpha_0$};
    \node[lbl, right=3pt] at (8) {$\alpha_2$};

\end{tikzpicture}
\end{center}

For this choice of simple roots for $\Lambda_{E_8}$ we have 
\begin{equation}
  E = 2 \alpha_1 + 4 \alpha_3 + 6 \alpha_4 + 3 \alpha_2 + 5 \alpha_5 + 
 4 \alpha_6 + 3 \alpha_7 + 2 \alpha_8 + \alpha_0 \, ,
\end{equation}
so that we can think of the $\alpha_i$ as the divisors associated to the fiber components of the resolved $II^*$ fiber of $S_{II^*}$.

Note that setting $m_i = m$ for all $i$ realizes $\mathfrak{su}(8)$, which is enhanced to $\mathfrak{so}(16)$ for $m_i=0 $ for all $i$. The $E_8$ flavor symmetry is only found in the UV where all mass parameters vanish. Using the form of the roots in~\eqref{eq:E8roots_so16weights}, we can conclude that the kernel of $N$ is trivial when restricted to 
\begin{equation}
\Lambda_{E_8} := \langle \alpha_1, \cdots, \alpha_8 \rangle_\Z.
\end{equation}

Hence every flux appearing in the field theoretic construction of flux tubes~\eqref{eq: field theory fluxes} is realized uniquely by a curve in $H_2(X_{S_{II^*}},\mathbb{Z})$.

\section{Local Calabi-Yau Manifolds} \label{sec:constructionofXS}

In this section we will describe some aspects of the geometry of $X_S$ in some more detail and explain how the maps associated with the domain walls and gluings discussed in Section \ref{sec: Field Theory sec2} are realized geometrically. The presentation of a rational elliptic surface as a Weierstrass fibration over $\P^1$ makes this hard to describe, so that we will first introduce a more suitable model. 

{This description is a particular complex structure deformation, which will make some of the aspects of the geometry more manifest, in particular the automorphisms. As we are considering automorphisms of a (partially resolved and thus still) singular elliptic fibration, the precise notion of automorphism remains to be developed in mathematics. We will here derive the automorphisms in a particular choice of complex structure, where the $II^*$ fiber is deformed to an $I_8$, and then conjecture that these automorphisms survive the limit back to the original geometry. The particular complex structure deformation is discussed in Section \ref{sec: Geometry of dP9 S}.}

\subsection{The dP$_9$ Surface $S_{0}$ and the Threefold $X_{S_{0}}$} \label{sec: Geometry of dP9 0}

The emergence of a 5d gauge theory description of the E-string theory is best understood from a ruling with reducible fibers on the contractable dP$_9$ surface contained in the E-string geometry. In order to relate the E-string geometry to the gauge theory description, we will first describe a particular geometric realization of a dP$_9$ surface in which both the ruling and the elliptic fibration are manifest, but the flavor symmetry is not. 

Consider the following hypersurface in $\P_x^1 \times \P_y^1 \times \P_z^1$
\begin{equation}\label{eq:dp9_p1cubed}
 S_{0}: \hspace{1cm} P_0(y,z) \, x_1 = Q_0(y,z)  \, x_2 \,,
\end{equation}
where $[x_1:x_2]$, $[y_1:y_2]$ and $[z_1:z_2]$ are homogeneous coordinates on the three $\P^1$s and $P_0$ and $Q_0$ are generic homogeneous polynomials 
of $[y_1:y_2]$ and $[z_1:z_2]$ of multidegree $(2,2)$.

As $[x_1:x_2]$ is uniquely determined everywhere except for the $8$ points where $P_0=Q_0=0$, we can think of $S_{0}$ as a blowup of $\P^1 \times \P^1$. Denoting the restriction of the hyperplane classes of $\P_y^1$ and $\P_z^1$ by $H_y$ and $H_z$ and the $8$ exceptional curves by $M_i$ then reproduces the inner form \eqref{eq:innerformS12I} with $F = H_z$, $\Sigma = H_y$, and realizes all of these generators as effective curves. The anti-canonical class of $S_{0}$ is correspondingly found to be
\begin{equation}
	-K_{S_{0}} = H_x = 2F + 2 \Sigma - \sum_{i=1}^8 M_i   \, .
\end{equation}

We can project $S_{0}$ to e.g. $\P^1_z$ by projecting the ambient space to $\P^1_z$. For a generic point on $\P^1_z$, the fiber is just a smooth $\P^1$
with divisor class $H_z = F$, but over the eight points where $P_0=Q_0=0$ the fiber splits into two components with classes $M_i$ and $F-M_i$. The threefold
\begin{equation}
    X_{S_{0}} := K_{S_{0}} \, ,
\end{equation}
is then readily interpreted as giving rise to an $SU(2)$ gauge theory with eight flavors and Coulomb branch parameters measured as in \eqref{eq: Vol(F)=Coulomb branch} and \eqref{eq: Vol(M)=Coulomb branch}. 

The fiber of a projection to $\P^1$ is an elliptic curve which degenerates over $12$ distinct points and has class $E = -K_{S_{0}}$. This is an elliptic fibration as any of the $M_i$ defines a global section, and we choose $\sigma_0 = M_1$ as the zero section. It follows that the rank of the Mordell-Weil group
\begin{equation}
\text{MW} \simeq H^{2}(S_0,\Z)/\langle \sigma_0, E \rangle \simeq -E_8 \, .
\end{equation}
must be $8$ and that its torsional subgroup must be trivial \cite{oguiso_shioda_rat_ell,schuett2010ellipticsurfaces}. Using the identification \eqref{eq: C vs alpha} we can hence find an associated section 
for any element $\alpha \in \Lambda_{E_8} = \text{span}_\Z \left\{ \alpha_1, \cdots, \alpha_8 \right\}$. 

Following \cite{oguiso_shioda_rat_ell,schuett2010ellipticsurfaces} the class of such a section is determined to be (see Appendix \ref{app:sections} for details)
\begin{equation}\label{eq:intersectionality}
\sigma_{\alpha}:= \sigma_0 + \alpha - \frac{\alpha^2}{2}E \, ,
\end{equation}
and the $\sigma_{\alpha}$ obey
\begin{equation}
\sigma_{\alpha}^2 = -1 \, , \hspace{1cm} \sigma_{\alpha} \cdot E = 1 \, .
\end{equation}
For example, the class $M_i$ is recovered as representing a section from the above by using $\alpha \in \Lambda_{E_8}$ with
\begin{equation}
\alpha = -2 \Sigma - 2 F + M_i - M_1 + \sum_{j=1}^8{M_j} \, .
\end{equation}

The group operation of the Mordell-Weil group is the addition of sections, which takes the form
\begin{equation}
\sigma_\alpha \boxplus \sigma_\beta = \sigma_{\alpha+\beta}\, . 
\end{equation}
These become automorphisms $\tau_\alpha$ of $S_{0}$ which act as a translation on every elliptic fiber. As $\tau_\alpha$ must map $\sigma_0$ to $\sigma_\alpha$, leave $E$ invariant, and act as an isometry on $H^2(S_0,\Z)$, the action of $\tau_\alpha$ is
\begin{equation}
\tau_\alpha:  \qquad 
\left\{\begin{aligned}
\sigma_0 &\rightarrow \sigma_\alpha = \sigma_0 + \alpha - \frac{\alpha^2}{2} E\\
E &\rightarrow E \\
\beta &\rightarrow \beta - (\alpha \cdot \beta) E \, ,
\end{aligned}\right.
\end{equation}
where $\beta$ is any element of $\Lambda_{E_8}$. The above implies that
a shift by $\sigma_\alpha$ acts on any section $\sigma_\beta$ as
\begin{equation}
\tau_\alpha: \qquad  \sigma_\beta \rightarrow \sigma_{\alpha+\beta}\, .
\end{equation}
In other words, lattice translations of $\Lambda_{E_8}$ are mapped by a group homomorphism to translations by a section and hence to automorphisms $\tau_\alpha$ of $S_0$. 

While $H^2(S_{0},\Z)$ has no $(-2)$ classes represented by holomorphic $\P^1$s, $S_0$ contains cycles which are represented by 2-spheres that generate the lattice $\Lambda_{E_8}$. Such cycles can e.g. be explicitely constructed along the lines of \cite{Gaberdiel:1997ud} from the singular fibers of the elliptic fibration on $S_0$. For specific choices of the polynomials $P_0$ and $Q_0$ some singular fibers may coincide and $S_0$ develops a singularity for which a curve $\alpha$ shrinks. Such loci exist in complex codimension one in the space of polynomial deformations, and transporting $S_0$ on a small loop around such a locus induces an automorphism $w_\alpha$ of $S_0$ which acts on any $\Gamma \in H^2(S_{0},\Z)$ by the Picard-Lefschetz transformation
\begin{equation}
 w_\alpha:   \Gamma \rightarrow \Gamma + \alpha \cdot \Gamma\, ,
\end{equation}
for any $\alpha \in \Lambda_{E_8}$. At this point one might still worry that we do not have access to all deformations required due to the specific realization of $S_0$ \eqref{eq:dp9_p1cubed}. However, we can alternatively realize the same rational elliptic surface $S_0$ as a Weierstrass model which gives us access to all required deformations and hence Weyl reflections for all $\alpha \in \Lambda_{E_8}$. Even though a generic $S_0$ will only give a restricted class of Weierstrass models, we can still explore the whole moduli space of Weierstrass models while still starting and ending with any given surface $S_0$, so that we find Weyl transformations for all $\alpha \in \Lambda_{E_8}$ also for $S_0$.

We can now describe the associated non-compact Calabi-Yau threefold $X_{S_{0}}$
as
\begin{equation}
    \begin{pmatrix}
        P_0 & -Q_0 \\
        u & v 
    \end{pmatrix}
    \begin{pmatrix}
        x_1 \\ x_2
    \end{pmatrix} = 0 \, ,
\end{equation}
where $u,v$ are coordinates on $\C^2$. This is a line bundle over $S_0$, and has trivial canonical bundle, so that it must be isomorphic to $K_{S_{0}}$. This threefold has an elliptic fibration over the blowup of $\C^2$ in a point, i.e. the base is given by 
\begin{equation}\label{eq: XS0 base}
    x_1 u + x_2 v = 0 \, ,
\end{equation}
and the compact surface $S_{0}$ is formed of the elliptic fiber over the exceptional $\P^1$ contained in the base. As $X_{S_{0}}$ is nothing but the canonical bundle over $S_{0}$, the automorphisms $\tau_\alpha$ and $w_\alpha$ of $S_{0}$ lift to automorphisms of $X_{S_{0}}$.

While $X_{S_0}$ neatly exhibits the ruling with reducible fibers expected for an $SU(2)$ gauge theory with 8 flavors, it fails to capture the flavor symmetry which emerges whenever some of the mass parameters $m_i$ are equal. The K\"ahler form \eqref{eq: correct Kahler form} was already written in a basis of divisors which are manifestly effective in $X_{S_0}$ and reproduces the correct prepotential when taken as the K\"ahler form for $X_{S_0}$. Setting $m_i = m_j$ for all $i,j = 1,\cdots,8$, all of the classes $\beta_i$ produce zero when integrating the K\"ahler form, so appear to have vanishing volume. However on $S_0$ and $X_{S_0}$, the $\beta_i$ do not correspond to effective curves, so that we cannot conclude that setting mass parameters equal produces any singularities. In fact, thinking of $S_0$ as a blowup of $\P^1 \times \P^1$ in eight points with exceptional divisors $M_i$, setting $m_i = m_j$ for all $i,j$ merely leads to all exceptional divisors having equal volume.

\subsection{The dP$_9$ Surface $S$ } \label{sec: Geometry of dP9 S}

In order to work in a model that at least makes an $SU(8)$ flavor symmetry manifest, we now tune the complex structure of $S_0$ to find a surface $S$ with an $I_8$ fiber (together with four $I_1$ fibers). This can be achieved by setting
\begin{equation}\label{eq:dp9_p1cubedv2}
    S: \qquad P(y,z) x_1 = Q(y,z) x_2 \, ,
\end{equation}
with 
\begin{equation}
    \begin{aligned}
        P(y,z)&= y_1^2 z_2^2 - z_1^2 y_2^2 - 4 y_2^2 z_2^2 \\
        Q(y,z)&= y_1 y_2 (z_1^2 +z_2^2) + z_1 z_2 (y_1^2 - y_2^2) \, .
    \end{aligned}
\end{equation}

Working out the associated Weierstrass model results in
\begin{equation}
    Y^2 = X^3 + f X + g \, ,
\end{equation}
with 
\begin{equation}
    \begin{aligned}
        f & = -432(x_1^4 + 3 x_2^4) \\
        g & = 1728(2 x_1^4 + 9 x_2^4) x_1^2 \, ,
    \end{aligned}
\end{equation}
so that the discriminant is
\begin{equation}
    \Delta = x_2^8 (x_1^2 - 2 x_1 x_2 + 2 x_2^2) (x_1^2 + 2 x_1 x_2  + 2 x_2^2)\, .
\end{equation}
This shows four fibers of type $I_1$ and one of type $I_8$. One can indeed check that $S$ has a singularity at $x_2 = z_2 = y_2 = 0$, while the fiber degenerates over $x_2= 0$ and the four other points indicated by the discriminant above.  

Upon resolution (which we have control over via the K\"ahler form) the zero section $\sigma_0 = M_1$ persists and the $\beta_i, i = 1, \cdots, 7$, together with $E + M_1 - M_8$  become the fiber components of the $I_8$ fiber and hence effective curves of $S$. Here, the last component being $E + M_1 - M_8$ is forced on us by the requirement that the classes of all fiber components of an $I_8$ fiber sum to the class $E$. Note that $\beta_1 = M_2 - M_1$ now plays the role of the `trivial' fiber component meeting the zero section. 
Since 
\begin{equation}
 X^3 + f X + g =   (X-12 x_1^2) (-288 x_1^4 - 1296 x_2^4 + 12 X x_1^4 + X^2 ) \,,
\end{equation}
factors, we can conclude that $S$ has a torsional $\Z_2$ section, so that the Mordell-Weil group must be $\Z\times \Z_2$. The Mordell-Weil groups of elliptic fibrations on rational elliptic surfaces with various collections of fibers have been studied extensively in the literature, and there are two inequivalent cases with four $I_1$ and one $I_8$ fibers which differ in the embedding of the components of the $I_8$ fiber into $E_8$ \cite{oguiso_shioda_rat_ell, schuett2010ellipticsurfaces}, and in their Mordell-Weil groups. Only one of these has a torsional Mordell-Weil group, and it is distinguished by the lattice of fiber components not meeting the zero section not being primitively embedded into $H_2(S,\Z)$. The assignment \eqref{eq:Eso16roots_so16weights} exactly corresponds to this case, as it implies that the embedding of $\langle \beta_2, \cdots \beta_7, E+M_1-M_8 \rangle_\Z$ into $H_2(S,\Z)$ is not primitive.\footnote{This can be seen by working out 
\begin{equation}
    E + M_1 - M_8 +\beta_2 + \beta_4 +\beta_6 
    = 2F + 2 \Sigma + 2M_3 + 2M_5 + 2M_6
\end{equation}
which can be divided by two in $H_2(S,\Z)$ but not in $\langle \beta_2, \cdots \beta_7, E+M_1-M_8 \rangle_\Z$.}
Hence \eqref{eq:dp9_p1cubedv2} has a fiber of type $I_8$ with (affine) components precisely corresponding to the $SU(8)$ flavor symmetry identified in \eqref{eq:Eso16roots_so16weights}. 

As the Mordell-Weil group of $S$ is $\Z\times \Z_2$, not all of the sections identified in $S_{0}$ can persist when specializing $P_0$ and $Q_0$ to $P$ and $Q$. One can easily check that the eight points where $P_0=Q_0=0$ all become coincident for $P=Q = 0$, so that all of the sections $M_i$ become coincident. By working out the proper transform of sections one can verify which sections survive, however it is easier to invoke the general theory about sections of rational elliptic surface as laid out in 
\cite{schuett2010ellipticsurfaces}. The upshot is that (see Appendix \ref{app:sections} for details) the Mordell-Weil group retains the sections\footnote{The Mordell-Weil group of $S$ contains other sections beyond
$\sigma_{\ell \alpha_1}$, but we'll only require these in the following.} 
\begin{equation}
    \sigma_{\ell \alpha_1} =  M_1 + \ell(\Sigma - F) + \ell^2 E
\end{equation}
for all $\ell \in \mathbb{Z}$ when deforming $S_0$ to $S$. Note that the zero section is still $\sigma_0 = M_1$. 

For any elliptic surface, the existence of non-trivial sections implies the existence of automorphisms which act fiberwise as translations. Here, the action on classes in $H^2(S,\mathbb{Z})$ is determined by the induced action on irreducible fiber components. As the sections $\sigma_{\ell \alpha_1}$ intersect the same fiber component $\beta_1 = M_2 - M_1$ of the $I_8$ fiber as $\sigma_0$, it turns out that the associated automorphisms $\tau_{\ell \alpha_1}$ must act trivially on the irreducible fiber components and, correspondingly, the lattice $\Lambda_{A_7}$. Furthermore, $\tau_{\ell \alpha_1}$ also translates any of the sections $\sigma_{k \alpha_1}$ to $\sigma_{(k + \ell) \alpha_1}$, so that we find the action on $H^2(S,\mathbb{Z})$ to be 
\begin{equation}
  \tau_{\ell \alpha_1}:\qquad \left\{\begin{aligned}
  \sigma_{k \alpha_1} &\rightarrow \sigma_{(k+\ell) \alpha_1} \\
      E & \rightarrow E \\
      \beta_i & \rightarrow \beta_i \quad i = 1, \cdots,7 \, .
  \end{aligned} \right.
\end{equation}
We can rewrite the action of $\tau_{\ell\alpha_1}$ in the form 
\begin{equation}\label{eq: translation by sigma alpha 1}
    \tau_{\ell \alpha_1}: \gamma \rightarrow \gamma - \ell (\alpha_1 \cdot \gamma)E + \ell(E \cdot \gamma)\alpha_1 + \ell^2(E \cdot \gamma)E \, ,
\end{equation}
for any $\gamma \in H^2(S,\Z)$. These maps form an abelian group with composition
\begin{equation}
    \tau_{\ell \alpha_1} \circ \tau_{\ell' \alpha_1} = \tau_{(\ell +\ell')\alpha_1}\, ,
\end{equation}
and lattice automorphisms of this form are called Eichler-Siegel transformations in the literature, see \cite{heckman2001moduli} for a discussion in the present context. 

The $\tau_{\ell \alpha_1}$ are not the only class of automorphism we can find for $S$. As discussed for $S_0$, we may have automorphisms acting on the homology of $S$ by Picard-Lefschetz transformations. However, deformations now only affect cycles corresponding to roots in $\Lambda_{E_8}$ which are not contained in the lattice $\Lambda_{A_7}$ spanned by fiber components. Hence only Picard-Lefschetz transformations $w_\alpha$ for all $\alpha \in \Lambda_{E_8}\setminus \Lambda_{A_7}$ persist as automorphisms of $S$.

\subsection{The Threefolds $X_S$ and $X_S^f$ } \label{sec: Geometry of XS}

As done for $S_0$, we now define the threefold $X_S$ as the
canonical bundle over $S$, which we can realize concretely as
\begin{equation}\label{eq: def of XS}
  X_S: \qquad  \begin{pmatrix}
        P & -Q \\
        u & v 
    \end{pmatrix}
    \begin{pmatrix}
        x_1 \\ x_2
    \end{pmatrix} = 0 \, ,
\end{equation}
with $u,v$ coordinates on $\C^2$. 

For the purposes of geometric engineering, we will be interested in a phase where all the $M_i$ are flopped out of the surface $S$. We can blow down all of the $M_i$ in $X_S$ by passing to the threefold $\tilde{X}_S$ defined by 
\begin{equation}
    \tilde{X}_S: \qquad P v = Q u \, .
\end{equation}
Note that this takes exactly the same form as $S$, except that now $u,v$ are coordinates in $\C^2$ rather than $\P^1$. With the projection to $(u,v)$ we find the same elliptic fibers over rays in $\C^2$ as we did over points of the $\P^1$ base of $S$, but the fibration becomes non-flat at $u=v=0$ where the fiber becomes $\P^1_y \times \P^1_z$. There is an $I_8$ fiber located over the locus $u=0$ with a singularity enhancement at the point $u=v=P=Q=0$ due to the $M_i$ having been blown down.

We can reach the flopped phase as
\begin{equation}
  X_S^f: \qquad  \begin{pmatrix}
        P & u \\
        -Q & v 
    \end{pmatrix}
    \begin{pmatrix}
        x_1 \\ x_2
    \end{pmatrix} = 0 \, ,
\end{equation}
where now the coordinates $[x_1:x_2]$ are sections of $\mathcal{O}$ and $\mathcal{O}_{\P^1_y}(2)\otimes \mathcal{O}_{\P^1_z}(2)$ instead. This fibration is also not flat, but has a copy of $\P^1_y \times \P^1_z$ located at $u=v=0$. 

On $X_S$, every curve $ \mathcal{C}$ on $S$ defines a non-compact divisor given by $K_S|_{\mathcal{C}}$. Alternatively, one may think of these divisors as appropriate intersections of \eqref{eq: def of XS} with the ambient space. The second description persists for $X_S^f$, but it becomes harder to describe these divisors concretely. Note however, that a given curve class on $S$ uniquely determines a corresponding divisor class on $X_S^f$. We can hence continue to label divisor classes on $X_S^f$ in terms of elements of $H_2(S,\Z) \simeq I_{1,9}$, which is what we will do in the following.  

As argued above, $S$ admits automorphisms $\tau_{\ell \alpha_1}$ for all $\ell \in \Z$ and $w_\alpha$ for all $\alpha \in \Lambda_{E_8}\setminus \Lambda_{A_7}$. The threefold $X_S$ is just the canonical bundle over $S$, so that these clearly lift to automorphisms of $X_S$. 

In order to flop $X_S$ to $X_S^f$ we first collapsed the curves $M_i$ to find $\tilde{X}_S$, but we might as well have collapsed the image of $M_i$ under any of the automorphisms of $X_S$ to produce a space which is isomorphic to $\tilde{X}_S$. By the same logic, we can flop any set of 8 $(-1)$-curves which is found by acting with an automorphism of $X_S$ on the $M_i$ to find an isomorphic space. This shows that the automorphisms $\tau_{\alpha_1}$ and $w_\alpha$ for all $\alpha \in \Lambda_{E_8}\setminus \Lambda_{A_7}$ also become automorphisms of $X_S^f$. In the end, the structure of $S$ uniquely determines both $X_S$ and $X_S^f$, so that it is not surprising that symmetries of $S$ become symmetries of both $X_S$ and $X_S^f$.

\subsection{Generalized Box Graphs and Automorphisms of $H_2(S,\Z)$} \label{sec: box graphs}

Using the generalized box graphs introduced in~\cite{Sabag:2022hyw}, we now investigate dual descriptions of the field theory and the implied automorphisms of $H^2(S,\Z) \simeq I_{1,9}$. We will formulate these in terms of the fiber $F$ and section $\Sigma$ of the ruling, and the exceptional divisors $M_i$. As the action of these relates dual field theories, it is reasonable to suspect (and obvious in some cases) that these arise from automorphisms of the geometry. As we do not ultimately require such a statement for our construction, we will not spell out this aspect in detail.

\paragraph{Box Graphs.}
Box graphs provide a diagrammatic way to understand Coulomb branch phases using a representation theoretic graphical notation ~\cite{Hayashi:2014kca,Hayashi:2013lra, Braun:2014kla,Braun:2015hkv,Apruzzi:2019enx}, which have a direct geometric Calabi-Yau interpretation: they model the relative K\"ahler cones. For 5d SCFTs these were used to e.g. classify all the consistent 5d SCFTs and identify their gauge theory and geometric descriptions in \cite{Apruzzi:2019opn,Apruzzi:2019enx,Apruzzi:2019vpe,Apruzzi:2019kgb}. 

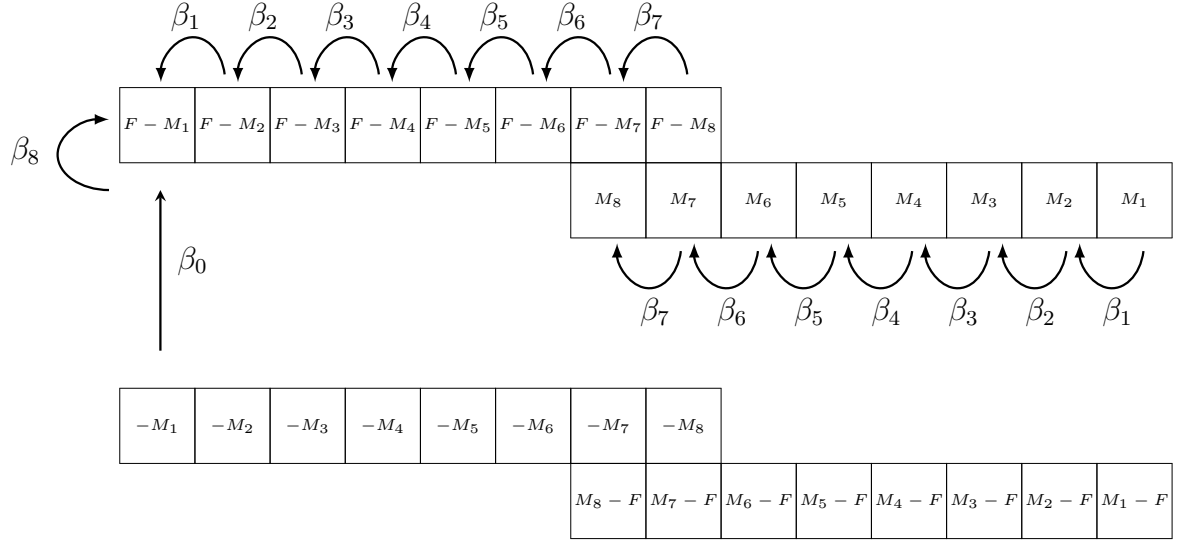
\begin{figure}
\center
\begin{tikzpicture}[baseline,scale=0.85]
\tikzstyle{every node}=[font=\scriptsize]

\node[] (5dLyt) at (0,0) {\ytableausetup{centertableaux,boxsize=28pt}
{\tiny\begin{ytableau}
 F-M_1 &  F-M_2 &  F-M_3 &  F-M_4 &  F-M_5 &  F-M_6 & F-M_7 & F-M_8 \\
\none[] & \none[] & \none[] & \none[] & \none[] & \none[] & M_8 & M_7 & M_6 & M_5 & M_4 & M_3 & M_2 & M_1 \\ \none[] \\ \none[] \\
 -M_1 & -M_2 & -M_3 & -M_4 & -M_5 & -M_6 & -M_7 & -M_8 \\
\none[] & \none[] & \none[] & \none[] & \none[] & \none[] & M_8 -F& M_7 -F& M_6 -F& M_5 -F& M_4 -F& M_3 -F& M_2 -F& M_1 -F\\ \none[] \\
\end{ytableau}}};

\draw[-latex,black,line width=0.3mm] (-8.3,2.5) arc
    [start angle=270,end angle=90,x radius=0.8cm,y radius =0.55cm] ;
\node[] (so16root8) at (-9.6,3.1) {\fontsize{12pt}{12pt}$\beta_8$};

\draw[-latex,black,line width=0.3mm] (-6.5,4.3) arc
    [start angle=10,end angle=180,x radius=0.5cm,y radius =0.7cm] ;
\draw[-latex,black,line width=0.3mm] (-5.3,4.3) arc
    [start angle=10,end angle=180,x radius=0.5cm,y radius =0.7cm] ;
\draw[-latex,black,line width=0.3mm] (-4.1,4.3) arc
    [start angle=10,end angle=180,x radius=0.5cm,y radius =0.7cm] ;
\draw[-latex,black,line width=0.3mm] (-2.9,4.3) arc
    [start angle=10,end angle=180,x radius=0.5cm,y radius =0.7cm] ;
\draw[-latex,black,line width=0.3mm] (-1.7,4.3) arc
    [start angle=10,end angle=180,x radius=0.5cm,y radius =0.7cm] ;
\draw[-latex,black,line width=0.3mm] (-0.5,4.3) arc
    [start angle=10,end angle=180,x radius=0.5cm,y radius =0.7cm] ;
\draw[-latex,black,line width=0.3mm] (0.7,4.3) arc
    [start angle=10,end angle=180,x radius=0.5cm,y radius =0.7cm] ;

\node[] (so16root1) at (-7.1,5.2) {\fontsize{12pt}{12pt}$\beta_1$};
\node[] (so16root2) at (-5.9,5.2) {\fontsize{12pt}{12pt}$\beta_2$};
\node[] (so16root3) at (-4.7,5.2) {\fontsize{12pt}{12pt}$\beta_3$};
\node[] (so16root4) at (-3.5,5.2) {\fontsize{12pt}{12pt}$\beta_4$};
\node[] (so16root5) at (-2.3,5.2) {\fontsize{12pt}{12pt}$\beta_5$};
\node[] (so16root6) at (-1.1,5.2) {\fontsize{12pt}{12pt}$\beta_6$};
\node[] (so16root7) at (0.1,5.2) {\fontsize{12pt}{12pt}$\beta_7$};

\draw[-latex,black,line width=0.3mm] (7.8,1.55) arc
    [start angle=350,end angle=180,x radius=0.5cm,y radius =0.7cm] ;
\draw[-latex,black,line width=0.3mm] (6.6,1.55) arc
    [start angle=350,end angle=180,x radius=0.5cm,y radius =0.7cm] ;
\draw[-latex,black,line width=0.3mm] (5.4,1.55) arc
    [start angle=350,end angle=180,x radius=0.5cm,y radius =0.7cm] ;
\draw[-latex,black,line width=0.3mm] (4.2,1.55) arc
    [start angle=350,end angle=180,x radius=0.5cm,y radius =0.7cm] ;
\draw[-latex,black,line width=0.3mm] (3,1.55) arc
    [start angle=350,end angle=180,x radius=0.5cm,y radius =0.7cm] ;
\draw[-latex,black,line width=0.3mm] (1.8,1.55) arc
    [start angle=350,end angle=180,x radius=0.5cm,y radius =0.7cm] ;
\draw[-latex,black,line width=0.3mm] (0.6,1.55) arc
    [start angle=350,end angle=180,x radius=0.5cm,y radius =0.7cm] ;

\node[] (so16root1) at (7.4,0.6) {\fontsize{12pt}{12pt}$\beta_1$};
\node[] (so16root2) at (6.2,0.6) {\fontsize{12pt}{12pt}$\beta_2$};
\node[] (so16root3) at (5,0.6) {\fontsize{12pt}{12pt}$\beta_3$};
\node[] (so16root4) at (3.8,0.6) {\fontsize{12pt}{12pt}$\beta_4$};
\node[] (so16root5) at (2.6,0.6) {\fontsize{12pt}{12pt}$\beta_5$};
\node[] (so16root6) at (1.4,0.6) {\fontsize{12pt}{12pt}$\beta_6$};
\node[] (so16root7) at (0.2,0.6) {\fontsize{12pt}{12pt}$\beta_7$};

\draw[-stealth,black,line width=0.3mm] (-7.5,0) -- (-7.5,2.5);
\node[] (su2root) at (-7,1.4) {\fontsize{12pt}{12pt}$\beta_0$};

\end{tikzpicture}
\caption{The weight diagram of the 5d KK-theory generated by the reduction of the 6d rank-1 E-string theory on a circle corresponding to the effective $\mathfrak{su}(2) +8F$ gauge description with symmetry $\mathfrak{su}(2) \oplus \mathfrak{so}(16)$. The weights are labeled by their corresponding curve classes.}
\label{fig: E-stringBG}
\end{figure}

The box graphs are essentially weight diagrams for the gauge and flavor symmetries of a 5d effective theory, decorated by signs, which are chosen such that the resulting sign-dressed weights form a consistent extended Coulomb branch chamber. Equivalently, in the geometry, the sign assignment for a curve associated to a weight $w$, $C_w$, indicates which of $\pm C_w$ is an effective curve. 
The box graphs are an efficient way to determine, using simple rules, all consistent sign assignments. To make the connection to the geometry explicit we use the identification between the geometric curves and the $\mathfrak{su}(2)\oplus \mathfrak{so}(16)$ weights~\eqref{eq: N map applied to effective curves} and roots~\eqref{eq:Eso16roots_so16weights} to label our box graphs.

The box graphs for the rank-1 E-string were discussed in \cite{Apruzzi:2019opn,Apruzzi:2019enx}.
The underlying weight diagram for the $(\mathbf{2},\mathbf{16})$ representation of the gauge and flavor symmetry $\mathfrak{su}(2) \oplus \mathfrak{so}(16)$ is shown in  Figure~\ref{fig: E-stringBG}. Adjacent boxes in the graph are connected by the addition or subtraction of simple roots $\beta_i$ as defined in~\eqref{eq:Eso16roots_so16weights}.

Box graphs are obtained by decorating each box, i.e.~$\mathfrak{so}(2) \oplus \mathfrak{so}(16)$ 
weight $w$, with a sign/color depending on the sign of 
$\langle (\phi, m_i), w \rangle $, 
where $(\phi,m_i)$ is a vector of extended Coulomb branch parameters, projected onto the weight $w$. 
The rules for the sign assignments of the standard box graph are simple: since weights to the left or above a box are reached by the addition of a (positive) simple root, boxes to the left or above a box with positive sign must have positive sign. The signs are indicated in terms of $\pm$ sign blue/yellow.

In the geometry we choose the K\"ahler form~\eqref{eq: correct Kahler form} as the vector of extended Coulomb branch data with which to color the boxes. In this way, the sign of a weight $w$ is determined geometrically by
\be
    \text{sign}(J \cdot C_w) \, ,
\ee
where $C_w$ is the corresponding curve. Thus the coloring of the boxes corresponds to which curves are effective: curves with positive volume are colored blue, and flopped curves with negative volume are colored yellow. 
For example, in the phase in which $-m_i<\phi<m_i$ and $\phi=0$, the curves $F-M_i$ are effective, and $M_i$ are at negative volume $-m_i<0$. So, for this gauge theory phase, we get the coloring of the decorated box graph shown in Figure~\ref{Fig: box graph first phase}.
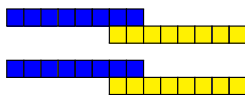
\begin{figure}
\centering
\begin{tikzpicture}
\node[] (5dLyt) at (-3,-2) {\ytableausetup{centertableaux,boxsize=6pt}
\begin{ytableau}
*(blue) & *(blue) & *(blue) & *(blue) & *(blue) & *(blue) & *(blue) & *(blue) \\
\none[] & \none[] & \none[] & \none[] & \none[] & \none[] & *(yellow) & *(yellow) & *(yellow) & *(yellow) & *(yellow) & *(yellow) & *(yellow) & *(yellow) \\ \none[] \\
*(blue) & *(blue) & *(blue) & *(blue) & *(blue) & *(blue) & *(blue) & *(blue) \\
\none[] & \none[] & \none[] & \none[] & \none[] & \none[] & *(yellow) & *(yellow) & *(yellow) & *(yellow) & *(yellow) & *(yellow) & *(yellow) & *(yellow) \\
\end{ytableau}};
\end{tikzpicture}
\caption{Box graph for the gauge theory phase $-m_i<\phi<m_i$, $\phi=0$.  }
\label{Fig: box graph first phase}
\end{figure}

\paragraph{Generalized Box Graphs.}
The generalized box graphs \cite{Sabag:2022hyw} appear when one tries to compare the box graphs for two different gauge theory phases related by a Weyl transformation, whilst keeping a fixed basis of weights. 
As such, the generalized box graph needs not obey the standard sign assignments of the box graph in that fixed basis of weights, but should map to a standard box graph after  relabeling by the appropriate Weyl action. 
More specifically, since the box graph in any phase of the ECB is defined modulo the Weyl action of the group, we only consider one Weyl chamber representative. 
If we fix the Weyl chamber for the theory in an initial gauge theory phase, then we have to choose a different Weyl chamber for the theory in another phase due to the non-trivial identification of ECB parameters. Examples of such generalized box graphs are shown in Figures~\ref{Fig: box diagram n4 second region} and~\ref{Fig: box diagram n3 second region}. 
The transformation also acts on the simple roots $\beta_i$, since they are built from the fundamental weights. Their image under this transformation will be denoted $\beta_i'$.
The weights in the transformed box graph are connected by the addition or subtraction of the transformed simple roots, not of the simple roots in the old basis. 

The sign assignment/coloring is of course consistent in the new basis of roots. 
We illustrate this with an example: look at the gauge theory phases in Figure~\ref{Fig: box diagram n3 second region}.
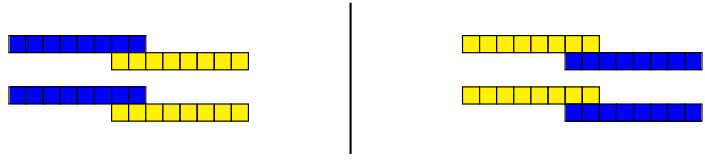
\begin{figure}
\centering
\begin{tikzpicture}
\node[] (5dLyt) at (-3,-2) {\ytableausetup{centertableaux,boxsize=6pt}
\begin{ytableau}
*(blue) & *(blue) & *(blue) & *(blue) & *(blue) & *(blue) & *(blue) & *(blue) \\
\none[] & \none[] & \none[] & \none[] & \none[] & \none[] & *(yellow) & *(yellow) & *(yellow) & *(yellow) & *(yellow) & *(yellow) & *(yellow) & *(yellow) \\ \none[] \\
*(blue) & *(blue) & *(blue) & *(blue) & *(blue) & *(blue) & *(blue) & *(blue) \\
\none[] & \none[] & \none[] & \none[] & \none[] & \none[] & *(yellow) & *(yellow) & *(yellow) & *(yellow) & *(yellow) & *(yellow) & *(yellow) & *(yellow) \\
\end{ytableau}};
    \node[] (5dLyt) at (3,-2) {\ytableausetup{centertableaux,boxsize=6pt}
\begin{ytableau}
*(yellow) & *(yellow) & *(yellow) & *(yellow) & *(yellow) & *(yellow) & *(yellow) & *(yellow) \\
\none[] & \none[] & \none[] & \none[] & \none[] & \none[] & *(blue) & *(blue) & *(blue) & *(blue) & *(blue) & *(blue) & *(blue) & *(blue) \\ \none[] \\
*(yellow) & *(yellow) & *(yellow) & *(yellow) & *(yellow) & *(yellow) & *(yellow) & *(yellow) \\
\none[] & \none[] & \none[] & \none[] & \none[] & \none[] & *(blue) & *(blue) & *(blue) & *(blue) & *(blue) & *(blue) & *(blue) & *(blue) \\
\end{ytableau}};
\draw[black, thick] (0,-3) -- (0,-1);
\end{tikzpicture}
\caption{Box graph description of the $\Phi$-gluing procedure of field theory. Both the matter curves $M_i$ and fiber $F$ are collapsed at the location of the gluing. The $2n=8$ curves $F-M_i$ are flopped out of the geometry and the 8 curves $M_i$ emerge from zero volume after passing through the gluing location.}
\label{Fig: box diagram n4 second region}
\end{figure}
\begin{figure}
\centering
\begin{tikzpicture}
\node[] (5dLyt) at (-3,-2) {\ytableausetup{centertableaux,boxsize=6pt}
\begin{ytableau}
*(blue) & *(blue) & *(blue) & *(blue) & *(blue) & *(blue) & *(blue) & *(blue) \\
\none[] & \none[] & \none[] & \none[] & \none[] & \none[] & *(yellow) & *(yellow) & *(yellow) & *(yellow) & *(yellow) & *(yellow) & *(yellow) & *(yellow) \\ \none[] \\
*(blue) & *(blue) & *(blue) & *(blue) & *(blue) & *(blue) & *(blue) & *(blue) \\
\none[] & \none[] & \none[] & \none[] & \none[] & \none[] & *(yellow) & *(yellow) & *(yellow) & *(yellow) & *(yellow) & *(yellow) & *(yellow) & *(yellow) \\
\end{ytableau}};
\node[] (5dLyt) at (3,-2) {\ytableausetup{centertableaux,boxsize=6pt}
\begin{ytableau}
*(yellow) & *(yellow) & *(yellow) & *(yellow) & *(yellow) & *(yellow) & *(blue) & *(blue) \\
\none[] & \none[] & \none[] & \none[] & \none[] & \none[] & *(yellow) & *(yellow) & *(blue) & *(blue) & *(blue) & *(blue) & *(blue) & *(blue) \\ \none[] \\
*(yellow) & *(yellow) & *(yellow) & *(yellow) & *(yellow) & *(yellow) & *(blue) & *(blue) \\
\none[] & \none[] & \none[] & \none[] & \none[] & \none[] & *(yellow) & *(yellow) & *(blue) & *(blue) & *(blue) & *(blue) & *(blue) & *(blue) \\
\end{ytableau}};
\draw[black, thick] (0,-3) -- (0,-1);
\end{tikzpicture}
\caption{Box graph description of the mixed-gluing procedure of field theory, where we have performed a $\Phi$ gluing for $2n=6$ of the chirals. Both the matter curves $M_{1,\cdots ,6}$ and fiber $F$ are collapsed at the location of the gluing. The 6 curves $F-M_{1,\cdots ,6}$ are flopped out of the geometry and the 6 curves $M_{1,\cdots ,6}$ emerge from zero volume after passing through the wall. The curves $F-M_{7,8}$ and $M_{7,8}$ remain invariant after passing through the gluing location.}
\label{Fig: box diagram n3 second region}
\end{figure}
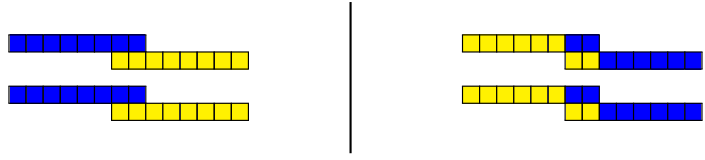
The gauge theory phase on the left appears consistent with the box graph coloring, but the right does not. However, observe that in terms of the transformed simple roots, the undecorated weight diagram takes the form shown in Figure~\ref{fig: E-stringBGn3}.

\begin{figure}[t]
\center
\begin{tikzpicture}[baseline,scale=0.85]
\tikzstyle{every node}=[font=\scriptsize]

\node[] (5dLyt) at (0,0) {\ytableausetup{centertableaux,boxsize=28pt}
{\tiny\begin{ytableau}
 F-M_1 &  F-M_2 &  F-M_3 &  F-M_4 &  F-M_5 &  F-M_6 & F-M_7 & F-M_8 \\
\none[] & \none[] & \none[] & \none[] & \none[] & \none[] & M_8 & M_7 & M_6 & M_5 & M_4 & M_3 & M_2 & M_1 \\ \none[] \\ \none[] \\
 -M_1 & -M_2 & -M_3 & -M_4 & -M_5 & -M_6 & -M_7 & -M_8 \\
\none[] & \none[] & \none[] & \none[] & \none[] & \none[] & M_8 -F& M_7 -F& M_6 -F& M_5 -F& M_4 -F& M_3 -F& M_2 -F& M_1 -F\\ \none[] \\
\end{ytableau}}};

\draw[-latex,black,line width=0.3mm] (-8.3,2.5) arc
    [start angle=270,end angle=90,x radius=0.8cm,y radius =0.55cm] ;
\node[] (so16root8) at (-9.6,3.1) {\fontsize{12pt}{12pt}$\beta_8'$};

\draw[latex-,black,line width=0.3mm] (-6.5,4.2) arc
    [start angle=0,end angle=170,x radius=0.5cm,y radius =0.7cm] ;
\draw[latex-,black,line width=0.3mm] (-5.3,4.2) arc
    [start angle=0,end angle=170,x radius=0.5cm,y radius =0.7cm] ;
\draw[latex-,black,line width=0.3mm] (-4.1,4.2) arc
    [start angle=0,end angle=170,x radius=0.5cm,y radius =0.7cm] ;
\draw[latex-,black,line width=0.3mm] (-2.9,4.2) arc
    [start angle=0,end angle=170,x radius=0.5cm,y radius =0.7cm] ;
\draw[latex-,black,line width=0.3mm] (-1.7,4.2) arc
    [start angle=0,end angle=170,x radius=0.5cm,y radius =0.7cm] ;
\draw[-latex,black,line width=0.3mm] (0.7,4.3) arc
    [start angle=10,end angle=180,x radius=0.5cm,y radius =0.7cm] ;

\draw[-latex,black, line width=0.3mm] (-1.6,2.85) .. controls (-1.2,-1) and (0.4,-1) .. (0.8,1.6);

\node[] (so16root1) at (-7.1,5.2) {\fontsize{12pt}{12pt}$\beta_1'$};
\node[] (so16root2) at (-5.9,5.2) {\fontsize{12pt}{12pt}$\beta_2'$};
\node[] (so16root3) at (-4.7,5.2) {\fontsize{12pt}{12pt}$\beta_3'$};
\node[] (so16root4) at (-3.5,5.2) {\fontsize{12pt}{12pt}$\beta_4'$};
\node[] (so16root5) at (-2.3,5.2) {\fontsize{12pt}{12pt}$\beta_5'$};
\node[] (so16root6) at (-1.7,0.6) {\fontsize{12pt}{12pt}$\beta_6'$};
\node[] (so16root7) at (0.1,5.2) {\fontsize{12pt}{12pt}$\beta_7'$};
\node[] (so16root6) at (1.7,5.2) {\fontsize{12pt}{12pt}$\beta_6'$};

\draw[latex-,black,line width=0.3mm] (7.8,1.65) arc
    [start angle=360,end angle=190,x radius=0.5cm,y radius =0.7cm] ;
\draw[latex-,black,line width=0.3mm] (6.6,1.65) arc
    [start angle=360,end angle=190,x radius=0.5cm,y radius =0.7cm] ;
\draw[latex-,black,line width=0.3mm] (5.4,1.65)         arc [start angle=360,end angle=190,x               radius=0.5cm,y radius =0.7cm] ;
\draw[latex-,black,line width=0.3mm] (4.2,1.65) arc
    [start angle=360,end angle=190,x radius=0.5cm,y radius =0.7cm] ;
\draw[latex-,black,line width=0.3mm] (3,1.65) arc
    [start angle=360,end angle=190,x radius=0.5cm,y radius =0.7cm] ;
\draw[-latex,black,line width=0.3mm] (0.6,1.55) arc
    [start angle=350,end angle=180,x radius=0.5cm,y radius =0.7cm] ;

\draw[-latex,black, line width=0.3mm] (-1.6,2.85) .. controls (-1.2,-1) and (0.4,-1) .. (0.8,1.6);

\draw[-latex,black, line width=0.3mm] (-0.8,4.2) .. controls (-0.4,6.8) and (1.2,6.8) .. (1.8,3.1);

\node[] (so16root1) at (7.4,0.6) {\fontsize{12pt}{12pt}$\beta_1'$};
\node[] (so16root2) at (6.2,0.6) {\fontsize{12pt}{12pt}$\beta_2'$};
\node[] (so16root3) at (5,0.6) {\fontsize{12pt}{12pt}$\beta_3'$};
\node[] (so16root4) at (3.8,0.6) {\fontsize{12pt}{12pt}$\beta_4'$};
\node[] (so16root5) at (2.6,0.6) {\fontsize{12pt}{12pt}$\beta_5'$};
\node[] (so16root7) at (0.2,0.6) {\fontsize{12pt}{12pt}$\beta_7'$};

\draw[-stealth,black,line width=0.3mm] (-7.5,0) -- (-7.5,2.5);
\node[] (su2root) at (-7,1.4) {\fontsize{12pt}{12pt}$\beta_0'$};

\end{tikzpicture}
\caption{The weight diagram of the gauge theory phase in the right of Figure~\ref{Fig: box diagram n3 second region}. Using this basis of roots, one sees that the coloring of the box graph is compatible with this basis of simple roots. }
\label{fig: E-stringBGn3}
\end{figure}
Starting with the blue weight $F-M_8$ of the gauge theory phase in the right of Figure~\ref{Fig: box diagram n3 second region}, we see that adding the positive roots $\beta_{i}'$ we reach the other blue boxes. Similarly, starting from the yellow box $F-M_1$, one sees that adding the positive roots $\beta_i'$ we reach the other yellow boxes. Finally, adding $\beta_8'=\beta_8$ maps us from the yellow colored weights $M_8$ and $M_7$ to the blue weights $F-M_7$ and $F-M_8$. Hence we see that the coloring of the weight diagram when using the new simple roots is consistent with the box graph coloring prescription.

The logic to derive automorphisms of $H^2(S,\Z)$ is now as follows. First, we construct transformations that act on the curves $F$ and $M_i$ that geometrically replicate the transformation on the box graph. Next, we lift that transformation to an automorphism of the homology lattice $H_2(X_S,\mathbb{Z})$ by demanding the transformed curves preserve self-intersection number and intersection with $E$. Note that intersections between curves in $X_S$ are defined by their intersection when restricted to $S$. Using this, we will be able to construct the compatible action of the box graph transformation on the section $\Sigma$ that lifts to an automorphism. 

The first type of transformation acts on the box graph coloring as 
\be \label{Fig: box graph for mus}
\begin{tikzpicture}
\node[] (5dLyt) at (-3,-2) {\ytableausetup{centertableaux,boxsize=6pt}
\begin{ytableau}
*(blue) & *(blue) & *(blue) & *(blue) & *(blue) & *(blue) & *(blue) & *(blue) \\
\none[] & \none[] & \none[] & \none[] & \none[] & \none[] & *(yellow) & *(yellow) & *(yellow) & *(yellow) & *(yellow) & *(yellow) & *(yellow) & *(yellow) \\ \none[] \\
*(blue) & *(blue) & *(blue) & *(blue) & *(blue) & *(blue) & *(blue) & *(blue) \\
\none[] & \none[] & \none[] & \none[] & \none[] & \none[] & *(yellow) & *(yellow) & *(yellow) & *(yellow) & *(yellow) & *(yellow) & *(yellow) & *(yellow) \\
\end{ytableau}};
\node[] (5dLyt) at (3,-2) {\ytableausetup{centertableaux,boxsize=6pt}
\begin{ytableau}
*(blue) & *(blue) & *(blue) & *(blue) & *(blue) & *(blue) & *(blue) & *(blue) \\
\none[] & \none[] & \none[] & \none[] & \none[] & \none[] & *(yellow) & *(yellow) & *(yellow) & *(yellow) & *(yellow) & *(yellow) & *(yellow) & *(yellow) \\ \none[] \\
*(blue) & *(blue) & *(blue) & *(blue) & *(blue) & *(blue) & *(blue) & *(blue) \\
\none[] & \none[] & \none[] & \none[] & \none[] & \none[] & *(yellow) & *(yellow) & *(yellow) & *(yellow) & *(yellow) & *(yellow) & *(yellow) & *(yellow) \\
\end{ytableau}};
\draw[black, thick] (0,-3) -- (0,-1);
\end{tikzpicture}
\ee
There are four obvious transformations which would map us from the left box graph to the right box graph in \eqref{Fig: box graph for mus}:
\begin{enumerate}
    \item The identity automorphism. 
    \item Interchange the top and bottom sets of $\mathfrak{so}(16)$ weights that are related by the $\mathfrak{su}(2)$ root. This yields a map on curves that does not preserve intersection with $E$, so does not lift to an automorphism. 
    \item Keep the first and fourth rows of the box graph fixed, equivalent to demanding 
    \be
        F'-M_i' = F-M_i \, , \qquad i=1,\cdots ,8 \, .
    \ee
    By the preservation of self-intersection numbers and intersection with $E$, such constraints uniquely force this map to correspond to the identity automorphism. 
    \item Keep the second and third rows of the box graph fixed, corresponding to 
    \be
    M_i' = M_i \, , \qquad i=1, \cdots ,8 \, .
    \ee
    This does not define a unique lift to $\text{Aut}(H_2(X_S,\mathbb{Z}))$. We can make an ansatz for the fiber and section which does not depend on the curves $M_i$ due to self-intersection number preservation. Satisfying both the self-intersection constraints and intersection with the elliptic fiber, we arrive at two possible consistent automorphisms of the fiber and section, namely the identity automorphism, and the one given by
    \be
       \Sigma \leftrightarrow F \, .
    \ee
    with all other classes left invariant. 
\end{enumerate}
Thus from this box graph transformation, we have shown there is precisely one non-trivial lift to $\text{Aut}(H_2(S,\mathbb{Z}))$ given by 
\be \label{eq: mus general definition}
\mu_{s}:  \qquad 
\left\{\ 
\begin{aligned}
\Sigma &\rightarrow \Sigma'= F\\
F &\rightarrow F'= \Sigma \\
M_i &\rightarrow M_i'=M_i \qquad i=1,\cdots 8 \, ,
\end{aligned}\right.
\ee
which captures the exchange of the fiber and section of the ruling, is of order 2, and leaves the class of the elliptic fiber $E$ invariant.

From field theory, we know that in the construction of the $(2^{2n},0^{8-2n})$ flux tori quivers, the box graphs undergo transformations at the location of the mixed gluing procedure, which can be summarized in Figures~\ref{Fig: box diagram n4 second region} and~\ref{Fig: box diagram n3 second region} for the examples $n=4$ and $n=3$ respectively.

For a general $n$, from the box graph one reads off the transformation 
\be 
\left\{\begin{aligned}
M_i &\rightarrow M_i'=F-M_i \qquad &i=1,\cdots ,2n \\
F-M_i &\rightarrow F'-M_i'=M_i \qquad &i=1,\cdots ,2n\\
M_i &\rightarrow M_i'=M_i \qquad \qquad &i=2n+1, \cdots ,8 \\
F-M_i &\rightarrow F'-M_i'=F-M_i \qquad \qquad &i=2n+1,\cdots ,8 \, ,
\end{aligned}\right.
\ee
from which we get the actions $F'=F$, $M_i'=F-M_i$ for $i=1, \cdots ,2n$, and $M_i'=M_i$ for $i=2n+1, \cdots ,8$. To construct the action on the section $\Sigma$, we demand that $\Sigma'^2 = 0$, $\Sigma'\cdot F'=1$, and $\Sigma'\cdot M_i'=0$, which uniquely determines the map
\be \label{eq: mutn general definition}
\mu_{t,n}:  \qquad 
\left\{\ 
\begin{aligned}
\Sigma &\rightarrow \Sigma'= \Sigma+n F - \sum_{i=1}^{2n}M_i\\
F &\rightarrow F'= F \\
M_i &\rightarrow M_i'=F-M_i \qquad i=1,\cdots 2n \\
M_i &\rightarrow M_i'=M_i \qquad \qquad i=2n+1\cdots 8 \, . \\
\end{aligned}\right.
\ee
A short computation shows that any of these are of order 2 as well. Note that $\mu_{t,n}$ furthermore leaves the class $E$ of the elliptic fiber invariant. 

Another approach to these transformations that comes directly from Section~\ref{sec: Field Theory sec2} is as follows. Observe that under the ECB parameter transformation~\eqref{eq:DWtransf} which occurs as one crosses a duality domain wall, the K\"ahler form~\eqref{eq: Kahler form} transforms as
\be
    J \to J' = (2\phi + m_{\lambda}^0) \Sigma+2\phi F +\sum_{i=1}^8 (m_i-\phi) M_i = \mu_s(J) \, .
\ee
This transformation of ECB parameters is captured geometrically by exchanging $\Sigma$ and $F$ which is exactly the action of $\mu_s$. Similarly, the action of~\eqref{eq:DWtransn} on the ECB parameters that occurs at the location of mixed gluing agrees with the action of $\mu_{t,n}$ on the K\"ahler form. Thus we have a second way of understanding these maps: they are capturing geometrically the action of the ECB parameter transformations that occur in field theory when domain walls are crossed on the K\"ahler form of the threefold.

\section{Geometric Realization of 4d $\cN=1$ Quiver Theories} 
\label{sec: Geometric Realization}

Our first goal in this section is to geometrically realize the flux tori quivers in Figure~\ref{fig:Torusq1intro}, for the cases $n=1,2,3,4$. After fibering the threefold over an interval, with the insertion of appropriate duality domain walls, we determine the resulting 4d $\mathcal{N}=1$ quivers for flux tori after closing the interval to form a circle. 

This requires showing that the geometries corresponding to the ends of the interval are isomorphic, so that the two ends can be glued. This is immediate for the case $n=4$, but we relegate an argument for the existence of such automorphisms for the remaining cases after we have introduced a more systematic algebraic procedure for constructing the geometry and gluing maps associated to a torus with arbitrary flux in Section \ref{sect:elemfluxtubes}.

\subsection{Domain Walls and their Local Geometry} \label{sec: domain walls and local geometry}

The two different types of automorphisms $\mu_s$ and $\mu_{t,n}$ in $\text{Aut}(I_{1,9})$ defined in~\eqref{eq: mus general definition} and \eqref{eq: mutn general definition} give rise to domain walls in field theory, which we refer to as $\mu_s$ and $\mu_{t,n}$ domain walls, respectively. For each domain wall, we define a \textbf{local building block} of arbitrary finite length $\varepsilon$, which we denote $\mu_s^{\varepsilon}$ and $\mu_{t,n}^{\varepsilon}$ respectively. Such a local building block is defined by the local geometry determined by its associated volume profile, and has the property that the field theory data on the two sides of the wall is related by applying $\mu_s$ or $\mu_{t,n}$. This implies that we must be in a configuration that is invariant under $\mu_s$ or $\mu_{t,n}$ at the middle point. Together with the field theory engineered by $X_S^f$ away from the domain wall, this allows us to understand which further fields and interactions are localized on the domain wall.

\paragraph{$\mu_s$ Domain Wall.} These domain walls are associated to S-duality, and can be understood as a form of fiber-base-duality. Invariance under the map 
$\mu_s$ hence requires fiber and base of the ruling to have equal volumes. 
The local geometry summarized by the volume profile in Figure~\ref{fig: mus local volumes} is what we denote as the local building block $\mu_s^{\varepsilon}$. 

Since we always work in a gauge theory phase in which the Coulomb branch parameter $\phi=0$, the fiber of the ruling is collapsed to zero volume on either side of the domain wall, which gives rise to an $A_1$ singularity sitting over the base of the ruling. Hence, we must have that the section $\Sigma$ collapses to zero volume at the location of the $\mu_s$ domain wall, where we apply $\mu_s$. We take the section to be at positive non-zero volume prior to collapse to give rise to a non-trivial 3-cycle in the geometry over which the collapsed fiber will be fibered. 

As $\phi=0$, all of the matter curves $M_i$ must be flopped out of the surface $S$ at the point we apply $\mu_{s}$, at volume 
\be
    \text{Vol}(M_i) = -\frac{1}{8} \text{Vol}(E) = -\frac{m_{\text{KK}}}{8} \, ,
\ee
as fixed by the volume of the elliptic fiber. 
{Here we are the special point in the extended Coulomb branch where all the masses $m_i$ are equal.}
After applying $\mu_s$, we switch the class of the section and the fiber, with fiber $F'$ and section $\Sigma'$ emerging to the right of the domain wall. We remain in the gauge theory phase in which the Coulomb branch parameter is zero, so $F'$ is collapsed to the right of the domain wall, and $\Sigma'$ grows to finite positive volume. This is summarized in the volume profile in Figure~\ref{fig: mus local volumes}. 
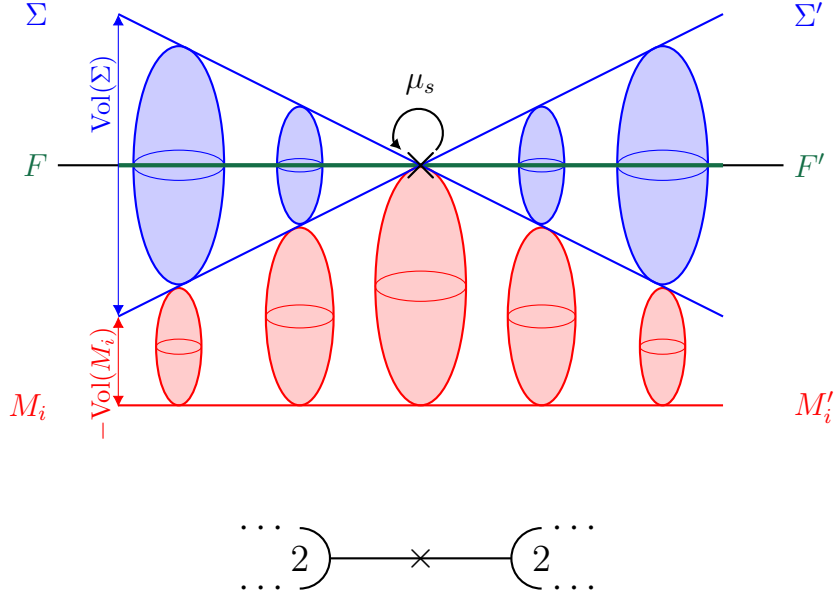
\begin{figure}
    \centering
    \begin{tikzpicture}[scale=0.8]
    \begin{scope}[shift={(0,0)}]
        \filldraw[color = red, fill=red!20, thick] (0,-2) ellipse (0.75 and 1.97);
        \draw[red, thin] (0,-2) ellipse (0.75 and 0.25);
        \filldraw[color = red, fill=red!20, thick] (-2,-2.5) ellipse (0.56 and 1.47);
        \draw[red, thin] (-2,-2.5) ellipse (0.56 and 0.19);
        \filldraw[color = red, fill=red!20, thick] (-4,-3) ellipse (0.375 and 0.97);
        \draw[red, thin] (-4,-3) ellipse (0.375 and 0.125);
        \filldraw[color = red, fill=red!20, thick] (2,-2.5) ellipse (0.56 and 1.47);
        \draw[red, thin] (2,-2.5) ellipse (0.56 and 0.19);
        \filldraw[color = red, fill=red!20, thick] (4,-3) ellipse (0.375 and 0.97);
        \draw[red, thin] (4,-3) ellipse (0.375 and 0.125);

        \filldraw[color = blue, fill=blue!20, thick] (-2,0) ellipse (0.375 and 0.97);
        \draw[blue, thin] (-2,0) ellipse (0.375 and 0.11);
        \filldraw[color = blue, fill=blue!20, thick] (-4,0) ellipse (0.75 and 1.97);
        \draw[blue, thin] (-4,0) ellipse (0.75 and 0.25);
         \filldraw[color = blue, fill=blue!20, thick] (2,0) ellipse (0.375 and 0.97);
        \draw[blue, thin] (2,0) ellipse (0.375 and 0.11);
        \filldraw[color = blue, fill=blue!20, thick] (4,0) ellipse (0.75 and 1.97);
        \draw[blue, thin] (4,0) ellipse (0.75 and 0.25);

        \draw[black, thick] (-6,0) -- (6,0);
        \draw[mygreen, ultra thick] (-5,0) -- (5,0);
        \draw[blue, thick] (-5,2.5) -- (0,0);
        \draw[blue, thick] (0,0) -- (5,2.5);
        \draw[blue, thick] (-5,-2.5) -- (0,0);
        \draw[blue, thick] (0,0) -- (5,-2.5);
        \draw[red, thick] (-5,-3.97) -- (5,-3.97);
        
        \filldraw[black] (0,0) circle (0pt) node{\huge{$\times$}};
        \filldraw[black] (0,1) circle (0pt) node[anchor=south]{\large{$\mu_s$}};
        \draw[thick, ->] (0.25,0.25) arc (-45:225:0.4);
        \filldraw[black] (-6,0) circle (0pt) node[anchor=east]{\textcolor{mygreen}{\large{$F$}}};
        \filldraw[black] (-6,2.5) circle (0pt) node[anchor=east]{\textcolor{blue}{\large{$\Sigma$}}};
        \filldraw[black] (-6,-4) circle (0pt) node[anchor=east]{\textcolor{red}{\large{$M_i$}}};

        \filldraw[black] (6,0) circle (0pt) node[anchor=west]{\textcolor{mygreen}{\large{$F'$}}};
        \filldraw[black] (6,2.5) circle (0pt) node[anchor=west]{\textcolor{blue}{\large{$\Sigma'$}}};
        \filldraw[black] (6,-4) circle (0pt) node[anchor=west]{\textcolor{red}{\large{$M_i'$}}};

         \draw[red, thin, <->] (-5,-2.5) -- (-5,-4);
         \draw[blue, thin, <->] (-5,-2.5) -- (-5,2.5);
        \filldraw[black] (-5.25,-2.5) circle (0pt) node[anchor=east,rotate=90]{\textcolor{red}{\small{$-\text{Vol}(M_i)$}}};
        \filldraw[black,] (-5.25,2) circle (0pt) node[anchor=east,rotate=90]{\textcolor{blue}{\small{$\text{Vol}(\Sigma)$}}};
\end{scope}
\begin{scope}[shift={(0,-7)}]
\draw[thick]  (-2,0) arc (-90:90:0.5);
\draw[thick]  (2,0) arc (270:90:0.5);

\draw[thick] (-1.5,0.5) -- (1.5,0.5) node[midway] {\Large $\times$};

\node at (-2,0.5) {\Large $2$};
\node at ( 2,0.5) {\Large $2$};

\node[anchor=west] at (2,1) {\Large $\cdots$};
\node[anchor=west] at (2,0.0) {\Large $\cdots$};

\node[anchor=east] at (-2,1) {\Large $\cdots$};
\node[anchor=east] at (-2,0.0) {\Large $\cdots$};
\end{scope}
    \end{tikzpicture}
    \caption{The two thimble geometry arising from the $\mu_s$ domain wall, denoted by $\mu_s^{\varepsilon}$. Each cross-sectional sphere represents a $\mathbb{P}^1$, the volumes of which are that of the curve classes labelled, and the location we apply the $\mu_s$ automorphism is marked with a $\times$.   }
    \label{fig: mus local volumes}
\end{figure}
The local geometry is that of two thimbles, over which an $A_1$ singularity is fibered, coming from the collapsed curve $F$. The matter curves $M_i$ all remain flopped out of the geometry and are merely spectators here. The $A_1$ singularity gives rise to $SU(2)$ quiver nodes. These nodes do not touch the $SU(8)$ quiver node in this interval, since the matter curves intersecting both the $A_1$ and $A_7$ roots are at finite volume. We therefore do not draw the $SU(8)$ quiver node in this region.

We have introduced a new notation for the quiver in Figure~\ref{fig: mus local volumes}, which should be viewed as a local piece of the 4d $\cN=1$ quiver. In this way, the quiver is meaningless until we close the profile to give a threefold fibration over $S^1$, as in Figure~\ref{fig: n4 volumes and quiver closed loop}. The local pieces of the quiver will glue together in the obvious way when we start to concatenate local building blocks.

\paragraph{$\mu_{t,n}$ Domain Walls.} These domain walls occur at the locations in which one performs a $\Phi$ or mixed gluing in field theory. To understand the geometry defined by the volume profile of the local building block shown in Figure~\ref{fig: local mutn dom wall}, which we denote by $\mu_{t,n}^{\varepsilon}$, recall the following action of $\mu_{t,n}$ on the matter curves 
\be 
\mu_{t,n}:  \qquad 
\left\{\begin{aligned}
M_i &\rightarrow M_i'=F-M_i \qquad i=1,\cdots, 2n \\
M_i &\rightarrow M_i'=M_i \qquad \qquad i=2n+1,\cdots, 8 \, .
\end{aligned}\right.
\ee
Since the gauge theory phase requires $F$ to be collapsed to zero volume, one can only apply the $\mu_{t,n}$ map at the location in which the volume of $M_i$ for $i=1,\cdots, 2n$ is zero. Moreover, to capture the field theory operation of mixed gluing at this location, we must require that the matter curves $M_i$ for $i=2n+1,\cdots, 8$ remain at non-zero volume at the location we apply the $\mu_{t,n}$ map. We further take the section $\Sigma$ to be at positive non-zero volume at the location we apply $\mu_{t,n}$.
\begin{figure}
    \centering
\begin{tikzpicture}[scale=0.8]
\begin{scope}[shift={(0,0)},yscale=0.8,xscale=1.2]
    \filldraw[color=blue, fill=blue!20, thick] (-2,0) ellipse (0.72 and 3.17);
    \draw[blue, thin] (-2,0) ellipse (0.72 and 0.22);

    \filldraw[color=blue, fill=blue!20, thick] (-4,0) ellipse (0.60 and 2.37);
    \draw[blue, thin] (-4,0) ellipse (0.60 and 0.17);

    \filldraw[color=blue, fill=blue!20, thick] (2,0) ellipse (0.72 and 3.17);
    \draw[blue, thin] (2,0) ellipse (0.72 and 0.22);

    \filldraw[color=blue, fill=blue!20, thick] (4,0) ellipse (0.60 and 2.37);
    \draw[blue, thin] (4,0) ellipse (0.60 and 0.17);

    \filldraw[color=blue, fill=blue!20, thick] (0,0) ellipse (0.84 and 3.96);
    \draw[blue, thin] (0,0) ellipse (0.84 and 0.27);

    \filldraw[color=orange, fill=orange!20, thick] (0,5) ellipse (0.40 and 0.97);
    \draw[orange, thin] (0,5) ellipse (0.40 and 0.12);

    \filldraw[color=orange, fill=orange!20, thick] (-2,4.6) ellipse (0.54 and 1.37);
    \draw[orange, thin] (-2,4.6) ellipse (0.54 and 0.16);

    \filldraw[color=orange, fill=orange!20, thick] (2,4.6) ellipse (0.54 and 1.37);
    \draw[orange, thin] (2,4.6) ellipse (0.54 and 0.16);

    \filldraw[color=orange, fill=orange!20, thick] (-4,4.2) ellipse (0.68 and 1.77);
    \draw[orange, thin] (-4,4.2) ellipse (0.68 and 0.21);

    \filldraw[color=orange, fill=orange!20, thick] (4,4.2) ellipse (0.68 and 1.77);
    \draw[orange, thin] (4,4.2) ellipse (0.68 and 0.21);

    \filldraw[color=red, fill=red!20, thick] (4,-3.2) ellipse (0.30 and 0.77);
    \draw[red, thin] (4,-3.2) ellipse (0.30 and 0.10);

    \filldraw[color=red, fill=red!20, thick] (2,-3.6) ellipse (0.20 and 0.38);
    \draw[red, thin] (2,-3.6) ellipse (0.20 and 0.05);

    \filldraw[color=red, fill=red!20, thick] (-4,-3.2) ellipse (0.30 and 0.77);
    \draw[red, thin] (-4,-3.2) ellipse (0.30 and 0.10);

    \filldraw[color=red, fill=red!20, thick] (-2,-3.6) ellipse (0.20 and 0.38);
    \draw[red, thin] (-2,-3.6) ellipse (0.20 and 0.05);

    \draw[black, thick] (-6,0) -- (6,0);
    \draw[mygreen, ultra thick] (-5,0) -- (5,0);

    \draw[blue, thick] (-5,2) -- (0,4);
    \draw[blue, thick] (0,4) -- (5,2);
    \draw[blue, thick] (-5,-2) -- (0,-3.97);
    \draw[blue, thick] (0,-3.97) -- (5,-2);

    \draw[red, thick] (-5,-3.97) -- (5,-3.97);
    \draw[orange, thick] (-5,5.97) -- (5,5.97); 

\draw[red, thin, <->] (-5,-2) -- (-5,-4);
         \draw[orange, thin, <->] (-5,2) -- (-5,6);
         \draw[blue, thin, <->] (-5,-2) -- (-5,2);
\end{scope}
\begin{scope}[shift={(0,0)}]

        \filldraw[black] (0,0) circle (0pt) node{\huge{$\times$}};
        \filldraw[black] (0,1) circle (0pt) node[anchor=south]{\large{$\mu_{t,n}$}};
        \draw[thick, ->] (0.25,0.25) arc (-45:225:0.4);
        \filldraw[black] (-7.5,0) circle (0pt) node[anchor=east]{\textcolor{mygreen}{\large{$F$}}};
        \filldraw[black] (-7.5,1.5) circle (0pt) node[anchor=east]{\textcolor{blue}{\large{$\Sigma$}}};
        \filldraw[black] (-7,-3) circle (0pt) node[anchor=east]{\textcolor{red}{\large{$M_{1,\cdots, 2n}$}}};
        \filldraw[black] (-7,4.5) circle (0pt) node[anchor=east]{\textcolor{orange}{\large{$M_{2n+1,\cdots, 8}$}}};

         \filldraw[black] (7.5,0) circle (0pt) node[anchor=west]{\textcolor{mygreen}{\large{$F'$}}};
        \filldraw[black] (7.5,1.5) circle (0pt) node[anchor=west]{\textcolor{blue}{\large{$\Sigma'$}}};
        \filldraw[black] (7,-3) circle (0pt) node[anchor=west]{\textcolor{red}{\large{$M_{1, \cdots ,2n}'$}}};
        \filldraw[black] (7,4.5) circle (0pt) node[anchor=west]{\textcolor{orange}{\large{$M_{2n+1,\cdots , 8}'$}}};

        \filldraw[black] (-6.5,-1.4) circle (0pt) node[anchor=east,rotate=90]{\textcolor{red}{\small{$-\text{Vol}(M_i)$}}};
        \filldraw[black] (-6.5,4.5) circle (0pt) node[anchor=east,rotate=90]{\textcolor{orange}{\small{$-\text{Vol}(M_i)$}}};
        \filldraw[black,] (-6.5,1.5) circle (0pt) node[anchor=east,rotate=90]{\textcolor{blue}{\small{$\text{Vol}(\Sigma)$}}};
    \end{scope}
    \begin{scope}[shift={(0,-9.5)}]

     \tikzstyle{every node}=[font=\scriptsize]
    \draw[thick] (-1,3.5) rectangle (1,4.5);
    \draw[thick] (-0.5,-1.5) rectangle (0.5,-0.5);
    \filldraw[black] (0,-1) circle (0pt) node{\Large{$2n$}};
    \filldraw[black] (0,4) circle (0pt) node{\Large{$8-2n$}};
    \filldraw[black] (0,1.5) circle (0pt) node{\Large{$2$}};
    \draw[thick] (0.5,2) -- (-0.5,2);
    \draw[thick] (0.5,1) -- (-0.5,1);
    \filldraw[black,anchor=west] (0.5,2) circle (0pt) node{\Large{$\cdots$}};
    \filldraw[black,anchor=west] (0.5,1) circle (0pt) node{\Large{$\cdots$}};
    \filldraw[black,anchor=east] (-0.5,2) circle (0pt) node{\Large{$\cdots$}};
    \filldraw[black,anchor=east] (-0.5,1) circle (0pt) node{\Large{$\cdots$}};
    \draw[thick, ->-] (0,1) -- (0,-0.5) ;
    
\end{scope}

\end{tikzpicture}
    \caption{The local geometry arising from the $\mu_{t,n}$ domain wall, which we denote by $\mu_{t,n}^{\varepsilon}$. The fibration of the section $\Sigma$ over the interval of finite size gives rise to a three-cylinder over which an $A_1$ singularity is fibered. The $\mathbb{P}^1$'s associated to the curve classes $M_{1\cdots 2n}$ collapse at the location of the domain wall, meaning the $A_{2n-1}$ singularity touches the $A_1$ singularity at this point. The location at which we apply the $\mu_{t,n}$ automorphism is marked with a $\times$.}
    \label{fig: local mutn dom wall}
\end{figure}
From Figure~\ref{fig: local mutn dom wall}, it is clear that at the location of $\mu_{t,n}$, the $A_7$ root $\beta_{2n}=M_{2n+1}-M_{2n}$ acquires a non-zero volume, whilst the other $A_7$ roots $\beta_{1, \cdots ,2n-1}$ and $\beta_{2n+1, \cdots ,7}$ remain at zero volume. These give rise to an $A_{2n-1}$ and $A_{7-2n}$ singularity, giving the splitting of the flavor symmetry seen in the quiver
\be
    SU(8) \rightarrow SU(2n) \times SU(8-2n){\times U(1)} \, .
\ee
At the location of the domain walls, only the matter curves $M_{1, \cdots ,2n}$ collapse to zero volume so we only get the chirals charged under the $SU(2n)$ symmetry shown in the quiver.

Note that at the location we apply the $\mu_{t,n}$ automorphism, the volume of the elliptic fiber is not enough to uniquely fix the volumes of the sections and the non-collapsed matter curves, only to fix a linear relation between them. For concreteness, we choose them to be
\be
\left\{\begin{aligned}
&\text{Vol}(\Sigma) = \frac{m_{\text{KK}}}{4} \, , \qquad \text{Vol}(M_i) = -\frac{m_{\text{KK}}}{2(8-2n)} \, , \quad i=2n+1, \cdots ,8 \, , \quad n\in \{ 1,2,3 \} \, , \\
&\text{Vol}(\Sigma) = \frac{m_{\text{KK}}}{2} \, , \hspace{8.7cm} n=4 \, ,
\end{aligned}\right.
\ee
which gives us the singularities described above.

\subsection{Stacking Domain Walls}

We now move on to describe how to stack the local building blocks $\mu^\varepsilon_s$ and $\mu^\varepsilon_{t,n}$ introduced above. As we have seen, for each of the domain walls we consider, the action on divisors in $H^2(X_S^f,\Z)\simeq I_{1,9}$, can be summarized by a linear map $\mu$ which is either $\mu_s$ or $\mu_{t,n}$. We can describe this linear action by a matrix, which we will also call $\mu$ and which is defined as
\begin{equation}
   V \rightarrow V' = \mu V
\end{equation}
where $V$ is a vector with entries a basis of $I_{1,9}$, e.g. we could take
$V = \left(F,\Sigma,M_1,\cdots, M_8\right)$. When stacking several domain walls each acting by a map $\mu_i$, the resulting map should then be found by multiplication of the matrices $\mu_i$. This logic further extends to Section~\ref{sect:elemfluxtubes} when we discuss automorphisms associated to Weyl transformations. 

Given a configuration of domain walls 
\begin{equation}
    \mathcal{D} = \mu_1^\varepsilon \cdots \mu_n^\varepsilon \, ,
\end{equation} which is composed of domain walls $\mu_i^\varepsilon$ which are consecutively located on a line, we may try to close the line to a circle by identifying the two ends and form a flux torus. The domain walls have been geometrized by fibering $X_S^f$ over an interval, so this forms a 7-dimensional space $\mathcal{T}$ in which $X_S^f$ is fibered over a circle. Similar to how we can think of a space that has undergone a continuous deformation back to itself by a monodromy map, the map that we need to close the circle is then
\begin{equation}
    \mathcal{M} = (\mu_1 \cdots \mu_n)^{-1} \, ,
\end{equation}
and we will refer to this as the associated monodromy map as well. Starting with a fixed basis of cycles on $X_S^f$, it is precisely the map $\mathcal{M}$ that is encountered by divisors of $X_S^f$ as we move once around the circle over which $X_S^f$ is fibered. 

As discussed in Section \ref{sec: Geometry of XS}, divisor classes on $X_S^f$ are in one-to-one correspondence to curve classes on $S$. We will hence describe the action of maps such as $\mathcal{M}$ in terms of a basis of $H_2(S,\Z) \simeq I_{1,9}$, with the understanding that any of these lifts to a divisor on $X_S^f$.

\subsection{Stacking Domain Walls: Flux $(-2^8)$} \label{sec: flux 2^8}

Having understood the two most basic ingredients of the field theory dictionary, namely the S-duality domain wall and the gluing procedure, we can now begin to stack alternating copies of our local building blocks $\mu_{s}^{\varepsilon}$ and $\mu_{t,n}^{\varepsilon}$. First, we will show how the concatenating of volume profiles allows us to reproduce the quiver in Figure~\ref{fig:Torusq1intro} for $n=4$. To concatenate the volume profiles, we glue the curves at the locations where the volumes of the curves on either side of the gluing point are equal. 

\paragraph{$\mu_{s}^{\varepsilon}\mu_{t,4}^{\varepsilon}\mu_s^{\varepsilon}$ Configuration.} Gluing together the volume profiles in Figures~\ref{fig: mus local volumes} and~\ref{fig: local mutn dom wall} as dictated by this set of maps gives us the configuration shown in Figure~\ref{fig: local mus4s dom wall}.
\begin{figure}
    \centering
    \begin{tikzpicture}[scale=0.8]
\begin{scope}[shift={(0,0)}]


    \filldraw[color=red, fill=red!20, thick] (-5.5,-2.75) ellipse (0.42 and 1.17);
    \draw[red, thin] (-5.5,-2.75) ellipse (0.42 and 0.18);

    \filldraw[color=red, fill=red!20, thick] (-4,-2) ellipse (0.52 and 1.92);
    \draw[red, thin] (-4,-2) ellipse (0.52 and 0.22);

    \filldraw[color=red, fill=red!20, thick] (-2.5,-2.75) ellipse (0.42 and 1.17);
    \draw[red, thin] (-2.5,-2.75) ellipse (0.42 and 0.18);

    \filldraw[color=red, fill=red!20, thick] (-1,-3.5) ellipse (0.28 and 0.42);
    \draw[red, thin] (-1,-3.5) ellipse (0.28 and 0.07);

   \filldraw[color=red, fill=red!20, thick] (5.5,-2.75) ellipse (0.42 and 1.17);
    \draw[red, thin] (5.5,-2.75) ellipse (0.42 and 0.18);

    \filldraw[color=red, fill=red!20, thick] (4,-2) ellipse (0.52 and 1.92);
    \draw[red, thin] (4,-2) ellipse (0.52 and 0.22);

    \filldraw[color=red, fill=red!20, thick] (2.5,-2.75) ellipse (0.42 and 1.17);
    \draw[red, thin] (2.5,-2.75) ellipse (0.42 and 0.18);

    \filldraw[color=red, fill=red!20, thick] (1,-3.5) ellipse (0.28 and 0.42);
    \draw[red, thin] (1,-3.5) ellipse (0.28 and 0.07);


    \filldraw[color=blue, fill=blue!20, thick] (-5.5,0) ellipse (0.46 and 1.42);
    \draw[blue, thin] (-5.5,0) ellipse (0.46 and 0.14);

    \filldraw[color=blue, fill=blue!20, thick] (-2.5,0) ellipse (0.46 and 1.40);
    \draw[blue, thin] (-2.5,0) ellipse (0.46 and 0.16);

    \filldraw[color=blue, fill=blue!20, thick] (-1,0) ellipse (0.62 and 2.88);
    \draw[blue, thin] (-1,0) ellipse (0.62 and 0.22);

    \filldraw[color=blue, fill=blue!20, thick] (1,0) ellipse (0.62 and 2.88);
    \draw[blue, thin] (1,0) ellipse (0.62 and 0.20);

    \filldraw[color=blue, fill=blue!20, thick] (2.5,0) ellipse (0.46 and 1.40);
    \draw[blue, thin] (2.5,0) ellipse (0.46 and 0.15);

    \filldraw[color=blue, fill=blue!20, thick] (5.5,0) ellipse (0.46 and 1.42);
    \draw[blue, thin] (5.5,0) ellipse (0.46 and 0.14);


    \draw[black, thick, ->] (-6.25,0) -- (6.25,0);
    \draw[mygreen, ultra thick] (-5.85,0) -- (5.85,0);

    \draw[red, thick] (-5.85,-3.92) -- (5.85,-3.92);


    \draw[blue, thick] (-5.85,1.85) -- (-4.00,0);
    \draw[blue, thick] (-5.85,-1.85) -- (-4.00,0);

    \draw[blue, thick] (-4.00,0) -- (0.00,3.92);
    \draw[blue, thick] (0.00,3.92) -- (4.00,0);

    \draw[blue, thick] (-4.00,0) -- (0.00,-3.92);
    \draw[blue, thick] (0.00,-3.92) -- (4.00,0);

    \draw[blue, thick] (4.00,0) -- (5.85,1.85);
    \draw[blue, thick] (4.00,0) -- (5.85,-1.85);


    \filldraw[black] (-4.00,0) circle (0pt) node{\huge{$\times$}};
    \filldraw[black] (0.00,0) circle (0pt) node{\huge{$\times$}};
    \filldraw[black] (4.00,0) circle (0pt) node{\huge{$\times$}};


    \filldraw[black] (-6.3,0) circle (0pt)
        node[anchor=east]{\textcolor{mygreen}{\large{$F$}}};

    \filldraw[black] (-6.3,2) circle (0pt)
        node[anchor=east]{\textcolor{blue}{\large{$\Sigma$}}};

    \filldraw[black] (-6.3,-4) circle (0pt)
        node[anchor=east]{\textcolor{red}{\large{$M_i$}}};

    \filldraw[black] (0,0.2) circle (0pt) node[anchor=south]{\large{$\mu_{t,4}$}};

    \filldraw[black] (-4,0.2) circle (0pt) node[anchor=south]{\large{$\mu_{s}$}};

    \filldraw[black] (4,0.2) circle (0pt) node[anchor=south]{\large{$\mu_{s}$}};

\end{scope}
\begin{scope}[shift={(0,-7)}]

    \tikzstyle{every node}=[font=\scriptsize]
    \draw[thick] (-0.5,-1.5) rectangle (0.5,-0.5);
    \filldraw[black] (0,-1) circle (0pt) node{\Large{$8$}};
    \filldraw[black] (0,1.5) circle (0pt) node{\Large{$2$}};
    \draw[thick](0,1.5) circle (0.5);
    \draw[thick, ->-] (0,1) -- (0,-0.5) ;
       
    \draw[thick]  (-6,1) arc (-90:90:0.5);
    \draw[thick]  (6,1) arc (270:90:0.5);

    \draw[thick] (-5.5,1.5) -- (-0.5,1.5);
    \draw[thick] (0.5,1.5) -- (5.5,1.5);
    \node at (-4,1.5) {\large $\times$};
    \node at (4,1.5) {\large $\times$};

\node at (-6,1.5) {\Large $2$};
\node at (6,1.5) {\Large $2$};

\node[anchor=west] at (6,2) {\Large $\cdots$};
\node[anchor=west] at (6,1) {\Large $\cdots$};

\node[anchor=east] at (-6,2) {\Large $\cdots$};
\node[anchor=east] at (-6,1) {\Large $\cdots$};
\end{scope}

    \end{tikzpicture}
    \caption{The local geometry arising from the stacking of the local building blocks $\mu_{s}^{\varepsilon}\mu_{t,4}^{\varepsilon}\mu_s^{\varepsilon}$. The fibration of the section $\Sigma$ over the interval whose endpoints are given by the location of the $\mu_s$ domain walls gives rise to a 3-sphere over which an $A_1$ singularity is fibered. The $\mathbb{P}^1$'s associated to the curve classes $M_i$ collapse at the location of the $\mu_{t,4}$ domain wall, meaning the $A_{7}$ singularity touches the $A_1$ singularity at this point.}
    \label{fig: local mus4s dom wall}
\end{figure}
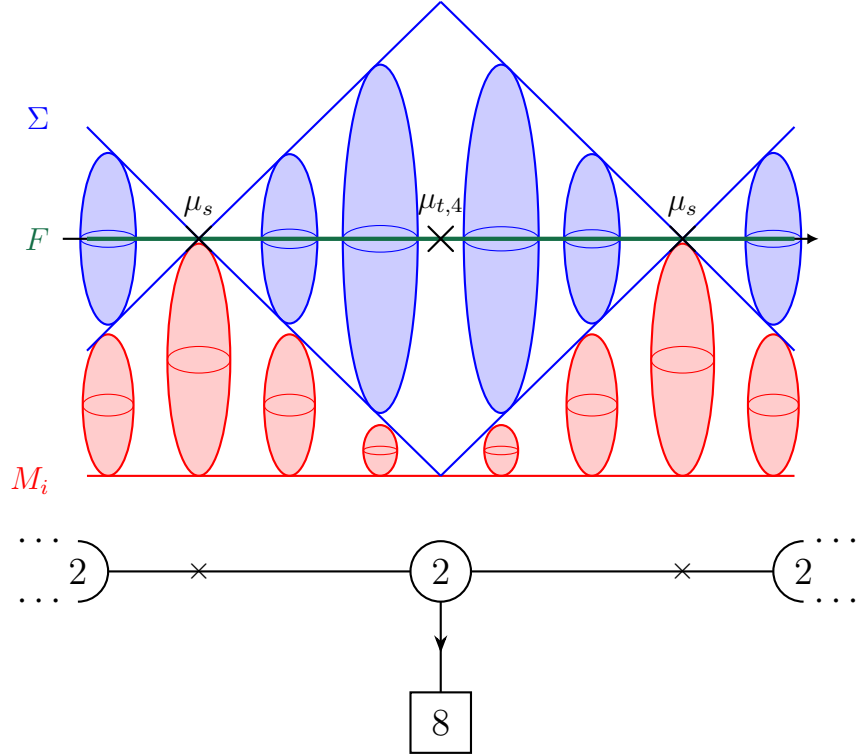




The compact 3-cycle coming from the fibration of the section $\Sigma$ over the interval between the two $\mu_s$ domain walls gives rise to a gauged $SU(2)$ node in the quiver. At the location of the $\mu_{t,4}$ domain wall, all of the matter curves collapse to zero volume, meaning the $A_7$ singularity touches the $A_1$ singularity of the collapsed fiber at this point, giving rise to the bifundamental matter seen in the quiver.

\paragraph{$\mu_s^{\varepsilon} \mu_{t,4}^{\varepsilon}\mu_s^{\varepsilon}\mu_{t,4}^{\varepsilon}\mu_s^{\varepsilon}$ Configuration.} Combining these local building blocks gives us the volume profile and quiver shown in Figure~\ref{fig: local mus4s4s dom wall}. There are two compact three-cycles over which the $A_1$ singularity coming from the collapsed curve class of the fiber is fibered. The $A_1$ singularity touches the matter curves $M_i$ at the locations where they collapse to zero volume. The second application of the $\mu_{t,4}$ map will interchange the curves $M_i$ and $F-M_i$, giving us the box graphs shown in Figure~\ref{Fig: box diagram n4 third region}. 

{As discussed in \cite{Braun:2023fqa}, one can argue for the existence of a cubic superpotential corresponding to the triangle loop of the quiver from the fact that the 3-cycles swept out by $M_i$ connect the three codimension-7 singular loci associated to the chirals involved in the interaction. This has the effect of identifying the $SU(8)$ of each $\mu_{t,4}^\varepsilon$ building block, so that we have a single $SU(8)$ in the final quiver.}

\begin{figure}[t]
    \centering
    \begin{tikzpicture}[scale=0.8]


    \filldraw[color=red, fill=red!20, thick] (-7.25,-2.125) ellipse (0.38 and 0.775);
    \draw[red, thin] (-7.25,-2.125) ellipse (0.38 and 0.14);
    \filldraw[color=red, fill=red!20, thick] (-4.75,-2.105) ellipse (0.38 and 0.79);
    \draw[red, thin] (-4.75,-2.125) ellipse (0.38 and 0.14);
    \filldraw[color=red, fill=red!20, thick] (-6,-1.5) ellipse (0.5 and 1.4);
    \draw[red, thin] (-6,-1.5) ellipse (0.5 and 0.2);
    \filldraw[color=red, fill=red!20, thick] (-1.25,-2.1) ellipse (0.38 and 0.79);
    \draw[red, thin] (-1.25,-2.125) ellipse (0.38 and 0.14);
    \filldraw[color=red, fill=red!20, thick] (1.25,-2.1) ellipse (0.38 and 0.79);
    \draw[red, thin] (1.25,-2.125) ellipse (0.38 and 0.14);
    \filldraw[color=red, fill=red!20, thick] (0,-1.5) ellipse (0.5 and 1.4);
    \draw[red, thin] (0,-1.5) ellipse (0.5 and 0.2);

    \filldraw[color=red, fill=red!20, thick] (4.75,-2.105) ellipse (0.38 and 0.79);
    \draw[red, thin] (4.75,-2.125) ellipse (0.38 and 0.14);
    \filldraw[color=red, fill=red!20, thick] (7.25,-2.125) ellipse (0.38 and 0.775);
    \draw[red, thin] (7.25,-2.125) ellipse (0.38 and 0.14);
    \filldraw[color=red, fill=red!20, thick] (6,-1.5) ellipse (0.5 and 1.4);
    \draw[red, thin] (6,-1.5) ellipse (0.5 and 0.2);


    \filldraw[color=blue, fill=blue!20, thick] (-7.25,0) ellipse (0.5 and 1.14);
    \draw[blue, thin] (-7.25,0) ellipse (0.5 and 0.2);

    \filldraw[color=blue, fill=blue!20, thick] (-4.75,0) ellipse (0.5 and 1.1);
    \draw[blue, thin] (-4.75,0) ellipse (0.5 and 0.2);

    \filldraw[color=blue, fill=blue!20, thick] (-1.25,0) ellipse (0.5 and 1.1);
    \draw[blue, thin] (-1.25,0) ellipse (0.5 and 0.2);

    \filldraw[color=blue, fill=blue!20, thick] (1.25,0) ellipse (0.5 and 1.1);
    \draw[blue, thin] (1.25,0) ellipse (0.5 and 0.2);

    \filldraw[color=blue, fill=blue!20, thick] (4.75,0) ellipse (0.5 and 1.1);
    \draw[blue, thin] (4.75,0) ellipse (0.5 and 0.2);

    \filldraw[color=blue, fill=blue!20, thick] (7.25,0) ellipse (0.5 and 1.14);
    \draw[blue, thin] (7.25,0) ellipse (0.5 and 0.2);

    \filldraw[color=blue, fill=blue!20, thick] (-3,0) ellipse (0.8 and 2.8);
    \draw[blue, thin] (-3,0) ellipse (0.8 and 0.4);

    \filldraw[color=blue, fill=blue!20, thick] (3,0) ellipse (0.8 and 2.8);
    \draw[blue, thin] (3,0) ellipse (0.8 and 0.4);

        \draw[thick] (-8.5,0) -- (8.5,0);
        \draw[color=mygreen, ultra thick] (-7.5,0) -- (7.5,0);
        \draw[color=red, thick] (-7.5,-2.9) -- (7.5,-2.9);
        \draw[color=blue, thick] (-7.5,1.5) -- (-6,0);
        \draw[color=blue, thick] (-7.5,-1.5) -- (-6,0);
        \draw[color=blue, thick] (-6,0) -- (-3,2.9);
        \draw[color=blue, thick] (-6,0) -- (-3,-2.9);
        \draw[color=blue, thick] (-3,2.9) -- (0,0);
        \draw[color=blue, thick] (-3,-2.9) -- (0,0);
        \draw[color=blue, thick] (0,0) -- (3,2.9);
        \draw[color=blue, thick] (0,0) -- (3,-2.9);
        \draw[color=blue, thick] (3,-2.9) -- (6,0);
        \draw[color=blue, thick] (3,2.9) -- (6,0);
        \draw[color=blue, thick] (6,0) -- (7.5,1.5);
        \draw[color=blue, thick] (6,0) -- (7.5,-1.5);

    \filldraw[black] (-8.5,0) circle (0pt)
        node[anchor=east]{\textcolor{mygreen}{\large{$F$}}};

    \filldraw[black] (-8.5,1.75) circle (0pt)
        node[anchor=east]{\textcolor{blue}{\large{$\Sigma$}}};

    \filldraw[black] (-8.5,-2.92) circle (0pt)
        node[anchor=east]{\textcolor{red}{\large{$M_i$}}};

    \filldraw[black] (-6.00,0.2) circle (0pt)
        node[anchor=south]{\large{$\mu_s$}};

    \filldraw[black] (-3,0.3) circle (0pt)
        node[anchor=south]{\large{$\mu_{t,4}$}};

    \filldraw[black] (0.00,0.2) circle (0pt)
        node[anchor=south]{\large{$\mu_s$}};

    \filldraw[black] (3,0.3) circle (0pt)
        node[anchor=south]{\large{$\mu_{t,4}$}};

    \filldraw[black] (6,0.2) circle (0pt)
        node[anchor=south]{\large{$\mu_s$}};

    \filldraw[black] (-6.00,0) circle (0pt) node{\huge{$\times$}};
    \filldraw[black] (-3.00,0) circle (0pt) node{\huge{$\times$}};
    \filldraw[black] (0.00,0) circle (0pt) node{\huge{$\times$}};
    \filldraw[black] (3.00,0) circle (0pt) node{\huge{$\times$}};
    \filldraw[black] (6.00,0) circle (0pt) node{\huge{$\times$}};

\begin{scope}[shift={(0,-6)}]

    \tikzstyle{every node}=[font=\scriptsize]
    \draw[thick] (-0.5,-1.5) rectangle (0.5,-0.5);
    \filldraw[black] (0,-1) circle (0pt) node{\Large{$8$}};
    \filldraw[black] (-3,1.5) circle (0pt) node{\Large{$2$}};
    \filldraw[black] (3,1.5) circle (0pt) node{\Large{$2$}};
    \draw[thick](3,1.5) circle (0.5);
    \draw[thick](-3,1.5) circle (0.5);
    \draw[thick, ->-] (-2.6,1.2) -- (-0.5,-0.7) ;
    \draw[thick, -<-] (2.6,1.2) -- (0.5,-0.7) ;
       
    \draw[thick]  (-7.25,1) arc (-90:90:0.5);
    \draw[thick]  (7.25,1) arc (270:90:0.5);

    \draw[thick] (-6.75,1.5) -- (-3.5,1.5);
    \draw[thick] (6.75,1.5) -- (3.5,1.5);
    \draw[thick] (-2.5,1.5) -- (2.5,1.5);
    \node at (-6,1.5) {\large $\times$};
    \node at (6,1.5) {\large $\times$};
    \node at (0,1.5) {\large $\times$};

\node at (-7.25,1.5) {\Large $2$};
\node at (7.25,1.5) {\Large $2$};

\node[anchor=west] at (7.25,2) {\Large $\cdots$};
\node[anchor=west] at (7.25,1) {\Large $\cdots$};

\node[anchor=east] at (-7.25,2) {\Large $\cdots$};
\node[anchor=east] at (-7.25,1) {\Large $\cdots$};
\end{scope}        
    \end{tikzpicture}
    \caption{The local geometry arising from the stacking of the $\mu_{s}^{\varepsilon}\mu_{t,4}^{\varepsilon}\mu_s^{\varepsilon}\mu_{t,4}^{\varepsilon}\mu_s^{\varepsilon}$ volume profiles. The fibration of the section $\Sigma$ over the intervals whose endpoints are given by the location of the $\mu_s$ domain walls gives rise to two 3-spheres over which an $A_1$ singularity is fibered. The $\mathbb{P}^1$'s associated to the curve classes $M_i$ collapse at the location of the $\mu_{t,4}$ domain walls, meaning the $A_{7}$ singularity touches the $A_1$ singularity at these points.}
    \label{fig: local mus4s4s dom wall}
\end{figure}
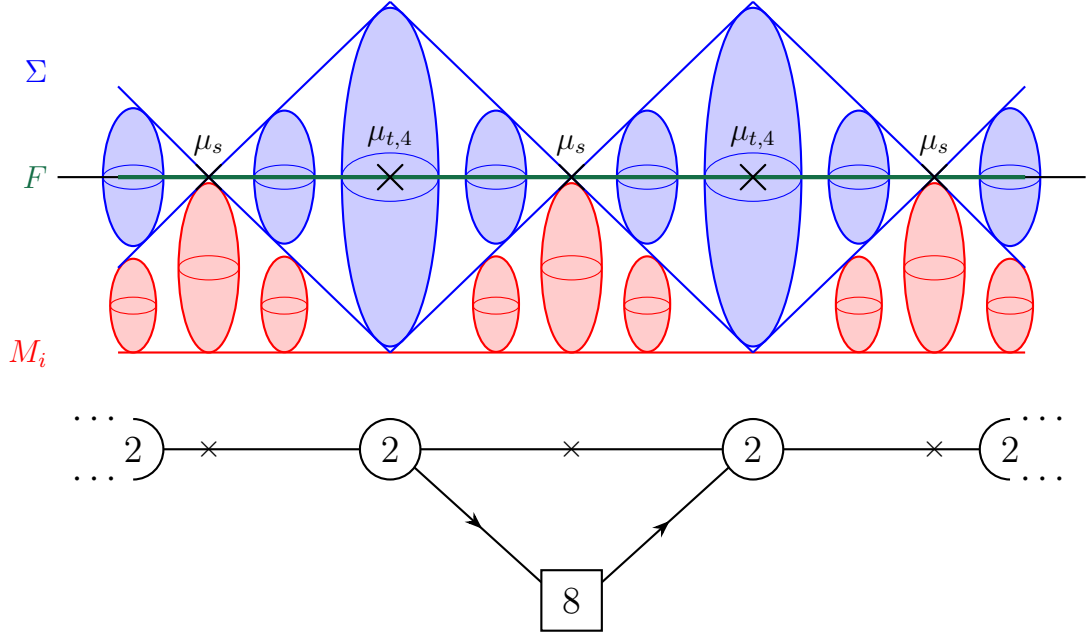
\begin{figure}[h]
\centering
\begin{tikzpicture}
\node[] (5dLyt) at (-3,-2) {\ytableausetup{centertableaux,boxsize=6pt}
\begin{ytableau}
*(blue) & *(blue) & *(blue) & *(blue) & *(blue) & *(blue) & *(blue) & *(blue) \\
\none[] & \none[] & \none[] & \none[] & \none[] & \none[] & *(yellow) & *(yellow) & *(yellow) & *(yellow) & *(yellow) & *(yellow) & *(yellow) & *(yellow) \\ \none[] \\
*(blue) & *(blue) & *(blue) & *(blue) & *(blue) & *(blue) & *(blue) & *(blue) \\
\none[] & \none[] & \none[] & \none[] & \none[] & \none[] & *(yellow) & *(yellow) & *(yellow) & *(yellow) & *(yellow) & *(yellow) & *(yellow) & *(yellow) \\
\end{ytableau}};
    \node[] (5dLyt) at (3,-2) {\ytableausetup{centertableaux,boxsize=6pt}
\begin{ytableau}
*(yellow) & *(yellow) & *(yellow) & *(yellow) & *(yellow) & *(yellow) & *(yellow) & *(yellow) \\
\none[] & \none[] & \none[] & \none[] & \none[] & \none[] & *(blue) & *(blue) & *(blue) & *(blue) & *(blue) & *(blue) & *(blue) & *(blue) \\ \none[] \\
*(yellow) & *(yellow) & *(yellow) & *(yellow) & *(yellow) & *(yellow) & *(yellow) & *(yellow) \\
\none[] & \none[] & \none[] & \none[] & \none[] & \none[] & *(blue) & *(blue) & *(blue) & *(blue) & *(blue) & *(blue) & *(blue) & *(blue) \\
\end{ytableau}};
\draw[black, thick] (0,-3) -- (0,-1);
\filldraw[black] (0,-3) circle (0pt) node[anchor=north]{$\mu_{t,4}$};
\node[] (5dLyt) at (9,-2) {\ytableausetup{centertableaux,boxsize=6pt}
\begin{ytableau}
*(blue) & *(blue) & *(blue) & *(blue) & *(blue) & *(blue) & *(blue) & *(blue) \\
\none[] & \none[] & \none[] & \none[] & \none[] & \none[] & *(yellow) & *(yellow) & *(yellow) & *(yellow) & *(yellow) & *(yellow) & *(yellow) & *(yellow) \\ \none[] \\
*(blue) & *(blue) & *(blue) & *(blue) & *(blue) & *(blue) & *(blue) & *(blue) \\
\none[] & \none[] & \none[] & \none[] & \none[] & \none[] & *(yellow) & *(yellow) & *(yellow) & *(yellow) & *(yellow) & *(yellow) & *(yellow) & *(yellow) \\
\end{ytableau}};
\draw[black, thick] (6,-3) -- (6,-1);
\filldraw[black] (6,-3) circle (0pt) node[anchor=north]{$\mu_{t,4}$};
\end{tikzpicture}
\caption{The box diagram for the configuration in Figure~\ref{fig: local mus4s4s dom wall}. Each application of $\mu_{t,4}$ inverts the coloring of the $\mathfrak{so}(16)$ weights.}
\label{Fig: box diagram n4 third region}
\end{figure}
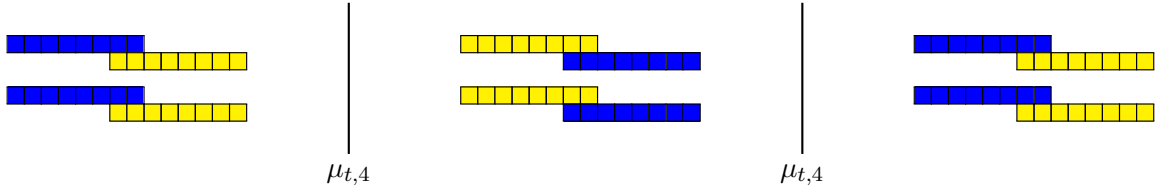

\paragraph{$(\mu_{t,4}^{\varepsilon}\mu_s^{\varepsilon}\mu_{t,4}^{\varepsilon}\mu_s^{\varepsilon})^2$ Configuration.} The final configuration we need to understand before closing up our line over which the threefold is fibered is the volume profile associated to 
\begin{equation}
(\mu_{t,4}^{\varepsilon}\mu_s^{\varepsilon}\mu_{t,4}^{\varepsilon}\mu_s^{\varepsilon})^2 \, ,
\end{equation}
shown in Figure~\ref{fig: local s4s4s4s4 domain wall}. 
\begin{figure}
    \centering
\begin{tikzpicture}[scale=0.9]


    \filldraw[color=red, fill=red!20, thick] (-8,-1.02) ellipse (0.3 and 0.97);
    \draw[red, thin] (-8,-1) ellipse (0.3 and 0.14);

    \filldraw[color=red, fill=red!20, thick] (-4,-1.02) ellipse (0.3 and 0.97);
    \draw[red, thin] (-4,-1) ellipse (0.3 and 0.14);

    \filldraw[color=red, fill=red!20, thick] (0,-1.02) ellipse (0.3 and 0.97);
    \draw[red, thin] (0,-1) ellipse (0.3 and 0.14);

    \filldraw[color=red, fill=red!20, thick] (4,-1.02) ellipse (0.3 and 0.97);
    \draw[red, thin] (4,-1) ellipse (0.3 and 0.14);

    \filldraw[color=red, fill=red!20, thick] (-9,-1.5) ellipse (0.2 and 0.49);
    \draw[red, thin] (-9,-1.5) ellipse (0.2 and 0.1);

    \filldraw[color=red, fill=red!20, thick] (-7,-1.53) ellipse (0.2 and 0.47);
    \draw[red, thin] (-7,-1.5) ellipse (0.2 and 0.1);

    \filldraw[color=red, fill=red!20, thick] (-5,-1.53) ellipse (0.2 and 0.47);
    \draw[red, thin] (-5,-1.5) ellipse (0.2 and 0.1);

    \filldraw[color=red, fill=red!20, thick] (-3,-1.53) ellipse (0.2 and 0.47);
    \draw[red, thin] (-3,-1.5) ellipse (0.2 and 0.1);

    \filldraw[color=red, fill=red!20, thick] (-1,-1.53) ellipse (0.2 and 0.47);
    \draw[red, thin] (-1,-1.5) ellipse (0.2 and 0.1);

    \filldraw[color=red, fill=red!20, thick] (1,-1.53) ellipse (0.2 and 0.47);
    \draw[red, thin] (1,-1.5) ellipse (0.2 and 0.1);

    \filldraw[color=red, fill=red!20, thick] (3,-1.53) ellipse (0.2 and 0.47);
    \draw[red, thin] (3,-1.5) ellipse (0.2 and 0.1);

    \filldraw[color=red, fill=red!20, thick] (5,-1.53) ellipse (0.2 and 0.47);
    \draw[red, thin] (5,-1.5) ellipse (0.2 and 0.1);

    \filldraw[color=red, fill=red!20, thick] (7,-1.53) ellipse (0.2 and 0.47);
    \draw[red, thin] (7,-1.5) ellipse (0.2 and 0.1);


    \filldraw[color=blue, fill=blue!20, thick] (-9,0) ellipse (0.3 and 0.95);
    \draw[blue, thin] (-9,0) ellipse (0.3 and 0.14);

    \filldraw[color=blue, fill=blue!20, thick] (-7,0) ellipse (0.3 and 0.95);
    \draw[blue, thin] (-7,0) ellipse (0.3 and 0.14);

    \filldraw[color=blue, fill=blue!20, thick] (-5,0) ellipse (0.3 and 0.95);
    \draw[blue, thin] (-5,0) ellipse (0.3 and 0.14);

    \filldraw[color=blue, fill=blue!20, thick] (-6,0) ellipse (0.4 and 1.95);
    \draw[blue, thin] (-6,0) ellipse (0.4 and 0.2);

    \filldraw[color=blue, fill=blue!20, thick] (-3,0) ellipse (0.3 and 0.95);
    \draw[blue, thin] (-3,0) ellipse (0.3 and 0.14);

    \filldraw[color=blue, fill=blue!20, thick] (-2,0) ellipse (0.4 and 1.95);
    \draw[blue, thin] (-2,0) ellipse (0.4 and 0.2);

    \filldraw[color=blue, fill=blue!20, thick] (-1,0) ellipse (0.3 and 0.95);
    \draw[blue, thin] (-1,0) ellipse (0.3 and 0.14);

    \filldraw[color=blue, fill=blue!20, thick] (1,0) ellipse (0.3 and 0.95);
    \draw[blue, thin] (1,0) ellipse (0.3 and 0.14);

    \filldraw[color=blue, fill=blue!20, thick] (2,0) ellipse (0.4 and 1.95);
    \draw[blue, thin] (2,0) ellipse (0.4 and 0.2);

    \filldraw[color=blue, fill=blue!20, thick] (3,0) ellipse (0.3 and 0.95);
    \draw[blue, thin] (3,0) ellipse (0.3 and 0.14);

    \filldraw[color=blue, fill=blue!20, thick] (5,0) ellipse (0.3 and 0.95);
    \draw[blue, thin] (5,0) ellipse (0.3 and 0.14);

    \filldraw[color=blue, fill=blue!20, thick] (6,0) ellipse (0.4 and 1.95);
    \draw[blue, thin] (6,0) ellipse (0.4 and 0.2);

    \filldraw[color=blue, fill=blue!20, thick] (7,0) ellipse (0.3 and 1);
    \draw[blue, thin] (7,0) ellipse (0.3 and 0.14);


    \draw[color=red, thick] (-9,-2) -- (7,-2);

    \draw[thick] (-10,0) -- (8,0);
    \draw[color=mygreen, ultra thick] (-9,0) -- (7,0);
    \draw[color=blue, thick] (-9,-1) -- (-8,0);
    \draw[color=blue, thick] (-8,0) -- (-6,-2);
    \draw[color=blue, thick] (-6,-2) -- (-4,0);
    \draw[color=blue, thick] (-4,0) -- (-2,-2);
    \draw[color=blue, thick] (-2,-2) -- (0,0);
    \draw[color=blue, thick] (0,0) -- (2,-2);
    \draw[color=blue, thick] (2,-2) -- (4,0);
    \draw[color=blue, thick] (4,0) -- (6,-2);
    \draw[color=blue, thick] (6,-2) -- (7,-1);

    \draw[color=blue, thick] (-9,1) -- (-8,0);
    \draw[color=blue, thick] (-8,0) -- (-6,2);
    \draw[color=blue, thick] (-6,2) -- (-4,0);
    \draw[color=blue, thick] (-4,0) -- (-2,2);
    \draw[color=blue, thick] (-2,2) -- (0,0);
    \draw[color=blue, thick] (0,0) -- (2,2);
    \draw[color=blue, thick] (2,2) -- (4,0);
    \draw[color=blue, thick] (4,0) -- (6,2);
    \draw[color=blue, thick] (6,2) -- (7,1);

   \filldraw[black] (-10,0) circle (0pt)
        node[anchor=east]{\textcolor{mygreen}{\large{$F$}}};

    \filldraw[black] (-10,1.00) circle (0pt)
        node[anchor=east]{\textcolor{blue}{\large{$\Sigma$}}};

    \filldraw[black] (-10,-2) circle (0pt)
        node[anchor=east]{\textcolor{red}{\large{$M_i$}}};

    \filldraw[black] (-8.00,0.2) circle (0pt)
        node[anchor=south]{\large{$\mu_s$}};

    \filldraw[black] (-6,0.3) circle (0pt)
        node[anchor=south]{\large{$\mu_{t,4}$}};

    \filldraw[black] (-4.00,0.2) circle (0pt)
        node[anchor=south]{\large{$\mu_s$}};

    \filldraw[black] (-2,0.3) circle (0pt)
        node[anchor=south]{\large{$\mu_{t,4}$}};

    \filldraw[black] (0,0.2) circle (0pt)
        node[anchor=south]{\large{$\mu_s$}};

    \filldraw[black] (2,0.2) circle (0pt)
        node[anchor=south]{\large{$\mu_{t,4}$}};

    \filldraw[black] (4,0.2) circle (0pt)
        node[anchor=south]{\large{$\mu_s$}};

    \filldraw[black] (6,0.2) circle (0pt)
        node[anchor=south]{\large{$\mu_{t,4}$}};

    \filldraw[black] (-8.00,0) circle (0pt) node{\Large{$\times$}};
    \filldraw[black] (-6.00,0) circle (0pt) node{\Large{$\times$}};
    \filldraw[black] (-4.00,0) circle (0pt) node{\Large{$\times$}};
    \filldraw[black] (-2.00,0) circle (0pt) node{\Large{$\times$}};
    \filldraw[black] (0.00,0) circle (0pt) node{\Large{$\times$}};
    \filldraw[black] (2.00,0) circle (0pt) node{\Large{$\times$}};
    \filldraw[black] (4.00,0) circle (0pt) node{\Large{$\times$}};
    \filldraw[black] (6.00,0) circle (0pt) node{\Large{$\times$}};


\begin{scope}[shift={(0,-5)}]
   \tikzstyle{every node}=[font=\scriptsize]
    \draw[thick] (-2.5,-1.5) rectangle (-1.5,-0.5);
    \filldraw[black] (-2,-1) circle (0pt) node{\Large{$8$}};
    \filldraw[black] (-9,1.5) circle (0pt) node{\Large{$2$}};
    \filldraw[black] (-6,1.5) circle (0pt) node{\Large{$2$}};
    \filldraw[black] (-2,1.5) circle (0pt) node{\Large{$2$}};
    \filldraw[black] (2,1.5) circle (0pt) node{\Large{$2$}};
    \filldraw[black] (6,1.5) circle (0pt) node{\Large{$2$}};
    \draw[thick](-6,1.5) circle (0.5);
    \draw[thick](-2,1.5) circle (0.5);
    \draw[thick](2,1.5) circle (0.5);
    \draw[thick, ->-] (-5.6,1.2) -- (-2.5,-0.7) ;
    \draw[thick, -<-] (-2,1) -- (-2,-0.5) ;
    \draw[thick, ->-] (1.6,1.2) -- (-1.5,-0.7) ;
    \draw[thick, -<-] (5.6,1.2) -- (-1.5,-1) ;
       
    \draw[thick]  (-9,1) arc (-90:90:0.5);
    \draw[thick]  (6,1) arc (270:90:0.5);

    \draw[thick] (-8.5,1.5) -- (-6.5,1.5);
    \draw[thick] (-5.5,1.5) -- (-2.5,1.5);
    \draw[thick] (-1.5,1.5) -- (1.5,1.5);
    \draw[thick] (2.5,1.5) -- (5.5,1.5);
    \node at (-8,1.5) {\large $\times$};
    \node at (-4,1.5) {\large $\times$};
    \node at (0,1.5) {\large $\times$};
    \node at (4,1.5) {\large $\times$};


\node[anchor=west] at (6,2) {\Large $\cdots$};
\node[anchor=west] at (6,1) {\Large $\cdots$};

\node[anchor=east] at (-9,2) {\Large $\cdots$};
\node[anchor=east] at (-9,1) {\Large $\cdots$};
    
\end{scope}

\end{tikzpicture}
    \caption{The local geometry arising from the concatenation 
    $(\mu_{t,4}^{\varepsilon}\mu_s^{\varepsilon}\mu_{t,4}^{\varepsilon}\mu_s^{\varepsilon})^2$.
    The fibration of the section $\Sigma$ over three compact intervals, whose endpoints are given by the location of the $\mu_s$ domain walls, gives rise to three 3-spheres over which an $A_1$ singularity is fibered. The $\mathbb{P}^1$'s associated to the curve classes $M_{i}$ collapse at the locations of the $\mu_{t,4}$ domain walls, meaning the $A_{7}$ singularity touches the $A_1$ singularity at these four points. }
    \label{fig: local s4s4s4s4 domain wall}
\end{figure}
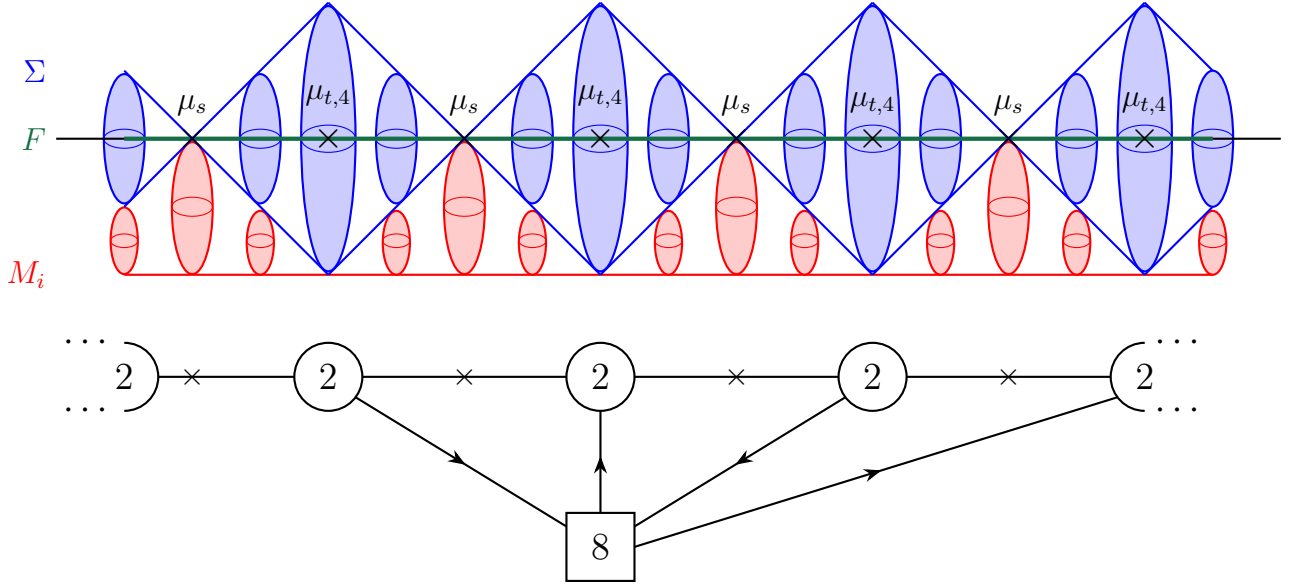
The box graph description for this is shown in Figure~\ref{Fig: box diagram n4 fifth region}. 
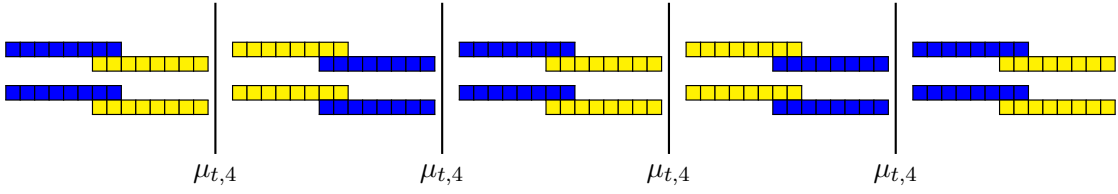
\begin{figure}
\centering
\begin{tikzpicture}
\node[] (5dLyt) at (-3,-2) {\ytableausetup{centertableaux,boxsize=5pt}
\begin{ytableau}
*(blue) & *(blue) & *(blue) & *(blue) & *(blue) & *(blue) & *(blue) & *(blue) \\
\none[] & \none[] & \none[] & \none[] & \none[] & \none[] & *(yellow) & *(yellow) & *(yellow) & *(yellow) & *(yellow) & *(yellow) & *(yellow) & *(yellow) \\ \none[] \\
*(blue) & *(blue) & *(blue) & *(blue) & *(blue) & *(blue) & *(blue) & *(blue) \\
\none[] & \none[] & \none[] & \none[] & \none[] & \none[] & *(yellow) & *(yellow) & *(yellow) & *(yellow) & *(yellow) & *(yellow) & *(yellow) & *(yellow) \\
\end{ytableau}};
    \node[] (5dLyt) at (0,-2) {\ytableausetup{centertableaux,boxsize=5pt}
\begin{ytableau}
*(yellow) & *(yellow) & *(yellow) & *(yellow) & *(yellow) & *(yellow) & *(yellow) & *(yellow) \\
\none[] & \none[] & \none[] & \none[] & \none[] & \none[] & *(blue) & *(blue) & *(blue) & *(blue) & *(blue) & *(blue) & *(blue) & *(blue) \\ \none[] \\
*(yellow) & *(yellow) & *(yellow) & *(yellow) & *(yellow) & *(yellow) & *(yellow) & *(yellow) \\
\none[] & \none[] & \none[] & \none[] & \none[] & \none[] & *(blue) & *(blue) & *(blue) & *(blue) & *(blue) & *(blue) & *(blue) & *(blue) \\
\end{ytableau}};
\draw[black, thick] (-1.5,-3) -- (-1.5,-1);
\filldraw[black] (-1.5,-3) circle (0pt) node[anchor=north]{$\mu_{t,4}$};
\node[] (5dLyt) at (3,-2) {\ytableausetup{centertableaux,boxsize=5pt}
\begin{ytableau}
*(blue) & *(blue) & *(blue) & *(blue) & *(blue) & *(blue) & *(blue) & *(blue) \\
\none[] & \none[] & \none[] & \none[] & \none[] & \none[] & *(yellow) & *(yellow) & *(yellow) & *(yellow) & *(yellow) & *(yellow) & *(yellow) & *(yellow) \\ \none[] \\
*(blue) & *(blue) & *(blue) & *(blue) & *(blue) & *(blue) & *(blue) & *(blue) \\
\none[] & \none[] & \none[] & \none[] & \none[] & \none[] & *(yellow) & *(yellow) & *(yellow) & *(yellow) & *(yellow) & *(yellow) & *(yellow) & *(yellow) \\
\end{ytableau}};
    \node[] (5dLyt) at (6,-2) {\ytableausetup{centertableaux,boxsize=5pt}
\begin{ytableau}
*(yellow) & *(yellow) & *(yellow) & *(yellow) & *(yellow) & *(yellow) & *(yellow) & *(yellow) \\
\none[] & \none[] & \none[] & \none[] & \none[] & \none[] & *(blue) & *(blue) & *(blue) & *(blue) & *(blue) & *(blue) & *(blue) & *(blue) \\ \none[] \\
*(yellow) & *(yellow) & *(yellow) & *(yellow) & *(yellow) & *(yellow) & *(yellow) & *(yellow) \\
\none[] & \none[] & \none[] & \none[] & \none[] & \none[] & *(blue) & *(blue) & *(blue) & *(blue) & *(blue) & *(blue) & *(blue) & *(blue) \\
\end{ytableau}};
\node[] (5dLyt) at (9,-2) {\ytableausetup{centertableaux,boxsize=5pt}
\begin{ytableau}
*(blue) & *(blue) & *(blue) & *(blue) & *(blue) & *(blue) & *(blue) & *(blue) \\
\none[] & \none[] & \none[] & \none[] & \none[] & \none[] & *(yellow) & *(yellow) & *(yellow) & *(yellow) & *(yellow) & *(yellow) & *(yellow) & *(yellow) \\ \none[] \\
*(blue) & *(blue) & *(blue) & *(blue) & *(blue) & *(blue) & *(blue) & *(blue) \\
\none[] & \none[] & \none[] & \none[] & \none[] & \none[] & *(yellow) & *(yellow) & *(yellow) & *(yellow) & *(yellow) & *(yellow) & *(yellow) & *(yellow) \\
\end{ytableau}};
\draw[black, thick] (1.5,-3) -- (1.5,-1);
\filldraw[black] (1.5,-3) circle (0pt) node[anchor=north]{$\mu_{t,4}$};
\draw[black, thick] (4.5,-3) -- (4.5,-1);
\filldraw[black] (4.5,-3) circle (0pt) node[anchor=north]{$\mu_{t,4}$};
\draw[black, thick] (7.5,-3) -- (7.5,-1);
\filldraw[black] (7.5,-3) circle (0pt) node[anchor=north]{$\mu_{t,4}$};
\end{tikzpicture}
\caption{The box diagram for the configuration in Figure~\ref{fig: local s4s4s4s4 domain wall}. 
}
\label{Fig: box diagram n4 fifth region}
\end{figure}

\paragraph{The $(\mu_{t,4}^{\varepsilon}\mu_s^{\varepsilon}\mu_{t,4}^{\varepsilon}\mu_s^{\varepsilon})^2$ Flux Torus.}
Finally, let us turn the flux tube with domain wall configuration $(\mu_{t,4}^{\varepsilon}\mu_s^{\varepsilon}\mu_{t,4}^{\varepsilon}\mu_s^{\varepsilon})^2$ 
into a flux torus. As this configuration realizes the flux 
$(-2^8) = N(4\alpha_1)$ we denote it by
\begin{equation}
\mathcal{D}[4 \alpha_1]:= (\mu_{t,4}^{\varepsilon}\mu_s^{\varepsilon}\mu_{t,4}^{\varepsilon}\mu_s^{\varepsilon})^2\, .
\end{equation}
In order to glue the two ends of the flux tube we hence need to identify them by
\begin{equation}
    \mathcal{M}[4 \alpha_1] := (\mu_{t,4}\mu_{s}\mu_{t,4}\mu_s)^{-2}    =
    (\mu_{s}\mu_{t,4}\mu_s\mu_{t,4})^2 \, .
\end{equation}
where we have used that $\mu_s^2 = \mu_{t,4}^2=1$. We can work out this map to be 
\be
\mathcal{M}[4 \alpha_1]:  \qquad 
\left\{\begin{aligned}
F &\rightarrow F+ 10 E  -4 \alpha_1 \\
\Sigma &\rightarrow \Sigma +6E - 4 \alpha_1 \\
M_i &\rightarrow M_i +4 E -2\alpha_1 \, , \qquad i=1\cdots 8 \, .
\end{aligned}\right. \, .
\ee 
This map can be expressed as 
\begin{equation} \label{eq: Action of M[-4a]}
\mathcal{M}[4 \alpha_1]: \Gamma \rightarrow \Gamma + 2(\alpha_1 \cdot \Gamma ) E - 2 (E \cdot \Gamma )\alpha_1 + 4 (E \cdot \Gamma) E
\end{equation}
for all $\Gamma \in I_{1,9}$. Here the inner form used is just the inner form on $I_{1,9}$. Comparing with \eqref{eq: translation by sigma alpha 1}, we see that 
$\mathcal{M}[4 \alpha_1] = \tau_{-\alpha_1}^2$, i.e. 
the map 
$\mathcal{M}[4 \alpha_1]$ is induced by an automorphism of 
$X_{S}^f$ associated with two translations along 
$-\alpha_1$. We can hence close the geometry to 
$\mathcal{T}[4\alpha_1]$ by this automorphism and have found the geometric analogue of the flux torus.

From the field theory perspective, the result of closing the volume profile in Figure~\ref{fig: local s4s4s4s4 domain wall} is the volume profile and quiver of the flux torus, shown in Figure~\ref{fig: n4 volumes and quiver closed loop}. 
In closing up the profile we have used an automorphism of the geometry, which must in particular identify the $A_1$ singularities associated with the two open $SU(2)$ nodes, which hence become identified and give rise to a single gauge $SU(2)$ node. Now that the geometry has been closed the associated 4d quiver is physically meaningful. 
\begin{figure}
    \centering
\begin{tikzpicture}[scale=0.7]


    \filldraw[color=blue, fill=blue!20, thick] (0,-3.6) ellipse (0.6 and 1.28);
    \draw[blue, thin] (0,-3.8) ellipse (0.58 and 0.2);

    \begin{scope}[shift={(-1.6,-3.49)}, rotate=-10]
        \filldraw[color=blue, fill=blue!20, thick] (0,0) ellipse (0.5 and 1.12);
        \draw[blue, thin] (0,-0.2) ellipse (0.48 and 0.17);
    \end{scope}

    \begin{scope}[shift={(1.6,-3.49)}, rotate=10]
        \filldraw[color=blue, fill=blue!20, thick] (0,0) ellipse (0.5 and 1.12);
        \draw[blue, thin] (0,-0.2) ellipse (0.48 and 0.17);
    \end{scope}

    \begin{scope}[shift={(-2.8,-3.28)}, rotate=-18]
        \filldraw[color=blue, fill=blue!20, thick] (0,0) ellipse (0.3 and 0.79);
        \draw[blue, thin] (0,-0.12) ellipse (0.3 and 0.1);
    \end{scope}

    \begin{scope}[shift={(2.8,-3.28)}, rotate=18]
        \filldraw[color=blue, fill=blue!20, thick] (0,0) ellipse (0.3 and 0.79);
        \draw[blue, thin] (0,-0.12) ellipse (0.3 and 0.1);
    \end{scope}


    \begin{scope}[shift={(0,3.6)}, rotate=180]
    \filldraw[color=blue, fill=blue!20, thick] (0,0) ellipse (0.6 and 1.28);
    \draw[blue, thin] (0,-0.2) ellipse (0.58 and 0.2);
   \end{scope}

    \begin{scope}[shift={(1.6,3.49)}, rotate=170]
        \filldraw[color=blue, fill=blue!20, thick] (0,0) ellipse (0.5 and 1.12);
        \draw[blue, thin] (0,-0.2) ellipse (0.48 and 0.17);
    \end{scope}

    \begin{scope}[shift={(-1.6,3.49)}, rotate=190]
        \filldraw[color=blue, fill=blue!20, thick] (0,0) ellipse (0.5 and 1.12);
        \draw[blue, thin] (0,-0.2) ellipse (0.48 and 0.17);
    \end{scope}

    \begin{scope}[shift={(2.8,3.28)}, rotate=162]
        \filldraw[color=blue, fill=blue!20, thick] (0,0) ellipse (0.3 and 0.79);
        \draw[blue, thin] (0,-0.12) ellipse (0.3 and 0.1);
    \end{scope}

    \begin{scope}[shift={(-2.8,3.28)}, rotate=198]
        \filldraw[color=blue, fill=blue!20, thick] (0,0) ellipse (0.3 and 0.79);
        \draw[blue, thin] (0,-0.12) ellipse (0.3 and 0.1);
    \end{scope}


    \begin{scope}[shift={(-5.87,0)}, rotate=90]
        \filldraw[color=blue, fill=blue!20, thick] (0,0) ellipse (0.3 and 0.74);
        \draw[blue, thin] (0,0.02) ellipse (0.29 and 0.1);
    \end{scope}

    \begin{scope}[shift={(-5.72,-1.02)}, rotate=112]
        \filldraw[color=blue, fill=blue!20, thick] (0,0) ellipse (0.27 and 0.67);
        \draw[blue, thin] (0,-0.02) ellipse (0.26 and 0.09);
    \end{scope}

    \begin{scope}[shift={(-5.72,1.02)}, rotate=68]
        \filldraw[color=blue, fill=blue!20, thick] (0,0) ellipse (0.27 and 0.67);
        \draw[blue, thin] (0,-0.02) ellipse (0.26 and 0.09);
    \end{scope}

    \begin{scope}[shift={(-5.3,-1.9)}, rotate=130]
        \filldraw[color=blue, fill=blue!20, thick] (0,0) ellipse (0.22 and 0.49);
        \draw[blue, thin] (0,-0.1) ellipse (0.22 and 0.08);
    \end{scope}

   \begin{scope}[shift={(-5.3,1.9)}, rotate=50]
        \filldraw[color=blue, fill=blue!20, thick] (0,0) ellipse (0.22 and 0.49);
        \draw[blue, thin] (0,-0.1) ellipse (0.22 and 0.08);
    \end{scope}


    \begin{scope}[shift={(5.87,0)}, rotate=-90]
        \filldraw[color=blue, fill=blue!20, thick] (0,0) ellipse (0.3 and 0.74);
        \draw[blue, thin] (0,0.02) ellipse (0.29 and 0.1);
    \end{scope}

    \begin{scope}[shift={(5.72,-1.02)}, rotate=-112]
        \filldraw[color=blue, fill=blue!20, thick] (0,0) ellipse (0.27 and 0.67);
        \draw[blue, thin] (0,-0.02) ellipse (0.26 and 0.09);
    \end{scope}

    \begin{scope}[shift={(5.72,1.02)}, rotate=-68]
        \filldraw[color=blue, fill=blue!20, thick] (0,0) ellipse (0.27 and 0.67);
        \draw[blue, thin] (0,-0.02) ellipse (0.26 and 0.09);
    \end{scope}

    \begin{scope}[shift={(5.3,-1.9)}, rotate=-130]
        \filldraw[color=blue, fill=blue!20, thick] (0,0) ellipse (0.22 and 0.49);
        \draw[blue, thin] (0,-0.1) ellipse (0.22 and 0.08);
    \end{scope}

   \begin{scope}[shift={(5.3,1.9)}, rotate=-50]
        \filldraw[color=blue, fill=blue!20, thick] (0,0) ellipse (0.22 and 0.49);
        \draw[blue, thin] (0,-0.1) ellipse (0.22 and 0.08);
    \end{scope}


    \begin{scope}[shift={(-4.31,-3.5)}, rotate=-30]
        \filldraw[color=red, fill=red!20, thick] (0,0) ellipse (0.25 and 0.56);
        \draw[red, thin] (0,0.02) ellipse (0.25 and 0.08);
    \end{scope}

    \begin{scope}[shift={(-5.9,-2.45)}, rotate=-50]
        \filldraw[color=red, fill=red!20, thick] (0,0) ellipse (0.14 and 0.31);
        \draw[red, thin] (0,-0.02) ellipse (0.14 and 0.04);
    \end{scope}

    \begin{scope}[shift={(-6.57,-1.34)}, rotate=-68]
        \filldraw[color=red, fill=red!20, thick] (0,0) ellipse (0.14 and 0.2);
        \draw[red, thin] (0,-0.02) ellipse (0.14 and 0.04);
    \end{scope}

    \begin{scope}[shift={(-3.2,-4.34)}, rotate=-20]
        \filldraw[color=red, fill=red!20, thick] (0,0) ellipse (0.2 and 0.33);
        \draw[red, thin] (0,-0.02) ellipse (0.2 and 0.07);
    \end{scope}

    \begin{scope}[shift={(-1.85,-4.83)}, rotate=-10]
        \filldraw[color=red, fill=red!20, thick] (0,0) ellipse (0.15 and 0.22);
        \draw[red, thin] (0,-0.02) ellipse (0.15 and 0.05);
    \end{scope}


    \begin{scope}[shift={(-4.31,3.5)}, rotate=30]
        \filldraw[color=red, fill=red!20, thick] (0,0) ellipse (0.25 and 0.56);
        \draw[red, thin] (0,0.02) ellipse (0.25 and 0.08);
    \end{scope}

    \begin{scope}[shift={(-5.9,2.45)}, rotate=50]
        \filldraw[color=red, fill=red!20, thick] (0,0) ellipse (0.14 and 0.31);
        \draw[red, thin] (0,-0.02) ellipse (0.14 and 0.04);
    \end{scope}

    \begin{scope}[shift={(-6.57,1.34)}, rotate=68]
        \filldraw[color=red, fill=red!20, thick] (0,0) ellipse (0.14 and 0.2);
        \draw[red, thin] (0,-0.02) ellipse (0.14 and 0.04);
    \end{scope}

    \begin{scope}[shift={(-3.2,4.34)}, rotate=20]
        \filldraw[color=red, fill=red!20, thick] (0,0) ellipse (0.2 and 0.33);
        \draw[red, thin] (0,-0.02) ellipse (0.2 and 0.07);
    \end{scope}

    \begin{scope}[shift={(-1.85,4.83)}, rotate=10]
        \filldraw[color=red, fill=red!20, thick] (0,0) ellipse (0.15 and 0.22);
        \draw[red, thin] (0,-0.02) ellipse (0.15 and 0.05);
    \end{scope}


    \begin{scope}[shift={(4.31,3.5)}, rotate=-30]
        \filldraw[color=red, fill=red!20, thick] (0,0) ellipse (0.25 and 0.56);
        \draw[red, thin] (0,0.02) ellipse (0.25 and 0.08);
    \end{scope}

    \begin{scope}[shift={(5.9,2.45)}, rotate=-50]
        \filldraw[color=red, fill=red!20, thick] (0,0) ellipse (0.14 and 0.31);
        \draw[red, thin] (0,-0.02) ellipse (0.14 and 0.04);
    \end{scope}

    \begin{scope}[shift={(6.57,1.34)}, rotate=-68]
        \filldraw[color=red, fill=red!20, thick] (0,0) ellipse (0.14 and 0.2);
        \draw[red, thin] (0,-0.02) ellipse (0.14 and 0.04);
    \end{scope}

    \begin{scope}[shift={(3.2,4.34)}, rotate=-20]
        \filldraw[color=red, fill=red!20, thick] (0,0) ellipse (0.2 and 0.33);
        \draw[red, thin] (0,-0.02) ellipse (0.2 and 0.07);
    \end{scope}

    \begin{scope}[shift={(1.85,4.83)}, rotate=-10]
        \filldraw[color=red, fill=red!20, thick] (0,0) ellipse (0.15 and 0.22);
        \draw[red, thin] (0,-0.02) ellipse (0.15 and 0.05);
    \end{scope}


    \begin{scope}[shift={(4.31,-3.5)}, rotate=30]
        \filldraw[color=red, fill=red!20, thick] (0,0) ellipse (0.25 and 0.56);
        \draw[red, thin] (0,0.02) ellipse (0.25 and 0.08);
    \end{scope}

    \begin{scope}[shift={(5.9,-2.45)}, rotate=50]
        \filldraw[color=red, fill=red!20, thick] (0,0) ellipse (0.14 and 0.31);
        \draw[red, thin] (0,-0.02) ellipse (0.14 and 0.04);
    \end{scope}

    \begin{scope}[shift={(6.57,-1.34)}, rotate=68]
        \filldraw[color=red, fill=red!20, thick] (0,0) ellipse (0.14 and 0.2);
        \draw[red, thin] (0,-0.02) ellipse (0.14 and 0.04);
    \end{scope}

    \begin{scope}[shift={(3.2,-4.34)}, rotate=20]
        \filldraw[color=red, fill=red!20, thick] (0,0) ellipse (0.2 and 0.33);
        \draw[red, thin] (0,-0.02) ellipse (0.2 and 0.07);
    \end{scope}

    \begin{scope}[shift={(1.85,-4.83)}, rotate=10]
        \filldraw[color=red, fill=red!20, thick] (0,0) ellipse (0.15 and 0.22);
        \draw[red, thin] (0,-0.02) ellipse (0.15 and 0.05);
    \end{scope}


    \draw[thick] (0,0) ellipse (6 and 4);
    \draw[color=mygreen, ultra thick] (0,0) ellipse (6 and 4);

    \draw[color=red, thick] (6.63,0) .. controls (8,3.1) and (2,6) .. (0,4.88);
    \draw[color=red, thick] (-6.63,0) .. controls (-8,-3.1) and (-2,-6) .. (0,-4.88);
    \draw[color=red, thick] (-6.63,0) .. controls (-8,3.1) and (-2,6) .. (0,4.88);
    \draw[color=red, thick] (6.63,0) .. controls (8,-3.1) and (2,-6) .. (0,-4.88);

    \draw[color=blue, thick] (-4,-2.98) .. controls (-2.8,-2.1) and (2.8,-2.1) .. (4,-2.98);
    \draw[color=blue, thick] (-4,-2.98) .. controls (-2.7,-5.5) and (2.7,-5.5) .. (4,-2.98);

    \draw[color=blue, thick] (-4,2.98) .. controls (-2.8,5.5) and (2.8,5.5) .. (4,2.98);
    \draw[color=blue, thick] (-4,2.98) .. controls (-2.7,2.1) and (2.7,2.1) .. (4,2.98);

    \draw[color=blue, thick] (-4,2.98) .. controls (-7.5,2.5) and (-7.5,-2.5) .. (-4,-2.98);
    \draw[color=blue, thick] (-4,2.98) .. controls (-5.5,2.2) and (-5.5,-2.2) .. (-4,-2.98);

    \draw[color=blue, thick] (4,2.98) .. controls (7.5,2.5) and (7.5,-2.5) .. (4,-2.98);
    \draw[color=blue, thick] (4,2.98) .. controls (5.5,2.2) and (5.5,-2.2) .. (4,-2.98);

    \node at (6,0) {\large $\times$};    
    \node at (0,-4) {\large $\times$};  
    \node at (-6,0) {\large $\times$};    
    \node at (0,4) {\large $\times$};  
    \node at (-4,2.98) {\large $\times$};    
    \node at (4,2.98) {\large $\times$};  
    \node at (-4,-2.98) {\large $\times$};    
    \node at (4,-2.98) {\large $\times$}; 

    \filldraw[black] (0,-5.9) circle (0pt)
        node[anchor=south]{\large{$\mu_{t,4}$}};

    \filldraw[black] (-6.8,0) circle (0pt)
        node[anchor=east]{\large{$\mu_{t,4}$}};

    \filldraw[black] (0,5.9) circle (0pt)
        node[anchor=north]{\large{$\mu_{t,4}$}};

    \filldraw[black] (6.8,0) circle (0pt)
        node[anchor=west]{\large{$\mu_{t,4}$}};

    \filldraw[black] (-5,4.3) circle (0pt)
        node{\large{$\mu_s$}};
    \filldraw[black] (5,4.3) circle (0pt)
        node{\large{$\mu_s$}};
    \filldraw[black] (-5,-4.3) circle (0pt)
        node{\large{$\mu_s$}};
    \filldraw[black] (5,-4.3) circle (0pt)
        node{\large{$\mu_s$}};

\begin{scope}[shift={(-5.8,-2.35)}, rotate=310,scale=0.57]
         \draw[ultra thick] (0,2) arc(90:270:-0.3cm and 1.4cm);
         \draw[ultra thick, style=dashed] (0,2) arc(90:270:0.3cm and 1.4cm);
\end{scope}

\filldraw[black] (-9.25,-3) circle (0pt) node[anchor=north]{\large{$(\mu_{s}\mu_{t,4}\mu_s \mu_{t,4})^2$}};
\draw[thick, -<_] (-6.75,-2.75) arc (-270:20:0.4);


\begin{scope}[shift={(-4.5,-8)}, scale=1.5]

 \tikzstyle{every node}=[font=\scriptsize]
       \node[draw, circle,thick] (p1) at (0,-2.2) {\fontsize{12pt}{12pt}\selectfont $2$};
       \node[draw, circle,thick] (p2) at (3,0) {\fontsize{12pt}{12pt}\selectfont $2$};
       \node[draw, circle,thick] (p3) at (6,-2.2) {\fontsize{12pt}{12pt}\selectfont $2$};
       \node[draw, rectangle,thick] (p4) at (3,-2.2) {\fontsize{12pt}{12pt}\selectfont $8$};
       \node[draw, circle,thick] (p5) at (3,-4.4) {\fontsize{12pt}{12pt}\selectfont $2$};
       \draw[-,thick] (p1) to  node[rotate=45]{\large $\times$}    (p2);
       \draw[-,thick] (p2) to  node[rotate=45]{\large $\times$}    (p3);
       \draw[-,thick] (p3) to  node[rotate=45]{\large $\times$}    (p5);
       \draw[-,thick] (p5) to  node[rotate=45]{\large $\times$}    (p1);
       \draw[->-,thick] (p1) to      (p4);
       \draw[->-,thick] (p3) to      (p4);
       \draw[-<-,thick] (p2) to      (p4);
       \draw[-<-,thick] (p5) to      (p4);
    
\end{scope}

\end{tikzpicture}
    \caption{The geometry of the $\mathcal{D}[4 \alpha_1] = (\mu_{t,4}^{\varepsilon}\mu_s^{\varepsilon} \mu_{t,4}^{\varepsilon} \mu_s^{\varepsilon})^2 $ flux torus, and associated 4d $\cN=1$ quiver of flux $\mathcal{F}=(-2^{8})$. There are four blue compact 3-spheres, over which sits an $A_1$ singularity. The red matter curves $M_i$ touch each compact 3-sphere exactly once, giving rise to chirals charged under both the $SU(2)$ gauge and $SU(8)$ flavor symmetry, as shown in the quiver.}
    \label{fig: n4 volumes and quiver closed loop}
\end{figure}
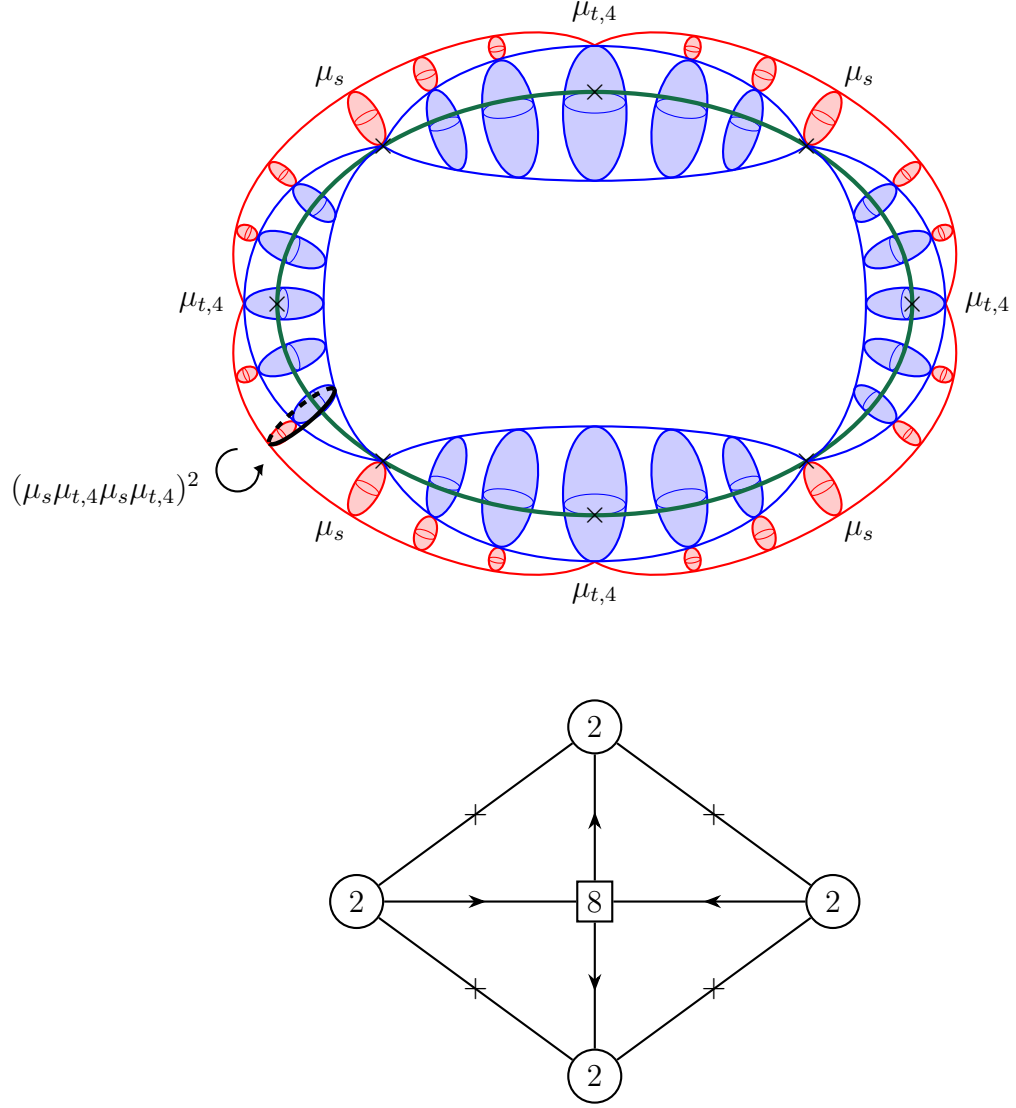

\paragraph{Enhanced Flavor Symmetry and Monodromy Map.}

The lattice $\Lambda_{E_8}$, which encodes the UV flavor symmetry of the 5D KK theory, is acted on non-trivially by the map 
$\mathcal{M}[4\alpha_1]$, which means that elements of $\Lambda_{E_8}$ encounter a non-trivial monodromy along the circle. The unbroken flavor symmetry is hence found by working out the sublattice
$\Lambda_{n=4} \subset \Lambda_{E_8}$ which is invariant under 
$\mathcal{M}[4\alpha_1]$. We can work out that the action on the roots of $E_8$ is 
\be\label{eq:Eshiftn4}
\mathcal{M}[4 \alpha_1]:  \qquad 
\left\{\begin{aligned}
\alpha_1 &\rightarrow  \alpha_1-4 E \\
\alpha_3 &\rightarrow  \alpha_3+2 E \\
\alpha_i &\rightarrow \alpha_i\, , \qquad \qquad i=2,4,5\cdots 8 \, ,
\end{aligned}\right. 
\ee
so that $\Lambda_{n=4} = \alpha_1^\perp \subset \Lambda_{E_8}$. This is an even, negative definite lattice of discriminant $2$ which determines it to be $\Lambda_{n=4} = -E_7$. Alternatively, one can work out that the invariant sublattice of $\Lambda_{E_8}$ is spanned by the simple roots
\be \label{eq: roots of invariant E7}
    \Lambda_{n=4} =\langle  \alpha_1+\alpha_2+2\alpha_3+2\alpha_4+\alpha_5, \alpha_2, \alpha_4, \alpha_5, \alpha_6, \alpha_7, \alpha_8 \rangle_{\mathbb{Z}} \, ,
\ee
which have (minus) the Cartan matrix of $E_7$ as their inner form. This is exactly the enhanced flavor symmetry~\eqref{eq:preserved flavor symms} appropriate for the quiver theory constructed in Figure~\ref{fig: n4 volumes and quiver closed loop}. 

As alluded in the Introduction and in Section \ref{sec: Field Theory sec2}, the flux tubes possess a $U(1)$ symmetry descending from the KK symmetry, which is broken to a finite subgroup when closing the tube to a torus due to anomaly and superpotential constraints \cite{Hwang:2021xyw,Sabag:2022hyw}. This was understood field-theoretically in \cite{Sabag:2022hyw} from the fact that the masses of the 5d hypers get shifted by integer multiples of $m_{\text{KK}}$ as we go around the circle. Remembering that $\text{Vol}(E)=m_{\text{KK}}$, we see that such a shift observed in field theory is also realized in the geometry as the map \eqref{eq:Eshiftn4}. In Appendix~\ref{app:KKshift} we explain more in details this agreement between field theory and geometry, also for the cases corresponding to other values of $n$ that we will discuss next.

As a variant of this construction, we may realize fluxes of the form
$2 \ell \alpha_1$ corresponding to $(-\ell^8)$ as 
$\mathcal{D}[2\ell \alpha_1] = (\mu_{t,4}^{\varepsilon}\mu_s^{\varepsilon}\mu_{t,4}^{\varepsilon}\mu_s^{\varepsilon})^\ell$, and find that 
$\mathcal{M}[2\ell \alpha_1] =  (\mu_{s}\mu_{t,4}\mu_s\mu_{t,4})^\ell$
\begin{equation}
    \mathcal{M}[2\ell  \alpha_1]: \Gamma \rightarrow \Gamma + \ell(\alpha_1 \cdot \Gamma ) E - \ell (E \cdot \Gamma )\alpha_1 + \ell^2 (E \cdot \Gamma) E
\end{equation}
so that 
$\mathcal{M}[2\ell  \alpha_1] = \tau_{-\alpha_1}^\ell$ and we can glue with the $\ell$'th power of the automorphism 
$\tau_{-\alpha_1}$ to find the geometry of
$\mathcal{T}[2\ell \alpha_1]$.

The action on $\Lambda_{E_8}$ is
\be 
\mathcal{M}[2\ell  \alpha_1] :  \qquad 
\left\{\begin{aligned}
\alpha_1 &\rightarrow \alpha_1 - 2 \ell E\\
 \alpha_3 &\rightarrow \alpha_3+ \ell E \\
\alpha_i &\rightarrow \alpha_i \qquad \qquad i=2,4,5\cdots 8 \, ,
\end{aligned}\right.
\ee
The preserved $E_7$ flavor symmetry is again present in these examples, as expedected from field theory. Moreover, we see again the shift by integer multiples of $E$, compatible with the field theory expectations.

{Note that for the case $n=4$ considered here, we can even consider the case $\ell=1$ and construct a quiver realized by $\mathcal{D}[2 \alpha_1]= \mu_{t,4}^{\varepsilon}\mu_s^{\varepsilon}\mu_{t,4}^{\varepsilon}\mu_s^{\varepsilon}$. Indeed, we can still close up this configuration by the monodromy map $\mathcal{M}[2 \alpha_1] = \tau_{-\alpha_1}$ corresponding by a single translation by $-\alpha_1$. This is compatible with the fact that the associated flux $(1^8)$ is still an $E_8$ flux that breaks it to $E_7\times U(1)$ and that is not a center flux, namely it corresponds to a unit flux for the $U(1)$ and no flux for the center of $E_7$. This will not be the case for other values of $n$ that we will consider next, since fluxes of the form $(-1^{2n},0^{8-2n})$ are center fluxes which we do not address in this work. We will discuss more general configurations related to non-center fluxes in Section \ref{sect:elemfluxtubes}.}

\subsection{Stacking Domain Walls: Flux $(-2^{2n},0^{8-2n})$} \label{sec: flux 2^2n}

\paragraph{$\mu_s^{\varepsilon} \mu_{t,n}^{\varepsilon} \mu_s^{\varepsilon}$ Configuration.} A useful component that we will require in the subsequent concatenation is to glue together the local building blocks in Figures~\ref{fig: mus local volumes} and~\ref{fig: local mutn dom wall} as dictated by this set of maps. This yields the configuration shown in Figure~\ref{fig: local musns dom wall}.
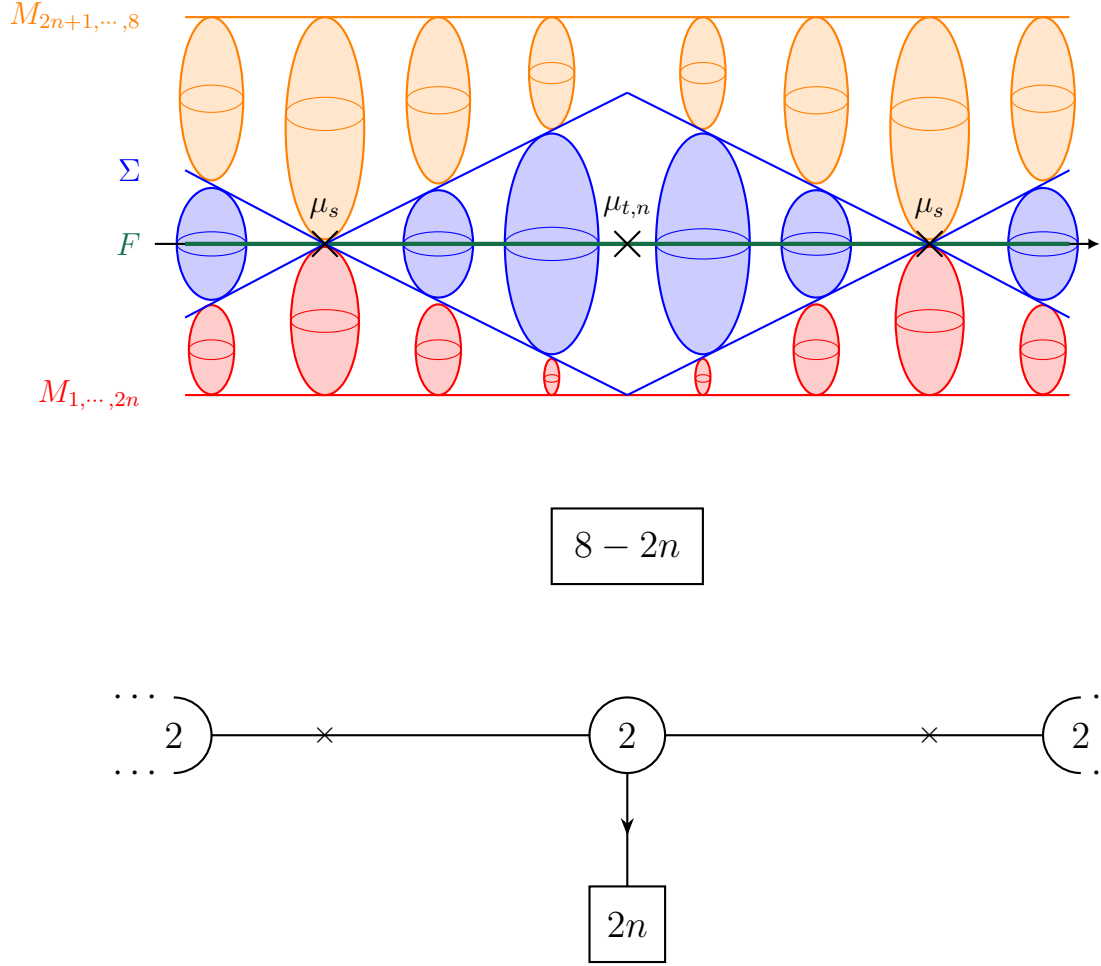
\begin{figure}
    \centering
    \begin{tikzpicture}
        \begin{scope}[shift={(0,0)}]

    \filldraw[color=red, fill=red!20, thick] (-5.5,-1.4) ellipse (0.3 and 0.59);
    \draw[red, thin] (-5.5,-1.4) ellipse (0.3 and 0.13);

    \filldraw[color=red, fill=red!20, thick] (-4,-1.02) ellipse (0.45 and 0.98);
    \draw[red, thin] (-4,-1.02) ellipse (0.45 and 0.15);

    \filldraw[color=red, fill=red!20, thick] (-2.5,-1.4) ellipse (0.3 and 0.6);
    \draw[red, thin] (-2.5,-1.4) ellipse (0.3 and 0.13);

    \filldraw[color=red, fill=red!20, thick] (-1,-1.76) ellipse (0.1 and 0.24);
    \draw[red, thin] (-1,-1.78) ellipse (0.1 and 0.05);

    \filldraw[color=red, fill=red!20, thick] (5.5,-1.4) ellipse (0.3 and 0.59);
    \draw[red, thin] (5.5,-1.4) ellipse (0.3 and 0.13);

    \filldraw[color=red, fill=red!20, thick] (4,-1.02) ellipse (0.45 and 0.98);
    \draw[red, thin] (4,-1.02) ellipse (0.45 and 0.15);

    \filldraw[color=red, fill=red!20, thick] (2.5,-1.4) ellipse (0.3 and 0.6);
    \draw[red, thin] (2.5,-1.4) ellipse (0.3 and 0.13);

    \filldraw[color=red, fill=red!20, thick] (1,-1.76) ellipse (0.1 and 0.24);
    \draw[red, thin] (1,-1.78) ellipse (0.1 and 0.05);


    \filldraw[color=orange, fill=orange!20, thick] (-5.5,1.92) ellipse (0.42 and 1.08);
    \draw[orange, thin] (-5.5,1.92) ellipse (0.42 and 0.18);

    \filldraw[color=orange, fill=orange!20, thick] (-4,1.53) ellipse (0.52 and 1.47);
    \draw[orange, thin] (-4,1.72) ellipse (0.52 and 0.22);

    \filldraw[color=orange, fill=orange!20, thick] (-2.5,1.9) ellipse (0.42 and 1.1);
    \draw[orange, thin] (-2.5,1.9) ellipse (0.42 and 0.18);

    \filldraw[color=orange, fill=orange!20, thick] (-1,2.26) ellipse (0.3 and 0.74);
    \draw[orange, thin] (-1,2.26) ellipse (0.3 and 0.14);

    \filldraw[color=orange, fill=orange!20, thick] (5.5,1.92) ellipse (0.42 and 1.08);
    \draw[orange, thin] (5.5,1.92) ellipse (0.42 and 0.18);

    \filldraw[color=orange, fill=orange!20, thick] (4,1.53) ellipse (0.52 and 1.47);
    \draw[orange, thin] (4,1.72) ellipse (0.52 and 0.22);

    \filldraw[color=orange, fill=orange!20, thick] (2.5,1.9) ellipse (0.42 and 1.1);
    \draw[orange, thin] (2.5,1.9) ellipse (0.42 and 0.18);

    \filldraw[color=orange, fill=orange!20, thick] (1,2.26) ellipse (0.3 and 0.74);
    \draw[orange, thin] (1,2.26) ellipse (0.3 and 0.14);

    \filldraw[color=blue, fill=blue!20, thick] (-5.5,0) ellipse (0.46 and 0.74);
    \draw[blue, thin] (-5.5,0) ellipse (0.46 and 0.16);

    \filldraw[color=blue, fill=blue!20, thick] (-2.5,0) ellipse (0.46 and 0.71);
    \draw[blue, thin] (-2.5,0) ellipse (0.46 and 0.16);

    \filldraw[color=blue, fill=blue!20, thick] (-1,0) ellipse (0.62 and 1.46);
    \draw[blue, thin] (-1,0) ellipse (0.62 and 0.22);

    \filldraw[color=blue, fill=blue!20, thick] (1,0) ellipse (0.62 and 1.46);
    \draw[blue, thin] (1,0) ellipse (0.62 and 0.20);

    \filldraw[color=blue, fill=blue!20, thick] (2.5,0) ellipse (0.46 and 0.71);
    \draw[blue, thin] (2.5,0) ellipse (0.46 and 0.15);

    \filldraw[color=blue, fill=blue!20, thick] (5.5,0) ellipse (0.46 and 0.74);
    \draw[blue, thin] (5.5,0) ellipse (0.46 and 0.16);


    \draw[black, thick, ->] (-6.25,0) -- (6.25,0);
    \draw[mygreen, ultra thick] (-5.85,0) -- (5.85,0);

    \draw[red, thick] (-5.85,-2) -- (5.85,-2);
    
    \draw[orange, thick] (-5.85,3) -- (5.85,3);


    \draw[blue, thick] (-5.85,0.975) -- (-4.00,0);
    \draw[blue, thick] (-5.85,-0.975) -- (-4.00,0);

    \draw[blue, thick] (-4.00,0) -- (0.00,2.00);
    \draw[blue, thick] (0.00,2.00) -- (4.00,0);

    \draw[blue, thick] (-4.00,0) -- (0.00,-2.00);
    \draw[blue, thick] (0.00,-2.00) -- (4.00,0);

    \draw[blue, thick] (4.00,0) -- (5.85,0.975);
    \draw[blue, thick] (4.00,0) -- (5.85,-0.975);


    \filldraw[black] (-4.00,0) circle (0pt) node{\huge{$\times$}};
    \filldraw[black] (0.00,0) circle (0pt) node{\huge{$\times$}};
    \filldraw[black] (4.00,0) circle (0pt) node{\huge{$\times$}};


    \filldraw[black] (-6.3,0) circle (0pt)
        node[anchor=east]{\textcolor{mygreen}{\large{$F$}}};

    \filldraw[black] (-6.3,0.975) circle (0pt)
        node[anchor=east]{\textcolor{blue}{\large{$\Sigma$}}};

    \filldraw[black] (-6.3,-2) circle (0pt)
        node[anchor=east]{\textcolor{red}{\large{$M_{1,\cdots ,2n}$}}};
    
    \filldraw[black] (-6.3,3) circle (0pt)
        node[anchor=east]{\textcolor{orange}{\large{$M_{2n+1,\cdots ,8}$}}};

    \filldraw[black] (0,0.2) circle (0pt) node[anchor=south]{\large{$\mu_{t,n}$}};

    \filldraw[black] (-4,0.2) circle (0pt) node[anchor=south]{\large{$\mu_{s}$}};

    \filldraw[black] (4,0.2) circle (0pt) node[anchor=south]{\large{$\mu_{s}$}};

\begin{scope}[shift={(0,-8)}]

    \tikzstyle{every node}=[font=\scriptsize]
    \draw[thick] (-0.5,-1.5) rectangle (0.5,-0.5);
    \filldraw[black] (0,-1) circle (0pt) node{\Large{$2n$}};
    \draw[thick] (-1,3.5) rectangle (1,4.5);
    \filldraw[black] (0,4) circle (0pt) node{\Large{$8-2n$}};
    \filldraw[black] (0,1.5) circle (0pt) node{\Large{$2$}};
    \draw[thick](0,1.5) circle (0.5);
    \draw[thick, ->-] (0,1) -- (0,-0.5) ;
       
    \draw[thick]  (-6,1) arc (-90:90:0.5);
    \draw[thick]  (6,1) arc (270:90:0.5);

    \draw[thick] (-5.5,1.5) -- (-0.5,1.5);
    \draw[thick] (0.5,1.5) -- (5.5,1.5);
    \node at (-4,1.5) {\large $\times$};
    \node at (4,1.5) {\large $\times$};

\node at (-6,1.5) {\Large $2$};
\node at (6,1.5) {\Large $2$};

\node[anchor=west] at (6,2) {\Large $\cdots$};
\node[anchor=west] at (6,1) {\Large $\cdots$};

\node[anchor=east] at (-6,2) {\Large $\cdots$};
\node[anchor=east] at (-6,1) {\Large $\cdots$};
\end{scope}

\end{scope}
    \end{tikzpicture}
    \caption{The local geometry arising from the $\mu_{s}^{\varepsilon}\mu_{t,n}^{\varepsilon}\mu_s^{\varepsilon}$ profile. The fibration of the section $\Sigma$ over the interval whose endpoints are given by the location of the $\mu_s$ domain walls, gives rise to a three-sphere over which an $A_1$ singularity is fibered. The $\mathbb{P}^1$'s associated to the curve classes $M_{1,\cdots ,2n}$ collapse at the location of the $\mu_{t,n}$ domain wall, meaning the $A_{2n-1}$ singularity touches the $A_1$ singularity at this point. The curves $M_{2n+1,\cdots ,8}$ do not collapse, meaning the $A_{7-2n}$ singularity does not touch the $A_1$ singularity.}
    \label{fig: local musns dom wall}
\end{figure}
Since the curves $M_{2n+1,\cdots ,8}$ do not collapse at any point in this profile, there is no bifundamental matter charged under the $SU(8-2n)$ flavor symmetry.

\paragraph{$(\mu_{t,n}^{\varepsilon} \mu_s^{\varepsilon}\mu_{t,4}^{\varepsilon}\mu_s^{\varepsilon} )^2 $ Configuration.} Next we present the situation involving two $\mu_{t,n}$ domain walls and two $\mu_{t,4}$ domain walls. The reason for alternating the maps $\mu_{t,n}$ and $\mu_{t,4}$ is to mimic the field theory operation of gluing quivers using an alternation of $\Phi$ and mixed gluings. Following an analogous logic to Section~\ref{sec: flux 2^8}, we associate to this configuration of local building blocks the volume profile and quiver in Figure~\ref{fig: local s3s4s3s4 domain wall}.

\begin{figure}
    \centering
\begin{tikzpicture}[scale=0.8]


    \filldraw[color=red, fill=red!20, thick] (-9,-0.75) ellipse (0.18 and 0.25);
    \draw[red, thin] (-9,-0.75) ellipse (0.18 and 0.07);

    \filldraw[color=red, fill=red!20, thick] (-8,-0.52) ellipse (0.27 and 0.48);
    \draw[red, thin] (-8,-0.52) ellipse (0.27 and 0.12);

    \filldraw[color=red, fill=red!20, thick] (-7,-0.76) ellipse (0.18 and 0.24);
    \draw[red, thin] (-7,-0.75) ellipse (0.18 and 0.07);

    \filldraw[color=red, fill=red!20, thick] (-4,-0.52) ellipse (0.27 and 0.48);
    \draw[red, thin] (-4,-0.52) ellipse (0.27 and 0.12);

    \filldraw[color=red, fill=red!20, thick] (-5,-0.76) ellipse (0.18 and 0.24);
    \draw[red, thin] (-5,-0.75) ellipse (0.18 and 0.07);

    \filldraw[color=red, fill=red!20, thick] (-3,-0.76) ellipse (0.18 and 0.24);
    \draw[red, thin] (-3,-0.75) ellipse (0.18 and 0.07);

    \filldraw[color=red, fill=red!20, thick] (0,-0.52) ellipse (0.27 and 0.48);
    \draw[red, thin] (0,-0.52) ellipse (0.27 and 0.12);

    \filldraw[color=red, fill=red!20, thick] (-1,-0.76) ellipse (0.18 and 0.24);
    \draw[red, thin] (-1,-0.75) ellipse (0.18 and 0.07);

    \filldraw[color=red, fill=red!20, thick] (1,-0.76) ellipse (0.18 and 0.24);
    \draw[red, thin] (1,-0.75) ellipse (0.18 and 0.07);

    \filldraw[color=red, fill=red!20, thick] (4,-0.52) ellipse (0.27 and 0.48);
    \draw[red, thin] (4,-0.52) ellipse (0.27 and 0.12);

    \filldraw[color=red, fill=red!20, thick] (3,-0.76) ellipse (0.18 and 0.24);
    \draw[red, thin] (3,-0.75) ellipse (0.18 and 0.07);

    \filldraw[color=red, fill=red!20, thick] (5,-0.76) ellipse (0.18 and 0.24);
    \draw[red, thin] (5,-0.75) ellipse (0.18 and 0.07);

    \filldraw[color=red, fill=red!20, thick] (7,-0.76) ellipse (0.18 and 0.24);
    \draw[red, thin] (7,-0.75) ellipse (0.18 and 0.07);


    \filldraw[color=blue, fill=blue!20, thick] (-9,0) ellipse (0.3 and 0.45);
    \draw[blue, thin] (-9,0) ellipse (0.3 and 0.14);

    \filldraw[color=blue, fill=blue!20, thick] (-7,0) ellipse (0.3 and 0.47);
    \draw[blue, thin] (-7,0) ellipse (0.3 and 0.14);

    \filldraw[color=blue, fill=blue!20, thick] (-5,0) ellipse (0.3 and 0.47);
    \draw[blue, thin] (-5,0) ellipse (0.3 and 0.14);

    \filldraw[color=blue, fill=blue!20, thick] (-6,0) ellipse (0.4 and 0.95);
    \draw[blue, thin] (-6,0) ellipse (0.4 and 0.2);

   \filldraw[color=blue, fill=blue!20, thick] (-3,0) ellipse (0.3 and 0.47);
    \draw[blue, thin] (-3,0) ellipse (0.3 and 0.14);

    \filldraw[color=blue, fill=blue!20, thick] (-1,0) ellipse (0.3 and 0.47);
    \draw[blue, thin] (-1,0) ellipse (0.3 and 0.14);

    \filldraw[color=blue, fill=blue!20, thick] (-2,0) ellipse (0.4 and 0.95);
    \draw[blue, thin] (-2,0) ellipse (0.4 and 0.2);

    \filldraw[color=blue, fill=blue!20, thick] (1,0) ellipse (0.3 and 0.47);
    \draw[blue, thin] (1,0) ellipse (0.3 and 0.14);

    \filldraw[color=blue, fill=blue!20, thick] (2,0) ellipse (0.4 and 0.95);
    \draw[blue, thin] (2,0) ellipse (0.4 and 0.2);

    \filldraw[color=blue, fill=blue!20, thick] (3,0) ellipse (0.3 and 0.47);
    \draw[blue, thin] (3,0) ellipse (0.3 and 0.14);

   \filldraw[color=blue, fill=blue!20, thick] (7,0) ellipse (0.3 and 0.47);
    \draw[blue, thin] (7,0) ellipse (0.3 and 0.14);

    \filldraw[color=blue, fill=blue!20, thick] (5,0) ellipse (0.3 and 0.47);
    \draw[blue, thin] (5,0) ellipse (0.3 and 0.14);

    \filldraw[color=blue, fill=blue!20, thick] (6,0) ellipse (0.4 and 0.95);
    \draw[blue, thin] (6,0) ellipse (0.4 and 0.2);


    \filldraw[color=orange, fill=orange!20, thick] (-9,1) ellipse (0.2 and 0.5);
    \draw[orange, thin] (-9,1) ellipse (0.2 and 0.08);

    \filldraw[color=orange, fill=orange!20, thick] (-8,0.69) ellipse (0.3 and 0.64);
    \draw[orange, thin] (-8,0.69) ellipse (0.3 and 0.14);


   \filldraw[color=orange, fill=orange!20, thick] (-7,0.85) ellipse (0.2 and 0.3);
    \draw[orange, thin] (-7,0.85) ellipse (0.2 and 0.08);

   \filldraw[color=orange, fill=orange!20, thick] (-2,1.25) ellipse (0.15 and 0.25);
    \draw[orange, thin] (-2,1.25) ellipse (0.15 and 0.05);

    \filldraw[color=orange, fill=orange!20, thick] (-3,0.94) ellipse (0.2 and 0.43);
    \draw[orange, thin] (-3,0.94) ellipse (0.2 and 0.08);

    \filldraw[color=orange, fill=orange!20, thick] (-4,0.65) ellipse (0.3 and 0.6);
    \draw[orange, thin] (-4,0.65) ellipse (0.3 and 0.14);

   \filldraw[color=orange, fill=orange!20, thick] (-5,0.83) ellipse (0.2 and 0.29);
    \draw[orange, thin] (-5,0.83) ellipse (0.2 and 0.08);


   \filldraw[color=orange, fill=orange!20, thick] (1,0.83) ellipse (0.2 and 0.29);
    \draw[orange, thin] (1,0.83) ellipse (0.2 and 0.08);



    \filldraw[color=orange, fill=orange!20, thick] (0,0.65) ellipse (0.3 and 0.6);
    \draw[orange, thin] (0,0.65) ellipse (0.3 and 0.14);


    \filldraw[color=orange, fill=orange!20, thick] (-1,0.94) ellipse (0.2 and 0.43);
    \draw[orange, thin] (-1,0.94) ellipse (0.2 and 0.08);


   \filldraw[color=orange, fill=orange!20, thick] (6,1.25) ellipse (0.15 and 0.25);
    \draw[orange, thin] (6,1.25) ellipse (0.15 and 0.05);


    \filldraw[color=orange, fill=orange!20, thick] (5,0.94) ellipse (0.2 and 0.43);
    \draw[orange, thin] (5,0.94) ellipse (0.2 and 0.08);


    \filldraw[color=orange, fill=orange!20, thick] (4,0.65) ellipse (0.3 and 0.6);
    \draw[orange, thin] (4,0.65) ellipse (0.3 and 0.14);


   \filldraw[color=orange, fill=orange!20, thick] (3,0.83) ellipse (0.2 and 0.29);
    \draw[orange, thin] (3,0.83) ellipse (0.2 and 0.08);




    \filldraw[color=orange, fill=orange!20, thick] (7,1) ellipse (0.2 and 0.5);
    \draw[orange, thin] (7,1) ellipse (0.2 and 0.08);


    \draw[thick] (-10,0) -- (8,0);
    \draw[color=mygreen, ultra thick] (-9,0) -- (7,0);
    \draw[color=blue, thick] (-9,-0.5) -- (-8,0);
    \draw[color=blue, thick] (-8,0) -- (-6,-1);
    \draw[color=blue, thick] (-6,-1) -- (-4,0);
    \draw[color=blue, thick] (-4,0) -- (-2,-1);
    \draw[color=blue, thick] (-2,-1) -- (0,0);
    \draw[color=blue, thick] (0,0) -- (2,-1);
    \draw[color=blue, thick] (2,-1) -- (4,0);
    \draw[color=blue, thick] (4,0) -- (6,-1);
    \draw[color=blue, thick] (6,-1) -- (7,-0.5);

    \draw[color=blue, thick] (-9,0.5) -- (-8,0);
    \draw[color=blue, thick] (-8,0) -- (-6,1);
    \draw[color=blue, thick] (-6,1) -- (-4,0);
    \draw[color=blue, thick] (-4,0) -- (-2,1);
    \draw[color=blue, thick] (-2,1) -- (0,0);
    \draw[color=blue, thick] (0,0) -- (2,1);
    \draw[color=blue, thick] (2,1) -- (4,0);
    \draw[color=blue, thick] (4,0) -- (6,1);
    \draw[color=blue, thick] (6,1) -- (7,0.5);

    \draw[color=red, thick] (-9,-1) -- (7,-1);

    \draw[color=orange, thick] (-9,1.5) -- (-6,1);
    \draw[color=orange, thick] (-2,1.5) -- (-6,1);
    \draw[color=orange, thick] (-2,1.5) -- (2,1);
    \draw[color=orange, thick] (6,1.5) -- (2,1);
    \draw[color=orange, thick] (6,1.5) -- (7,1.5);

   \filldraw[black] (-10,0) circle (0pt)
        node[anchor=east]{\textcolor{mygreen}{\large{$F$}}};

    \filldraw[black] (-10,0.5) circle (0pt)
        node[anchor=east]{\textcolor{blue}{\large{$\Sigma$}}};

    \filldraw[black] (-9.5,-1) circle (0pt)
        node[anchor=east]{\textcolor{red}{\large{$M_{1 \cdots 2n}$}}};

    \filldraw[black] (-9.5,1.5) circle (0pt)
        node[anchor=east]{\textcolor{orange}{\large{$M_{2n+1 \cdots 8}$}}};

    \filldraw[black] (-8.00,0.2) circle (0pt)
        node[anchor=south]{\large{$\mu_s$}};

    \filldraw[black] (-6,0.3) circle (0pt)
        node[anchor=south]{\large{$\mu_{t,4}$}};

    \filldraw[black] (-4.00,0.2) circle (0pt)
        node[anchor=south]{\large{$\mu_s$}};

    \filldraw[black] (-2,0.3) circle (0pt)
        node[anchor=south]{\large{$\mu_{t,n}$}};

    \filldraw[black] (0,0.2) circle (0pt)
        node[anchor=south]{\large{$\mu_s$}};

    \filldraw[black] (2,0.2) circle (0pt)
        node[anchor=south]{\large{$\mu_{t,4}$}};

    \filldraw[black] (4,0.2) circle (0pt)
        node[anchor=south]{\large{$\mu_s$}};

    \filldraw[black] (6,0.2) circle (0pt)
        node[anchor=south]{\large{$\mu_{t,n}$}};

    \filldraw[black] (-8.00,0) circle (0pt) node{\Large{$\times$}};
    \filldraw[black] (-6.00,0) circle (0pt) node{\Large{$\times$}};
    \filldraw[black] (-4.00,0) circle (0pt) node{\Large{$\times$}};
    \filldraw[black] (-2.00,0) circle (0pt) node{\Large{$\times$}};
    \filldraw[black] (0.00,0) circle (0pt) node{\Large{$\times$}};
    \filldraw[black] (2.00,0) circle (0pt) node{\Large{$\times$}};
    \filldraw[black] (4.00,0) circle (0pt) node{\Large{$\times$}};
    \filldraw[black] (6.00,0) circle (0pt) node{\Large{$\times$}};


\begin{scope}[shift={(0,-7)}]
   \tikzstyle{every node}=[font=\scriptsize]
    \draw[thick] (-2.5,-1.5) rectangle (-1.5,-0.5);
    \filldraw[black] (-2,-1) circle (0pt) node{\Large{$2n$}};
    \draw[thick] (-3,3.5) rectangle (-1,4.5);
    \filldraw[black] (-2,4) circle (0pt) node{\Large{$8-2n$}};
    \filldraw[black] (-9,1.5) circle (0pt) node{\Large{$2$}};
    \filldraw[black] (-6,1.5) circle (0pt) node{\Large{$2$}};
    \filldraw[black] (-2,1.5) circle (0pt) node{\Large{$2$}};
    \filldraw[black] (2,1.5) circle (0pt) node{\Large{$2$}};
    \filldraw[black] (6,1.5) circle (0pt) node{\Large{$2$}};
    \draw[thick](-6,1.5) circle (0.5);
    \draw[thick](-2,1.5) circle (0.5);
    \draw[thick](2,1.5) circle (0.5);
    \draw[thick, -<-] (-5.6,1.2) -- (-2.5,-0.7) ;
    \draw[thick, ->-] (-3,4) -- (-5.6,1.8) ;
    \draw[thick, -<-] (1.6,1.2) -- (-1.5,-0.7) ;
    \draw[thick, ->-] (5.6,1.2) -- (-1.5,-1) ;
   \draw[thick, -<-] (-2,1) -- (-2,-0.5) ;
    \draw[thick, -<-] (-1,4) -- (1.6,1.8) ;

    \draw[thick]  (-9,1) arc (-90:90:0.5);
    \draw[thick]  (6,1) arc (270:90:0.5);

    \draw[thick] (-8.5,1.5) -- (-6.5,1.5);
    \draw[thick] (-5.5,1.5) -- (-2.5,1.5);
    \draw[thick] (-1.5,1.5) -- (1.5,1.5);
    \draw[thick] (2.5,1.5) -- (5.5,1.5);
    \node at (-8,1.5) {\large $\times$};
    \node at (-4,1.5) {\large $\times$};
    \node at (0,1.5) {\large $\times$};
    \node at (4,1.5) {\large $\times$};


\node[anchor=west] at (6,2) {\Large $\cdots$};
\node[anchor=west] at (6,1) {\Large $\cdots$};

\node[anchor=east] at (-9,2) {\Large $\cdots$};
\node[anchor=east] at (-9,1) {\Large $\cdots$};
    
\end{scope}

\end{tikzpicture}
    \caption{The local geometry arising from the concatenation $(\mu_{t,n}^{\varepsilon}\mu_s^{\varepsilon}\mu_{t,4}^{\varepsilon}\mu_s^{\varepsilon})^2$. The fibration of the section $\Sigma$ over three compact intervals, whose endpoints are given by the location of the $\mu_s$ domain walls, gives rise to three 3-spheres over which an $A_1$ singularity is fibered. The $\mathbb{P}^1$'s associated to the curve classes $M_{1,\cdots , 2n}$ collapse at the locations of the $\mu_{t,n}$ and $\mu_{t,4}$ domain walls, meaning the $A_{2n-1}$ singularity touches the $A_1$ singularity at these four points. The curves $M_{2n+1, \cdots 8}$ only collapse at two points, meaning the $A_{8-2n}$ singularity touches only two of the $A_1$ singularities. }
    \label{fig: local s3s4s3s4 domain wall}
\end{figure}

There are four locations in which the matter curves $M_{1\cdots 2n}$ collapse, giving four sets of chirals charged under the $SU(2n)$ flavor symmetry. The matter curves $M_{2n+1 \cdots 8}$ only collapse at two points, giving two sets of chirals charged under the $SU(8-2n)$ flavor symmetry, as can be seen in the quiver. We further associate to this configuration the box graph description shown in Figure~\ref{Fig: box diagram n3 fifth region}. 
\begin{figure}
\centering
\begin{tikzpicture}
\node[] (5dLyt) at (-3,-2) {\ytableausetup{centertableaux,boxsize=5pt}
\begin{ytableau}
*(blue) & *(blue) & *(blue) & *(blue) & *(blue) & *(blue) & *(blue) & *(blue) \\
\none[] & \none[] & \none[] & \none[] & \none[] & \none[] & *(yellow) & *(yellow) & *(yellow) & *(yellow) & *(yellow) & *(yellow) & *(yellow) & *(yellow) \\ \none[] \\
*(blue) & *(blue) & *(blue) & *(blue) & *(blue) & *(blue) & *(blue) & *(blue) \\
\none[] & \none[] & \none[] & \none[] & \none[] & \none[] & *(yellow) & *(yellow) & *(yellow) & *(yellow) & *(yellow) & *(yellow) & *(yellow) & *(yellow) \\
\end{ytableau}};
\node[] (5dLyt) at (0,-2) {\ytableausetup{centertableaux,boxsize=5pt}
\begin{ytableau}
*(yellow) & *(yellow) & *(yellow) & *(yellow) & *(yellow) & *(yellow) & *(yellow) & *(yellow) \\
\none[] & \none[] & \none[] & \none[] & \none[] & \none[] & *(blue) & *(blue) & *(blue) & *(blue) & *(blue) & *(blue) & *(blue) & *(blue) \\ \none[] \\
*(yellow) & *(yellow) & *(yellow) & *(yellow) & *(yellow) & *(yellow) & *(yellow) & *(yellow) \\
\none[] & \none[] & \none[] & \none[] & \none[] & \none[] & *(blue) & *(blue) & *(blue) & *(blue) & *(blue) & *(blue) & *(blue) & *(blue) \\
\end{ytableau}};
\draw[black, thick] (-1.5,-3) -- (-1.5,-1);
\filldraw[black] (-1.5,-3) circle (0pt) node[anchor=north]{$\mu_{t,4}$};
\node[] (5dLyt) at (3,-2) {\ytableausetup{centertableaux,boxsize=5pt}
\begin{ytableau}
*(blue) & *(blue) & *(blue) & *(blue) & *(blue) & *(blue) & *(yellow) & *(yellow) \\
\none[] & \none[] & \none[] & \none[] & \none[] & \none[] & *(blue) & *(blue) & *(yellow) & *(yellow) & *(yellow) & *(yellow) & *(yellow) & *(yellow) \\ \none[] \\
*(blue) & *(blue) & *(blue) & *(blue) & *(blue) & *(blue) & *(yellow) & *(yellow) \\
\none[] & \none[] & \none[] & \none[] & \none[] & \none[] & *(blue) & *(blue) & *(yellow) & *(yellow) & *(yellow) & *(yellow) & *(yellow) & *(yellow) \\
\end{ytableau}};
    \node[] (5dLyt) at (6,-2) {\ytableausetup{centertableaux,boxsize=5pt}
\begin{ytableau}
*(yellow) & *(yellow) & *(yellow) & *(yellow) & *(yellow) & *(yellow) & *(blue) & *(blue) \\
\none[] & \none[] & \none[] & \none[] & \none[] & \none[] & *(yellow) & *(yellow) & *(blue) & *(blue) & *(blue) & *(blue) & *(blue) & *(blue) \\ \none[] \\
*(yellow) & *(yellow) & *(yellow) & *(yellow) & *(yellow) & *(yellow) & *(blue) & *(blue) \\
\none[] & \none[] & \none[] & \none[] & \none[] & \none[] & *(yellow) & *(yellow) & *(blue) & *(blue) & *(blue) & *(blue) & *(blue) & *(blue) \\
\end{ytableau}};
\node[] (5dLyt) at (9,-2) {\ytableausetup{centertableaux,boxsize=5pt}
\begin{ytableau}
*(blue) & *(blue) & *(blue) & *(blue) & *(blue) & *(blue) & *(blue) & *(blue) \\
\none[] & \none[] & \none[] & \none[] & \none[] & \none[] & *(yellow) & *(yellow) & *(yellow) & *(yellow) & *(yellow) & *(yellow) & *(yellow) & *(yellow) \\ \none[] \\
*(blue) & *(blue) & *(blue) & *(blue) & *(blue) & *(blue) & *(blue) & *(blue) \\
\none[] & \none[] & \none[] & \none[] & \none[] & \none[] & *(yellow) & *(yellow) & *(yellow) & *(yellow) & *(yellow) & *(yellow) & *(yellow) & *(yellow) \\
\end{ytableau}};
\draw[black, thick] (1.5,-3) -- (1.5,-1);
\filldraw[black] (1.5,-3) circle (0pt) node[anchor=north]{$\mu_{t,3}$};
\draw[black, thick] (4.5,-3) -- (4.5,-1);
\filldraw[black] (4.5,-3) circle (0pt) node[anchor=north]{$\mu_{t,4}$};
\draw[black, thick] (7.5,-3) -- (7.5,-1);
\filldraw[black] (7.5,-3) circle (0pt) node[anchor=north]{$\mu_{t,3}$};
\end{tikzpicture}
\caption{Box diagram for the configuration in Figure~\ref{fig: local s3s4s3s4 domain wall} for the examples $n=3$. The action of $\mu_{t,3}$ is to flop six of the effective curves, whilst the action of $\mu_{t,4}$ is to flop eight of the effective curves. }
\label{Fig: box diagram n3 fifth region}
\end{figure}
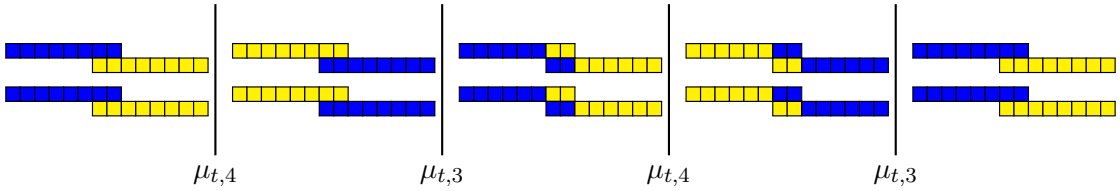

\paragraph{The $(\mu_{t,n}^{\varepsilon} \mu_s^{\varepsilon}\mu_{t,4}^{\varepsilon}\mu_s^{\varepsilon} )^2 $ Flux Torus.} 

Finally, we glue the profile in Figure~\ref{fig: local s3s4s3s4 domain wall} using the monodromy map $\mathcal{M}[\mathcal{F}_n]=( \mu_s\mu_{t,4}\mu_s\mu_{t,n})^2$. The label $\mathcal{F}_n$ is the geometric flux that satisfies $N(\mathcal{F}_n)=(-2^{2n},0^{8-2n})$, and the specific form of $\mathcal{F}_n \in H_2(X_S,\Z)$ will be explained in Section~\ref{sect:concatenating}. The volume profile for this configuration is displayed in Figure~\ref{fig: n3 volumes and quiver closed}. 
\begin{figure}
    \centering
\begin{tikzpicture}[scale=0.7]


    \filldraw[color=blue, fill=blue!20, thick] (0,-3.6) ellipse (0.6 and 1.28);
    \draw[blue, thin] (0,-3.8) ellipse (0.58 and 0.2);

    \begin{scope}[shift={(-1.6,-3.49)}, rotate=-10]
        \filldraw[color=blue, fill=blue!20, thick] (0,0) ellipse (0.5 and 1.12);
        \draw[blue, thin] (0,-0.2) ellipse (0.48 and 0.17);
    \end{scope}

    \begin{scope}[shift={(1.6,-3.49)}, rotate=10]
        \filldraw[color=blue, fill=blue!20, thick] (0,0) ellipse (0.5 and 1.12);
        \draw[blue, thin] (0,-0.2) ellipse (0.48 and 0.17);
    \end{scope}

    \begin{scope}[shift={(-2.8,-3.28)}, rotate=-18]
        \filldraw[color=blue, fill=blue!20, thick] (0,0) ellipse (0.3 and 0.79);
        \draw[blue, thin] (0,-0.12) ellipse (0.3 and 0.1);
    \end{scope}

    \begin{scope}[shift={(2.8,-3.28)}, rotate=18]
        \filldraw[color=blue, fill=blue!20, thick] (0,0) ellipse (0.3 and 0.79);
        \draw[blue, thin] (0,-0.12) ellipse (0.3 and 0.1);
    \end{scope}


    \begin{scope}[shift={(0,3.6)}, rotate=180]
    \filldraw[color=blue, fill=blue!20, thick] (0,0) ellipse (0.6 and 1.28);
    \draw[blue, thin] (0,-0.2) ellipse (0.58 and 0.2);
   \end{scope}

    \begin{scope}[shift={(1.6,3.49)}, rotate=170]
        \filldraw[color=blue, fill=blue!20, thick] (0,0) ellipse (0.5 and 1.12);
        \draw[blue, thin] (0,-0.2) ellipse (0.48 and 0.17);
    \end{scope}

    \begin{scope}[shift={(-1.6,3.49)}, rotate=190]
        \filldraw[color=blue, fill=blue!20, thick] (0,0) ellipse (0.5 and 1.12);
        \draw[blue, thin] (0,-0.2) ellipse (0.48 and 0.17);
    \end{scope}

    \begin{scope}[shift={(2.8,3.28)}, rotate=162]
        \filldraw[color=blue, fill=blue!20, thick] (0,0) ellipse (0.3 and 0.79);
        \draw[blue, thin] (0,-0.12) ellipse (0.3 and 0.1);
    \end{scope}

    \begin{scope}[shift={(-2.8,3.28)}, rotate=198]
        \filldraw[color=blue, fill=blue!20, thick] (0,0) ellipse (0.3 and 0.79);
        \draw[blue, thin] (0,-0.12) ellipse (0.3 and 0.1);
    \end{scope}


    \begin{scope}[shift={(-5.87,0)}, rotate=90]
        \filldraw[color=blue, fill=blue!20, thick] (0,0) ellipse (0.3 and 0.74);
        \draw[blue, thin] (0,0.02) ellipse (0.29 and 0.1);
    \end{scope}

    \begin{scope}[shift={(-5.72,-1.02)}, rotate=112]
        \filldraw[color=blue, fill=blue!20, thick] (0,0) ellipse (0.27 and 0.67);
        \draw[blue, thin] (0,-0.02) ellipse (0.26 and 0.09);
    \end{scope}

    \begin{scope}[shift={(-5.72,1.02)}, rotate=68]
        \filldraw[color=blue, fill=blue!20, thick] (0,0) ellipse (0.27 and 0.67);
        \draw[blue, thin] (0,-0.02) ellipse (0.26 and 0.09);
    \end{scope}

    \begin{scope}[shift={(-5.3,-1.9)}, rotate=130]
        \filldraw[color=blue, fill=blue!20, thick] (0,0) ellipse (0.22 and 0.49);
        \draw[blue, thin] (0,-0.1) ellipse (0.22 and 0.08);
    \end{scope}

   \begin{scope}[shift={(-5.3,1.9)}, rotate=50]
        \filldraw[color=blue, fill=blue!20, thick] (0,0) ellipse (0.22 and 0.49);
        \draw[blue, thin] (0,-0.1) ellipse (0.22 and 0.08);
    \end{scope}


    \begin{scope}[shift={(5.87,0)}, rotate=-90]
        \filldraw[color=blue, fill=blue!20, thick] (0,0) ellipse (0.3 and 0.74);
        \draw[blue, thin] (0,0.02) ellipse (0.29 and 0.1);
    \end{scope}

    \begin{scope}[shift={(5.72,-1.02)}, rotate=-112]
        \filldraw[color=blue, fill=blue!20, thick] (0,0) ellipse (0.27 and 0.67);
        \draw[blue, thin] (0,-0.02) ellipse (0.26 and 0.09);
    \end{scope}

    \begin{scope}[shift={(5.72,1.02)}, rotate=-68]
        \filldraw[color=blue, fill=blue!20, thick] (0,0) ellipse (0.27 and 0.67);
        \draw[blue, thin] (0,-0.02) ellipse (0.26 and 0.09);
    \end{scope}

    \begin{scope}[shift={(5.3,-1.9)}, rotate=-130]
        \filldraw[color=blue, fill=blue!20, thick] (0,0) ellipse (0.22 and 0.49);
        \draw[blue, thin] (0,-0.1) ellipse (0.22 and 0.08);
    \end{scope}

   \begin{scope}[shift={(5.3,1.9)}, rotate=-50]
        \filldraw[color=blue, fill=blue!20, thick] (0,0) ellipse (0.22 and 0.49);
        \draw[blue, thin] (0,-0.1) ellipse (0.22 and 0.08);
    \end{scope}


    \begin{scope}[shift={(-4.31,-3.5)}, rotate=-30]
        \filldraw[color=red, fill=red!20, thick] (0,0) ellipse (0.25 and 0.56);
        \draw[red, thin] (0,0.02) ellipse (0.25 and 0.08);
    \end{scope}

    \begin{scope}[shift={(-5.9,-2.45)}, rotate=-50]
        \filldraw[color=red, fill=red!20, thick] (0,0) ellipse (0.14 and 0.31);
        \draw[red, thin] (0,-0.02) ellipse (0.14 and 0.04);
    \end{scope}

    \begin{scope}[shift={(-6.57,-1.34)}, rotate=-68]
        \filldraw[color=red, fill=red!20, thick] (0,0) ellipse (0.14 and 0.2);
        \draw[red, thin] (0,-0.02) ellipse (0.14 and 0.04);
    \end{scope}

    \begin{scope}[shift={(-3.2,-4.34)}, rotate=-20]
        \filldraw[color=red, fill=red!20, thick] (0,0) ellipse (0.2 and 0.33);
        \draw[red, thin] (0,-0.02) ellipse (0.2 and 0.07);
    \end{scope}

    \begin{scope}[shift={(-1.85,-4.83)}, rotate=-10]
        \filldraw[color=red, fill=red!20, thick] (0,0) ellipse (0.15 and 0.22);
        \draw[red, thin] (0,-0.02) ellipse (0.15 and 0.05);
    \end{scope}


    \begin{scope}[shift={(-4.31,3.5)}, rotate=30]
        \filldraw[color=red, fill=red!20, thick] (0,0) ellipse (0.25 and 0.56);
        \draw[red, thin] (0,0.02) ellipse (0.25 and 0.08);
    \end{scope}

    \begin{scope}[shift={(-5.9,2.45)}, rotate=50]
        \filldraw[color=red, fill=red!20, thick] (0,0) ellipse (0.14 and 0.31);
        \draw[red, thin] (0,-0.02) ellipse (0.14 and 0.04);
    \end{scope}

    \begin{scope}[shift={(-6.57,1.34)}, rotate=68]
        \filldraw[color=red, fill=red!20, thick] (0,0) ellipse (0.14 and 0.2);
        \draw[red, thin] (0,-0.02) ellipse (0.14 and 0.04);
    \end{scope}

    \begin{scope}[shift={(-3.2,4.34)}, rotate=20]
        \filldraw[color=red, fill=red!20, thick] (0,0) ellipse (0.2 and 0.33);
        \draw[red, thin] (0,-0.02) ellipse (0.2 and 0.07);
    \end{scope}

    \begin{scope}[shift={(-1.85,4.83)}, rotate=10]
        \filldraw[color=red, fill=red!20, thick] (0,0) ellipse (0.15 and 0.22);
        \draw[red, thin] (0,-0.02) ellipse (0.15 and 0.05);
    \end{scope}


    \begin{scope}[shift={(4.31,3.5)}, rotate=-30]
        \filldraw[color=red, fill=red!20, thick] (0,0) ellipse (0.25 and 0.56);
        \draw[red, thin] (0,0.02) ellipse (0.25 and 0.08);
    \end{scope}

    \begin{scope}[shift={(5.9,2.45)}, rotate=-50]
        \filldraw[color=red, fill=red!20, thick] (0,0) ellipse (0.14 and 0.31);
        \draw[red, thin] (0,-0.02) ellipse (0.14 and 0.04);
    \end{scope}

    \begin{scope}[shift={(6.57,1.34)}, rotate=-68]
        \filldraw[color=red, fill=red!20, thick] (0,0) ellipse (0.14 and 0.2);
        \draw[red, thin] (0,-0.02) ellipse (0.14 and 0.04);
    \end{scope}

    \begin{scope}[shift={(3.2,4.34)}, rotate=-20]
        \filldraw[color=red, fill=red!20, thick] (0,0) ellipse (0.2 and 0.33);
        \draw[red, thin] (0,-0.02) ellipse (0.2 and 0.07);
    \end{scope}

    \begin{scope}[shift={(1.85,4.83)}, rotate=-10]
        \filldraw[color=red, fill=red!20, thick] (0,0) ellipse (0.15 and 0.22);
        \draw[red, thin] (0,-0.02) ellipse (0.15 and 0.05);
    \end{scope}


    \begin{scope}[shift={(4.31,-3.5)}, rotate=30]
        \filldraw[color=red, fill=red!20, thick] (0,0) ellipse (0.25 and 0.56);
        \draw[red, thin] (0,0.02) ellipse (0.25 and 0.08);
    \end{scope}

    \begin{scope}[shift={(5.9,-2.45)}, rotate=50]
        \filldraw[color=red, fill=red!20, thick] (0,0) ellipse (0.14 and 0.31);
        \draw[red, thin] (0,-0.02) ellipse (0.14 and 0.04);
    \end{scope}

    \begin{scope}[shift={(6.57,-1.34)}, rotate=68]
        \filldraw[color=red, fill=red!20, thick] (0,0) ellipse (0.14 and 0.2);
        \draw[red, thin] (0,-0.02) ellipse (0.14 and 0.04);
    \end{scope}

    \begin{scope}[shift={(3.2,-4.34)}, rotate=20]
        \filldraw[color=red, fill=red!20, thick] (0,0) ellipse (0.2 and 0.33);
        \draw[red, thin] (0,-0.02) ellipse (0.2 and 0.07);
    \end{scope}

    \begin{scope}[shift={(1.85,-4.83)}, rotate=10]
        \filldraw[color=red, fill=red!20, thick] (0,0) ellipse (0.15 and 0.22);
        \draw[red, thin] (0,-0.02) ellipse (0.15 and 0.05);
    \end{scope}


    \begin{scope}[shift={(0,-1.8)}, rotate=180]
        \filldraw[color=orange, fill=orange!20, thick] (0,0) ellipse (0.2 and 0.5);
        \draw[orange, thin] (0,-0.02) ellipse (0.2 and 0.05);
    \end{scope}

    \begin{scope}[shift={(-3.69,-2.34)}, rotate=330]
        \filldraw[color=orange, fill=orange!20, thick] (0,0) ellipse (0.25 and 0.655);
        \draw[orange, thin] (0,0.02) ellipse (0.25 and 0.08);
    \end{scope}

    \begin{scope}[shift={(-2.5,-2.17)}, rotate=340]
        \filldraw[color=orange, fill=orange!20, thick] (0,0) ellipse (0.2 and 0.34);
        \draw[orange, thin] (0,-0.02) ellipse (0.2 and 0.04);
    \end{scope}

    \begin{scope}[shift={(-1.36,-2.02)}, rotate=350]
        \filldraw[color=orange, fill=orange!20, thick] (0,0) ellipse (0.18 and 0.34);
        \draw[orange, thin] (0,0) ellipse (0.18 and 0.04);
    \end{scope}

    \begin{scope}[shift={(-4.63,-1.39)}, rotate=300]
        \filldraw[color=orange, fill=orange!20, thick] (0,0) ellipse (0.2 and 0.32);
        \draw[orange, thin] (0,-0.02) ellipse (0.2 and 0.05);
    \end{scope}

    \begin{scope}[shift={(-4.92,-0.71)}, rotate=285]
        \filldraw[color=orange, fill=orange!20, thick] (0,0) ellipse (0.1 and 0.14);
        \draw[orange, thin] (0,-0.02) ellipse (0.1 and 0.03);
    \end{scope}


    \begin{scope}[shift={(3.69,-2.34)}, rotate=210]
        \filldraw[color=orange, fill=orange!20, thick] (0,0) ellipse (0.25 and 0.655);
        \draw[orange, thin] (0,0.02) ellipse (0.25 and 0.08);
    \end{scope}

    \begin{scope}[shift={(2.5,-2.17)}, rotate=200]
        \filldraw[color=orange, fill=orange!20, thick] (0,0) ellipse (0.2 and 0.34);
        \draw[orange, thin] (0,-0.02) ellipse (0.2 and 0.04);
    \end{scope}

    \begin{scope}[shift={(1.36,-2.02)}, rotate=190]
        \filldraw[color=orange, fill=orange!20, thick] (0,0) ellipse (0.18 and 0.34);
        \draw[orange, thin] (0,0) ellipse (0.18 and 0.04);
    \end{scope}

    \begin{scope}[shift={(4.63,-1.39)}, rotate=240]
        \filldraw[color=orange, fill=orange!20, thick] (0,0) ellipse (0.2 and 0.32);
        \draw[orange, thin] (0,-0.02) ellipse (0.2 and 0.05);
    \end{scope}

    \begin{scope}[shift={(4.92,-0.71)}, rotate=255]
        \filldraw[color=orange, fill=orange!20, thick] (0,0) ellipse (0.1 and 0.14);
        \draw[orange, thin] (0,-0.02) ellipse (0.1 and 0.03);
    \end{scope}


    \begin{scope}[shift={(0,1.8)}, rotate=0]
        \filldraw[color=orange, fill=orange!20, thick] (0,0) ellipse (0.2 and 0.5);
        \draw[orange, thin] (0,-0.02) ellipse (0.2 and 0.05);
    \end{scope}

    \begin{scope}[shift={(3.69,2.34)}, rotate=150]
        \filldraw[color=orange, fill=orange!20, thick] (0,0) ellipse (0.25 and 0.655);
        \draw[orange, thin] (0,0.02) ellipse (0.25 and 0.08);
    \end{scope}

    \begin{scope}[shift={(2.5,2.17)}, rotate=160]
        \filldraw[color=orange, fill=orange!20, thick] (0,0) ellipse (0.2 and 0.34);
        \draw[orange, thin] (0,-0.02) ellipse (0.2 and 0.04);
    \end{scope}

    \begin{scope}[shift={(1.36,2.02)}, rotate=170]
        \filldraw[color=orange, fill=orange!20, thick] (0,0) ellipse (0.18 and 0.34);
        \draw[orange, thin] (0,0) ellipse (0.18 and 0.04);
    \end{scope}

    \begin{scope}[shift={(4.63,1.39)}, rotate=120]
        \filldraw[color=orange, fill=orange!20, thick] (0,0) ellipse (0.2 and 0.32);
        \draw[orange, thin] (0,-0.02) ellipse (0.2 and 0.05);
    \end{scope}

    \begin{scope}[shift={(4.92,0.71)}, rotate=105]
        \filldraw[color=orange, fill=orange!20, thick] (0,0) ellipse (0.1 and 0.14);
        \draw[orange, thin] (0,-0.02) ellipse (0.1 and 0.03);
    \end{scope}


    \begin{scope}[shift={(-3.69,2.34)}, rotate=30]
        \filldraw[color=orange, fill=orange!20, thick] (0,0) ellipse (0.25 and 0.655);
        \draw[orange, thin] (0,0.02) ellipse (0.25 and 0.08);
    \end{scope}

    \begin{scope}[shift={(-2.5,2.17)}, rotate=20]
        \filldraw[color=orange, fill=orange!20, thick] (0,0) ellipse (0.2 and 0.34);
        \draw[orange, thin] (0,-0.02) ellipse (0.2 and 0.04);
    \end{scope}

    \begin{scope}[shift={(-1.36,2.02)}, rotate=10]
        \filldraw[color=orange, fill=orange!20, thick] (0,0) ellipse (0.18 and 0.34);
        \draw[orange, thin] (0,0) ellipse (0.18 and 0.04);
    \end{scope}

    \begin{scope}[shift={(-4.63,1.39)}, rotate=60]
        \filldraw[color=orange, fill=orange!20, thick] (0,0) ellipse (0.2 and 0.32);
        \draw[orange, thin] (0,-0.02) ellipse (0.2 and 0.05);
    \end{scope}

    \begin{scope}[shift={(-4.92,0.71)}, rotate=75]
        \filldraw[color=orange, fill=orange!20, thick] (0,0) ellipse (0.1 and 0.14);
        \draw[orange, thin] (0,-0.02) ellipse (0.1 and 0.03);
    \end{scope}
    

    \draw[thick] (0,0) ellipse (6 and 4);
    \draw[color=mygreen, ultra thick] (0,0) ellipse (6 and 4);

    \draw[color=red, thick] (6.63,0) .. controls (8,3.1) and (2,6) .. (0,4.88);
    \draw[color=red, thick] (-6.63,0) .. controls (-8,-3.1) and (-2,-6) .. (0,-4.88);
    \draw[color=red, thick] (-6.63,0) .. controls (-8,3.1) and (-2,6) .. (0,4.88);
    \draw[color=red, thick] (6.63,0) .. controls (8,-3.1) and (2,-6) .. (0,-4.88);

    \draw[color=blue, thick] (-4,-2.98) .. controls (-2.8,-2.1) and (2.8,-2.1) .. (4,-2.98);
    \draw[color=blue, thick] (-4,-2.98) .. controls (-2.7,-5.5) and (2.7,-5.5) .. (4,-2.98);

    \draw[color=blue, thick] (-4,2.98) .. controls (-2.8,5.5) and (2.8,5.5) .. (4,2.98);
    \draw[color=blue, thick] (-4,2.98) .. controls (-2.7,2.1) and (2.7,2.1) .. (4,2.98);

    \draw[color=blue, thick] (-4,2.98) .. controls (-7.5,2.5) and (-7.5,-2.5) .. (-4,-2.98);
    \draw[color=blue, thick] (-4,2.98) .. controls (-5.5,2.2) and (-5.5,-2.2) .. (-4,-2.98);

    \draw[color=blue, thick] (4,2.98) .. controls (7.5,2.5) and (7.5,-2.5) .. (4,-2.98);
    \draw[color=blue, thick] (4,2.98) .. controls (5.5,2.2) and (5.5,-2.2) .. (4,-2.98);

    

    \draw[color=orange, thick] (-5.1,0) .. controls (-4,-3) and (-1,-1.5) .. (0,-1.3);
    \draw[color=orange, thick] (5.1,0) .. controls (4,-3) and (1,-1.5) .. (0,-1.3);
    \draw[color=orange, thick] (5.1,0) .. controls (4,3) and (1,1.5) .. (0,1.3);
    \draw[color=orange, thick] (-5.1,0) .. controls (-4,3) and (-1,1.5) .. (0,1.3);

    \node at (6,0) {\large $\times$};    
    \node at (0,-4) {\large $\times$};  
    \node at (-6,0) {\large $\times$};    
    \node at (0,4) {\large $\times$};  
    \node at (-4,2.98) {\large $\times$};    
    \node at (4,2.98) {\large $\times$};  
    \node at (-4,-2.98) {\large $\times$};    
    \node at (4,-2.98) {\large $\times$}; 

    \filldraw[black] (0,-5.9) circle (0pt)
        node[anchor=south]{\large{$\mu_{t,n}$}};

    \filldraw[black] (-6.8,0) circle (0pt)
        node[anchor=east]{\large{$\mu_{t,4}$}};

    \filldraw[black] (0,5.9) circle (0pt)
        node[anchor=north]{\large{$\mu_{t,n}$}};

    \filldraw[black] (6.8,0) circle (0pt)
        node[anchor=west]{\large{$\mu_{t,4}$}};

    \filldraw[black] (-5,4.3) circle (0pt)
        node{\large{$\mu_s$}};
    \filldraw[black] (5,4.3) circle (0pt)
        node{\large{$\mu_s$}};
    \filldraw[black] (-5,-4.3) circle (0pt)
        node{\large{$\mu_s$}};
    \filldraw[black] (5,-4.3) circle (0pt)
        node{\large{$\mu_s$}};

\begin{scope}[shift={(2.9,-3.25)}, rotate=200,scale=0.725]
         \draw[ultra thick] (0,2) arc(90:270:-0.3cm and 2cm);
         \draw[ultra thick, style=dashed] (0,2) arc(90:270:0.3cm and 2cm);
\end{scope}

\filldraw[black] (4,-5.8) circle (0pt) node[anchor=north]{\large{$(\mu_s \mu_{t,4}\mu_{s}\mu_{t,n})^2$}};

\begin{scope}[shift={(3.4,2.2)}, rotate=67,scale=1]
    \draw[thick, -<_] (-6.75,-2.75) arc (-270:20:0.4);
\end{scope}


\begin{scope}[shift={(-4.5,-11)}, scale=1.5]

 \tikzstyle{every node}=[font=\scriptsize]
       \node[draw, circle,thick] (p1) at (0,-2.2) {\fontsize{12pt}{12pt}\selectfont $2$};
       \node[draw, circle,thick] (p2) at (3,0) {\fontsize{12pt}{12pt}\selectfont $2$};
       \node[draw, circle,thick] (p3) at (6,-2.2) {\fontsize{12pt}{12pt}\selectfont $2$};
       \node[draw, rectangle,thick] (p4) at (3,-2.2) {\fontsize{12pt}{12pt}\selectfont $2n$};
       \node[draw, circle,thick] (p5) at (3,-4.4) {\fontsize{12pt}{12pt}\selectfont $2$};
       \node[draw, rectangle,thick] (p6) at (3,1.6) {\fontsize{12pt}{12pt}\selectfont $8-2n$};
       \draw[-,thick] (p1) to  node[rotate=45]{\large $\times$}    (p2);
       \draw[-,thick] (p2) to  node[rotate=45]{\large $\times$}    (p3);
       \draw[-,thick] (p3) to  node[rotate=45]{\large $\times$}    (p5);
       \draw[-,thick] (p5) to  node[rotate=45]{\large $\times$}    (p1);
       \draw[->-,thick] (p1) to      (p4);
       \draw[->-,thick] (p3) to      (p4);
       \draw[-<-,thick] (p2) to      (p4);
       \draw[-<-,thick] (p5) to      (p4);
       \draw[-<-,thick] (p1) to      (0,1.6);
       \draw[-,thick] (0,1.6) to      (p6);
       \draw[->-,thick] (p3) to      (6,1.6);
       \draw[-,thick] (6,1.6) to      (p6);
    
\end{scope}

\end{tikzpicture}
    \caption{The geometry of the $(\mu_{t,n}^{\varepsilon}\mu_s^{\varepsilon}\mu_{t,4}^{\varepsilon}\mu_s^{\varepsilon}  )^2 $ flux torus, and the associated 4d $\cN=1$ quiver of flux $\mathcal{F}=(-2^{2n},0^{8-2n})$. There are four compact three-cycles given by the fibration of the section over the circle. An $A_1$ singularity from the collapsed fiber sits over the whole circle. The curves $M_{1\cdots 2n}$ collapse at four points in which they touch the $A_1$ singularity, giving rise to the four sets of bifundamental hypers charged under the $SU(2n)$ flavor symmetry. The curves $M_{2n+1\cdots 8}$ collapse only at two points, giving rise to the two sets of hypers charged under the $SU(8-2n)$ flavor symmetry.}
    \label{fig: n3 volumes and quiver closed}
\end{figure}
In closing up the profile in Figure~\ref{fig: local s3s4s3s4 domain wall}, we glue the two open $SU(2)$ nodes, giving rise to the fourth gauged $SU(2)$ node in the flux torus quiver. 

Specializing to the case of $n=3$, we can calculate the action of the monodromy map on our curves. We will reserve the discussion of the cases $n=1,2$ for Section~\ref{sec: Examples} when we have developed more systematic algebraic machinery. We use the geometric flux $\mathcal{F}_3=4\Sigma-2F-2M_7-2M_8$ such that $N(\mathcal{F}_3)=(-2^6,0^2)$. The action of the monodromy map $\mathcal{M}[\mathcal{F}_3]=(\mu_s\mu_{t,4}\mu_{s}\mu_{t,3})^2$ on the curves is
\be \label{eq: Curves after domain wall n3}
\mathcal{M}[\mathcal{F}_3]:  \qquad 
\left\{\begin{aligned}
F &\rightarrow F+8E-\mathcal{F}_3\\
\Sigma &\rightarrow\Sigma+ 5E-\mathcal{F}_3 \\
M_i &\rightarrow M_i+ 3E-\frac{1}{2}\mathcal{F}_3  \, , \qquad i=1\cdots 6 \\
M_i &\rightarrow M_i+4 E-\frac{1}{2} \mathcal{F}_3\, , \qquad i=7,8 \, ,
\end{aligned}\right.
\ee
and can be expressed as 
\be
    \mathcal{M}[\mathcal{F}_3]: \Gamma \rightarrow \Gamma+\frac{1}{2}(\mathcal{F}_3 \cdot \Gamma)E-\frac{1}{2}(E\cdot \Gamma)\mathcal{F}_3+3(E\cdot \Gamma)E \, , 
\ee
for all $\Gamma \in I_{1,9}$. We will explain in Sections~\ref{sect:elemfluxtubes} and~\ref{sect:concatenating} why such a map is an isomorphism of the threefold $X_S^f$. We can hence close the geometry to $\mathcal{T}[\mathcal{F}_3]$ using this automorphism, giving us the geometric analogue of the flux torus.

\paragraph{Enhanced Flavor Symmetry and Monodromy Map.} 

To determine the enhanced flavor symmetry, we need the action of the monodromy map on the $E_8$ roots
\begin{equation} \label{eq: n=3 root shifts}
\mathcal{M}[\mathcal{F}_3]:  \qquad 
\left\{
\begin{aligned}
    \alpha_1 & \rightarrow \alpha_1  -3E   \\
    \alpha_4 & \rightarrow \alpha_4  +E   \\
    \alpha_i &\rightarrow \alpha_i \, ,  \qquad \qquad i= 2,3,5,\cdots,8 \, .
\end{aligned}\right.
\end{equation}
From here, we can determine the lattice $\Lambda_{n=3}$ that is invariant under $\mathcal{M}[\mathcal{F}_3]$ is
\be
\ba
    \Lambda_{n=3} = \langle &\alpha_1 +  \alpha_2+2\alpha_3+3\alpha_4+2\alpha_5+\alpha_6, \alpha_2, \alpha_3,\alpha_5,\alpha_6,\alpha_7,\alpha_8 \rangle_{\mathbb{Z}} \, ,
\ea
\ee
written in a basis of simple roots. One can furthermore verify that this lattice coincides with the lattice $(\mathcal{F}_3)^{\perp}$, the lattice of vectors perpendicular to $\mathcal{F}_3$. The inner form on this lattice is (minus) the Cartan of $E_6 \times A_1$. This is exactly the enhanced flavor symmetry~\eqref{eq:preserved flavor symms} appropriate for the quiver theory constructed in Figure~\ref{fig: n3 volumes and quiver closed}. Moreover, we again see a shift of the $E_8$ roots by integer multiples of $E$, which is compatible with the shift of the hyper masses by $m_{\text{KK}}$ observed in the field theory \cite{Sabag:2022hyw}, see Appendix~\ref{app:KKshift} for more details. 
 
Lastly, a variant of this construction allows us to realize the fluxes $(-2q^{6},0^2)$ by concatenating the building blocks $(\mu_{t,3}^{\varepsilon}\mu_s^{\varepsilon}\mu_{t,4}^{\varepsilon}\mu_s^{\varepsilon})^{2q}$. From the associated monodromy map one recovers the $E_8$ root shift
\be \label{eq: explicit form of 2qn4 monodromy map}
\mathcal{M}[q\mathcal{F}_3] = (\mu_s\mu_{t,4}\mu_s\mu_{t,3})^{2q} :  \qquad 
\left\{\begin{aligned}
\alpha_1 &\rightarrow \alpha_1 - 3q E\\
 \alpha_4 &\rightarrow \alpha_4 +qE \\
\alpha_i &\rightarrow \alpha_i \qquad \qquad i=2,3,5,\cdots,8 \, ,
\end{aligned}\right.
\ee
The enhanced non-abelian flavor symmetry is again $E_6 \times A_1$ and the shift by $E$ is still compatible with the field theory observations.

\subsection{Building Blocks}\label{sect:elemfluxtubes}

The most basic flux torus constructed above arises as the configuration
\begin{equation}
      \mathcal{D}[2\alpha_1] := (\mu_{t,4}^{\varepsilon}\mu_s^{\varepsilon})^2 \, .
\end{equation}
which we call the \textbf{building block} associated with the flux $2\alpha_1$. As we have seen, the associated monodromy map $\mathcal{M} = (\mu_s \mu_{t,4})^2 $ used to glue the flux tube to a torus acts on $H^2(X_S^f,\Z) \simeq I_{1,9}$ as
\begin{equation}\label{eq ES for -2alpha1}
\mathcal{M}[ 2\alpha_1] \Gamma = \Gamma + (\alpha_1 \cdot \Gamma ) E -  (E \cdot \Gamma )\alpha_1 +  (E \cdot \Gamma) E  \, .  
\end{equation}
As we argued in Section \ref{sec: Geometry of XS}, $X_S^f$ admits automorphisms $w_\beta$ which act as Weyl reflections on $I_{1,9}$ for all $\beta \in \Lambda_{E_8} \setminus \Lambda_{A_7}$. We can hence act with such maps to convert 
$\mathcal{D}[2\alpha_1]$ into building blocks for other fluxes. 
Using the fact that the inner form as well as $E$ are invariant under Weyl reflections and furthermore $w_\beta^2=1$, \eqref{eq ES for -2alpha1} implies that
\begin{equation}\label{eq:conjfluxformula_alpha1}
    w_\beta \mathcal{M}[2 \alpha_1]w_\beta\,\,\Delta = \Delta +\left(w_\beta \alpha_1 \cdot \Delta\right) E   - \left(E \cdot \Delta \right) w_\beta\alpha_1 
    + \left(E \cdot \Delta \right)E \, .
\end{equation}
where we have written $\Delta = w_{\beta} \Gamma $. As \eqref{eq ES for -2alpha1} holds for any $\Gamma \in I_{1,9}$, the above holds for any $\Delta \in I_{1,9}$ as well. We can think of the above as a building block associated to the flux 
$2 w_\beta \alpha_1$, with the associated automorphism characterized by the map 
\begin{equation}
\mathcal{M}[2 w_\beta \alpha_1] = w_\beta\mathcal{M}[2 \alpha_1]w_\beta \, .
\end{equation}
For $\beta \in \Lambda_{E_8} \setminus \Lambda_{A_7}$, both $w_\beta$ and 
$\mathcal{M}[2 \alpha_1]$ are automorphisms of $X_S^f$, so that the same must hold true for 
$\mathcal{M}[2 w_\beta \alpha_1]$. We can hence identify the geometries corresponding to the ends of the interval.  

The above becomes particularly simple if $\beta \in \Lambda_{D_8} \setminus \Lambda_{A_7}$ (see \eqref{eq D8 lattice}) as this implies that
\begin{equation}
    [w_\beta,\mu_{t,4}] = 0 \, .  
\end{equation}
We can then write 
\begin{equation}
\mathcal{M}[2 w_{\beta} \alpha_1] = w_{\beta} ( \mu_s \mu_{t,4})^2 w_{\beta}
=   ( w_{\beta} w_{\alpha_1} w_{\beta} \,\mu_{t,4})^2 
=   (w_{w_\beta\alpha_1} \,\mu_{t,4}  )^2 
\end{equation}
where we have used that $\mu_s = w_{\alpha_1}$ and $ w_{\beta} w_\alpha w_{\beta}= w_{w_{\beta} \alpha}$ a Weyl reflection by $w_\beta \alpha$. What we have just seen is that such domain wall configurations are constructed in exactly the same way as 
$\mathcal{D}[2 \alpha_1]$ with the replacement of $\mu_s$ by a Weyl reflection by a different root which is $w_\beta\alpha_1$. 

The set of roots $\{\alpha_+\}$ of the form 
$\alpha_+ = w_{\beta} \alpha_1$
for $ \beta \in  \Lambda_{D_8} \setminus \Lambda_{A_7}$ is equal to\footnote{As $\beta \cdot F =0$ for all roots $\beta \in \Lambda_{D_8}$,
any element of the Weyl group of $\Lambda_{D_8}$ acts trivially on $F$.
Hence we have that $-1 = F \cdot \alpha_1 = w_{\beta} F \cdot (w_{\beta} \alpha_1) = F \cdot (w_{\beta} \alpha_1)$. By inspection one sees that indeed all roots with the property $\alpha_+ \cdot F = 1$ can be reached by an appropriate $w_{\beta}$. }  
\begin{equation}
   R_+ := \{\alpha_+ \in \Lambda_{E_8} \, |\, \alpha_+^2 = -2\quad \alpha_+ \cdot F = 1 \} \subset I_{1,9} \, .
\end{equation}
and for all such roots we can construct an automorphism 
\begin{equation}\label{eq:defdomainwallplus}
    \mathcal{M}[2 \alpha_+] := \left(w_{\alpha_+} \mu_{t,4}  \right)^2 \, ,
\end{equation}
which realizes the building block 
\begin{equation} \label{eq: Rplus building block}
      \mathcal{D}[2\alpha_+] := (\mu_{t,4}^{\varepsilon} w_{\alpha_+}^{\varepsilon})^2 \, .
\end{equation}
with flux 
$2\alpha_+$. We will comment shortly what is meant by $w_{\alpha_+}^{\varepsilon}$. 

Likewise, we can consider 
\begin{equation}
   R_- := \{\alpha_- \in \Lambda_{E_8} \, |\, \alpha_-^2 = -2\quad \alpha_- \cdot F = -1 \} \subset I_{1,9} \, .
\end{equation}
For any such root, we can construct an automorphism for the flux 
$2\alpha_-$ by 
taking the inverse of 
$\mathcal{M}[2\alpha_+]$ where $\alpha_- = -\alpha_+\in R_+$, i.e. 
\begin{equation}\label{eq:defdomainwallminus}
 \mathcal{M}[-2 \alpha_+] =  \mathcal{M}[2 \alpha_-] := \left(\mu_{t,4}  w_{\alpha_-} \right)^2
\end{equation}
where we have used that $w_{\alpha_-} = w_{\alpha_+} $. This provides another set of building blocks 
\be
    \mathcal{D}[2\alpha_-] := (w_{\alpha_-}^{\varepsilon} \mu_{t,4}^{\varepsilon})^2 \, .
\ee 

We may summarize the set of all building blocks as follows: let $R = R_+ \cup R_-$. For any element $\alpha \in R$, there is an associated automorphism defined as above which acts on curves $\Gamma \in H_2(X_S,\Z)$ as
\begin{equation}\label{eq:fluxformula_alpha}
    \mathcal{M}[2 \alpha] \,\,\Gamma = \Gamma + \left(\alpha \cdot \Gamma\right) E - \left(E \cdot \Gamma \right) \alpha + \left(E \cdot \Gamma \right)E \, .
\end{equation}
One can check that this formula is compatible with $\mathcal{M}[2\alpha]^{-1} = \mathcal{M}[-2\alpha]$, which we show in the next section. To each such automorphism, we associate the building block 
$\mathcal{D}[2\alpha]$ of flux 
$2\alpha$. As we have seen, these are all composed of Weyl reflections $w_\beta$ by roots $\beta \in \Lambda_{D_8}\setminus \Lambda_{A_7}$ and 
$\mathcal{M}[2\alpha_1] = \tau_{-\alpha_1}$, so that there is an associated automorphism of $X_S^f$ in each case and we glue to find the geometric realization of a flux torus.

Let us now elaborate on the local building blocks $w_{\alpha}^{\varepsilon}$ for $\alpha \in R$. In Appendix~\ref{sec: App Gluing Profiles} we prove that we can always uniquely write
\be \label{eq: unique way to rewrite Weyl}
    w_{\alpha} = \mu_{t,\mathbf{n}} \mu_s \mu_{t,\mathbf{n}} \, ,
\ee
for any $\alpha \in R$ and an automorphism of $H_2(S,\Z)$ which generalizes $\mu_{t,n}$. This result rests on the fact that for any such $\alpha$ we have
\be \label{eq: generic alpha in R}
    \alpha = \pm \Sigma \pm \left( \frac{1}{2}|\mathbf{n}|-1 \right)F\mp\sum_{i=1}^8  n_i M_i \, , \qquad \alpha \in R_{\pm} \, ,
\ee
for a unique $\mathbf{n} \in \mathcal{S} := \{ (n_i) \in \mathbb{Z}_2^8 \, | \, \sum_{i=1}^8 n_i = 0 \, \, \, \text{mod} \, \, 2 \}$. For each such $\mathbf{n} \in \mathcal{S}$, we can then define
\be \label{eq: more general mutn def}
\mu_{t,\mathbf{n}} :  \qquad 
\left\{\begin{aligned}
F &\rightarrow F'=F \\
 \Sigma &\rightarrow \Sigma' = \Sigma+\frac{1}{2}|\mathbf{n}|F-\sum_{i=1}^8 n_i M_i \\
M_i &\rightarrow M_i' = F-M_i \qquad \ \qquad \qquad \qquad  n_i \equiv 1 \\
M_i &\rightarrow M_i' = M_i\qquad \ \qquad \qquad \qquad \qquad  n_i \equiv 0 \, .
\end{aligned}\right.
\ee
We will therefore define the local building blocks associated to a Weyl reflection $w_{\alpha}$ for $\alpha \in R$ as
\be
    w_{\alpha}^{\varepsilon} := \mu_{t,\bold{n}}^{\varepsilon} \mu_s^{\varepsilon} \mu_{t,\bold{n}}^{\varepsilon} \, .
\ee
Note that compared to the quivers associated to $\mu_s=w_{\alpha_1}$, those associated to a generic Weyl reflection $w_\alpha$ contain $2|\bold{n}|$ extra light chirals. Once the building block $w_{\alpha}^{\varepsilon}$ is concatenated with the building block $\mu_{t,4}^{\varepsilon}$ as $(\mu_{t,4}^{\varepsilon} w_{\alpha}^{\varepsilon})$ however, there are no additional unexpected light states. The building block $(\mu_{t,4}^{\varepsilon} w_{\alpha}^{\varepsilon})$ gives rise to 8 light chiral multiplets in the quiver, as expected. To see this, we use the important result of Appendix~\ref{sec: App Gluing Profiles}
\be
    \mu_{t,\bold{n}}^{\varepsilon} \mu_{t,\bold{m}}^{\varepsilon}  = \mu_{t,\bold{n}+\bold{m}}^{\varepsilon} \, ,
\ee
that implies we can write $(\mu_{t,4}^{\varepsilon} w_{\alpha}^{\varepsilon})=(\mu_{t,(1^8)+\bold{n}}^{\varepsilon}\mu_s^{\varepsilon} \mu_{t,\bold{n}}^{\varepsilon})$, where the addition is defined component-wise mod 2. It is clear that the quiver associated to this building block contains 8 light chirals. Moreover, the mod 2 addition has an important physical interpretation: it captures the action of integrating out massive chirals. Indeed, if both $n_i,m_i=1$, then we have a pair of chirals charged under the same flavor symmetry but in conjugate representations, meaning they generate a superpotential mass term, so are integrated out in the IR effective theory. 

Using~\eqref{eq: unique way to rewrite Weyl}, and the further basic properties we explain in Appendix~\ref{sec: App Gluing Profiles}, one can always rewrite the building block 
$\mathcal{D}[2\alpha]$
as an alternating stacking of local building blocks $\mu_s^{\varepsilon}$ and $\mu_{t,\mathbf{n}}^{\varepsilon}$. We know the volume profile for the local building block $\mu_s^{\varepsilon}$ as shown in Figure~\ref{fig: mus local volumes}, and the volume profile for the local building block $\mu_{t,\mathbf{n}}^{\varepsilon}$ as shown in Figure~\ref{fig: local mutn dom wall} up to a relabeling of the curves $M_i$. Thus we always know the volume profiles for each local building block in 
$\mathcal{D}[2\alpha]$, meaning we can always glue such volume profiles together in the way demonstrated in Sections~\ref{sec: flux 2^8} and~\ref{sec: flux 2^2n}.

\subsection{Concatenating Building Blocks}\label{sect:concatenating}

Given that we managed to construct the geometric analogue for each flux torus of flux 
$2\alpha$ with $\alpha \in R$, it is tempting to speculate that these can be concatenated to realize more general fluxes. That this is indeed the case will be the purpose of the present section. We begin by stating the general result:
\begin{Theorem}[]
\emph{For any $\mathcal{F}$ such that 
$\mathcal{F}/2 \in \Lambda_{E_8} \subset H_2(X_S,\Z)$ there is an associated flux domain wall which acts on any element $\Gamma \in I_{1,9}$ as the Eichler-Siegel transformation
\begin{equation}\label{eq:fluxformula_general}
    \mathcal{M}[\mathcal{F}] \Gamma = \Gamma + \frac{1}{2} (\mathcal{F} \cdot \Gamma) E - \frac12(E \cdot \Gamma) \mathcal{F} - \frac18(E\cdot \Gamma) (\mathcal{F}\cdot \mathcal{F}) E\, ,
\end{equation}
and this map is induced from the action of an automorphism of $X_S^f$. Moreover, 
this map defines a group homomorphism
\be \label{eq: homomorphism property}
    \mathcal{M} : \left\{ \mathcal{F} \,\, |\,\, \mathcal{F}/2 \in \Lambda_{E_8} \right\} \rightarrow \mbox{Aut}(I_{1,9}) \, .
\ee
Here, we can think of $I_{1,9}$ as $H_2(S,\Z) \simeq H^2(X_S^f,\Z)$. }
\end{Theorem}

\proof: 

Comparing with \eqref{eq:fluxformula_alpha} immediately shows that the above formula works for any
$\mathcal{F} = 2 \alpha$ with $\alpha \in R$. Furthermore, a short computation shows that $\mathcal{M}[\mathcal{F}]$ acts as an isometry on 
$I_{1,9}$:
\begin{equation}
    A \cdot B = (\mathcal{M}[\mathcal{F}] A ) \cdot (\mathcal{M}[\mathcal{F}] B)
\end{equation}
for all $A,B \in I_{1,9}$. Now for any $\mathcal{F},\mathcal{F}'$ with $\mathcal{F}/2 $ and $\mathcal{F}'/2$ $ \in \Lambda_{E_8}$ one can now work out 
\begin{equation}
\begin{aligned}
        \mathcal{M}[\mathcal{F}]\mathcal{M}[\mathcal{F}'] \Gamma = \mathcal{M}[\mathcal{F} + \mathcal{F}']\, \Gamma
\end{aligned}
\end{equation}
which shows the homomorphism property. These two are in fact well-known properties of Eichler-Siegel transformations, see e.g. \cite{heckman2001moduli} for a discussion in the context of rational elliptic surfaces.

As a consequence, we can decompose
\begin{equation}
    \mathcal{M}[\mathcal{F}] = \prod_i \mathcal{M}[2 \alpha_i] \, ,
\end{equation}
and this expression is well-defined for any flux of the form\footnote{We allow the case where $\alpha_i = \alpha_j$, so that $\mathcal{F}$ can be any linear combination of roots in $R$.}  
$\mathcal{F} = \sum_i 2 \alpha_i$ with $\alpha_i \in R$, and acts on $I_{1,9}$ as \eqref{eq:fluxformula_general}. We have already argued that 
$\mathcal{M}[2 \alpha]$ for $\alpha \in R$ is induced by an automorphism, so that the composition of these gives an automorphism inducing $\mathcal{M}[\mathcal{F}]$.

It only remains to be shown that the $\Z$-span of the set $R$ is $\Lambda_{E_8}$. For convenience, let us pass to the representation theoretic description, where we use the $\mathfrak{so}(16)$ weight basis to write any element 
$\alpha \in R$ as
\begin{equation}
\alpha = N^{-1}(\ell) \, , \qquad  \ell = \left(\ell_1,\ell_2,\ell_3,\ell_4,\ell_5,\ell_6,\ell_7,\ell_8 \right) \in \Lambda_w(\mathfrak{so}(16))\, .
\end{equation}

The fluxes $N(R)$ realize all such elements with the property that $\ell_i = \pm \tfrac12$ and $\sum \ell_k \in 2 \Z$. As any root of $E_8$ can be written as a linear combination of such roots, the $\Z$-span of $R$ is hence equal to $\Lambda_{E_8}$. It follows that we can build domain walls for any flux $\mathcal{F}$ such that $\mathcal{F}/2 \in \Lambda_{E_8}$. Given any weight of $\Lambda_w(\mathfrak{so}(16))$, the associated flux $\mathcal{F}$ is furthermore unique since the kernel of the map $N$ trivializes when restricted to $\Lambda_{E_8}$.\begin{flushright}\qedsymbol{} \end{flushright}

As a consequence of the above, we can close up the circle and find a variety $\mathcal{T}[\mathcal{F}]$ as a circle fibration of $X_S^f$ with monodromy 
$\mathcal{M}[\mathcal{F}]$ for any $\mathcal{F}$ with $\mathcal{F}/2 \in \Lambda_{E_8}$. We can describe the building blocks out of which this is formed by expanding 
$\mathcal{F}=\sum_{i}2 \alpha_i$ for $\alpha_i \in R$, and concatenating the associate building blocks: 
\be
    \mathcal{D}[\mathcal{F}] :=\prod_i \mathcal{D}[2\alpha_i] \, ,
\ee
which in turn allows to reconstruct the associated quiver. One immediate consequence of this is that the quivers have an even number of chirals charged under each gauge node, meaning it has vanishing Witten anomaly~\cite{Witten:1982fp}. 

As compactification of M-Theory on $\mathcal{T}[\mathcal{F}]$ geometrically engineers a 4d $\mathcal{N}=1$ field theory, it is natural to conjecture that the topological space $\mathcal{T}[\mathcal{F}]$ admits a $G_2$ metric. 

Determining the roots invariant under the automorphism 
$\mathcal{M}[2\alpha]$ is very simple as well. Observe that~\eqref{eq:fluxformula_general} simplifies when acting on a root $\gamma \in \Lambda_{E_8} \subset H_2(X_S,\mathbb{Z})$ where it acts as
\be
    \mathcal{M}[\mathcal{F}]\gamma = \gamma + \left( \frac{\mathcal{F}}{2}\cdot \gamma \right) E \, .
\ee
Hence the lattice of invariant roots of $\Lambda_{E_8}$ is always given by $\mathcal{F}^{\perp}$, which is exactly what is expected for compactification on a torus with flux $N(\mathcal{F})$. 

The fluxes $\mathcal{F}$ such $\mathcal{F}/2 \in \Lambda_{E_8}$ are precisely non-center fluxes. These are fluxes breaking $E_8\to H\times U(1)^k$, where $H$ denotes the non-abelian part of the residual subgroup, which corresponds to fluxes in $U(1)^k$ but not the center of $H$. In the present work we only focus on such non-center fluxes, however we point out that there exist more general fluxes in $\Lambda_{E_8}$. For example, we have already discussed that
\begin{equation}
    \left(-\frac{1}{2},-\frac{1}{2},-\frac{1}{2},-\frac{1}{2},-\frac{1}{2},-\frac{1}{2},-\frac{1}{2},-\frac{1}{2}\right)
\end{equation}
breaks
\begin{equation}
    E_8\to E_7\times U(1)\,.
\end{equation}
In terms of this subgroup, this flux corresponds to a fractional flux in the $U(1)$ that needs to be compensated by a flux in the center of $E_7$, which can be characterized by a non-trivial Stiefel-Whitney class in $E_7/\Z_2$. As a consequence, the actual preserved symmetry is smaller. In this case, we can at most preserve an $F_4$ subgroup of $E_7$ \cite{Kim:2017toz}. It would be interesting to investigate how these more general fluxes appear in our construction.

\subsection{Examples} \label{sec: Examples}

As examples, let us show that how to construct the field theory fluxes $(-2^8),(-2^6,0^2),(-2^4,0^4)$, and $(-2^2,0^6)$ and how to determine the preserved flavor symmetries using the domain wall prescription to see that it coincides with the geometric understanding we developed in Sections~\ref{sec: flux 2^8} and ~\ref{sec: flux 2^2n}. 

\paragraph{Flux $(-2^8)$.} We have already analyzed the monodromy map 
$\mathcal{M}[4\alpha_1] = \mathcal{M}[2\alpha_1]^2$. Indeed, since 
$2\alpha_1\in R$, the geometric flux 
$\mathcal{F}_4=4\alpha_1$ is the correct geometric flux to realize $N(\mathcal{F}_4)=(-2^8)$. Moreover, we verified that the monodromy map had the simple action~\eqref{eq: Action of M[-4a]}, in agreement with~\eqref{eq:fluxformula_general}. We have that $(\mathcal{F}_4)^{\perp} \cong E_7$ as before.

\paragraph{Flux $(-2^6,0^2)$.} We can decompose the flux $(-2^6,0^2)$ as
\begin{equation}
    (-2^6,0^2) = 2 \left( -\frac{1}{2}^8 \right) + 2 \left( -\frac{1}{2}^6, \frac{1}{2}^2 \right) \, .
\end{equation}
Clearly we would associate the root 
$2\alpha_1$ to the first term. For the second term, we associate it to the root
$2\hat{\alpha}_3$ such that 
$N(2 \alpha_1+2\hat{\alpha}_3)=(-2^6,0^2)$, which in this case is given explicitly by
\begin{equation} \label{eq:hat alpha 3}
   \hat{\alpha}_3  = \Sigma-M_7-M_8 \, .
\end{equation}
The geometric flux 
$\mathcal{F}_{3}:=2 \alpha_1+2\hat{\alpha}_3$ corresponds to the $E_8$ element $N(\mathcal{F}_{3})=(-2^6,0^2)$, and by the homomorphism~\eqref{eq: homomorphism property}, the monodromy map is given by
\be
   \mathcal{M}[\mathcal{F}_{3}] =  \mathcal{M}[2\hat{\alpha}_3] \mathcal{M}[2\alpha_1] \, .
\ee   
An explicit computation shows that
\begin{equation}
    \mu_{t,4}\mu_{t,3}\mu_s \mu_{t,3}\mu_{t,4} = w_{\Sigma-M_7-M_8} \, ,
\end{equation}
which we insert into~\eqref{eq: Rplus building block} to derive
\be
    \mathcal{M}[\mathcal{F}_3] = (\mu_s \mu_{t,4} \mu_s \mu_{t,3})^2 \, .
\ee
Thus we recover the monodromy map used in Section~\ref{sec: flux 2^2n} to close the flux torus, and the preserved flavor symmetry $(\mathcal{F}_3)^{\perp} \cong E_6 \times A_1$.

\paragraph{Flux $(-2^4,0^4)$.} We can decompose the flux $(-2^4,0^4)$ as
\be
(-2^4,0^4)=2(-\frac{1}{2}^8)+2(-\frac{1}{2}^4,\frac{1}{2}^4)
\ee
and so associate the root 
$2\hat{\alpha}_2$ to the second term, with $\hat{\alpha}_2$ given by
\begin{equation} \label{eq:hat alpha 2}
    \hat{\alpha}_2=\Sigma+F-\sum_{i=5}^8 M_i\, .
\end{equation}
We therefore associate the geometric flux 
$\mathcal{F}_{2}=2\alpha_1+2\hat{\alpha}_2$, whose $E_8$ element is $N(\mathcal{F}_{2})=(-2^4,0^4)$. The associated monodromy map for this flux is
\be
    \mathcal{M}[\mathcal{F}_{2}] = \mathcal{M}[2\hat{\alpha}_2]\mathcal{M}[2\alpha_1]\, .
\ee    
We can further verify that
\begin{equation}
    \mu_{t,4}\mu_{t,2}\mu_s \mu_{t,2}\mu_{t,4} = w_{ \Sigma+F-\sum_{i=5}^8 M_i} \, ,
\end{equation}
and use~\eqref{eq: Rplus building block} to derive
\be
    \mathcal{M}[\mathcal{F}_2] = (\mu_{s} \mu_{t,4} \mu_s \mu_{t,2})^2 \, .
\ee
This monodromy map corresponds to the quiver constructed by closing the building block $\mathcal{D}[\mathcal{F}_2]=(\mu_{t,2}^{\varepsilon}\mu_s^{\varepsilon}\mu_{t,4}^{\varepsilon}\mu_s^{\varepsilon}  )^2 $, as shown in Figure~\ref{fig: n3 volumes and quiver closed}. Using this flux, we can determine the invariant sublattice of $\Lambda_{E_8}$ is
\be 
   (\mathcal{F}_{2})^{\perp}= \langle 
   \alpha_1+\alpha_2+\alpha_3+2\alpha_4+2\alpha_5+2\alpha_6+\alpha_7,\alpha_2,\alpha_3,\alpha_4,\alpha_5,\alpha_7,\alpha_8 \rangle_{\mathbb{Z}} \, .
\ee
The inner form generated by this lattice is the Cartan matrix of $D_7$, which is exactly the enhanced flavor symmetry~\eqref{eq:preserved flavor symms} appropriate for this quiver theory.

\paragraph{Flux $(-2^2,0^6)$.} To conclude our list of examples, we consider the flux 
\be
(-2^2,0^6)=2 (-\frac{1}{2}^8)+2(-\frac{1}{2}^2,\frac{1}{2}^6)\,.
\ee
To the second entry, we associate the root
\begin{equation} \label{eq:hat alpha 1}
    \hat{\alpha}_1 = \Sigma+2F-\sum_{i=3}^8 M_i \, .
\end{equation}
For the geometric flux 
$\mathcal{F}_{1}=2\alpha_1+2\hat{\alpha}_1$, with corresponding $E_8$ element $N(\mathcal{F}_{1})=(-2^2,0^6)$, we associate to it the monodromy map
\be
    \mathcal{M}[\mathcal{F}_{1}] =\mathcal{M}[2\hat{\alpha}_1] \mathcal{M}[2\alpha_1]\, .
\ee
We can further verify that
\begin{equation}
    \mu_{t,4}\mu_{t,1}\mu_s \mu_{t,1}\mu_{t,4}=w_{\Sigma+2F-\sum_{i=3}^8 M_i} \, ,
\end{equation}
meaning we can similarly write the monodromy map as
\be
    \mathcal{M}[\mathcal{F}_{1}] =(\mu_{s} \mu_{t,4} \mu_s \mu_{t,1})^2 \, .
\ee
This monodromy map is appropriate to close up the building block $\mathcal{D}[\mathcal{F}_1]=(\mu_{t,1}^{\varepsilon}\mu_s^{\varepsilon}\mu_{t,4}^{\varepsilon}\mu_s^{\varepsilon}  )^2 $. The preserved flavor symmetry for this is given by the lattice of roots orthogonal to $\mathcal{F}_1$
\be \label{eq: n=1 lattice of inv roots}
    (\mathcal{F}_{1})^{\perp} = \langle \alpha_1+\alpha_3+\alpha_4+\alpha_5+\alpha_6+\alpha_7+\alpha_8,\alpha_2,\alpha_3,\alpha_4,\alpha_5,\alpha_6,\alpha_7 \rangle_{\mathbb{Z}} \, ,
\ee
which generates the Cartan matrix of $E_7$. This is exactly the enhanced flavor symmetry~\eqref{eq:preserved flavor symms} appropriate for the quiver theory constructed in Figure~\ref{fig: n3 volumes and quiver closed} for $n=1$ \\

Note that comparing the invariant root lattices $(\mathcal{F}_{4})^{\perp}$ and $(\mathcal{F}_{1})^{\perp}$ we see that despite them both being isomorphic to an $E_7$ root lattice, they are spanned by different invariant roots, and therefore correspond to different $E_7$ flavor symmetries. We comment on this isomorphism from the perspective of IR dualities next.

\subsection{IR Duality from Geometry}

The IR duality discussed in Section~\ref{sec: Field Theory sec2} can be seen geometrically by relating the monodromy maps $\mathcal{M}[\mathcal{F}_{1}]$ and 
$\mathcal{M}[2\alpha_1]$ by a Weyl reflection. This captures the identification between the fluxes $(-1^8)$ and $(-2^2,0^6)$. Take the simple root
\be
    \alpha_{+} = \Sigma+2 F - \sum_{i=3}^8 M_i \in R_+\, .
\ee
Then one can show that the two fluxes are obtained by the action of the Weyl reflection 
$w_{\alpha_+}(\mathcal{F}_{1})=2\alpha_1$, meaning that the two $E_7$ sublattices are related by this Weyl transformation. Applying this Weyl reflection to the invariant roots explicitly gives
\be
\ba
    w_{\alpha_+}(\mathcal{F}_{1}^{\perp}) = \langle &2 \alpha_1+2 \alpha_2+4\alpha_3+5\alpha_4+4 \alpha_5+3\alpha_6+2\alpha_7+\alpha_8, \alpha_2\cr
    &-\alpha_1-2\alpha_2-2\alpha_3-4\alpha_4-3 \alpha_5-2\alpha_6-\alpha_7,\alpha_4,\alpha_5,\alpha_6,\alpha_7  \rangle_{\mathbb{Z}} \, .
\ea
\ee
One can then show that by suitable linear combinations of the simple roots we can write the above as
\be
    w_{\alpha_+}(\mathcal{F}_{1}^{\perp}) =  \langle \alpha_1+\alpha_2+2\alpha_3+2 \alpha_4+\alpha_5, \alpha_2,\alpha_4,\alpha_5,\alpha_6,\alpha_7,\alpha_8 \rangle_{\mathbb{Z}} = (2\alpha_1)^{\perp}  \, .
\ee
To understand how the geometric maps are related, it can be verified that 
\be \label{eq: Seiberg duality for flux maps}
    w_{\alpha_+} \mathcal{M}[\mathcal{F}_{1}] w_{\alpha_+} = \mathcal{M}[2\alpha_1] \, ,
\ee
which can be done using~\eqref{eq:conjfluxformula_alpha1} and~\eqref{eq:fluxformula_alpha}, establishing that the two isomorphisms are in fact related by conjugation by a Weyl reflection.

One can actually understand the application of Seiberg duality at the level of local geometry. Observe from Figure~\ref{fig:Seiberg} that the local geometries of each quiver correspond to the local building blocks $\mu_s^{\varepsilon} \mu_{t,1}^{\varepsilon}\mu_s^{\varepsilon}$ and $\mu_{t,1}^{\varepsilon}\mu_s^{\varepsilon} \mu_{t,1}^{\varepsilon}$, respectively, as shown in Figure~\ref{fig: local Seiberg duality geometry}. The corresponding automorphisms are related by a Weyl reflection with respect to the $\mathfrak{so}(16)$ simple root $\beta_1=M_2-M_1$
\be \label{eq: Local geometric Seiberg}
    \mu_s \mu_{t,1} \mu_s = w_{\beta_1} \mu_{t,1} \mu_s \mu_{t,1} \, .
\ee
The action of this Weyl reflection swaps $M_1 \leftrightarrow M_2$ and leaves all other curves invariant. More generally, one can understand local Seiberg duality by the identity
\be
    \mu_s \mu_{t,\bold{n}} \mu_s = w_{\beta(\bold{n})} \mu_{t,\bold{n}} \mu_s \mu_{t,\bold{n}} \, , \qquad |\bold{n}|=2 \, ,
\ee
where we have defined $\beta(\bold{n})=M_i-M_j$ with $i,j$ labeling the two non-zero components of $\bold{n}$.

\begin{figure}
    \centering
    \begin{tikzpicture}[scale=0.6, transform shape]
        \begin{scope}[shift={(0,0)}]

    \filldraw[color=red, fill=red!20, thick] (-5.5,-1.4) ellipse (0.3 and 0.59);
    \draw[red, thin] (-5.5,-1.4) ellipse (0.3 and 0.13);

    \filldraw[color=red, fill=red!20, thick] (-4,-1.02) ellipse (0.45 and 0.98);
    \draw[red, thin] (-4,-1.02) ellipse (0.45 and 0.15);

    \filldraw[color=red, fill=red!20, thick] (-2.5,-1.4) ellipse (0.3 and 0.6);
    \draw[red, thin] (-2.5,-1.4) ellipse (0.3 and 0.13);

    \filldraw[color=red, fill=red!20, thick] (-1,-1.76) ellipse (0.1 and 0.24);
    \draw[red, thin] (-1,-1.78) ellipse (0.1 and 0.05);

    \filldraw[color=red, fill=red!20, thick] (5.5,-1.4) ellipse (0.3 and 0.59);
    \draw[red, thin] (5.5,-1.4) ellipse (0.3 and 0.13);

    \filldraw[color=red, fill=red!20, thick] (4,-1.02) ellipse (0.45 and 0.98);
    \draw[red, thin] (4,-1.02) ellipse (0.45 and 0.15);

    \filldraw[color=red, fill=red!20, thick] (2.5,-1.4) ellipse (0.3 and 0.6);
    \draw[red, thin] (2.5,-1.4) ellipse (0.3 and 0.13);

    \filldraw[color=red, fill=red!20, thick] (1,-1.76) ellipse (0.1 and 0.24);
    \draw[red, thin] (1,-1.78) ellipse (0.1 and 0.05);


    \filldraw[color=orange, fill=orange!20, thick] (-5.5,1.92) ellipse (0.42 and 1.08);
    \draw[orange, thin] (-5.5,1.92) ellipse (0.42 and 0.18);

    \filldraw[color=orange, fill=orange!20, thick] (-4,1.53) ellipse (0.52 and 1.47);
    \draw[orange, thin] (-4,1.72) ellipse (0.52 and 0.22);

    \filldraw[color=orange, fill=orange!20, thick] (-2.5,1.9) ellipse (0.42 and 1.1);
    \draw[orange, thin] (-2.5,1.9) ellipse (0.42 and 0.18);

    \filldraw[color=orange, fill=orange!20, thick] (-1,2.26) ellipse (0.3 and 0.74);
    \draw[orange, thin] (-1,2.26) ellipse (0.3 and 0.14);

    \filldraw[color=orange, fill=orange!20, thick] (5.5,1.92) ellipse (0.42 and 1.08);
    \draw[orange, thin] (5.5,1.92) ellipse (0.42 and 0.18);

    \filldraw[color=orange, fill=orange!20, thick] (4,1.53) ellipse (0.52 and 1.47);
    \draw[orange, thin] (4,1.72) ellipse (0.52 and 0.22);

    \filldraw[color=orange, fill=orange!20, thick] (2.5,1.9) ellipse (0.42 and 1.1);
    \draw[orange, thin] (2.5,1.9) ellipse (0.42 and 0.18);

    \filldraw[color=orange, fill=orange!20, thick] (1,2.26) ellipse (0.3 and 0.74);
    \draw[orange, thin] (1,2.26) ellipse (0.3 and 0.14);

    \filldraw[color=blue, fill=blue!20, thick] (-5.5,0) ellipse (0.46 and 0.74);
    \draw[blue, thin] (-5.5,0) ellipse (0.46 and 0.16);

    \filldraw[color=blue, fill=blue!20, thick] (-2.5,0) ellipse (0.46 and 0.71);
    \draw[blue, thin] (-2.5,0) ellipse (0.46 and 0.16);

    \filldraw[color=blue, fill=blue!20, thick] (-1,0) ellipse (0.62 and 1.46);
    \draw[blue, thin] (-1,0) ellipse (0.62 and 0.22);

    \filldraw[color=blue, fill=blue!20, thick] (1,0) ellipse (0.62 and 1.46);
    \draw[blue, thin] (1,0) ellipse (0.62 and 0.20);

    \filldraw[color=blue, fill=blue!20, thick] (2.5,0) ellipse (0.46 and 0.71);
    \draw[blue, thin] (2.5,0) ellipse (0.46 and 0.15);

    \filldraw[color=blue, fill=blue!20, thick] (5.5,0) ellipse (0.46 and 0.74);
    \draw[blue, thin] (5.5,0) ellipse (0.46 and 0.16);


    \draw[black, thick, ->] (-6.25,0) -- (6.25,0);
    \draw[mygreen, ultra thick] (-5.85,0) -- (5.85,0);

    \draw[red, thick] (-5.85,-2) -- (5.85,-2);
    
    \draw[orange, thick] (-5.85,3) -- (5.85,3);


    \draw[blue, thick] (-5.85,0.975) -- (-4.00,0);
    \draw[blue, thick] (-5.85,-0.975) -- (-4.00,0);

    \draw[blue, thick] (-4.00,0) -- (0.00,2.00);
    \draw[blue, thick] (0.00,2.00) -- (4.00,0);

    \draw[blue, thick] (-4.00,0) -- (0.00,-2.00);
    \draw[blue, thick] (0.00,-2.00) -- (4.00,0);

    \draw[blue, thick] (4.00,0) -- (5.85,0.975);
    \draw[blue, thick] (4.00,0) -- (5.85,-0.975);


    \filldraw[black] (-4.00,0) circle (0pt) node{\huge{$\times$}};
    \filldraw[black] (0.00,0) circle (0pt) node{\huge{$\times$}};
    \filldraw[black] (4.00,0) circle (0pt) node{\huge{$\times$}};


    \filldraw[black] (0,0.2) circle (0pt) node[anchor=south]{\huge{$\mu_{t,1}$}};

    \filldraw[black] (-4,0.2) circle (0pt) node[anchor=south]{\huge{$\mu_{s}$}};

    \filldraw[black] (4,0.2) circle (0pt) node[anchor=south]{\huge{$\mu_{s}$}};

\end{scope}


        \begin{scope}[shift={(14,0)}]


    \filldraw[color=red, fill=red!20, thick] (-5.5,-1.64) ellipse (0.2 and 0.36);
    \draw[red, thin] (-5.5,-1.64) ellipse (0.2 and 0.08);

    \filldraw[color=red, fill=red!20, thick] (-2.5,-1.64) ellipse (0.2 and 0.36);
    \draw[red, thin] (-2.5,-1.64) ellipse (0.2 and 0.08);

    \filldraw[color=red, fill=red!20, thick] (-1,-1.27) ellipse (0.28 and 0.73);
    \draw[red, thin] (-1,-1.27) ellipse (0.28 and 0.11);

    \filldraw[color=red, fill=red!20, thick] (5.5,-1.64) ellipse (0.2 and 0.36);
    \draw[red, thin] (5.5,-1.64) ellipse (0.2 and 0.08);

    \filldraw[color=red, fill=red!20, thick] (2.5,-1.64) ellipse (0.2 and 0.36);
    \draw[red, thin] (2.5,-1.64) ellipse (0.2 and 0.08);

    \filldraw[color=red, fill=red!20, thick] (1,-1.27) ellipse (0.28 and 0.73);
    \draw[red, thin] (1,-1.27) ellipse (0.28 and 0.11);

     \filldraw[color=orange, fill=orange!20, thick] (-5.5,2.13) ellipse (0.3 and 0.87);
    \draw[orange, thin] (-5.5,2.13) ellipse (0.3 and 0.14);

    \filldraw[color=orange, fill=orange!20, thick] (-4,2.5) ellipse (0.22 and 0.5);
    \draw[orange, thin] (-4,2.5) ellipse (0.22 and 0.1);

     \filldraw[color=orange, fill=orange!20, thick] (-2.5,2.13) ellipse (0.3 and 0.87);
    \draw[orange, thin] (-2.5,2.13) ellipse (0.3 and 0.14);

    \filldraw[color=orange, fill=orange!20, thick] (-1,1.77) ellipse (0.5 and 1.23);
    \draw[orange, thin] (-1,1.77) ellipse (0.5 and 0.18);

     \filldraw[color=orange, fill=orange!20, thick] (5.5,2.13) ellipse (0.3 and 0.87);
    \draw[orange, thin] (5.5,2.13) ellipse (0.3 and 0.14);

    \filldraw[color=orange, fill=orange!20, thick] (4,2.5) ellipse (0.22 and 0.5);
    \draw[orange, thin] (4,2.5) ellipse (0.22 and 0.1);

     \filldraw[color=orange, fill=orange!20, thick] (2.5,2.13) ellipse (0.3 and 0.87);
    \draw[orange, thin] (2.5,2.13) ellipse (0.3 and 0.14);

    \filldraw[color=orange, fill=orange!20, thick] (1,1.77) ellipse (0.5 and 1.23);
    \draw[orange, thin] (1,1.77) ellipse (0.5 and 0.18);


    \filldraw[color=blue, fill=blue!20, thick] (-5.5,0) ellipse (0.46 and 1.23);
    \draw[blue, thin] (-5.5,0) ellipse (0.46 and 0.14);

    \filldraw[color=blue, fill=blue!20, thick] (-2.5,0) ellipse (0.46 and 1.23);
    \draw[blue, thin] (-2.5,0) ellipse (0.46 and 0.16);
    
    \filldraw[color=blue, fill=blue!20, thick] (-1,0) ellipse (0.3 and 0.48);
    \draw[blue, thin] (-1,0) ellipse (0.3 and 0.1);

    \filldraw[color=blue, fill=blue!20, thick] (-4,0) ellipse (0.7 and 1.95);
    \draw[blue, thin] (-4,0) ellipse (0.7 and 0.22);

    \filldraw[color=blue, fill=blue!20, thick] (5.5,0) ellipse (0.46 and 1.23);
    \draw[blue, thin] (5.5,0) ellipse (0.46 and 0.14);

    \filldraw[color=blue, fill=blue!20, thick] (2.5,0) ellipse (0.46 and 1.23);
    \draw[blue, thin] (2.5,0) ellipse (0.46 and 0.16);
    
    \filldraw[color=blue, fill=blue!20, thick] (1,0) ellipse (0.3 and 0.48);
    \draw[blue, thin] (1,0) ellipse (0.3 and 0.1);

    \filldraw[color=blue, fill=blue!20, thick] (4,0) ellipse (0.7 and 1.95);
    \draw[blue, thin] (4,0) ellipse (0.7 and 0.22);


    \draw[black, thick, ->] (-6.25,0) -- (6.25,0);
    \draw[mygreen, ultra thick] (-5.85,0) -- (5.85,0);

    \draw[red, thick] (-5.85,-2) -- (5.85,-2);
    
    \draw[orange, thick] (-5.85,3) -- (5.85,3);


       \draw[blue, thick] (-5.85,1.075) -- (-4.00,2);
    \draw[blue, thick] (-5.85,-1.075) -- (-4.00,-2);

    \draw[blue, thick] (-4.00,2) -- (0.00,0);
    \draw[blue, thick] (0.00,0) -- (4.00,2);

    \draw[blue, thick] (-4.00,-2) -- (0.00,0);
    \draw[blue, thick] (0.00,0) -- (4.00,-2);

    \draw[blue, thick] (4.00,2) -- (5.85,1.075);
    \draw[blue, thick] (4.00,-2) -- (5.85,-1.075);


    \filldraw[black] (-4.00,0) circle (0pt) node{\huge{$\times$}};
    \filldraw[black] (0.00,0) circle (0pt) node{\huge{$\times$}};
    \filldraw[black] (4.00,0) circle (0pt) node{\huge{$\times$}};


    \filldraw[black] (0,0.2) circle (0pt) node[anchor=south]{\huge{$\mu_{s}$}};

    \filldraw[black] (-4,0.2) circle (0pt) node[anchor=south]{\huge{$\mu_{t,1}$}};

    \filldraw[black] (4,0.2) circle (0pt) node[anchor=south]{\huge{$\mu_{t,1}$}};
\end{scope}


\begin{scope}[shift={(0,-6)}]

    \tikzstyle{every node}=[font=\scriptsize]
    \draw[thick] (-0.5,-1.5) rectangle (0.5,-0.5);
    \filldraw[black] (0,-1) circle (0pt) node{\huge{$2$}};
    \filldraw[black] (0,1.5) circle (0pt) node{\huge{$2$}};
    \draw[thick](0,1.5) circle (0.5);
    \draw[thick] (0,1) -- (0,-0.5) ;
       
    \draw[thick]  (-6,1) arc (-90:90:0.5);
    \draw[thick]  (6,1) arc (270:90:0.5);

    \draw[thick] (-5.5,1.5) -- (-0.5,1.5);
    \draw[thick] (0.5,1.5) -- (5.5,1.5);
    \node at (-4,1.5) {\huge $\times$};
    \node at (4,1.5) {\huge $\times$};

\node at (-6,1.5) {\huge $2$};
\node at (6,1.5) {\huge $2$};

\node[anchor=west] at (6,2) {\Large $\cdots$};
\node[anchor=west] at (6,1) {\Large $\cdots$};

\node[anchor=east] at (-6,2) {\Large $\cdots$};
\node[anchor=east] at (-6,1) {\Large $\cdots$};
\end{scope}
\begin{scope}[shift={(14,-6)}]

    \tikzstyle{every node}=[font=\scriptsize]
    \draw[thick] (-0.5,-1.5) rectangle (0.5,-0.5);
    \filldraw[black] (0,-1) circle (0pt) node{\huge{$2$}};
       
    \draw[thick]  (-4,1) arc (-90:90:0.5);
    \draw[thick]  (4,1) arc (270:90:0.5);

    \draw[thick] (-3.5,1.5) -- (3.5,1.5);
    \node at (0,1.5) {\huge $\times$};
    \draw[thick] (-3.6,1.22) -- (-0.5,-1);
    \draw[thick] (3.6,1.22) -- (0.5,-1);

\node at (-4,1.5) {\huge $2$};
\node at (4,1.5) {\huge $2$};

\node[anchor=west] at (4,2) {\Large $\cdots$};
\node[anchor=west] at (4,1) {\Large $\cdots$};

\node[anchor=east] at (-4,2) {\Large $\cdots$};
\node[anchor=east] at (-4,1) {\Large $\cdots$};

\end{scope}

    \end{tikzpicture}
    \caption{The local geometries interchanged by application of Seiberg duality. }
    \label{fig: local Seiberg duality geometry}
\end{figure}
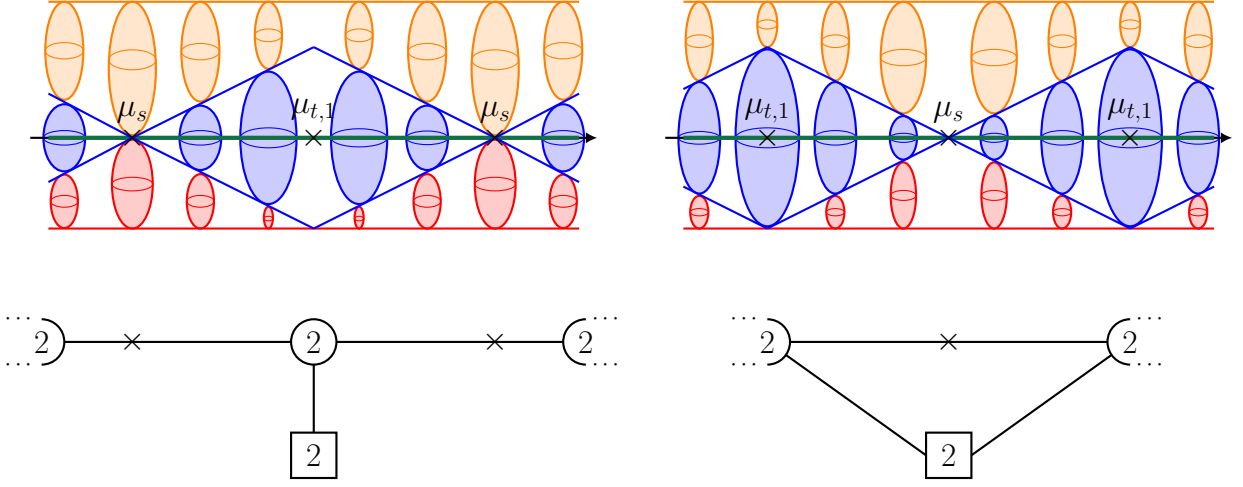

We can spell out the flux $(-1^8)\leftrightarrow (-2^2,0^6)$ example using the local application of Seiberg duality to show that the field theoretic computation in Figure~\ref{fig:Seiberg n=14} agrees with the geometric calculation. We will do this at the level of the unfolded quiver. Observe that we have the following two applications of Seiberg duality on the unfolded quiver shown in Figure~\ref{fig:2 Seibergs}.
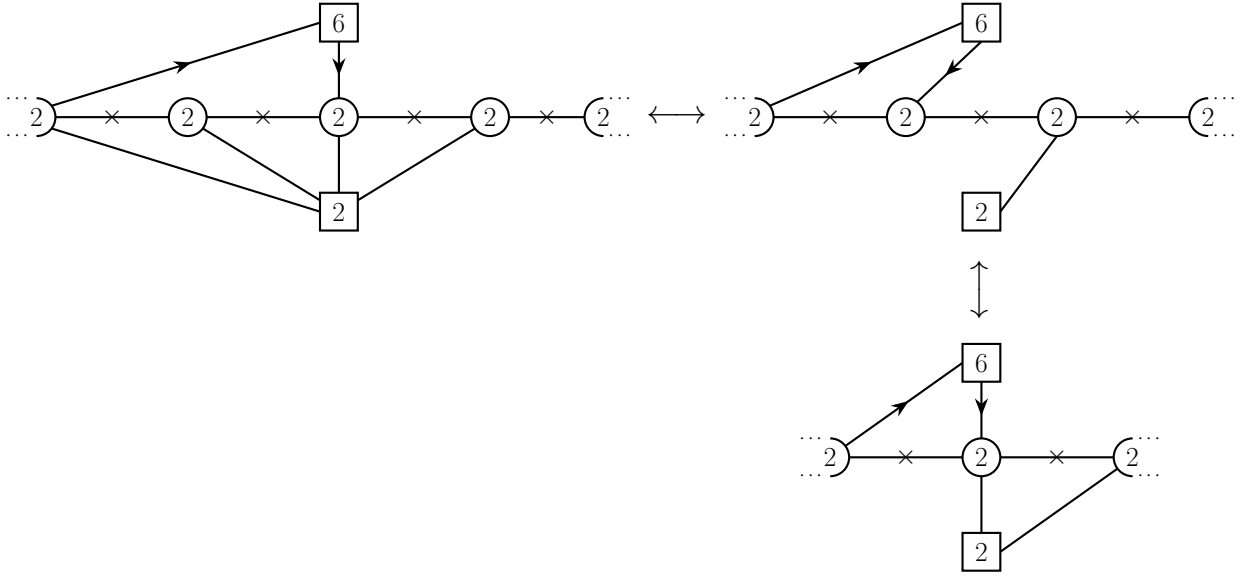
\begin{figure}
    \centering
\begin{tikzpicture}[scale=0.5, transform shape]
    
\begin{scope}[shift={(0,0)}]
   \tikzstyle{every node}=[font=\scriptsize]
    \draw[thick] (2.5,-1.5) rectangle (1.5,-0.5);
    \filldraw[black] (2,-1) circle (0pt) node{\huge{$2$}};
    \draw[thick] (2.5,3.5) rectangle (1.5,4.5);
    \filldraw[black] (2,4) circle (0pt) node{\huge{$6$}};
    \filldraw[black] (9,1.5) circle (0pt) node{\huge{$2$}};
    \filldraw[black] (6,1.5) circle (0pt) node{\huge{$2$}};
    \filldraw[black] (2,1.5) circle (0pt) node{\huge{$2$}};
    \filldraw[black] (-2,1.5) circle (0pt) node{\huge{$2$}};
    \filldraw[black] (-6,1.5) circle (0pt) node{\huge{$2$}};
    \draw[thick](6,1.5) circle (0.5);
    \draw[thick](2,1.5) circle (0.5);
    \draw[thick](-2,1.5) circle (0.5);
    \draw[thick] (5.6,1.2) -- (2.5,-0.7) ;
    \draw[thick] (2,1) -- (2,-0.5) ;
    \draw[thick] (-1.6,1.2) -- (1.5,-0.7) ;
    \draw[thick] (-5.6,1.2) -- (1.5,-1) ;
    \draw[thick, -<-] (2,2) -- (2,3.5) ;
    \draw[thick, -<-] (1.5,4) -- (-5.6,1.8) ;

    \draw[thick]  (9,1) arc (270:90:0.5);
    \draw[thick]  (-6,1) arc (-90:90:0.5);

    \draw[thick] (8.5,1.5) -- (6.5,1.5);
    \draw[thick] (5.5,1.5) -- (2.5,1.5);
    \draw[thick] (1.5,1.5) -- (-1.5,1.5);
    \draw[thick] (-2.5,1.5) -- (-5.5,1.5);
    \node at (7.5,1.5) {\huge $\times$};
    \node at (4,1.5) {\huge $\times$};
    \node at (0,1.5) {\huge $\times$};
    \node at (-4,1.5) {\huge $\times$};


\node[anchor=east] at (-6,2) {\Large $\cdots$};
\node[anchor=east] at (-6,1) {\Large $\cdots$};

\node[anchor=west] at (9,2) {\Large $\cdots$};
\node[anchor=west] at (9,1) {\Large $\cdots$};

\node[anchor=west] at (10,1.5) {\Huge $\longleftrightarrow$};
    
\end{scope}


\begin{scope}[shift={(19,0)}]

  \tikzstyle{every node}=[font=\scriptsize]
    \draw[thick] (0.5,-1.5) rectangle (-0.5,-0.5);
    \filldraw[black] (0,-1) circle (0pt) node{\huge{$2$}};
    \draw[thick] (0.5,3.5) rectangle (-0.5,4.5);
    \filldraw[black] (0,4) circle (0pt) node{\huge{$6$}};
    \filldraw[black] (6,1.5) circle (0pt) node{\huge{$2$}};
    \filldraw[black] (2,1.5) circle (0pt) node{\huge{$2$}};
    \filldraw[black] (-2,1.5) circle (0pt) node{\huge{$2$}};
    \filldraw[black] (-6,1.5) circle (0pt) node{\huge{$2$}};
    \draw[thick](2,1.5) circle (0.5);
    \draw[thick](-2,1.5) circle (0.5);
    \draw[thick] (2,1) -- (0.5,-1) ;
    \draw[thick, -<-] (-1.7,1.9) -- (0,3.5) ;
    \draw[thick, -<-] (-0.5,4) -- (-5.6,1.8) ;

    \draw[thick]  (6,1) arc (270:90:0.5);
    \draw[thick]  (-6,1) arc (-90:90:0.5);

    \draw[thick] (5.5,1.5) -- (2.5,1.5);
    \draw[thick] (1.5,1.5) -- (-1.5,1.5);
    \draw[thick] (-2.5,1.5) -- (-5.5,1.5);
    \node at (4,1.5) {\huge $\times$};
    \node at (0,1.5) {\huge $\times$};
    \node at (-4,1.5) {\huge $\times$};


\node[anchor=east] at (-6,2) {\Large $\cdots$};
\node[anchor=east] at (-6,1) {\Large $\cdots$};

\node[anchor=west] at (6,2) {\Large $\cdots$};
\node[anchor=west] at (6,1) {\Large $\cdots$};

\node[anchor=west,rotate=90] at (0,-4) {\Huge $\longleftrightarrow$};

\end{scope}


\begin{scope}[shift={(19,-9)}]
  \tikzstyle{every node}=[font=\scriptsize]
    \draw[thick] (0.5,-1.5) rectangle (-0.5,-0.5);
    \filldraw[black] (0,-1) circle (0pt) node{\huge{$2$}};
    \draw[thick] (0.5,3.5) rectangle (-0.5,4.5);
    \filldraw[black] (0,4) circle (0pt) node{\huge{$6$}};
    \filldraw[black] (-4,1.5) circle (0pt) node{\huge{$2$}};
    \filldraw[black] (0,1.5) circle (0pt) node{\huge{$2$}};
    \filldraw[black] (4,1.5) circle (0pt) node{\huge{$2$}};
    \draw[thick](0,1.5) circle (0.5);
    \draw[thick] (3.6,1.2) -- (0.5,-1) ;
    \draw[thick, -<-] (0,2) -- (0,3.5) ;
    \draw[thick] (0,1) -- (0,-0.5) ;
    \draw[thick, -<-] (-0.5,4) -- (-3.6,1.8) ;

    \draw[thick]  (4,1) arc (270:90:0.5);
    \draw[thick]  (-4,1) arc (-90:90:0.5);

    \draw[thick] (3.5,1.5) -- (0.5,1.5);
    \draw[thick] (-0.5,1.5) -- (-3.5,1.5);
    \node at (-2,1.5) {\huge $\times$};
    \node at (2,1.5) {\huge $\times$};


\node[anchor=east] at (-4,2) {\Large $\cdots$};
\node[anchor=east] at (-4,1) {\Large $\cdots$};

\node[anchor=west] at (4,2) {\Large $\cdots$};
\node[anchor=west] at (4,1) {\Large $\cdots$};

\end{scope}
\end{tikzpicture}
    \caption{Two applications of Seiberg duality on the unfolded quiver.}
    \label{fig:2 Seibergs}
\end{figure}
The first and final quivers correspond to the maps
\be
    \mu_{s}\mu_{t,(1^2,0^6)} \mu_s \mu_{t,(1^8)} \mu_s \mu_{t,(1^2,0^6)} \mu_s \mu_{t,(1^8)}
     \, , \qquad  \mu_{t,(1^2,0^6)} \mu_s \mu_{t,(1^8)}  \mu_s\mu_{t,(0^2,1^6)} \, ,
\ee
respectively. Let us apply Seiberg duality twice at the level of the automorphisms
\be
\ba
    \mu_{s}\mu_{t,(1^2,0^6)} \mu_s \mu_{t,(1^8)} \mu_s \mu_{t,(1^2,0^6)} \mu_s \mu_{t,(1^8)} &= w_{\beta_1} \mu_{t,(1^2,0^6)} \mu_s \mu_{t,(1^2,0^6)} \mu_{t,(1^8)}\mu_s \mu_{t,(1^2,0^6)} \mu_s \mu_{t,(1^8)} \cr
    &= w_{\beta_1} \mu_{t,(1^2,0^6)} \mu_s \mu_{t,(1^2,0^6)} \mu_{t,(1^8)} w_{\beta_1} \mu_{t,(1^2,0^6)} \mu_s \mu_{t,(1^2,0^6)} \mu_{t,(1^8)} \cr
    &= \mu_{t,(1^2,0^6)} \mu_s \mu_{t,(1^8)} \mu_s \mu_{t,(0^2,1^6)} \, ,
\ea
\ee
where the first two equalities were applications of~\eqref{eq: Local geometric Seiberg}, and the final equality used the properties of $\mu_{t,\bold{n}}$ maps in Appendix~\ref{sec: App Gluing Profiles}, and that the Weyl transformation $w_{\beta_1}$ commutes with all automorphisms in the expression. Thus we can recover the geometric Seiberg duality of two flux tori quivers by local applications of the geometric Seiberg duality.

Thus far, we have shown that Seiberg duality relates automorphisms by Weyl transformations using a root $\beta_1 \in \Lambda_{D_8} \setminus \Lambda_{A_7}$. As we have seen, these are among the automorphisms of $X_S^f$, so that we can think of the associated varieties $\mathcal{T}[(-1^8)]$ and $\mathcal{T}[(-2^2,0^6)]$
as being isomorphic.

\section{Discussion and Outlook}

\paragraph{Summary of the Construction.} \label{sec: summary of construction}
We conclude our discussion of the geometric construction of flux tori by providing a recipe for constructing a topological $G_2$-holonomy manifold and associated 4d $\cN=1$ quiver starting solely from a non-zero field theory flux $f \in E_8$ with $f/2 \in E_8$:  
\begin{enumerate}
    \item Starting from $f\in E_8$, construct the unique geometric flux $\mathcal{F}\in \Lambda_{E_8} \subset H_2(X_S,\mathbb{Z})$ satisfying $N(\mathcal{F})=f$. 
    \item Determine a possible decomposition of $\mathcal{F}=\sum_{i=1}^n 2 \alpha_i$ with $\alpha_i \in R$. 
    \item Glue together the building blocks 
    $\mathcal{D}[2\alpha_i]$
    to construct the $\mathcal{D}[\mathcal{F}]$ concatenation of domain walls, following the logic of Section~\ref{sec: flux 2^2n}. This provides a non-trivial fibration of the threefold $X_S^f$ over an interval. 
    \item Glue the threefold at the ends of the line interval together by applying the automorphism associated to the monodromy map $\mathcal{M}[\mathcal{F}]$ to find the 7-dimensional space  $\mathcal{T}[\mathcal{F}]$. This is realized as a fibration of the threefold $X_S^f$ over the circle $S^1$. Read off the quiver from the corresponding glued volume profile.
    \item Determine the preserved flavor symmetry inside $E_8$ by calculating the set of roots of $\mathcal{F}^{\perp} \subset \Lambda_{E_8}$ that are perpendicular to the geometric flux $\mathcal{F}$.  
\end{enumerate}

\paragraph{Path in the Moduli Space.}
The geometry of the moduli space of the threefold $X_S$ is sketched in Figure~\ref{fig: n4 moduli space path}, with the path we take corresponding to the flux torus $\mathcal{T}[4\alpha_1]$ drawn in black.
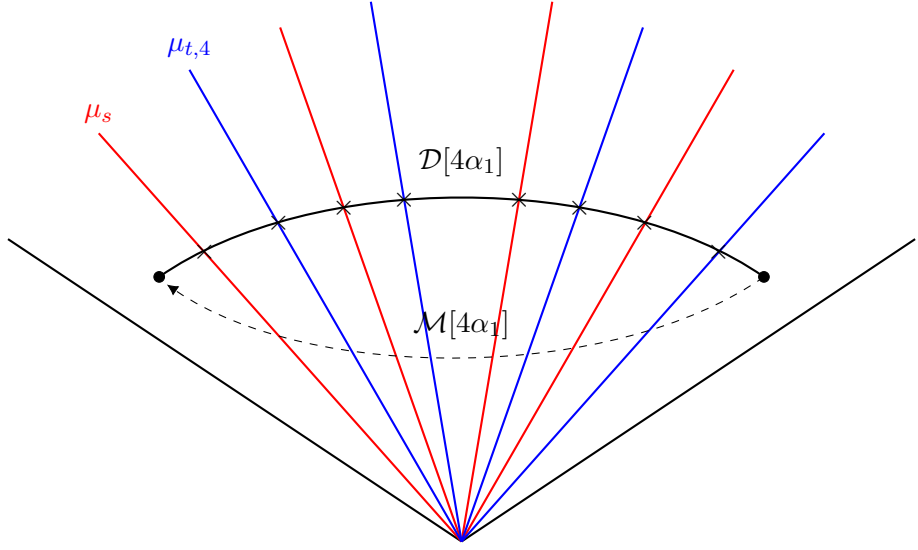
\begin{figure}
\centering
\begin{tikzpicture}
    \draw[black, thick] (0,0) -- (6,4) ;
    \draw[black, thick] (0,0) -- (-6,4) ;
    \draw[red, thick] (0,0) -- (-4.8,5.4);
    \draw[red, thick] (0,0) -- (-2.4,6.8);
    \draw[red, thick] (0,0) -- (1.2,7.14);
    \draw[red, thick] (0,0) -- (3.6,6.24);
    \draw[blue, thick] (0,0) -- (4.8,5.4);
    \draw[blue, thick] (0,0) -- (2.4,6.8);
    \draw[blue, thick] (0,0) -- (-1.2,7.14);
    \draw[blue, thick] (0,0) -- (-3.6,6.24);
    \filldraw[black] (-4.8,5.4) circle (0pt) node[anchor=south]{\textcolor{red}{$\mu_s$}};
    \filldraw[black] (-3.6,6.24) circle (0pt) node[anchor=south]{\textcolor{blue}{$\mu_{t,4}$}};
    \filldraw[black] (-4,3.5) circle (2pt) node[anchor=south]{};
    \filldraw[black] (4,3.5) circle (2pt) node[anchor=south]{};
    \draw[black, thick] (-4,3.5) .. controls (-2,4.9) and (2,4.9) .. (4,3.5);
    \draw[black, dashed, ->] (4,3.5) .. controls (2,2.1) and (-2,2.1) .. (-3.9,3.4);
    \filldraw[black] (-3.4,4.1) circle (0pt) node[anchor=north]{$\times$};
    \filldraw[black] (-2.42,4.48) circle (0pt) node[anchor=north]{$\times$};
    \filldraw[black] (-1.56,4.68) circle (0pt) node[anchor=north]{$\times$};
    \filldraw[black] (-0.76,4.78) circle (0pt) node[anchor=north]{$\times$};
    \filldraw[black] (3.4,4.1) circle (0pt) node[anchor=north]{$\times$};
    \filldraw[black] (2.42,4.48) circle (0pt) node[anchor=north]{$\times$};
    \filldraw[black] (1.56,4.68) circle (0pt) node[anchor=north]{$\times$};
    \filldraw[black] (0.76,4.78) circle (0pt) node[anchor=north]{$\times$};
    \filldraw[black] (0,4.7) circle (0pt) node[anchor=south]{$\mathcal{D}[4\alpha_1]$};
    \filldraw[black] (0,3.2) circle (0pt) node[anchor=north]{$\mathcal{M}[4\alpha_1]$};
\end{tikzpicture}
\caption{The path $\mathcal{T}[4\alpha_1]$ we take through the K\"ahler cone of the threefold $X_S$ to construct the $(-2^8)$ flux domain wall. We cross red and blue walls corresponding to the location of the $\mu_s$ and $\mu_{t,4}$ domain walls, respectively. The path $\mathcal{D}[4\alpha_1]$ starts at the first black dot on the left, crosses walls in a continuous way, then closes back on itself using the monodromy map $\mathcal{M}[4\alpha_1]$. }
\label{fig: n4 moduli space path}
\end{figure}

The moduli space contains many walls, corresponding to the special locations in which we apply $\mu_s$ and $\mu_{t,\mathbf{n}}$ automorphisms. Crossing the same domain wall twice in succession is a contractible path in this space, since each automorphism squares to the identity. 

 After traversing a path in the moduli space, we glue back to our starting point using a monodromy map, which is in particular an element of $GL(10,\mathbb{Z})$. Such a path through the moduli space is non-contractible provided one uses a non-trivial monodromy map to glue closed the profile. Since the monodromy map must be an element of $GL(10,\mathbb{Z})$, if the monodromy map for a given path is not the identity, then it cannot be continuously deformed to the identity due to the integrality of the map. We conclude that all of the continuous deformations of a given path through the moduli space must have the same monodromy map.

\paragraph{Future questions.}
 We did not discuss M2-instanton corrections, which  become important as we move to the singular point in the geometry that  corresponds to the SCFT. It would be interesting to get a quantitative handle on these corrections and explore whether they invalidate the IR description as an SCFT. 

We already commented on the issue of the cubic superpotential that is relevant in these quivers, see also part I
 \cite{Braun:2023fqa}. This should be generated by M2-instantons and it would be important to get a better understanding of these contributions, as they would substantiate the enhanced flavor symmetries, that we expect from the field theory construction.

We have merely proposed a topological $G_2$ manifold construction. There are various assumptions, which would be very good to substantiate beyond the reliance on the field theory. For example, we constructed the automorphisms in the particular choice of complex structure where there is only an $I_8$ fiber, assuming that these would persist  when we tune back to the singular geometry with $II^*$ fiber. It would be important to develop a theory of the automorphisms of partially resolved elliptic fibrations. The physics seems to predict that these retain the automorphisms of the $I_8$ model.  
Moreover, the most  challenging open problem is to show that the topological $G_2$-manifolds we have constructed are actually (singular degenerations of) $G_2$-manifolds. To do so, one needs to construct a torsion-free $G_2$-holonomy metric on these spaces.

 There are obvious generalizations of the construction of topologically $G_2$-manifolds: the first is to generalize this to arbitrary flux Riemann surface compactifications. 
The novelty here is understanding the 7-dimensional geometry associated to the pair of pants decomposition. Once this is understood, since we know how the tubes work, one should be able to $G_2$-geometrize all flux Riemann surface compactifications of 6d $\mathcal{N}=(1,0)$ theories.

\subsection*{Acknowledgments}
We thank Evyatar Sabag for collaboration on Part I and initial collaboration. We thank Bobby Acharya for discussions. 
OL is supported by an STFC studentship. 
The work of SSN  is  supported  in part by an EPSRC Open Fellowship EP/X01276X/1 (Schafer-Nameki) and the  STFC consolidated grant ST/X000761/1.

\appendix

\section{KK-Shifts from Field Theory}
\label{app:KKshift}

As we have seen in Section \ref{sec: Geometric Realization}, when constructing a torus model, the volumes of the curves associated with the $E_8$ roots should be identified up to a shift by an integer multiple of $\mathrm{Vol}(E)=m_{\text{KK}}$. This is consistent with the field theory observations of \cite{Sabag:2022hyw}, where it was argued the the KK symmetry of the tube models is broken to some finite subgroup when forming a torus, since the mass parameters of the 5d hypers gets shifted by multiples of $m_{\text{KK}}$ as we go around the cirle. In this appendix we explain more in details how the geometric and the field theoretic pictures are compatible.

Let $\mathcal{F}_n$ denote the geometric flux used for the $q=1$ torus model,
so that
\be\label{eq:app-flux-weight}
  N(\mathcal{F}_n)
  = \bigl(\underbrace{-2,\ldots,-2}_{2n},
  \underbrace{0,\ldots,0}_{8-2n}\bigr)\,,
  \qquad n=1,\ldots,4\,.
\ee
For a root $\gamma\in\Lambda_{E_8}$, write
$N(\gamma)=(\lambda_1,\ldots,\lambda_8)$.  The pairing convention following
\eqref{eq: N homomorphism definition} gives
\be\label{eq:app-flux-pairing}
  \frac{1}{2}\mathcal{F}_n\mathbin{\cdot}\gamma
  = \sum_{i=1}^{2n}\lambda_i\,.
\ee
Consequently, the general monodromy formula \eqref{eq:fluxformula_general}
reduces on the $E_8$ root lattice to
\be\label{eq:app-root-monodromy}
  \mathcal{M}[q\mathcal{F}_n](\gamma)
  = \gamma+q\left(\sum_{i=1}^{2n}\lambda_i\right)E\,.
\ee

We now express the same statement in fugacity variables.  Associate to a
weight $\lambda=(\lambda_1,\ldots,\lambda_8)$ the monomial
\be\label{eq:app-weight-fugacity}
  \chi_{\lambda}(\boldsymbol{v})
  :=\prod_{i=1}^{8}v_i^{\lambda_i}\,.
\ee
The simple roots in the convention of \eqref{eq:E8roots_so16weights} therefore
correspond to
\be\label{eq:rootsfug}
\ba
  \alpha_1&:\quad \prod_{i=1}^{8}v_i^{-1/2}\,,
  &\alpha_2&:\quad v_7v_8^{-1}\,,
  &\alpha_3&:\quad v_7v_8\,,\\
  \alpha_4&:\quad v_6v_7^{-1}\,,
  &\alpha_5&:\quad v_5v_6^{-1}\,,
  &\alpha_6&:\quad v_4v_5^{-1}\,,\\
  \alpha_7&:\quad v_3v_4^{-1}\,,
  &\alpha_8&:\quad v_2v_3^{-1}\,.
\ea
\ee
In particular, the half-powers in the monomial for $\alpha_1$ are required by
its spinor weight; replacing them by integer powers changes the normalization
by a factor of two.

The field-theory analysis of \cite{Sabag:2022hyw} is naturally written in a
different basis for the two abelian symmetries of the tube.  In that basis,
which we denote temporarily by $u_i$ and $f$, the rescaling for
the flux $q\mathcal{F}_n$ takes the form
\be\label{eq:app-old-fugacity-shift}
  u_i\longmapsto
  \begin{cases}
    u_i f^{q(2n-4)}\,, & i=1,\ldots,2n\,,\\
    u_i f^{2nq}\,, & i=2n+1,\ldots,8\,.
  \end{cases}
\ee
To compare directly with the main text, we reverse the orientation of the KK
generator and absorb the resulting common factor into the diagonal abelian
flavor fugacity
\be
 f=f^{-1}\,,\qquad v_i=u_i\prod_{j=1}^8u_j^{-\frac{1}{4}}\,.
\ee
The rescaling of the new fugacities $v_i$ is then
\be\label{eq:app-adapted-fugacity-shift}
  v_i\longmapsto
  \begin{cases}
    v_i f^{4q}\,, & i=1,\ldots,2n\,,\\
    v_i\,, & i=2n+1,\ldots,8\,.
  \end{cases}
\ee
It follows immediately that a root monomial transforms as
\be\label{eq:app-general-root-fugacity-shift}
  \chi_{N(\gamma)}(\boldsymbol{v})
  \longmapsto
  f^{4q\sum_{i=1}^{2n}\lambda_i}
  \chi_{N(\gamma)}(\boldsymbol{v})\,.
\ee
Comparing \eqref{eq:app-general-root-fugacity-shift} with
\eqref{eq:app-root-monodromy}, the field-theory and geometric shifts agree with
the convention
\be\label{eq:idEf}
  E\longleftrightarrow f^4\,.
\ee
For completeness at $q=1$ the only shifted simple roots are
\be\label{eq:app-simple-root-shifts}
\ba
  n=1:&\qquad \alpha_1\longmapsto\alpha_1-E
  \qquad \alpha_8\longmapsto\alpha_8+E\\
  n=2:&\qquad \alpha_1\longmapsto\alpha_1-2E
  \qquad \alpha_6\longmapsto\alpha_6+E\\
  n=3:&\qquad \alpha_1\longmapsto\alpha_1-3E
  \qquad \alpha_4\longmapsto\alpha_4+E\\
  n=4:&\qquad \alpha_1\longmapsto\alpha_1-4E
  \qquad \alpha_3\longmapsto\alpha_3+2E\,,
\ea
\ee
where the remaining roots are invariant.  The last two
lines reproduce \eqref{eq: n=3 root shifts} and \eqref{eq:Eshiftn4},
respectively, while the first two follow from the same general monodromy
formula.  For higher $q$, every coefficient of $E$ in
\eqref{eq:app-simple-root-shifts} is multiplied by $q$.

\section{The Mordell-Weil Groups of $S_0$ and $S$}\label{app:sections}

In this Appendix we show how to compute the classes of sections on the elliptic fibrations on $S$ and $S_0$ using \cite{schuett2010ellipticsurfaces} (see Section 11). 

\subsection{Determining Classes of Sections}

For an elliptic fibration, the classes of the zero section $\sigma_0$, elliptic fiber $E$ and irreducible fiber components $E_i$ not meeting the zero section define the trivial lattice
\begin{equation}
    T := \langle \sigma_0,E,E_i \rangle_\Z
\end{equation}
and its orthogonal complement $L := T^\perp$ is called the essential lattice. There is a group homomorphism $\varphi$ from the Mordell-Weil group to $L\otimes \mathbb{Q}$ given by
\begin{equation}\label{eq:sections schutts hioda}
    \varphi(s) = \sigma_s - \sigma_0 - \left(\sigma_s \cdot \sigma_0 +1 \right) E
    - \sum_{ij} E_i\, A_{ij} \, (\sigma_s \cdot E_j)
\end{equation}
where $s$ denotes the section as an abstract element of the Mordell-Weil group, $\sigma_s$ is its class, and $A$ is the inverse of the intersection product between the $E_i$. Together with information about which fiber components a given section meets, this allows to find the class $\sigma_s$ of a given element $s$ of the Mordell-Weil group. Torsional elements of the Mordell-Weil group are known to be in the kernel of $\varphi$. 

The identification 
\begin{equation}
    s \leftrightarrow \sigma_s
\end{equation}
between the Mordell-Weil group and curve classes is a homomorphism modulo the trivial lattice $T$. However, the curve class of any section has the properties that 
\begin{equation}\label{eq: section intersection conditions}
    \begin{aligned}
        \sigma_s \cdot E & = 1 \\ 
        \sigma_s \cdot \sigma_s & = -1 \\
        \sigma_s \cdot E_i &= 1 \qquad \text{for a single $i$} \\
        \sigma_s \cdot E_j &= 0 \qquad \text{for all others} \, .
    \end{aligned}
\end{equation}
The number of these conditions equals the rank of $T$, so that they allow to find the class of the sum (under the Mordell-Weil group) of sections. However, this requires knowledge about which fiber component the sum of any two given sections meets. This can be worked out as follows: reducible elliptic fibers have a multiplicative (in the $I_n$ case) or additive (for all others) group structure. The fiber component meeting the zero section (called the trivial fiber component) plays the role of the identity element here. The map $\Psi$ which takes any element of the Mordell-Weil group to the fiber component it meets turns out to be a group homomorphism (see Lemma 7.4 of \cite{schuett2010ellipticsurfaces}), a fact which can be understood from the action of translations by a section on reducible fiber components. $\Psi$ is injective for torsional sections, so that they must always meet a non-zero fiber component. 

With the information about which $E_i$ is met by $\sigma_{s \boxplus s'}$, one can work out the class of 
\begin{equation}
\sigma_{s \boxplus s'} = \sigma_s + \sigma_{s'} \quad \text{mod}\quad  T \, ,
\end{equation}
by imposing the conditions \eqref{eq: section intersection conditions}.

\subsection{Sections on $S_0$}

For $S_0$ there are no reducible fibers and $L = \Lambda_{-E_8}$. We can hence write 
\begin{equation}
    \varphi(s) = \sigma - \sigma_0 - \left(\sigma \cdot \sigma_0 +1 \right) E = \alpha
\end{equation}
for $\alpha \in \Lambda_{-E_8}$. Squaring both sides and using that any section must meet the fiber exactly once and square to $-1$ we find
\begin{equation}
    \alpha \cdot \alpha = -2 \sigma \cdot \sigma_0 -2
\end{equation}
so that 
\begin{equation}
    \sigma_\alpha = \sigma_0 + \alpha - \frac{\alpha^2}{2} E \, .
\end{equation}
As $T = \left\{E,\sigma_0\right\}$ in this case and there are no $E_i$, it follows that sections are simply added as
\begin{equation}
\sigma_{\alpha+\beta} = \sigma_0 + \alpha +\beta - \frac{(\alpha+\beta)^2}{2} E \, . 
\end{equation}

\subsection{Sections on $S$}

On $S$ we have an $I_8$ fiber for which the fiber component $\beta_1$ meets the zero section. The other components are $E_i = \beta_{i+1} = M_{i+2}-M_{i+1}$ for $i\leq6$ and $E_7 = E - \sum_{i=1}^7\beta_i = E + M_1 - M_8$. This is summarized by the white nodes in the diagram:
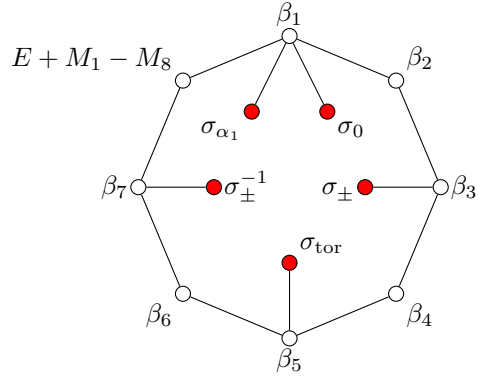
\begin{center}
    \begin{tikzpicture}[
    node distance=1.5cm,
    main/.style={circle, draw, fill=white, inner sep=2pt},
    section/.style={circle, draw, fill=red, inner sep=2pt},
    label style/.style={font=\small, color=black}
]

    \node[main] (v0) at (0, 2) {};
    \node[main] (v1) at (1.41, 1.41) {};
    \node[main] (v2) at (2, 0) {};
    \node[main] (v3) at (1.41, -1.41) {};
    \node[main] (v4) at (0, -2) {};
    \node[main] (v5) at (-1.41, -1.41) {};
    \node[main] (v6) at (-2, 0) {};
    \node[main] (v7) at (-1.41, 1.41) {};

    \node[section] (s0) at (0.5,1) {};
    \node[section] (sa) at (-0.5,1) {};
    \node[section] (spm) at (1,0) {};
    \node[section] (spmi) at (-1,0) {};
    \node[section] (st) at (0,-1) {};

    \draw (v0) -- (v1);
    \draw (v1) -- (v2);
    \draw (v2) -- (v3);
    \draw (v3) -- (v4);
    \draw (v4) -- (v5);
    \draw (v5) -- (v6);
    \draw (v6) -- (v7);
    \draw (v7) -- (v0);

    \draw (v0) -- (s0);
    \draw (v0) -- (sa);
    \draw (v2) -- (spm);
    \draw (v6) -- (spmi);
    \draw (v4) -- (st);

    \node[above, label style] at (v0) {$\beta_1$};
    \node[above right, label style] at (v1) {$\beta_2$};
    \node[right, label style] at (v2) {$\beta_3$};
    \node[below right, label style] at (v3) {$\beta_4$};
    \node[below, label style] at (v4) {$\beta_5$};
    \node[below left, label style] at (v5) {$\beta_6$};
    \node[left, label style] at (v6) {$\beta_7$};
    \node[above left, label style] at (v7) {$E + M_1 - M_8$};

    \node[below right, label style] at (s0) {$\sigma_0$};
    \node[below left, label style] at (sa) {$\sigma_{\alpha_1}$};
    \node[left, label style] at (spm) {$\sigma_\pm^{{\color{white}1}}$};
    \node[right, label style] at (spmi) {$\sigma_\pm^{-1}$};
    \node[above right, label style] at (st) {$\sigma_{\text{tor}}$};
    
\end{tikzpicture}
\end{center}
The trivial lattice $T$ is generated by $\sigma_0,E$ together with the $E_i$. As the Mordell-Weil group is $\Z \oplus \Z_2$ so there must be one more non-torsional section besides $\sigma_0$. On $S$ there is a fiber of type $I_8$ so that $A^{-1}$ is minus the Cartan matrix of $SU(8)$, and the fiber components not meeting the zero section are $E_i = \beta_{i+1}$. The inner form between $\{\sigma_0,E,\sigma,E_i | i = 1,\cdots ,7\}$ must have determinant $4$, which makes $H^2(S,\Z)$ unimodular taking into account the torsional section. This determines $\sigma$ to meet the fiber component $E_2 = \beta_3$ (another equivalent option is $E_6=\beta_7$). Note that $\beta_1$ is not among the $E_i$ as it meets the zero section $\sigma_0 = C_1$ drawn in red.

As the lattice $L$ is generated by $F-\Sigma$ for $S$, \eqref{eq:sections schutts hioda} implies 
\begin{equation}
\varphi(s) = q(F-\Sigma) =  \sigma - E + \frac14(2F + 2 \Sigma -4M_2-4M_3) 
\end{equation}
so that we need $q=\pm 1/2$ for $\sigma$ to be integral and square to $-1$ and find the two sections $s_\pm$ with classes
\begin{equation}
\begin{aligned}
  \sigma_{+} &:=  M_1 + M_2 + M_3 - \Sigma + E   \\
  \sigma_{-} &:=  M_1 + M_2 + M_3 - F + E \, .
\end{aligned}
\end{equation}
The sum of these two sections is such that 
$\varphi(s_{+} \boxplus s_{-} ) = 0$ which implies that it must be the torsional section. As $\sigma_{\pm}$ both meet $E_2 = \beta_3$, the multiplicative group structure of the reducible fiber implies that $s_{+} \boxplus s_{-}$ meets the fiber component $E_4$.
This is also found by noting that the torsional subgroup is a $\Z_2$ acting multiplicatively on the cyclical group of fiber components. Implementing the intersection conditions then determines 
\begin{equation}
\sigma_{\text{tor}}  = \sigma_{+} \boxplus \sigma_{-} = M_1 -F -\Sigma + M_2 + M_3 + M_4 + M_5 + E \, .
\end{equation}
This can also be found from \eqref{eq:sections schutts hioda} using the information that the torsional section must meet the fiber component $E_4 = \beta_5$. 

The multiplicative structure of the $I_8$ fiber determines that the inverse $s_\pm^{-1}$ of $s_\pm$ (they obey $s_\pm \boxplus s_\pm^{-1} = s_0$ with $s_0$ the zero section) correspond to curve classes $\sigma_{\pm}^{-1}$ which meet the fiber component $E_6 = \beta_7$.
It follows that $s_- \boxplus s_+^{-1} = s_- \boxminus s_+$ only meets the trivial fiber component, so that 
\begin{equation}
\sigma_{\alpha_1}:= \sigma_{s_- \boxminus s_+} = M_1 + \Sigma - F + E  \,,
\end{equation}
is determined by \eqref{eq: section intersection conditions} together with 
\begin{equation}
    \sigma_{s_- \boxminus s_+} = \sigma_+ - \sigma_- \quad \text{mod} \quad T\, .
\end{equation}

\section{Generalities of Gluing Geometric Profiles} \label{sec: App Gluing Profiles}

Here we provide all of the details necessary to prove that one can always write the building blocks $\mathcal{D}[\mathcal{F}]$ as an alternating product of $\mu_s^{\varepsilon}$ and $\mu_{t,\mathbf{n}}^{\varepsilon}$ local building blocks, as claimed in Section~\ref{sect:elemfluxtubes}. 

We start with a definition of the $\mu_{t,\mathbf{n}}$ automorphisms. Define the set $\mathcal{S} := \{ ( n_i ) \in \mathbb{Z}_2^8 \, | \, \sum_{i=1}^8 n_i = 0 \, \, \, \text{mod} \, \, 2 \}$, with the following structure:
\begin{enumerate}
    \item The set $\mathcal{S}$ is closed under elementwise addition $\text{mod} \, \, 2$. Namely for elements $\mathbf{n},\mathbf{m} \in \mathcal{S}$
    \be
        \mathbf{n}+\mathbf{m} := (n_1+m_1 \, \, \,  \text{mod} \, \, 2, n_2+m_2 \, \, \,  \text{mod} \, \, 2, \cdots , n_8+m_8 \, \, \,  \text{mod} \, \, 2,  )  \in \mathcal{S}\, .
    \ee
    \item We can define the size of an element $\mathbf{n}\in \mathcal{S}$ as 
    \be
        | \mathbf{n} | := \sum_{i=1}^8 n_i \in \{ 0,2,4,6,8 \} \, .
    \ee
    
\end{enumerate}
For each such $\mathbf{n} \in \mathcal{S}$, one can define the following $\text{dP}_9$ automorphism
\be \label{eq: very general mutn def}
\mu_{t,\mathbf{n}} :  \qquad 
\left\{\begin{aligned}
F &\rightarrow F'=F \\
 \Sigma &\rightarrow \Sigma' = \Sigma+\frac{1}{2}|\mathbf{n}|F-\sum_{i=1}^8 n_i M_i \\
M_i &\rightarrow M_i' = F-M_i \qquad \ \qquad \qquad \qquad  n_i\equiv 1 \\
M_i &\rightarrow M_i' = M_i\qquad \ \qquad \qquad \qquad \qquad  n_i \equiv 0 \, .
\end{aligned}\right.
\ee
The argument that such maps are automorphisms follows from Section~\ref{sec:constructionofXS}, up to a relabeling of the curves classes of the blowups. These maps obey the following properties
\begin{enumerate}
    \item $\mu_{t,\mathbf{n}} \mu_{t,\mathbf{m}} = \mu_{t,\mathbf{n}+\mathbf{m}}$.
    \item $[\mu_{t,\mathbf{n}},\mu_{t,\mathbf{m}}]=0$. 
    \item $\mu_{t,\mathbf{n}}^2=1$. 
\end{enumerate}

\paragraph{Proof:} We prove the first property, from which the second and third property follow. Observe that the composed map acts as
\be 
\mu_{t,\mathbf{n}}  \mu_{t,\mathbf{m}} :  \qquad 
\left\{\begin{aligned}
F &\rightarrow F''=F' \\
 \Sigma &\rightarrow \Sigma'' = \Sigma'+\frac{1}{2}|\mathbf{m}|F'-\sum_{i=1}^8 m_i M_i' \\
M_i &\rightarrow M_i'' = F'-M_i' \qquad \ \qquad \qquad \qquad  &m_i \equiv 1 \\
M_i &\rightarrow M_i'' = M_i'  &m_i \equiv 0 \, ,
\end{aligned}\right.
\ee
where the primed curves denote the image of $F,\Sigma,M_i$ under $\mu_{t,\mathbf{n}}$. Clearly $F''=F'=F$. Next, observe that 
\be
\ba
    \Sigma'' &= \Sigma+\frac{1}{2}|\mathbf{n}|F-\sum_{i=1}^8 n_i M_i+\frac{1}{2}|\mathbf{m}|F-\sum_{i=1}^8 n_i m_i(F-M_i) -\sum_{i=1}^8 (1-n_i) m_i  M_i \cr
    &= \Sigma + \frac{1}{2}\left(|\mathbf{m}|+|\mathbf{n}|-2 \sum_{i=1}^8 n_i m_i \right)F - \sum_{i=1}^8 [n_i(1-m_i)+ (1-n_i)m_i] M_i  \cr
    &= \Sigma +\frac{1}{2}|\mathbf{n}+\mathbf{m}|F - \sum_{i=1}^8 [n_i(1-m_i)+ (1-n_i)m_i] M_i \, .
\ea
\ee
The sum over the matter curves has coefficient 1 when $(\mathbf{n}+\mathbf{m})_i \equiv 1$, and zero otherwise. Finally, observe that 
\be 
M_i'' =   \qquad 
\left\{\begin{aligned}
 &F-(F-M_i)=M_i &n_{i},m_i\equiv 1 \\
 &F-M_i &m_i \equiv 1, \, n_i \equiv 0 \\
&F-M_i &n_i \equiv 1 \, , m_i \equiv 0 \\
&M_i  &n_i,m_i \equiv 0 
\end{aligned}\right. \qquad  = \qquad 
\left\{\begin{aligned}
 &F-M_i &(\mathbf{n}+\mathbf{m})_i \equiv 1 \\
 &M_i &(\mathbf{n}+\mathbf{m})_i \equiv 0 \, ,
\end{aligned}\right.
\ee
allowing us to conclude $\mu_{t,\mathbf{n}} \mu_{t,\mathbf{m}} = \mu_{t,\mathbf{n}+\mathbf{m}}$. \\

We can promote this result to building blocks $\mu_{t,\mathbf{n}}^{\varepsilon} \mu_{t,\mathbf{m}}^{\varepsilon}=\mu_{t,\mathbf{n}+\mathbf{m}}^{\varepsilon}$. This fact is a direct geometric analogue of integrating out massive chirals in field theory. Observe that the building block $\mu_{t,\mathbf{n}}^{\varepsilon}\mu_{t,\mathbf{m}}^{\varepsilon}$ that one would naturally associate to such maps gives rise to a quiver containing a single gauge $SU(2)$ node. When $n_i=m_i=0$, the $i$'th chiral is not charged under the flavor symmetry, so is absent in the local $\bold{n}+\bold{m}$ quiver. When $n_i=1$, $m_i=0$, the $i$'th chiral is charged under the flavor symmetry, and so we get a chiral in the local $\bold{n}+\bold{m}$ quiver. However, when both $n_i=m_i=1$, we get a pair of chirals charged under the flavor symmetry, meaning they generate a superpotential term which gives them both a mass, so are integrated out in the IR. Thus we do not get a chiral charged under the flavor symmetry in the local $\bold{n}+\bold{m}$ quiver when $n_i+m_i=0 \mod{2}$, in exact agreement with the building block structure of $\mu_{t,\bold{n}+\bold{m}}^{\varepsilon}$. The properties $2.$ and $3.$ therefore extend to local building blocks as well. \\

The next identity we prove is that the most general form of a root $\alpha \in R$ takes the form
\be 
    \alpha = \pm \Sigma \pm \left( \frac{1}{2}|\mathbf{n}|-1 \right)F\mp\sum_{i=1}^8 n_i M_i \, , \qquad \alpha \in R_{\pm} \, ,
\ee
with $\mathbf{n} \in \mathcal{S}$. This provides a bijection between $R_
{\pm}$ and $\mathcal{S}$. 

\paragraph{Proof:} To prove this, observe that if we take the most generic $\alpha \in H_2(X_S, \mathbb{Z})$, we can write it as
\be \label{eq: generic form of a curve}
    \alpha = a \Sigma + b F - \sum_{i=1}^8 d_i M_i \, .
\ee
There are now multiple constraints coming from $\alpha \in R_{\pm}$
\be
    \alpha^2=-2 \, , \qquad \alpha \cdot F = \pm 1 \, , \qquad \alpha \cdot E =0\, , \qquad 4a +\sum_{i=1}^8 d_i \equiv 0 \mod{2} \, ,
\ee
for $a,b,d_i\in \mathbb{Z}$, where the final constraint comes from demanding $N(\alpha)\in E_8$. The second constraint fixes $a=\pm1$, and the first and third constraints fix 
\be
    2b = \mp2 \pm \sum_{i=1}^8 d_i^2 \, , \qquad 2b= \mp2+\sum_{i=1}^8 d_i \, ,
\ee
uniquely fixing $b$ in terms of $d_i$. These two equations lead to
\be
    \sum_{i=1}^8 d_i(1\mp d_i) = 0 \, .
\ee
Now we use the fact that for any integer $p \in \mathbb{Z}$, $p(1-p) \le 0$, and $p(1+p)\ge 0$. Thus since each term in the sum is either non-positive or non-negative, it must be true that each term is identically zero. This fixes $d_i \in \{ 0,1 \}$ for $\alpha \in R_+$, and $d_i \in \{ 0,-1 \}$ for $\alpha \in R_-$. Last, we apply the fourth constraint to determine that only an even number of $d_i$ are allowed to be non-zero. Thus we can define $\mathbf{n}\equiv (d_i) \in \mathcal{S}$. Solving for $b$ and substituting into~\eqref{eq: generic form of a curve} yields 
\be
    \alpha = \pm\Sigma \pm\left( \frac{1}{2}|\mathbf{n}|-1 \right)F \mp \sum_{i=1}^8 n_i M_i \, , 
\ee
as required. Thus for any $\alpha$, there is a unique $\mathbf{n}(\alpha)$ that gives rise to such an element of $R$. \\

Using this result, we now prove the following: given a Weyl reflection $w_{\alpha}$, with $\alpha \in R$, we can always write its action as
\be
    w_{\alpha} = \mu_{t,\mathbf{n}} \mu_s \mu_{t,\mathbf{n}} \, ,
\ee
for the unique $\mathbf{n}$ satisfying 
\be 
    \alpha = \pm \Sigma \pm \left( \frac{1}{2}|\mathbf{n}|-1 \right)F\mp\sum_{i=1}^8 n_i M_i \, , \qquad \alpha \in R_{\pm} \, .
\ee

\paragraph{Proof:} The logic is as follows. We first suppose that $\mu_{t,\mathbf{n}}\mu_s\mu_{t,\mathbf{n}}=w_{\alpha}$ for some $\alpha \in R$. We then fix the form of $\alpha=\alpha(\mathbf{n})$ in terms of $\mathbf{n}$ using the intersection between $\alpha$ and $F$. Using this, we show the proposed $\alpha(\mathbf{n})$ is consistent with the action on the section and matter curves. Lastly, we invert the relation using our previous bijection to uniquely construct the corresponding $\mathbf{n}$. To prove this, we first construct the action of $\mu_{t,\mathbf{n}}\mu_s\mu_{t,\mathbf{n}}$ on the curves of the threefold
\be 
\mu_{t,\mathbf{n}} \mu_s \mu_{t,\mathbf{n}} :  \qquad 
\left\{\begin{aligned}
F &\rightarrow \Sigma +\frac{1}{2}|\mathbf{n}|F-\sum_{i=1}^8 n_i M_i \, , \\
 \Sigma &\rightarrow \frac{1}{2}|\mathbf{n}| \Sigma+\left(\frac{1}{2}|\mathbf{n}|-1\right)^2F-\left(\frac{1}{2}|\mathbf{n}|-1\right)\sum_{i=1}^8 n_i M_i \\
M_i &\rightarrow \Sigma +\left(\frac{1}{2}|\mathbf{n}|-1\right)F+M_i-\sum_{j=1}^8 n_j M_j \qquad \ \qquad \qquad \qquad  &n_i \equiv 1 \\
M_i &\rightarrow M_i  &n_i \equiv 0 \, ,
\end{aligned}\right.
\ee
for a generic $\mathbf{n} \in \mathcal{S}$. Next, we compare this to the form of the Weyl reflection on the fiber $F$, which we can do since we know that $\alpha \in R_\pm$ implies $\alpha \cdot F=\pm1$
\be
    w_{\alpha}(F) = F+(\alpha \cdot F)\alpha = F\pm \alpha \, .
\ee
From here we immediately read off $\alpha$ as 
\be
    \alpha = \pm \Sigma \pm \left( \frac{1}{2}|\mathbf{n}|-1 \right)F\mp\sum_{i=1}^8 n_i M_i \, , \qquad \alpha \in R_{\pm} \, .
\ee
Next, we show that this $\alpha(\mathbf{n})$ is consistent with the section and matter curves 
\be
\ba
    w_{\alpha(\mathbf{n})}(\Sigma) &= \Sigma \pm \left( \frac{1}{2}|\mathbf{n}|-1 \right)\alpha(\mathbf{n}) = \frac{1}{2}|\mathbf{n}|\Sigma+\left( \frac{1}{2}|\mathbf{n}|-1 \right)^2 F -\left( \frac{1}{2}|\mathbf{n}|-1 \right)\sum_{i=1}^8 n_i M_i \, , \cr
    w_{\alpha(\mathbf{n})}(M_i) &= M_i \pm \alpha(\mathbf{n}) = \Sigma +\left( \frac{1}{2}|\mathbf{n}|-1 \right)F+M_i-\sum_{j=1}^8 n_j M_j \, , &n_i \equiv 1\cr
    w_{\alpha(\mathbf{n})}(M_i) &= M_i &n_i \equiv 0 \, ,
\ea
\ee
as required. To conclude, we use the fact that there exists a bijection between $\mathcal{S}$ and $R_{\pm}$, to give the unique $\mathbf{n} \in \mathcal{S}$ satisfying $w_{\alpha}=\mu_{t,\mathbf{n}}\mu_s \mu_{t,\mathbf{n}}$.

\bibliographystyle{ytphys}
\small 
\bibliography{FM}

\providecommand{\href}[2]{#2}\begingroup\raggedright\begin{thebibliography}{10}

\bibitem{Seiberg:1996bd}
N.~Seiberg, ``{Five-dimensional SUSY field theories, nontrivial fixed points
  and string dynamics},''
  \href{http://dx.doi.org/10.1016/S0370-2693(96)01215-4}{{\em Phys. Lett. B}
  {\bfseries 388} (1996) 753--760},
  \href{http://arxiv.org/abs/hep-th/9608111}{{\ttfamily arXiv:hep-th/9608111}}.

\bibitem{Seiberg:1996vs}
N.~Seiberg and E.~Witten, ``{Comments on string dynamics in six-dimensions},''
  \href{http://dx.doi.org/10.1016/0550-3213(96)00189-7}{{\em Nucl. Phys. B}
  {\bfseries 471} (1996) 121--134},
  \href{http://arxiv.org/abs/hep-th/9603003}{{\ttfamily arXiv:hep-th/9603003}}.

\bibitem{DelZotto:2014hpa}
M.~Del~Zotto, J.~J. Heckman, A.~Tomasiello, and C.~Vafa, ``{6d Conformal
  Matter},'' \href{http://dx.doi.org/10.1007/JHEP02(2015)054}{{\em JHEP}
  {\bfseries 02} (2015) 054}, \href{http://arxiv.org/abs/1407.6359}{{\ttfamily
  arXiv:1407.6359 [hep-th]}}.

\bibitem{Bhardwaj:2015xxa}
L.~Bhardwaj, ``{Classification of 6d $ \mathcal{N}=\left(1,0\right) $ gauge
  theories},'' \href{http://dx.doi.org/10.1007/JHEP11(2015)002}{{\em JHEP}
  {\bfseries 11} (2015) 002}, \href{http://arxiv.org/abs/1502.06594}{{\ttfamily
  arXiv:1502.06594 [hep-th]}}.

\bibitem{Heckman:2013pva}
J.~J. Heckman, D.~R. Morrison, and C.~Vafa, ``{On the Classification of 6D
  SCFTs and Generalized ADE Orbifolds},''
  \href{http://dx.doi.org/10.1007/JHEP05(2014)028}{{\em JHEP} {\bfseries 05}
  (2014) 028}, \href{http://arxiv.org/abs/1312.5746}{{\ttfamily arXiv:1312.5746
  [hep-th]}}. [Erratum: JHEP 06, 017 (2015)].

\bibitem{Heckman:2015bfa}
J.~J. Heckman, D.~R. Morrison, T.~Rudelius, and C.~Vafa, ``{Atomic
  Classification of 6D SCFTs},''
  \href{http://dx.doi.org/10.1002/prop.201500024}{{\em Fortsch. Phys.}
  {\bfseries 63} (2015) 468--530},
  \href{http://arxiv.org/abs/1502.05405}{{\ttfamily arXiv:1502.05405
  [hep-th]}}.

\bibitem{Xie:2017pfl}
D.~Xie and S.-T. Yau, ``{Three dimensional canonical singularity and five
  dimensional $ \mathcal{N} $ = 1 SCFT},''
  \href{http://dx.doi.org/10.1007/JHEP06(2017)134}{{\em JHEP} {\bfseries 06}
  (2017) 134}, \href{http://arxiv.org/abs/1704.00799}{{\ttfamily
  arXiv:1704.00799 [hep-th]}}.

\bibitem{Jefferson:2018irk}
P.~Jefferson, S.~Katz, H.-C. Kim, and C.~Vafa, ``{On Geometric Classification
  of 5d SCFTs},'' \href{http://dx.doi.org/10.1007/JHEP04(2018)103}{{\em JHEP}
  {\bfseries 04} (2018) 103}, \href{http://arxiv.org/abs/1801.04036}{{\ttfamily
  arXiv:1801.04036 [hep-th]}}.

\bibitem{Bhardwaj:2018yhy}
L.~Bhardwaj and P.~Jefferson, ``{Classifying $5d$ SCFTs via $6d$ SCFTs: Rank
  one},'' \href{http://dx.doi.org/10.1007/JHEP07(2019)178}{{\em JHEP}
  {\bfseries 07} (2019) 178}, \href{http://arxiv.org/abs/1809.01650}{{\ttfamily
  arXiv:1809.01650 [hep-th]}}. [Addendum: JHEP 01, 153 (2020)].

\bibitem{Bhardwaj:2018vuu}
L.~Bhardwaj and P.~Jefferson, ``{Classifying 5d SCFTs via 6d SCFTs: Arbitrary
  rank},'' \href{http://dx.doi.org/10.1007/JHEP10(2019)282}{{\em JHEP}
  {\bfseries 10} (2019) 282}, \href{http://arxiv.org/abs/1811.10616}{{\ttfamily
  arXiv:1811.10616 [hep-th]}}.

\bibitem{Apruzzi:2019vpe}
F.~Apruzzi, C.~Lawrie, L.~Lin, S.~Sch{\"a}fer-Nameki, and Y.-N. Wang, ``{5d
  Superconformal Field Theories and Graphs},''
  \href{http://dx.doi.org/10.1016/j.physletb.2019.135077}{{\em Phys. Lett. B}
  {\bfseries 800} (2020) 135077},
  \href{http://arxiv.org/abs/1906.11820}{{\ttfamily arXiv:1906.11820
  [hep-th]}}.

\bibitem{Apruzzi:2019opn}
F.~Apruzzi, C.~Lawrie, L.~Lin, S.~Sch{\"a}fer-Nameki, and Y.-N. Wang, ``{Fibers
  add Flavor, Part I: Classification of 5d SCFTs, Flavor Symmetries and BPS
  States},'' \href{http://dx.doi.org/10.1007/JHEP11(2019)068}{{\em JHEP}
  {\bfseries 11} (2019) 068}, \href{http://arxiv.org/abs/1907.05404}{{\ttfamily
  arXiv:1907.05404 [hep-th]}}.

\bibitem{Apruzzi:2019enx}
F.~Apruzzi, C.~Lawrie, L.~Lin, S.~Sch\"afer-Nameki, and Y.-N. Wang, ``{Fibers
  add Flavor, Part II: 5d SCFTs, Gauge Theories, and Dualities},''
  \href{http://dx.doi.org/10.1007/JHEP03(2020)052}{{\em JHEP} {\bfseries 03}
  (2020) 052}, \href{http://arxiv.org/abs/1909.09128}{{\ttfamily
  arXiv:1909.09128 [hep-th]}}.

\bibitem{Apruzzi:2019kgb}
F.~Apruzzi, S.~Schafer-Nameki, and Y.-N. Wang, ``{5d SCFTs from Decoupling and
  Gluing},'' \href{http://dx.doi.org/10.1007/JHEP08(2020)153}{{\em JHEP}
  {\bfseries 08} (2020) 153}, \href{http://arxiv.org/abs/1912.04264}{{\ttfamily
  arXiv:1912.04264 [hep-th]}}.

\bibitem{Acharya:2001gy}
B.~S. Acharya and E.~Witten, ``{Chiral fermions from manifolds of G(2)
  holonomy},'' \href{http://arxiv.org/abs/hep-th/0109152}{{\ttfamily
  arXiv:hep-th/0109152}}.

\bibitem{Halverson:2014tya}
J.~Halverson and D.~R. Morrison, ``{The landscape of M-theory compactifications
  on seven-manifolds with G$_{2}$ holonomy},''
  \href{http://dx.doi.org/10.1007/JHEP04(2015)047}{{\em JHEP} {\bfseries 04}
  (2015) 047}, \href{http://arxiv.org/abs/1412.4123}{{\ttfamily arXiv:1412.4123
  [hep-th]}}.

\bibitem{Braun:2017ryx}
A.~P. Braun and M.~Del~Zotto, ``{Mirror Symmetry for $G_2$-Manifolds: Twisted
  Connected Sums and Dual Tops},''
  \href{http://dx.doi.org/10.1007/JHEP05(2017)080}{{\em JHEP} {\bfseries 05}
  (2017) 080}, \href{http://arxiv.org/abs/1701.05202}{{\ttfamily
  arXiv:1701.05202 [hep-th]}}.

\bibitem{Braun:2016igl}
A.~P. Braun, ``{Tops as building blocks for G$_{2}$ manifolds},''
  \href{http://dx.doi.org/10.1007/JHEP10(2017)083}{{\em JHEP} {\bfseries 10}
  (2017) 083}, \href{http://arxiv.org/abs/1602.03521}{{\ttfamily
  arXiv:1602.03521 [hep-th]}}.

\bibitem{Braun:2017uku}
A.~P. Braun and S.~Sch{\"a}fer-Nameki, ``{Compact, Singular $G_2$-Holonomy
  Manifolds and M/Heterotic/F-Theory Duality},''
  \href{http://dx.doi.org/10.1007/JHEP04(2018)126}{{\em JHEP} {\bfseries 04}
  (2018) 126}, \href{http://arxiv.org/abs/1708.07215}{{\ttfamily
  arXiv:1708.07215 [hep-th]}}.

\bibitem{Braun:2017csz}
A.~P. Braun and M.~Del~Zotto, ``{Towards Generalized Mirror Symmetry for
  Twisted Connected Sum $G_2$ Manifolds},''
  \href{http://dx.doi.org/10.1007/JHEP03(2018)082}{{\em JHEP} {\bfseries 03}
  (2018) 082}, \href{http://arxiv.org/abs/1712.06571}{{\ttfamily
  arXiv:1712.06571 [hep-th]}}.

\bibitem{Braun:2018joh}
A.~P. Braun and S.~Sch{\"a}fer-Nameki, ``{Spin(7)-manifolds as generalized
  connected sums and 3d $\mathcal{N}=1$ theories},''
  \href{http://dx.doi.org/10.1007/JHEP06(2018)103}{{\em JHEP} {\bfseries 06}
  (2018) 103}, \href{http://arxiv.org/abs/1803.10755}{{\ttfamily
  arXiv:1803.10755 [hep-th]}}.

\bibitem{Braun:2018vhk}
A.~P. Braun, S.~Cizel, M.~H\"ubner, and S.~Sch\"afer-Nameki, ``{Higgs bundles
  for M-theory on $G_{2}$-manifolds},''
  \href{http://dx.doi.org/10.1007/JHEP03(2019)199}{{\em JHEP} {\bfseries 03}
  (2019) 199}, \href{http://arxiv.org/abs/1812.06072}{{\ttfamily
  arXiv:1812.06072 [hep-th]}}.

\bibitem{Braun:2018fdp}
A.~P. Braun, M.~Del~Zotto, J.~Halverson, M.~Larfors, D.~R. Morrison, and
  S.~Sch{\"a}fer-Nameki, ``{Infinitely many M2-instanton corrections to
  M-theory on G$_{2}$-manifolds},''
  \href{http://dx.doi.org/10.1007/JHEP09(2018)077}{{\em JHEP} {\bfseries 09}
  (2018) 077}, \href{http://arxiv.org/abs/1803.02343}{{\ttfamily
  arXiv:1803.02343 [hep-th]}}.

\bibitem{Acharya:2018nbo}
B.~S. Acharya, A.~P. Braun, E.~E. Svanes, and R.~Valandro, ``{Counting
  associatives in compact $G_{2}$ orbifolds},''
  \href{http://dx.doi.org/10.1007/JHEP03(2019)138}{{\em JHEP} {\bfseries 03}
  (2019) 138}, \href{http://arxiv.org/abs/1812.04008}{{\ttfamily
  arXiv:1812.04008 [hep-th]}}.

\bibitem{Barbosa:2019bgh}
R.~Barbosa, M.~Cveti{\v{c}}, J.~J. Heckman, C.~Lawrie, E.~Torres, and
  G.~Zoccarato, ``{T-branes and $G_2$ backgrounds},''
  \href{http://dx.doi.org/10.1103/PhysRevD.101.026015}{{\em Phys. Rev. D}
  {\bfseries 101} no.~2, (2020) 026015},
  \href{http://arxiv.org/abs/1906.02212}{{\ttfamily arXiv:1906.02212
  [hep-th]}}.

\bibitem{Acharya:2019svi}
B.~S. Acharya, R.~L. Bryant, and S.~Salamon, ``{A circle quotient of a $G_2$
  cone},'' \href{http://dx.doi.org/10.1016/j.difgeo.2020.101681}{{\em Differ.
  Geom. Appl.} {\bfseries 73} (2020) 101681},
  \href{http://arxiv.org/abs/1910.09518}{{\ttfamily arXiv:1910.09518
  [math.DG]}}.

\bibitem{Acharya:2020vmg}
B.~S. Acharya, L.~Foscolo, M.~Najjar, and E.~E. Svanes, ``{New
  G$_{2}$-conifolds in M-theory and their field theory interpretation},''
  \href{http://dx.doi.org/10.1007/JHEP05(2021)250}{{\em JHEP} {\bfseries 05}
  (2021) 250}, \href{http://arxiv.org/abs/2011.06998}{{\ttfamily
  arXiv:2011.06998 [hep-th]}}.

\bibitem{Acharya:2021rvh}
B.~S. Acharya, A.~Kinsella, and D.~R. Morrison, ``{Non-perturbative heterotic
  duals of M-theory on G$_{2}$ orbifolds},''
  \href{http://dx.doi.org/10.1007/JHEP11(2021)065}{{\em JHEP} {\bfseries 11}
  (2021) 065}, \href{http://arxiv.org/abs/2106.03886}{{\ttfamily
  arXiv:2106.03886 [hep-th]}}.

\bibitem{Acharya:2023bth}
B.~S. Acharya, M.~Del~Zotto, J.~J. Heckman, M.~Hubner, and E.~Torres,
  ``{Junctions, edge modes, and $G_2$-holonomy orbifolds},''
  \href{http://dx.doi.org/10.4310/bpam.2024.v1.n1.a5}{{\em Beijing J. Pure
  Appl. Math.} {\bfseries 1} no.~1, (2024) 273--371},
  \href{http://arxiv.org/abs/2304.03300}{{\ttfamily arXiv:2304.03300
  [hep-th]}}.

\bibitem{Braun:2023fqa}
A.~P. Braun, E.~Sabag, M.~Sacchi, and S.~Schafer-Nameki, ``{$G_2$-manifolds
  from 4d N=1 theories, part I: Domain walls},''
  \href{http://dx.doi.org/10.21468/SciPostPhys.17.4.102}{{\em SciPost Phys.}
  {\bfseries 17} no.~4, (2024) 102},
  \href{http://arxiv.org/abs/2304.01193}{{\ttfamily arXiv:2304.01193
  [hep-th]}}.

\bibitem{Gaiotto:2015usa}
D.~Gaiotto and S.~S. Razamat, ``{$ \mathcal{N}=1 $ theories of class $
  {\mathcal{S}}_k $},'' \href{http://dx.doi.org/10.1007/JHEP07(2015)073}{{\em
  JHEP} {\bfseries 07} (2015) 073},
  \href{http://arxiv.org/abs/1503.05159}{{\ttfamily arXiv:1503.05159
  [hep-th]}}.

\bibitem{Ohmori:2015pua}
K.~Ohmori, H.~Shimizu, Y.~Tachikawa, and K.~Yonekura, ``{6d $\mathcal{N}=(1,0)$
  theories on $T^2$ and class S theories: Part I},''
  \href{http://dx.doi.org/10.1007/JHEP07(2015)014}{{\em JHEP} {\bfseries 07}
  (2015) 014}, \href{http://arxiv.org/abs/1503.06217}{{\ttfamily
  arXiv:1503.06217 [hep-th]}}.

\bibitem{Ohmori:2015pia}
K.~Ohmori, H.~Shimizu, Y.~Tachikawa, and K.~Yonekura, ``{6d
  $\mathcal{N}=\left(1,\;0\right) $ theories on S$^{1}$ /T$^{2}$ and class S
  theories: part II},'' \href{http://dx.doi.org/10.1007/JHEP12(2015)131}{{\em
  JHEP} {\bfseries 12} (2015) 131},
  \href{http://arxiv.org/abs/1508.00915}{{\ttfamily arXiv:1508.00915
  [hep-th]}}.

\bibitem{Razamat:2016dpl}
S.~S. Razamat, C.~Vafa, and G.~Zafrir, ``{4d $ \mathcal{N}=1 $ from 6d (1,
  0)},'' \href{http://dx.doi.org/10.1007/JHEP04(2017)064}{{\em JHEP} {\bfseries
  04} (2017) 064}, \href{http://arxiv.org/abs/1610.09178}{{\ttfamily
  arXiv:1610.09178 [hep-th]}}.

\bibitem{Bah:2017gph}
I.~Bah, A.~Hanany, K.~Maruyoshi, S.~S. Razamat, Y.~Tachikawa, and G.~Zafrir,
  ``{4d $ \mathcal{N}=1 $ from 6d $ \mathcal{N}=\left(1,0\right) $ on a torus
  with fluxes},'' \href{http://dx.doi.org/10.1007/JHEP06(2017)022}{{\em JHEP}
  {\bfseries 06} (2017) 022}, \href{http://arxiv.org/abs/1702.04740}{{\ttfamily
  arXiv:1702.04740 [hep-th]}}.

\bibitem{Kim:2017toz}
H.-C. Kim, S.~S. Razamat, C.~Vafa, and G.~Zafrir, ``{E-String Theory on Riemann
  Surfaces},'' \href{http://dx.doi.org/10.1002/prop.201700074}{{\em Fortsch.
  Phys.} {\bfseries 66} no.~1, (2018) 1700074},
  \href{http://arxiv.org/abs/1709.02496}{{\ttfamily arXiv:1709.02496
  [hep-th]}}.

\bibitem{Kim:2018bpg}
H.-C. Kim, S.~S. Razamat, C.~Vafa, and G.~Zafrir, ``{D-type Conformal Matter
  and SU/USp Quivers},'' \href{http://dx.doi.org/10.1007/JHEP06(2018)058}{{\em
  JHEP} {\bfseries 06} (2018) 058},
  \href{http://arxiv.org/abs/1802.00620}{{\ttfamily arXiv:1802.00620
  [hep-th]}}.

\bibitem{Kim:2018lfo}
H.-C. Kim, S.~S. Razamat, C.~Vafa, and G.~Zafrir, ``{Compactifications of ADE
  conformal matter on a torus},''
  \href{http://dx.doi.org/10.1007/JHEP09(2018)110}{{\em JHEP} {\bfseries 09}
  (2018) 110}, \href{http://arxiv.org/abs/1806.07620}{{\ttfamily
  arXiv:1806.07620 [hep-th]}}.

\bibitem{Razamat:2018gro}
S.~S. Razamat and G.~Zafrir, ``{Compactification of 6d minimal SCFTs on Riemann
  surfaces},'' \href{http://dx.doi.org/10.1103/PhysRevD.98.066006}{{\em Phys.
  Rev. D} {\bfseries 98} no.~6, (2018) 066006},
  \href{http://arxiv.org/abs/1806.09196}{{\ttfamily arXiv:1806.09196
  [hep-th]}}.

\bibitem{Zafrir:2018hkr}
G.~Zafrir, ``{On the torus compactifications of Z$_{2}$ orbifolds of E-string
  theories},'' \href{http://dx.doi.org/10.1007/JHEP10(2019)040}{{\em JHEP}
  {\bfseries 10} (2019) 040}, \href{http://arxiv.org/abs/1809.04260}{{\ttfamily
  arXiv:1809.04260 [hep-th]}}.

\bibitem{Ohmori:2018ona}
K.~Ohmori, Y.~Tachikawa, and G.~Zafrir, ``{Compactifications of 6d $N = (1, 0)$
  SCFTs with non-trivial Stiefel-Whitney classes},''
  \href{http://dx.doi.org/10.1007/JHEP04(2019)006}{{\em JHEP} {\bfseries 04}
  (2019) 006}, \href{http://arxiv.org/abs/1812.04637}{{\ttfamily
  arXiv:1812.04637 [hep-th]}}.

\bibitem{Chen:2019njf}
J.~Chen, B.~Haghighat, S.~Liu, and M.~Sperling, ``{4d $N$=1 from 6d D-type
  $N$=(1,0)},'' \href{http://dx.doi.org/10.1007/JHEP01(2020)152}{{\em JHEP}
  {\bfseries 01} (2020) 152}, \href{http://arxiv.org/abs/1907.00536}{{\ttfamily
  arXiv:1907.00536 [hep-th]}}.

\bibitem{Pasquetti:2019hxf}
S.~Pasquetti, S.~S. Razamat, M.~Sacchi, and G.~Zafrir, ``{Rank $Q$ E-string on
  a torus with flux},''
  \href{http://dx.doi.org/10.21468/SciPostPhys.8.1.014}{{\em SciPost Phys.}
  {\bfseries 8} no.~1, (2020) 014},
  \href{http://arxiv.org/abs/1908.03278}{{\ttfamily arXiv:1908.03278
  [hep-th]}}.

\bibitem{Baume:2021qho}
F.~Baume, M.~J. Kang, and C.~Lawrie, ``{Two 6D origins of 4D SCFTs: Class S and
  6D (1,{\,}0) on a torus},''
  \href{http://dx.doi.org/10.1103/PhysRevD.106.086003}{{\em Phys. Rev. D}
  {\bfseries 106} no.~8, (2022) 086003},
  \href{http://arxiv.org/abs/2106.11990}{{\ttfamily arXiv:2106.11990
  [hep-th]}}.

\bibitem{Sabag:2022hyw}
E.~Sabag and M.~Sacchi, ``{A 5d perspective on the compactifications of 6d
  SCFTs to 4d $ \mathcal{N} $ = 1 SCFTs},''
  \href{http://dx.doi.org/10.1007/JHEP12(2022)017}{{\em JHEP} {\bfseries 12}
  (2022) 017}, \href{http://arxiv.org/abs/2208.03331}{{\ttfamily
  arXiv:2208.03331 [hep-th]}}.

\bibitem{Giacomelli:2023qyc}
S.~Giacomelli and R.~Savelli, ``{{\ensuremath{\mathscr{N}}} = 1 SCFTs from
  F-theory on Orbifolds},''
  \href{http://dx.doi.org/10.1007/JHEP08(2023)129}{{\em JHEP} {\bfseries 08}
  (2023) 129}, \href{http://arxiv.org/abs/2304.11148}{{\ttfamily
  arXiv:2304.11148 [hep-th]}}.

\bibitem{Razamat:2022gpm}
S.~S. Razamat, E.~Sabag, O.~Sela, and G.~Zafrir, ``{Aspects of 4d
  supersymmetric dynamics and geometry},''
  \href{http://dx.doi.org/10.21468/SciPostPhysLectNotes.78}{{\em SciPost Phys.
  Lect. Notes} {\bfseries 78} (2024) 1},
  \href{http://arxiv.org/abs/2203.06880}{{\ttfamily arXiv:2203.06880
  [hep-th]}}.

\bibitem{heckman2001moduli}
G.~Heckman and E.~Looijenga, ``The moduli space of rational elliptic
  surfaces,'' {\em Algebraic Geometry 2000, Azumino} {\bfseries 36} (2001)
  185--248.

\bibitem{Gaiotto:2015una}
D.~Gaiotto and H.-C. Kim, ``{Duality walls and defects in 5d $ \mathcal{N}=1 $
  theories},'' \href{http://dx.doi.org/10.1007/JHEP01(2017)019}{{\em JHEP}
  {\bfseries 01} (2017) 019}, \href{http://arxiv.org/abs/1506.03871}{{\ttfamily
  arXiv:1506.03871 [hep-th]}}.

\bibitem{Ganor:1996pc}
O.~J. Ganor, D.~R. Morrison, and N.~Seiberg, ``{Branes, Calabi-Yau spaces, and
  toroidal compactification of the N=1 six-dimensional E(8) theory},''
  \href{http://dx.doi.org/10.1016/S0550-3213(96)00690-6}{{\em Nucl. Phys. B}
  {\bfseries 487} (1997) 93--127},
  \href{http://arxiv.org/abs/hep-th/9610251}{{\ttfamily arXiv:hep-th/9610251}}.

\bibitem{Intriligator:1997pq}
K.~A. Intriligator, D.~R. Morrison, and N.~Seiberg, ``{Five-dimensional
  supersymmetric gauge theories and degenerations of Calabi-Yau spaces},''
  \href{http://dx.doi.org/10.1016/S0550-3213(97)00279-4}{{\em Nucl. Phys. B}
  {\bfseries 497} (1997) 56--100},
  \href{http://arxiv.org/abs/hep-th/9702198}{{\ttfamily arXiv:hep-th/9702198}}.

\bibitem{Hwang:2021xyw}
C.~Hwang, S.~S. Razamat, E.~Sabag, and M.~Sacchi, ``{Rank $Q$ E-string on
  spheres with flux},''
  \href{http://dx.doi.org/10.21468/SciPostPhys.11.2.044}{{\em SciPost Phys.}
  {\bfseries 11} no.~2, (2021) 044},
  \href{http://arxiv.org/abs/2103.09149}{{\ttfamily arXiv:2103.09149
  [hep-th]}}.

\bibitem{Seiberg:1994bz}
N.~Seiberg, ``{Exact results on the space of vacua of four-dimensional SUSY
  gauge theories},'' \href{http://dx.doi.org/10.1103/PhysRevD.49.6857}{{\em
  Phys. Rev. D} {\bfseries 49} (1994) 6857--6863},
  \href{http://arxiv.org/abs/hep-th/9402044}{{\ttfamily arXiv:hep-th/9402044}}.

\bibitem{Seiberg:1994pq}
N.~Seiberg, ``{Electric - magnetic duality in supersymmetric nonAbelian gauge
  theories},'' \href{http://dx.doi.org/10.1016/0550-3213(94)00023-8}{{\em Nucl.
  Phys. B} {\bfseries 435} (1995) 129--146},
  \href{http://arxiv.org/abs/hep-th/9411149}{{\ttfamily arXiv:hep-th/9411149}}.

\bibitem{Serre:1993}
J.-P. Serre, {\em A course in arithmetic}.
\newblock Springer, 1993.

\bibitem{oguiso_shioda_rat_ell}
K.~Oguiso and T.~Shioda, ``{The Mordell-Weil Lattice of a Rational Elliptic
  Surface},'' {\em Comment. Math. Univ. St. Pauli.} {\bfseries 40} (1991) 83.

\bibitem{schuett2010ellipticsurfaces}
M.~Schuett and T.~Shioda, ``Elliptic surfaces,'' 2010.
\newblock \url{https://arxiv.org/abs/0907.0298}.

\bibitem{Gaberdiel:1997ud}
M.~R. Gaberdiel and B.~Zwiebach, ``{Exceptional groups from open strings},''
  \href{http://dx.doi.org/10.1016/S0550-3213(97)00841-9}{{\em Nucl. Phys. B}
  {\bfseries 518} (1998) 151--172},
  \href{http://arxiv.org/abs/hep-th/9709013}{{\ttfamily arXiv:hep-th/9709013}}.

\bibitem{Hayashi:2014kca}
H.~Hayashi, C.~Lawrie, D.~R. Morrison, and S.~Schafer-Nameki, ``{Box Graphs and
  Singular Fibers},'' \href{http://dx.doi.org/10.1007/JHEP05(2014)048}{{\em
  JHEP} {\bfseries 05} (2014) 048},
  \href{http://arxiv.org/abs/1402.2653}{{\ttfamily arXiv:1402.2653 [hep-th]}}.

\bibitem{Hayashi:2013lra}
H.~Hayashi, C.~Lawrie, and S.~Schafer-Nameki, ``{Phases, Flops and F-theory:
  SU(5) Gauge Theories},''
  \href{http://dx.doi.org/10.1007/JHEP10(2013)046}{{\em JHEP} {\bfseries 10}
  (2013) 046}, \href{http://arxiv.org/abs/1304.1678}{{\ttfamily arXiv:1304.1678
  [hep-th]}}.

\bibitem{Braun:2014kla}
A.~P. Braun and S.~Schafer-Nameki, ``{Box Graphs and Resolutions I},''
  \href{http://dx.doi.org/10.1016/j.nuclphysb.2016.02.002}{{\em Nucl. Phys. B}
  {\bfseries 905} (2016) 447--479},
  \href{http://arxiv.org/abs/1407.3520}{{\ttfamily arXiv:1407.3520 [hep-th]}}.

\bibitem{Braun:2015hkv}
A.~P. Braun and S.~Schafer-Nameki, ``{Box Graphs and Resolutions II: From
  Coulomb Phases to Fiber Faces},''
  \href{http://dx.doi.org/10.1016/j.nuclphysb.2016.02.001}{{\em Nucl. Phys. B}
  {\bfseries 905} (2016) 480--530},
  \href{http://arxiv.org/abs/1511.01801}{{\ttfamily arXiv:1511.01801
  [hep-th]}}.

\bibitem{Witten:1982fp}
E.~Witten, ``{An SU(2) Anomaly},''
  \href{http://dx.doi.org/10.1016/0370-2693(82)90728-6}{{\em Phys. Lett. B}
  {\bfseries 117} (1982) 324--328}.

\end{thebibliography}\endgroup

\end{document}